\documentclass[12pt]{amsart}
\usepackage{graphicx}
\usepackage{graphicx}
\newtheorem{thm}{Theorem}[section]

\newtheorem{lemma}[thm]{Lemma}

\newtheorem{remark}[thm]{Remark}
\newtheorem{example}[thm]{Example}

\newtheorem{conjecture}[thm]{Conjecture}
\usepackage{amsmath}
\usepackage{float}
\usepackage{amsxtra}
\usepackage{mathtools}
\usepackage{amscd}
\usepackage{amsthm}
\usepackage{stix}
\usepackage{amsfonts}
\usepackage{setspace}
\usepackage{amssymb}
\usepackage{eucal}
\usepackage{adjustbox}
\usepackage{subcaption}
\usepackage{tikzit}
\usepackage{blindtext}
\usepackage{hyperref}
\usepackage[toc,page]{appendix}
\usepackage[font={small}]{caption}

\tikzstyle{white dot}=[fill=white, draw=black, shape=circle, scale=0.5]
\tikzstyle{black dot}=[fill=black, draw=black, shape=circle, scale=0.5]
\tikzstyle{red dot center}=[fill=red, draw=red, shape=circle, scale=0.1]
\tikzstyle{small black dot}=[fill=black, draw=black, shape=circle, scale=0.1]
\tikzstyle{big red dot}=[fill=red, draw=black, shape=circle, scale=.8]

\tikzstyle{white filled}=[-, fill=white, thick]
\tikzstyle{gray filled}=[-, fill={rgb,255: red,128; green,128; blue,128}, thick]
\tikzstyle{red path}=[-, fill={rgb,255: red,128; green,128; blue,128}, draw=red, line width=3.8pt]
\tikzstyle{dashed edges}=[-, dashed]
\tikzstyle{thick line}=[-, thick, fill={rgb,255: red,128; green,128; blue,128}]
\tikzstyle{highlight}=[-, ultra thick, preaction={draw,yellow,-,
double=yellow,
double distance=\pgflinewidth,
}]
\tikzstyle{highlight blue}=[-, ultra thick, preaction={draw,blue,-,
double=blue,
double distance=0.5\pgflinewidth,}]
\tikzstyle{thick arrow}=[->, ultra thick]
\tikzstyle{direct arrow}=[->]
\tikzstyle{double arrow}=[<->, thick]
\tikzstyle{Blue edges}=[-, draw=blue, ultra thick, fill={rgb,255: red,128; green,128; blue,128}]
\tikzstyle{thick Blue}=[-, draw=blue, line width=2mm]

\newcommand{\bmb}{\left( \begin{array}{rr}}
\newcommand{\enm}{\end{array}\right)}
\newcommand{\cD}{\mathcal D}

\newcommand{\ti}{{\tilde i}}
\newcommand{\tj}{{\tilde j}}
\newcommand{\tk}{{\tilde k}}
\newcommand{\tl}{{\tilde l}}
\newcommand{\tm}{{\tilde m}}
\newcommand{\tn}{{\tilde n}}

\newcommand{\Z}{{\mathbb Z}}

\newcommand{\bk}{{\mathbf k}}

\newcommand{\al}{{\alpha}}

\newcommand{\ds}{\displaystyle}

\numberwithin{equation}{section}

\begin{document}

\title{Arctic curves of the $T$-system with Slanted Initial Data}
\author{Philippe Di Francesco$^*$}
\address{Department of Mathematics, University of Illinois, Urbana, IL 61821, U.S.A. 
%and \break
%Institut de physique th\'eorique, Universit\'e Paris Saclay, 
%CEA, CNRS, F-91191 Gif-sur-Yvette, FRANCE\hfill
\break  e-mail: philippe@illinois.edu\footnote{Corresponding Author}
}
\author{Hieu Trung Vu}
\address{Department of Mathematics, University of Illinois, Urbana, IL 61821, U.S.A. 
\break  e-mail: hvu@illinois.edu}
\begin{abstract}
We study the $T$-system of type $A_\infty$, also known as the octahedron recurrence/equation, 
viewed as a $2+1$-dimensional discrete evolution equation. Generalizing earlier work on arctic curves for
the Aztec Diamond obtained from solutions of the octahedron recurrence with ``flat" initial data, 
we consider initial data along parallel ``slanted" planes perpendicular to an arbitrary admissible direction 
$(r,s,t)\in \Z_+^3$. The corresponding solutions of the $T$-system are interpreted as partition functions of 
dimer models on some suitable ``pinecone" graphs introduced by Bousquet-Melou, Propp, and West in 2009. 
The $T$-system formulation and 
some exact solutions in uniform or periodic cases allow us to explore the thermodynamic limit of the 
corresponding dimer models and to derive exact arctic curves separating the various phases of the system.
This direct approach bypasses the standard general theory of dimers using the Kasteleyn matrix approach
and uses instead the theory of Analytic Combinatorics in Several Variables, by focusing on a linear 
system obeyed by the dimer density generating function.
% \vskip .2in
% \noindent{${}^*$ Corresponding Author.}
\end{abstract}

\maketitle
\date{\today}
\tableofcontents

\section{Introduction}

	The $T$-system, also known as the octahedron recurrence, is a system of non-linear equations describing the time evolution of a quantity $T_{i,j,k}$ indexed by the $\Z^3$ lattice, where $(i,j)$ are thought of as discrete space coordinates, and $k$ a discrete time. The $T$-system originated in the context of integrable quantum spin chains, as a functional relation between transfer matrices \cite{KUNIBA_1994}, \cite{KUNIBA1994_2}. The $T$-system was more recently reinterpreted in the framework of cluster algebras \cite{clusDFK} as a particular set of mutations in an infinite rank system. As a consequence, solutions display the Laurent phenomenon: the solution can be expressed in terms of any admissible initial data as a Laurent polynomial with non-negative integer coefficients. 
%	
%	The history of the subject dates all the way back to C.N. Yang \cite{Yang}, originates in Yang-Baxter quantum integrable systems. After that, there was a series of research in the area of mathematical physics considering families of solvable vertex models by Kulish, Kirillov, Reshetikhin, Klümper and Pearce \cite{Kulish1981}, \cite{Pearce}, \cite{Kirillov1987}, \cite{Kirillov1987_2}. The term $T$-system was named by Kuniba, Nakanishi and Suzuki \cite{KUNIBA_1994}, \cite{KUNIBA1994_2}, \cite{Kuniba_2011} in the context of quantum spin chain in 1994. It was shown to be a particular set of mutations within some infinite rank cluster algebra \cite{DF2}, and as such to enjoy the Laurent property.
This system displays rich combinatorial properties, depending on the choice of initial data/boundary conditions, such as discrete integrability \cite{DFKposit}, and periodicity properties \cite{DFPKedem,MR2931902,MR2649278,MR3029995,MR3029994}. In particular, the $T$-system with periodic boundary conditions is related to the pentagram map, an integrable dynamical system on polygons of the projective plane \cite{Kedem_2015}, and its higher generalizations \cite{GSTV}. 

We also consider the combinatorics of dimer configurations, i.e. perfect matchings of (planar) graphs. The \emph{perfect matching} of a graph $G$ is a subgraph of $G$ such that every vertex belongs to exactly one edge. In the case of the graph on the $\Z^2$ lattice, perfect matchings can be visualized dually as domino tilings, i.e. tilings by means of $2 \times 1$ and $1 \times 2$,rectangles. A method for counting the number of domino tilings of a finite domain of $\Z^2$  was devised  independently by Kasteleyn \cite{Kasteleyn} and by Fisher and Temperley \cite{Temperley1961DimerPI}. 

For suitable domain/graph shapes, dimer models display the so-called arctic phenomenon: when the domain/graph is scaled to a very large size, in typical configurations there is a sharp separation between ``frozen" regions of the domain with regular lattice-like dimer configurations and
``liquid" regions where the dimers are disordered, eventually converging to an ``arctic curve". 
The simplest instance is the arctic circle theorem for the uniform domino tiling of large Aztec diamonds \cite{ProppShor,CEP}. 
A general theory of arctic curves in dimer problems was developed by Kenyon, Okounkov and Sheffield, building on Kasteleyn's solution, and establishes a connection to solutions of the complex B\"urgers equation \cite{OK1,OK2}.

%
%	\begin{thm}[Kasteleyn]
%		The number of domino tilings of a square lattice is $\sqrt{|det K|}$ where $K$ is the weighted adjacency matrix of the graph $G$, with horizontal edges weighted 1 and vertical edges weighted $i$ (the complex number)
%	\end{thm}
%
%The theory developed by Kasteleyn initiated a long string of research, culminating in Ref.~ \cite{CEP,EKLP,ProppWestB-M} who considered the domino tiling of Aztec diamonds and other families of planar domains. 
% In 2003, a combinatorial method to solve this class of dimer problems was given by Kuo \cite{Kuo1} and later generalized by Ciucu \cite{Ciucu}. 

	The $T$-system solutions with suitable initial conditions can be interpreted in terms of various combinatorial objects such as tessellations of the triangular lattice and families of non-intersecting lattice paths. 
Significant progress was made by Speyer \cite{Speyer}, who worked out the general solution in terms of a weighted dimer model on a suitable graph (see also \cite{DiFrancesco2}).  
In addition to computing exact dimer partition functions, this interpretation of $T$-system solutions provides a tool  to investigate asymptotic properties of the corresponding dimer models. This was first applied to the domino tilings of the Aztec diamond for various types of (periodic) weights \cite{DiFrancesco1}, by considering the solutions of the $T$-system with ``flat initial data" assignments providing a weighting of the dimer model. This work uses the recent progress in the area of Analytic Combinatorics in Several Variables (ACSV) which provides analytic tools to study the asymptotic enumeration of combinatorial objects with rational multivariate generating functions \cite{pemantle_wilson_2013,PW2005,BP1,PW2004}. 
%	These series of paper considers the problem of finding the asymptotic of a multivariate sequences, specifically counting or enumerating some quantities like our $T$-system solutions. 
%In \cite{DiFrancesco1} the theory of ACSV was successfully applied to the determination of the artic curve phenomenon for the dimer partition functions corresponding to $T$-system solutions with uniform and periodic ``flat" initial data. 
Indeed, the crucial
ingredient in \cite{DiFrancesco1} is the fact that the average local dimer density $\rho_{i,j,k}$ at point $(i,j,k)$, which vanishes in crystalline phases and is non-trivial in liquid phases, has an explicit 
rational generating function in 3 variables, the denominator of which governs the behavior of $\rho_{i,j,k}$ when
$i,j,k\to\infty$ with $(i/k,j/k)\to (u,v)$ finite, eventually yielding via ACSV the arctic curve for the rescaled 
model in the $(u,v)$ plane.  Similar and further results were also found by the more traditional Kasteleyn method \cite{MR4634340,MR4610279,MR3479561}. However, the $T$-system approach advocated in this paper is in a
sense much more direct, as it simply relies on exactly solving the $T$-system for given initial data that determine the geometry and weights of the corresponding dimer problem. Moreover it provides exact explicit formulas for dimer density generating functions, instrumental in the study of the arctic phenomenon. A detailed comparison between the results of this paper and the Kasteleyn method will be done elsewhere \cite{workinprogress}.
 
The aim of the present paper is to study solutions of  the $T$-system with different initial data, giving rise to different dimer models, that also display an artic phenomenon, which we investigate by use of  ACSV. Our initial data are along collections of ($2t$) parallel planes perpendicular to a fixed direction $(r,s,t)\in \Z^3$ in the $(i,j,k)$ space-time. The corresponding solutions of the $T$-system are 
 interpreted as the partition functions of weighted dimer configurations of so-called pinecones \cite{ProppWestB-M}, certain families of bipartite planar graphs with square and hexagonal inner faces only. 
We find new solutions of the $T$-system corresponding to different uniform dimer weights along each initial data slanted plane, for which an arctic phenomenon occurs. We show this by 
computing  explicit rational generating functions for the corresponding dimer density $\rho_{i,j,k}$ at point $(i,j,k)$. 
As before, the singularities of the latter determine the arctic curves for the corresponding dimer models. We then explore non-uniform but periodic initial data along slanted planes, and in the exactly solvable cases we obtain higher order linear systems for the local density, leading to more involved arctic curves in the same spirit as \cite{DiFrancesco1}. 

Finally, we show that a given $T$-system solution for a given $(r,s,t)-$slanted initial data also provides
insights on dimer models arising from in any other $(\tilde r,\tilde s,\tilde t)-$slanted initial data given by the values taken by the previous solution along the corresponding new set of parallel planes. By construction, the new dimer model also displays an arctic phenomenon with its own arctic curve, which we view as a holographic image of the former.

% We again apply the theory of ACSV, study of the denominator of the density generating function of the $T$-system allows to explore the singularity structure of the dimer models in the thermodynamic limit of large size, and to confirm their phase structure, displaying frozen, disordered and facet-like phases separated by generalized arctic curves. 

The paper is organized as follows.

In Section 2, we recall known facts on $T$-system solutions and their interpretation in terms of dimer models. We define the $(r,s,t)-$slanted initial data and show their relation to dimer models on the pinecone graphs of \cite{ProppWestB-M}.
Section 3 is devoted to the case of uniform but specific initial values along each initial data plane, and proceeds as follows: we first present the exact solution of the $T$-system, then derive the local dimer density, which we finally analyze via ACSV to get the arctic curve.
We then follow the same procedure for non-uniform but 2x2-periodic initial data within each plane in  Sections 4 and 5.
Section 4 is devoted to the exact solution and its periodicity properties. Section 5 deals with the local dimer density, and the computation of the associated arctic curves. In particular, like in \cite{DiFrancesco1}, a new ``facet" dimer phase emerges as a consequence of the staggering of initial data. 
In Section 6, we describe the holographic principle, which allows to ``view" any $(r,s,t)-$slanted solution from a different
$(\tilde r,\tilde s,\tilde t)$ point of view. 
Section 7 is devoted to a discussion of the detailed structure of the facet phase, a 3D formulation of the holographic principle,  
and a few concluding remarks. 

Some cumbersome explicit expressions for systems and arctic curves of this paper are available online \cite{mathfiles}.

\noindent{\bf Acknowledgments.} We thank Greg Musiker for his valuable comments and feedback on the problem and David Speyer for enlightening discussions during the workshop ``Dimers: Combinatorics, Representation Theory and Physics", CUNY Graduate Center, N.-Y., Aug. 14-25, 2023. 
HTV is supported by the David G. Bourgin Mathematics Fellowship and the University of Illinois at Urbana-Champaign Campus Research Board. 
PDF acknowledges support from the Morris and Gertrude Fine endowment and the Simons Foundation travel grant MPS-TSM00002262, as well as the NSF-RTG grant DMS-1937241.

\section{$T$-system and dimers}

\subsection{General setting and slanted plane initial data}\label{general setting}

The $T$-system or octahedron relation is the following recursion relation for variables $T_{i,j,k}>0$,
$i,j,k\in \Z$
\begin{equation}\label{Tsys}
T_{i,j,k+1}\, T_{i,j,k-1}=T_{i+1,j,k}\, T_{i-1,j,k}+T_{i,j+1,k}\, T_{i,j-1,k} .\end{equation}
It may be interpreted as a discrete time $k$ evolution for the variable $T$, expressing its value at the time $k+1$ vertex of an octahedron 
in terms of the values at the 4 vertices at time $k$ and at a single vertex at time $k-1$. It is also interpreted as a particular 
mutation in an infinite rank cluster algebra. As the $T$-system clearly conserves the parity of $i+j+k$, we may restrict our study to solutions
subject to the additional condition $i+j+k=0$ mod 2. This condition will always be assumed implicitly unless otherwise specified.

The solution $T_{i,j,k}$ is unique once we fix admissible initial data along any given ``stepped surface" $\bk$ made of the vertices $(i_0,j_0,k_{i_0,j_0})$,
$i_0,j_0\in \Z$, where the height function $k_{i,j}:\Z^2\to\Z$ obeys $|k_{i+1,j}-k_{i,j}|=|k_{i,j+1}-k_{i,j}|=1$ for all $i,j\in \Z$. 
The initial data assignments read 
\begin{equation}\label{init}
T_{i_0,j_0,k_{i_0,j_0}}=t_{i_0,j_0}\, , \qquad (i_0,j_0\in \Z)
\end{equation}
for some fixed initial variables $t_{i_0,j_0}>0$, $i_0,j_0\in \Z$.

In this paper, we consider solutions of the $A_\infty$ $T$-system
subject to initial data along $(r,s,t)$-slanted parallel planes
$$(P_m)=\{(i,j,k)\, \vert \,  r i+s j+t k=m\} $$
for some fixed integers $r,s,t\geq 0$ such that $t>\max(r,s)$ and $\gcd(r,s,t)=1$. Throughout the paper, without loss of generality we shall also assume $r\leq s$, as the converse is easily reached upon interchanging $i\leftrightarrow j$. Note that when $r,s,t$ are odd, only even values of $m$ occur, as $i+j+k$ is even.

It is easy to see that
an admissible initial data set for the $T$-system consists of specifying the values of $T_{i,j,k}$
along $2t$ consecutive parallel planes $(P_m)$, $m=0,1,2,...,2t-1$. These form a particular stepped surface, by noting that 
neighboring points $(i,j,k)$ and $(i\mp 1,j,k\pm 1)$ belong respectively to planes $(P_{m})$, $m=ri+s j+t k$ and $(P_{m\pm(t-r)})$ 
while $(i,j\mp 1,k\pm 1)$ belong to $(P_{m\pm (t-s)})$.
Moreover, using the $T$-system as a 
recursion relation in the discrete variable $k$, we may write
$$T_{i,j,k+1}=\frac{T_{i+1,j,k}\, T_{i-1,j,k}+T_{i,j+1,k}\, T_{i,j-1,k}}{T_{i,j,k-1}} .$$
The point $(i,j,k+1)$ belongs to the plane $(P_{M})$ for $M=ri+sj+t(k+1)$. The above relation shows that $T_{i,j,k+1}$ is determined by values of $T$ on the 5 other planes: $(P_{M+r-t}),(P_{M-r-t}),(P_{M+s-t})$, $(P_{M-s-t}),(P_{M-2t})$.
We may therefore use the relation recursively to obtain all values of $T$ in $(P_M)$ from the data on $(P_{M-1}),(P_{M-2}),\ldots,(P_{M-2t})$.

The stepped surface corresponding to $(r,s,t)$-slanted planes initial conditions reads
\begin{equation}\label{kdef} k_{i,j}=\frac{1}{t}\times \left\{\begin{matrix}{\rm Mod}(ri+sj,2t)-ri-sj  & {\rm if}\, i+j=0\, [2]\\
{\rm Mod}(ri+sj+t,2t)-ri-sj  & {\rm if}\, i+j=1\, [2] \end{matrix} \right. 
\end{equation}
where we have identified the index $m={\rm Mod}(ri+sj+t\,{\rm Mod}(i+j,2),2t)$ of the plane $P_m$ containing the point $(i,j,k_{i,j})$.
Equivalently, we have
\begin{equation}\label{otherk}
k_{i,j}=\left\{ \begin{matrix} 
-2\left\lfloor \frac{ri+sj}{2t} \right\rfloor & {\rm if}\ \ i+j \ {\rm even}\\
1-2\left\lfloor \frac{ri+sj+t}{2t} \right\rfloor & {\rm otherwise}
\end{matrix}\right.
\end{equation}
\begin{remark}\label{minimality}
The latter expression allows to characterize alternatively the $(r,s,t)$-slanted stepped surface as the {\it lowest} stepped surface lying above the plane $(P_0):$ $ri+s i+tk=0$, $(i,j,k)\in \Z^3$. By lowest we mean that any ``down" mutation sending some point  $(i,j,k)\to (i,j,k-2)$ 
will end strictly below the $(P_0)$ plane.
\end{remark}

We may finally also write a single parity-independent formula for $k_{i,j}$ in the form
\begin{equation}
\label{goodk}
k_{i,j}=i+j-2 \left\lfloor \frac{ri+sj+t(i+j)}{2t} \right\rfloor .
\end{equation}

\subsection{Solution as a dimer model partition function}\label{solution as dimer}\label{secsol}

In \cite{DiFrancesco1} and \cite{DiFrancesco2}, it was shown that the solution $T_{i,j,k}$ to the $T$-system subject to initial conditions of the form \eqref{init} on some stepped surface $\bk$ is the partition function of a dimer model on a bipartite graph obtained as follows. 

First, from the octahedral mutation interpretation, we note that the solution $T_{i,j,k}$ may be expressed in terms
of a {\it finite subset} of the initial data \eqref{init}, namely that lying in the cone
$|x-i|+|y-j|\leq |z-k|$, $(x,y,z)\in \Z^3$ with apex $(i,j,k)$. Let $\mathcal D={\mathcal D}_{i,j,k}^{r,s,t}$ denote the intersection of the initial data stepped surface with this cone.
\begin{figure}[H]
\centering
		\begin{subfigure}{\linewidth}
		%\ctikzfig{tikz/tessellation}
		\includegraphics[width=\linewidth]{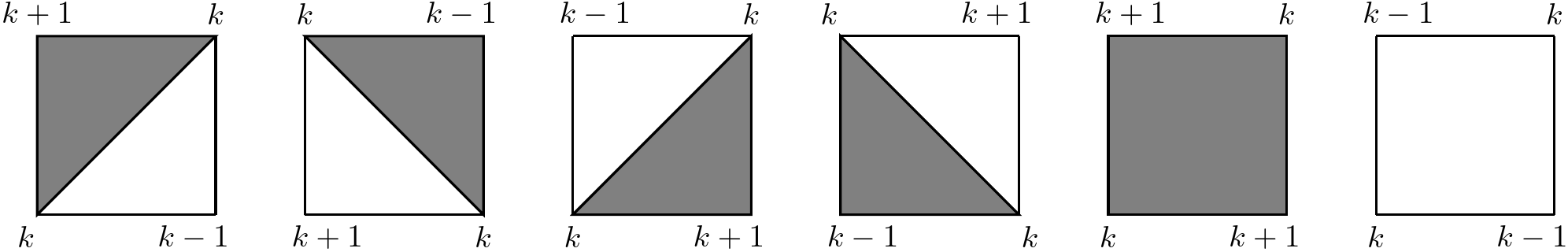}
		\end{subfigure}
		\caption{The six possible face configurations for the $k$ coordinate of any stepped surface $\bk$ and the associated black/white square/triangle  tessellation.}
		\label{fig:rule}
\end{figure}
Further recording the $k$-coordinates of the points in $\mathcal D$, and applying the dictionary of Fig.~\ref{fig:rule}, allows to construct a tessellation with black/white triangles and squares of the projection of $\mathcal D$ onto the $(x,y)$ plane. It turns out that the $(r,s,t)$-slanted tessellated stepped surfaces are very special:

\begin{thm}\label{restricthm}
Each vertex of an $(r,s,t)$-slanted tessellated stepped surface may only have one of the five possible environments depicted below:
\begin{figure}[H]
\centering
\includegraphics[width=12.cm]{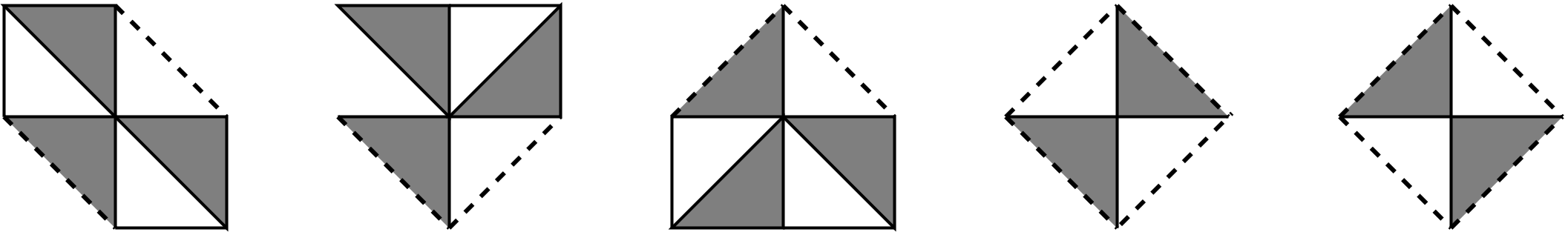}
\end{figure}
\end{thm}
To prove the Theorem, we note that the general stepped surface conditions 
$|k_{i+1,j}-k_{i,j}|=|k_{i,j+1}-k_{i,j}|=1$ ($i,j\in \Z$) would give rise to $2^4=16$ possible environments 
for the vertex $(i,j,k_{i,j})$. However, these rules 
are further restricted by eq.
\eqref{otherk} as follows.

\begin{lemma}\label{restrictlem}
The $(r,s,t)$-slanted stepped surfaces obey the further conditions: 
\begin{equation}\label{restrictone}
k_{i-1,j}-k_{i+1,j}\in \{0,2\} , \qquad k_{i,j-1}-k_{i,j+1}\in \{0,2\} , \qquad k_{i+1,j}-k_{i,j+1}\in \{0,2\}
\end{equation}
\end{lemma}
\begin{proof}
We show that the value $-2$, allowed by the general stepped surface condition, is ruled out here.
Using \eqref{goodk}, we have first:
$$k_{i-1,j}-k_{i+1,j}=-2+2\left( \left\lfloor \frac{ri+sj+t(i+j)+r+t}{2t} \right\rfloor-\left\lfloor \frac{ri+sj+t(i+j)-r-t}{2t} \right\rfloor\right)  .
$$
The arguments of the integer parts differ by $2\frac{r+t}{2t}>1$, hence the difference cannot be $0$, which implies 
that $k_{i-1,j}-k_{i+1,j}\neq -2$. The same reasoning leads to $k_{i,j-1}-k_{i,j+1}\neq -2$. Finally, using again \eqref{goodk}:
$$k_{i+1,j}-k_{i,j+1}=2\left(\left\lfloor \frac{ri+sj+t(i+j)+s+t}{2t} \right\rfloor-\left\lfloor \frac{ri+sj+t(i+j)+r+t}{2t} \right\rfloor \right) $$
The arguments of the integer parts differ by $\frac{s-r}{2t}>0$, which implies $k_{i+1,j}-k_{i,j+1}\geq 0$ and rules out the value $-2$.
\end{proof}

We are now ready to prove Theorem \ref{restricthm}. The first two conditions of \eqref{restrictone} rule out the following nine possible vertex environments:
\begin{figure}[H]
\centering
\includegraphics[width=12.cm]{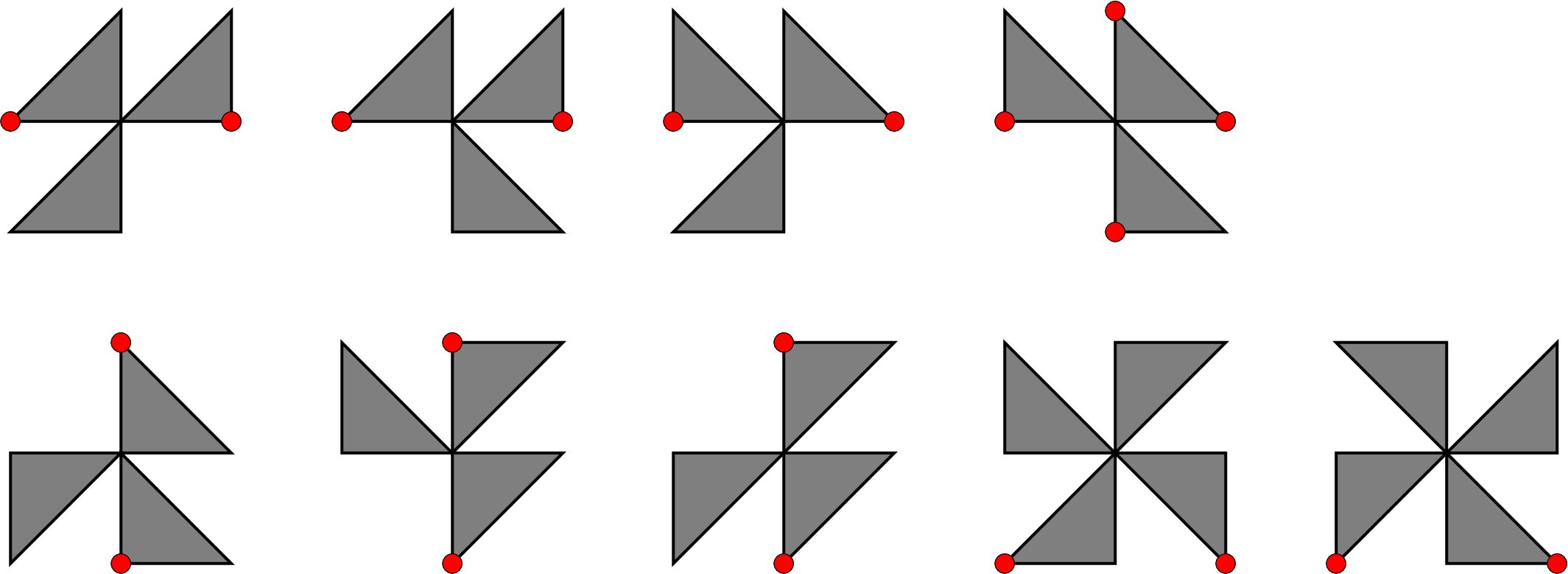}
\end{figure}
\noindent where the horizontal (resp. vertical) pairs of dots indicate the violation of the first (resp. second) restrictions of \eqref{restrictone}. Finally the third restriction in \eqref{restrictone} rules out the first face configuration of Fig. \ref{fig:rule},
and therefore the following two vertex environments:
\begin{figure}[H]
\centering
\includegraphics[width=5.5cm]{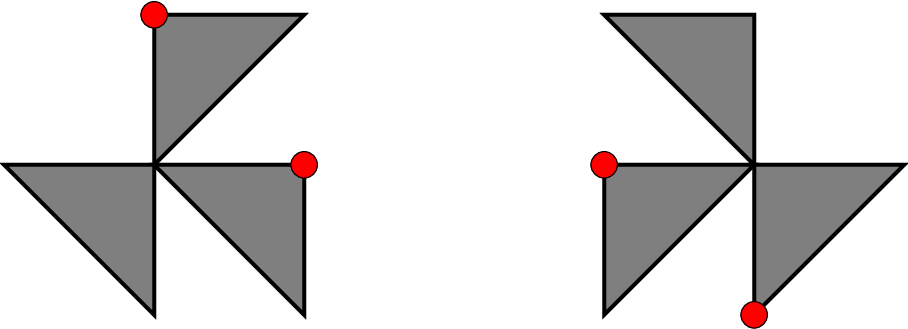}
\end{figure}
\noindent where we also indicated by a pair of dots the heights violating the third restriction of \eqref{restrictone}. Having ruled out 11 possible environments, we are left with the $16-11=5$ stated in the Theorem. 

\begin{figure}
		
		\begin{subfigure}{.3 \linewidth}
		\centering
			\includegraphics[scale=.1]{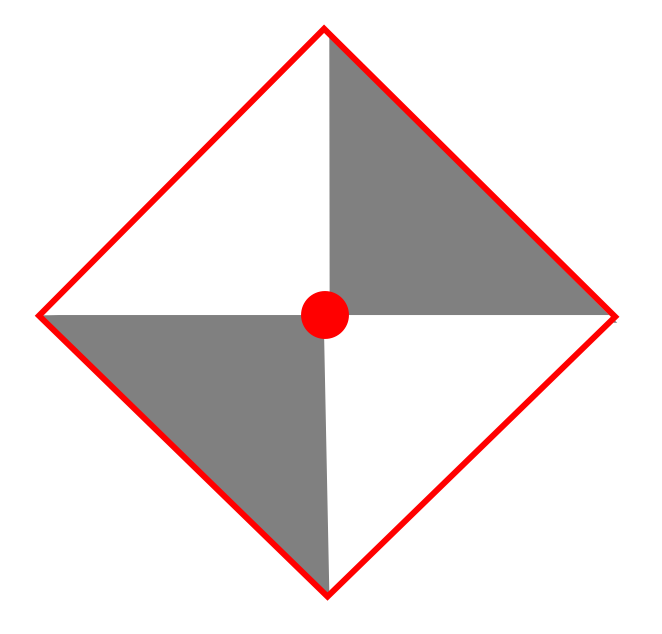}
			\subcaption[t]{$(r,s,t)=(1,1,3)$, $k=2$}
		\end{subfigure}
	\begin{subfigure}{.3 \linewidth}
	\centering 
		\includegraphics[scale=.15]{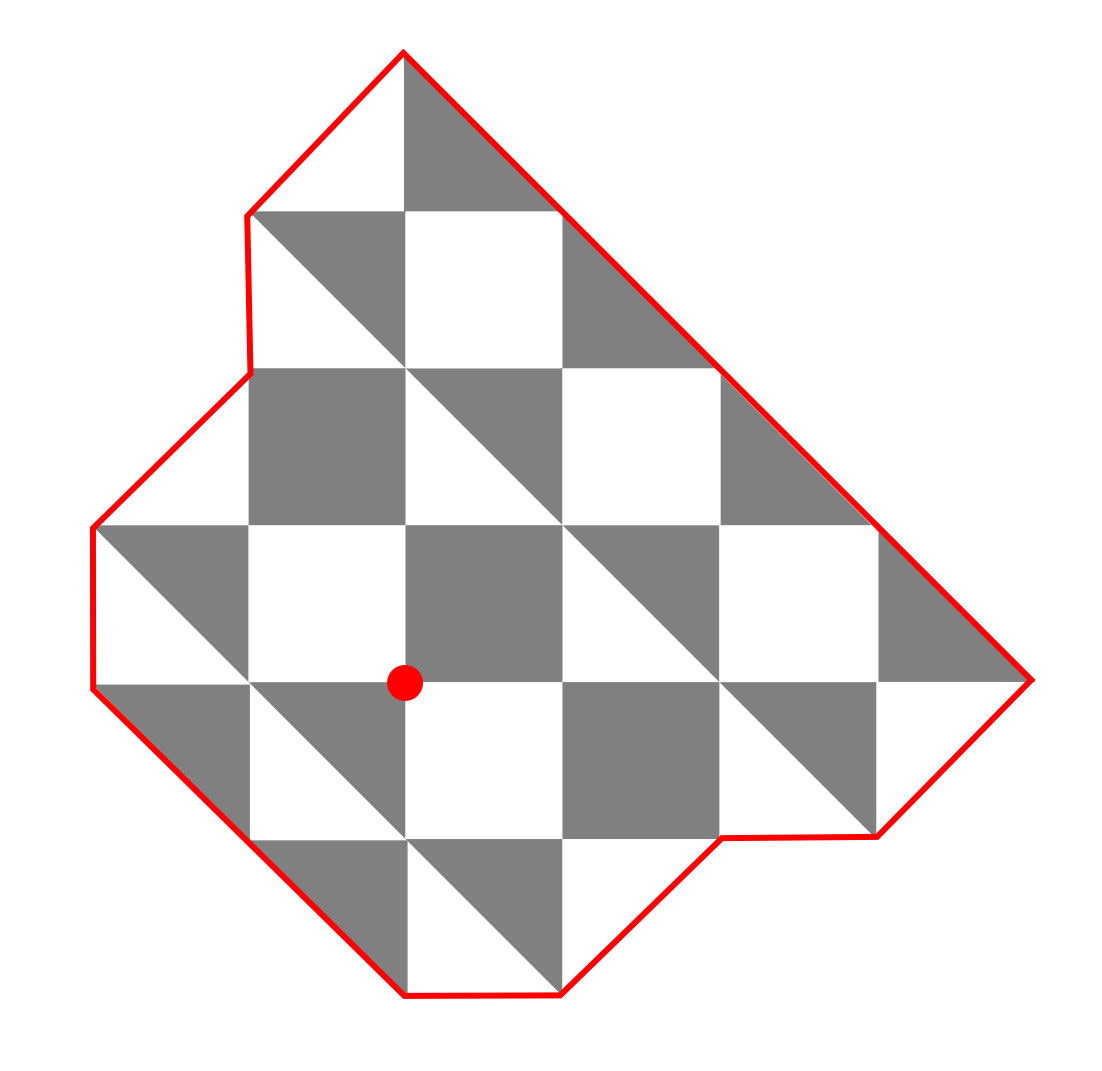}
		\subcaption[t]{ $(r,s,t)=(1,1,3)$, $k=4$}
	\end{subfigure}
	\begin{subfigure}{.3 \linewidth}
	\centering 
	\includegraphics[scale=.2]{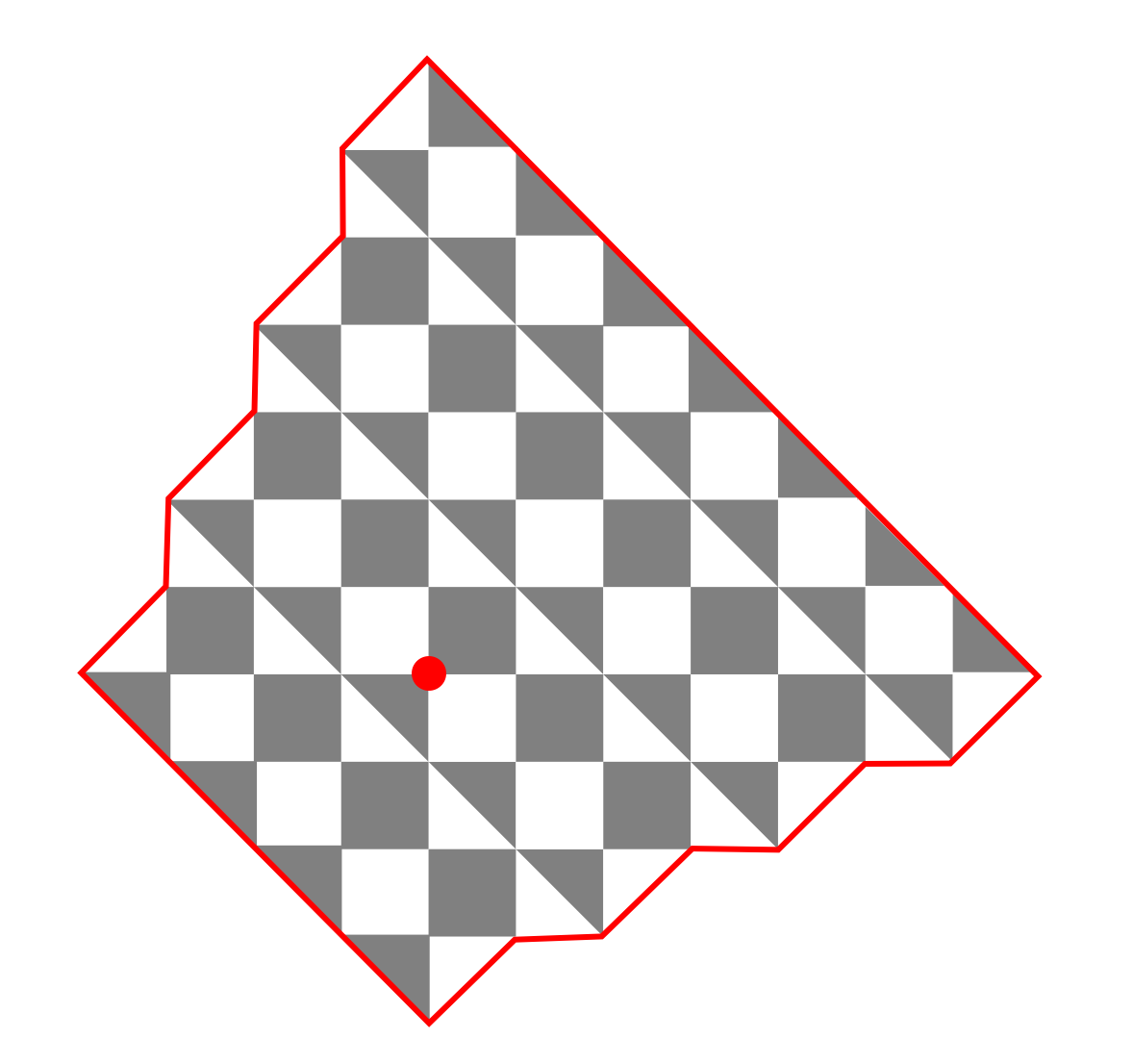}
	\subcaption[t]{$(r,s,t)=(1,1,3)$, $k=6$}
	\end{subfigure}\\
	\begin{subfigure}{.3 \linewidth}
	\centering
	\includegraphics[scale=.1]{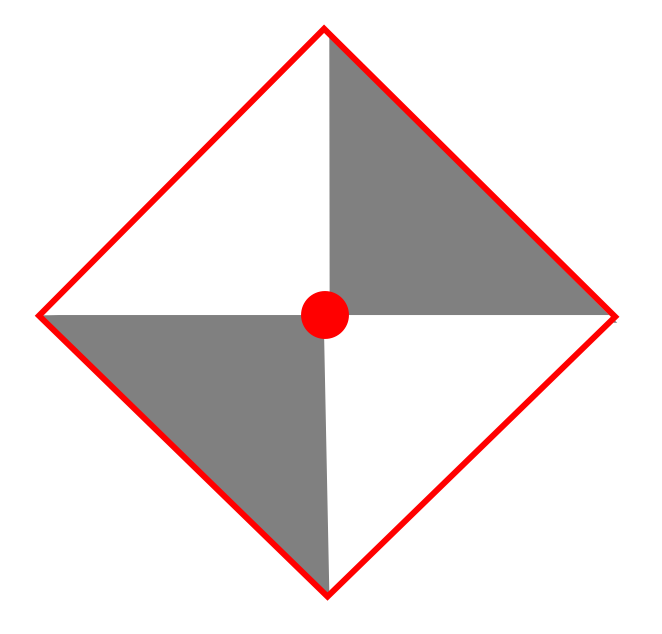}
	\subcaption[t]{$(r,s,t)=(1,2,3)$, $k=2$}
	\end{subfigure}
	\begin{subfigure}{.3 \linewidth}
	\centering
	\includegraphics[scale=.15]{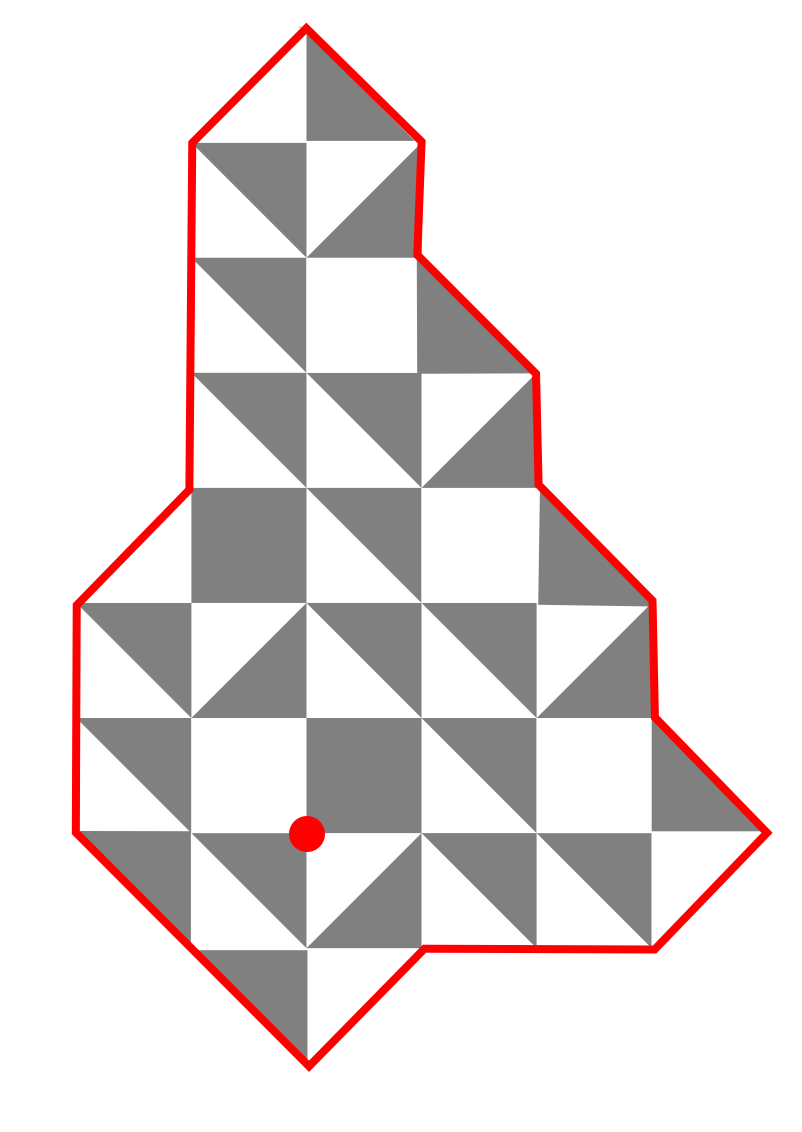}
	\subcaption[t]{$(r,s,t)=(1,2,3)$, $k=4$}
	\end{subfigure}
	\begin{subfigure}{.3 \linewidth}
	\centering
	\includegraphics[scale=.2]{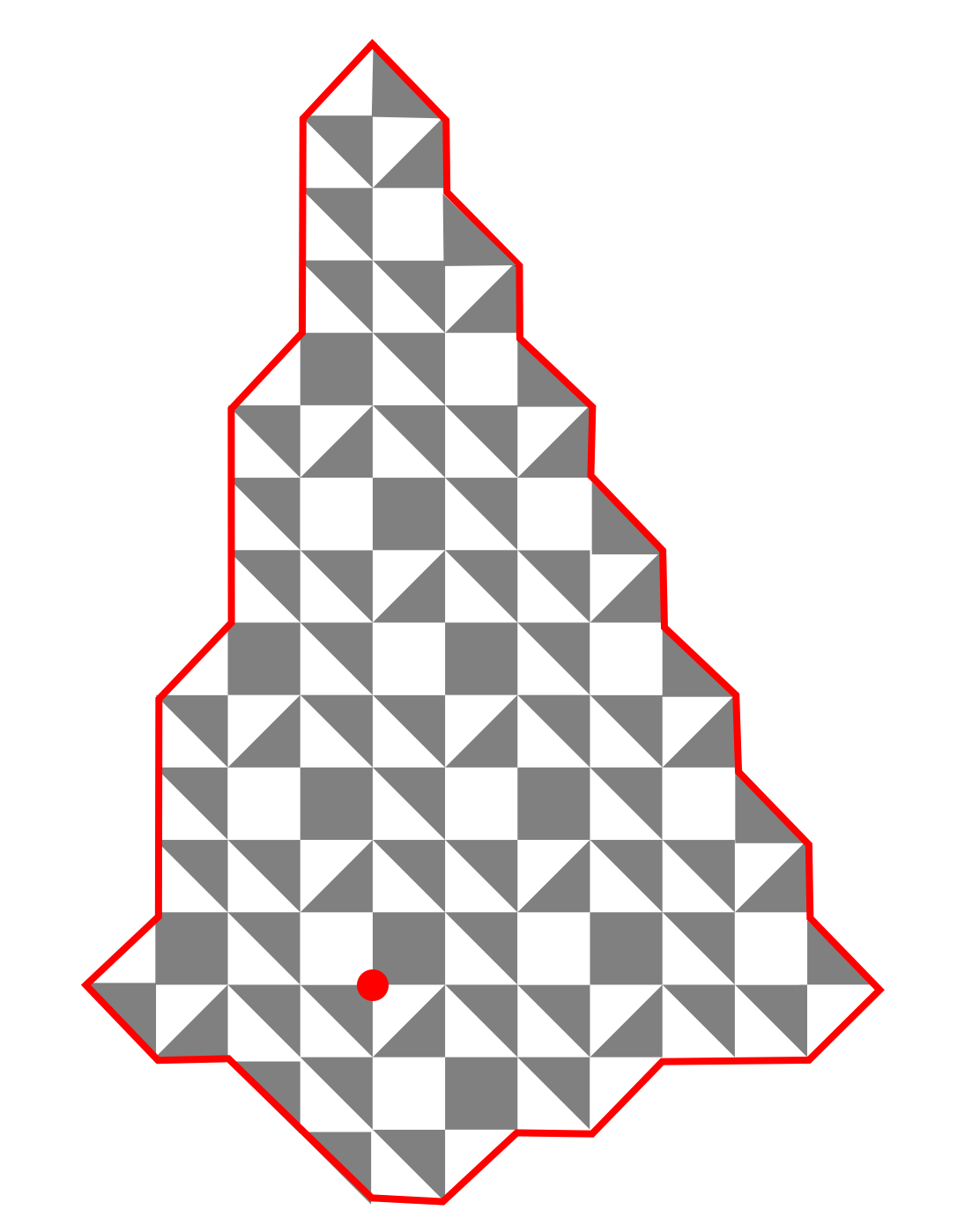}
	\subcaption[t]{$(r,s,t)=(1,2,3)$, $k=6$}
	\end{subfigure}
	\caption{The tessellation domain ${\mathcal D}_{i,j,k}^{r,s,t}$ for $(r,s,t)=(1,1,3)$ and $(r,s,t)=(1,2,3)$, centered at the point $(i,j,k_{i,j})$ (marked by a {$\textcolor{red}{\bullet}$}), equal to the projection of the point $(i,j,k)$ onto the initial data stepped surface $\bk$.}	
	\label{tessellation progress 246} 
	\end{figure}

We represent a few sample tessellated domains $\mathcal D$ in Fig. \ref{tessellation progress 246}, in the cases $(r,s,t)=(1,1,3)$ and $(r,s,t)=(1,2,3)$, for solutions $T_{0,0,k}$,
$k=2,4,6$.

From the tessellated domain $\mathcal D$, one can construct the bipartite dual graph $\mathcal G={\mathcal G}_{i,j,k}^{r,s,t}$, by assigning black/white bicolored vertices ${\rm \raisebox{-.02in}{ \scalebox{1.4}{$\bullet$}}}/ \circ$ corresponding to the color of the faces in the original triangulation. Faces of $\mathcal G$ are labeled by coordinates $(x,y)$
of the dual vertex $(x,y,k_{x,y})\in \mathcal D$. 
A {\it dimer configuration} on $\mathcal G$ is an independent set of edges of $G$ such that every vertex of $\mathcal G$ belongs to exactly one edge. The edges in this set can be thought as occupied by dimers, usually represented as thickened edges of $\mathcal G$.

\begin{thm}{\cite{DiFrancesco2}}\label{mainthm}
		The solution of the $T$-system with slanted initial data is expressed as:
			$$T_{i,j,k}=\sum_{{\rm dimer}\, {\rm configs.}\, D\atop
			{\rm on}\, \mathcal G} \prod_{{\rm faces}\, (x,y)\atop {\rm of}\, G}\begin{cases}
				(t_{x,y})^{v_{x,y}/2-1-N_{x,y}(D)} & (x,y) \  \text{interior faces} \\
				(t_{x,y})^{1-N_{x,y}(D)} & (x,y) \  \text{boundary faces} 
			\end{cases} $$
where the sum extends over all dimer configurations $D$ on the dual graph $\mathcal G$, while $v_{x,y}$ is the valency of the face $(x,y)$ and $N_{x,y}(D)\in \{0,1,...,v_{x,y}\}$ denotes the number of dimers occupying the edges at the boundary of the face $(x,y)$. The initial data $t_{x,y}$'s serve as local Boltzmann weights for the dimer model.	
\end{thm}
\begin{proof}
The proof of the theorem proceeds similarly to the case of $4-6-8$-graph in Theorem 3.10 of \cite{DiFrancesco2}, with the extra restrictions of Lemma \ref{restrictlem}. 
%The proof requires a bijection between dimers and non-intersecting lattice paths, using networks constructed via a matrix solution of the general $T$-system. 
%This will be the subject of another work, where we will discuss non-intersecting lattice paths in more detail.  
\end{proof}

In this paper, we apply Theorem \ref{mainthm} to the particular case of $(r,s,t)$ slanted initial data. The corresponding bipartite graphs $\mathcal G$ actually already appeared in the literature \cite{ProppWestB-M} under the name of ``pinecones". The precise connection is given in the next section.

\begin{example}\label{explaining core}
\begin{figure}
\begin{subfigure}{.2\linewidth}
\centering
\includegraphics[scale=.2]{figures/tess_113_4.png}
\subcaption[t]{}
\end{subfigure}
\hskip 2cm
\begin{subfigure}{.2\linewidth}
\centering
\includegraphics[scale=.17]{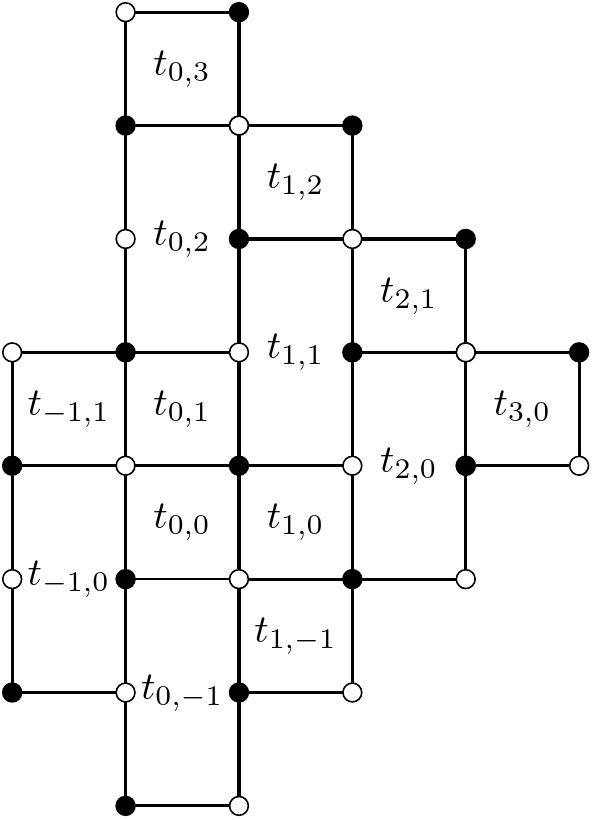}
\subcaption[t]{}
\end{subfigure}
\hskip 2cm
\begin{subfigure}{.2\linewidth}
\centering
\includegraphics[scale=.17]{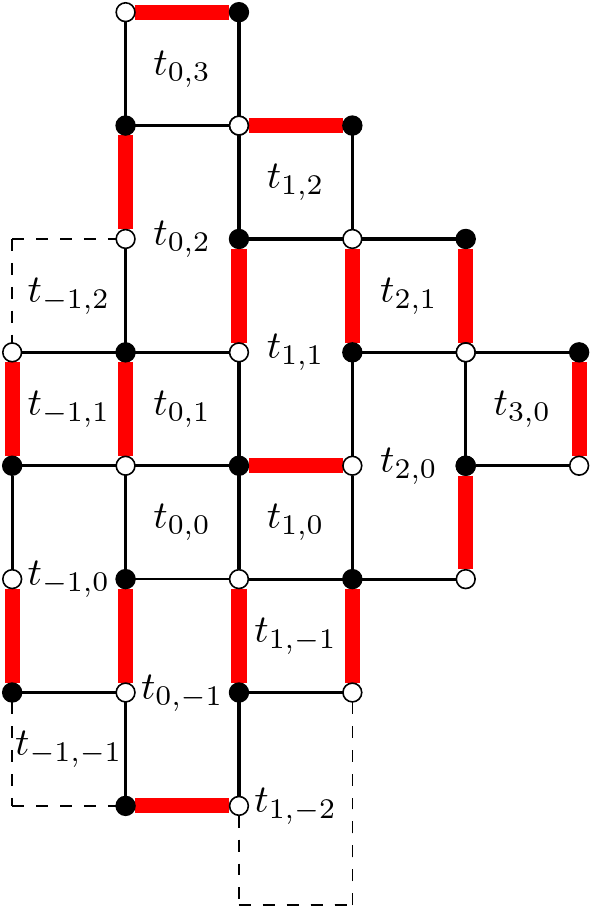}
\subcaption[t]{}
\end{subfigure}
\caption{The tessellated domain (A) for $T_{0,0,4}$ with $(1,1,3)$-slanted initial data, the corresponding dual 
graph (B) $\mathcal{G}_{1,1,3}^{0,0,4}$ and a sample dimer configuration (C).}
\label{113_dimer_basic}
\end{figure}
We consider the solution $T_{0,0,4}$ of the $(1,1,3)$-slanted T-system.
The tessellated domain ${\mathcal D}_{1,1,3}^{0,0,4}$ is represented in Fig. ~\ref{113_dimer_basic} (A).
The dual graph ${\mathcal G}_{1,1,3}^{0,0,4}$ together with its face weights $t_{a,b}$ is represented  in Fig.~ \ref{113_dimer_basic} (B). Finally, we show a sample dimer configuration in Fig. ~\ref{113_dimer_basic} (C),
corresponding to the contribution $\ds\frac{t_{-1,-1} t_{-1,2} t_{0,0} t_{1,-2} t_{2,0} t_{2,2}}{t_{-1,1} t_{0,-1} t_{1,-1}t_{1,1} t_{2,1}}$ to the partition function $T_{0,0,4}$, expressed as a Laurent polynomial of the initial data $t_{a,b}$.
\end{example}
%\begin{example}
	
%   \begin{figure}[H]
%   	\includegraphics[scale=.12]{figures/T004_example.png}
%   	\caption{A sample dimer configuration for $T_{0,0,4}$ in the case of $(1,1,3)$-slanted initial data. }
%   	\label{T004_example_one_term}
%   \end{figure}
%\end{example}

\subsection{Slanted plane initial data and dimers on pinecones}\label{pinecone section}
\begin{figure}\centering
	\includegraphics[width=14.cm]{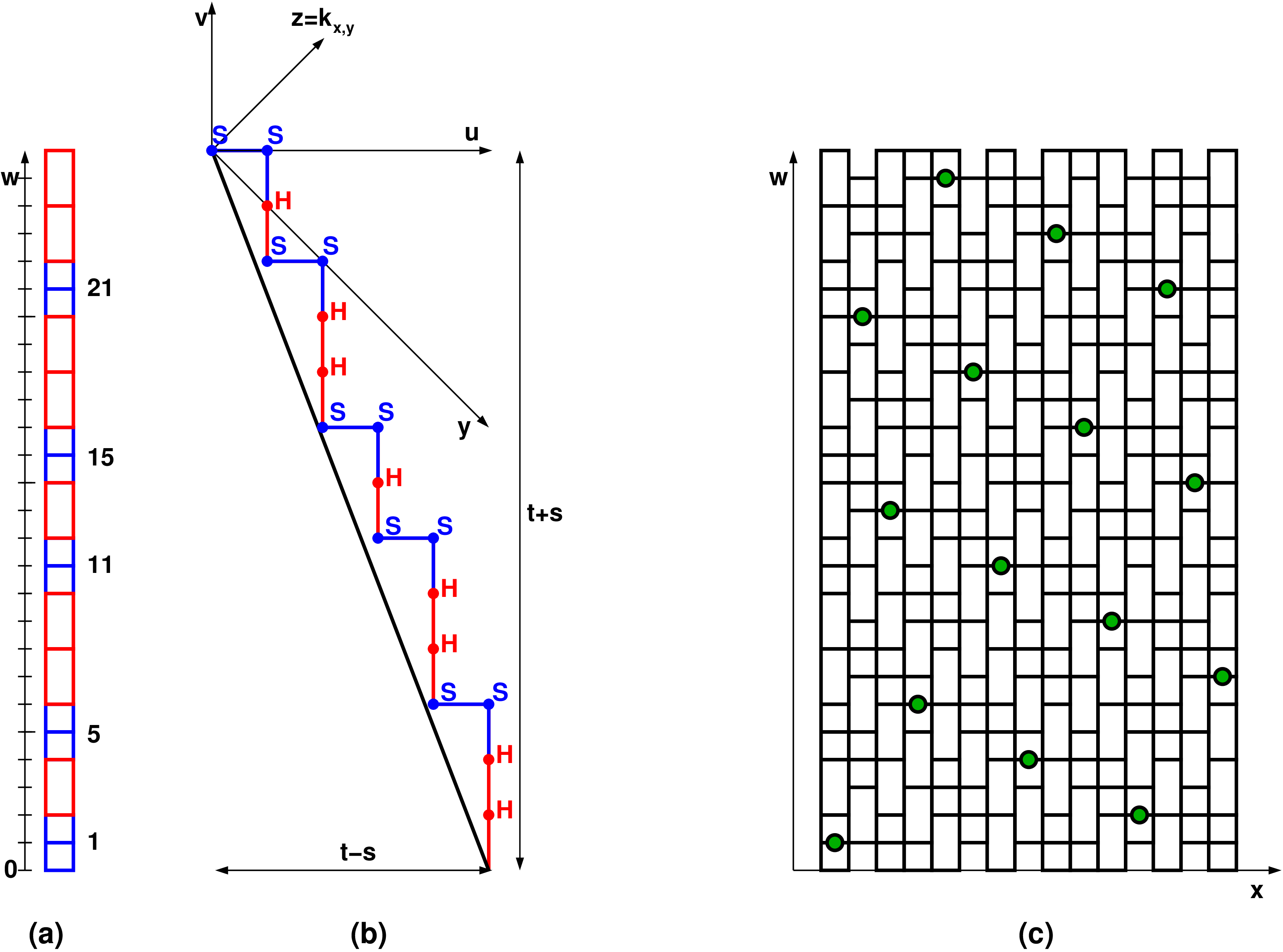}
	\caption{(a) The fundamental domain of a vertical slice of the bicolored graph dual  to the tessellated $(2,4,9)$-slanted stepped surface. 
	(b) A $x=$const.  section of the stepped surface $\bk$ (shown as a minimal path above the line $sy+t z=-rx$). We label each vertex by S/H for square/hexagon
	for the type of each corresponding dual face. (c)The infinite bicolored graph dual to the tessellated $(2,4,9)$-slanted infinite stepped surface: dots indicate 
	the periodicity lattice for the graph.}
	\label{path249}
\end{figure}

\subsubsection{$(r,s,t)$-dual graph structure}\label{dual graph struture}
Let us first describe the structure of the dual graphs $\mathcal G$ to the $(r,s,t)$-slanted initial data tessellations of previous section, 
which we shall call $(r,s,t)$-slanted graphs for short. As is clear from the discussion of Section \ref{secsol}, each such graph is a finite subset
of the infinite (doubly periodic) graph dual to the tessellation of the infinite stepped surface $\bk$, namely that corresponding to only retaining 
faces $(x,y)$ within the cone $|x-i|+|y-j|\leq |z-k|$, with $z=k_{x,y}$. Let us first describe this infinite graph.
The restriction Theorem \ref{restricthm} implies dually that the faces of any $(r,s,t)$-slanted graph may only be hexagons or squares of the following types:
\begin{figure}[H]
\centering
\includegraphics[width=6.cm]{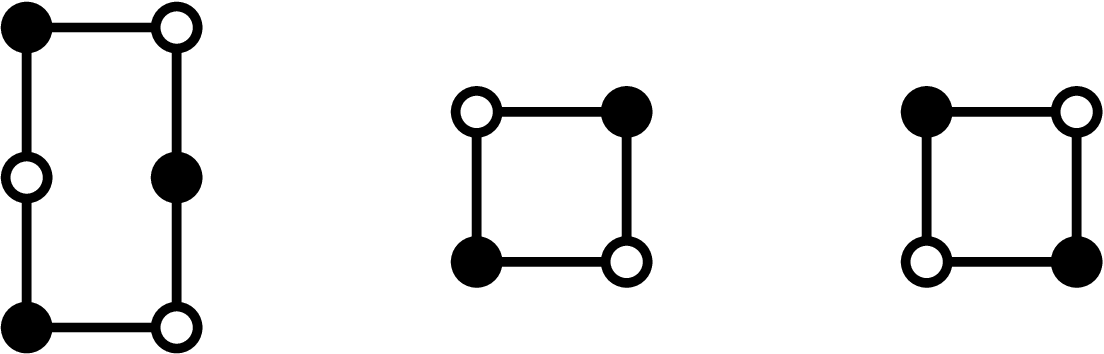}
\end{figure}
\noindent where the hexagon corresponds to the first three vertex environments of Theorem \ref{restricthm}, and the two squares to the two remaining environments. These faces are naturally arranged into columns of faces $(i,j)$, $j\in \Z$ (strips of width 1), each column a succession of hexagonal/square faces (see Fig. \ref{path249} (a) for the $(r,s,t)=(2,4,9)$ example). Bi-colorability imposes that squares always go by pairs, and we may alternatively view any vertical strip as made of only hexagons with some horizontal edges added in the middle.

By the minimality property of the stepped surface $\bk$ (see Remark \ref{minimality}), we deduce that each vertical plane section (of the form $x=$constant) of the tessellation of $\bk$ reduces to an infinite minimal path in the integer $(y,z)$ plane (with fixed parity of $y+z=x$ mod 2)
and with up/down steps $(1,\pm 1)$. The path is
described by its vertex coordinates  $(y,z=k_{x,y})$, and minimality means it is the lowest path above the line $s y+t z=-r x$.
In turn, each vertex of this path corresponds to a square or an hexagonal face of the corresponding vertical strip in the dual. 
More pecisely, each ``double descent" of the form $(k_{x-1,y},k_{x,y},k_{x+1,y})=(\mu,\mu-1,\mu-2)$ gives rise to a hexagon at $(x,y)$, while each ``down-up"
$(k_{x-1,y},k_{x,y},k_{x+1,y})=(\mu,\mu-1,\mu)$ and each ``up-down" $(k_{x-1,y},k_{x,y},k_{x+1,y})=(\mu,\mu+1,\mu)$ give rise to squares at $(x,y)$.
The H/S sequence is moreover $2t$-periodic, due to
the relation $k_{x,y+2t}=k_{x,y}-2s$.
%\textcolor{red}{this is not really true right? because by the new single-parity formula $k_{x,y}=k_{x,y+2t}-t-s$. 
%It should be that this is the periodicity of the configuration that the strip is the same modulo $2t$}, 
%we may restrict to a segment of length $2t$ in $y$.

We have represented the correspondence between the 
(minimal) path $(y,k_{x,y})$ and the vertical slice structure in Fig. \ref{path249} (b) and (a) for the case $(r,s,t)=(2,4,9)$. In Fig. \ref{path249} (b), vertices in the middle of a double descent are marked H (for hexagonal face in the dual) while those in the middle of an ``up-down" or ``down-up" are marked S (for square).
Upon changing to coordinates $(u,v)$ (by a rotation of $45^\circ$, see Fig. \ref{path249} (b)), the minimal path has steps $(1,0)$ and $(0,-1)$, and
connects the origin to the point $(t-s,-(t+s))$, while staying above the line $y=-\frac{s+t}{t-s}x$ that joins them.

To make the contact with the pinecone graphs of \cite{ProppWestB-M}, it is best to describe the above successions of squares and hexagons as a sequence of hexagons, with odd horizontal edges (connecting a white vertex on the left to a black vertex on the right) added in the middle of 
certain hexagons so as to form pairs of consecutive squares. In \cite{ProppWestB-M}, the positions of these added edges are recorded through
functions $L,U$. In our infinite $(r,s,t)$ slanted bipartite graph,
the positions (i.e. vertical coordinates) of the occupied odd horizontal edges in the vertical slice (depicted in blue in Fig. \ref{path249} (a)) are given 
by $w=1-2v$, for $v$ the coordinates of the vertices of the path closest to the line as in Fig. \ref{path249} (b), namely
$$w=1+2 \left\lfloor\frac{t+s}{t-s} u \right\rfloor \quad (u\in \Z).$$
(e.g. $w=1+2 \left\lfloor \frac{13}{5} u \right\rfloor=1,5,11,15,21$ for $u\in [0,t-s)$ in Fig. \ref{path249} (a)).
We conclude that the succession of  square and hexagonal faces is uniquely determined by $(s,t)$. 

However the {\it relative} positions of successive vertical slices  depends additionally on the value of $r$ as well. Indeed, the picture described so far remains identical in any other slice, except for the fact that the origin is a function of the position of the vertical plane $x$. Using the value \eqref{goodk} for $k_{x,y}=z$, and performing the change to rotated coordinates $u,v$ with:
$$u=\frac{z+x+y}{2},\qquad v=\frac{z-x-y}{2}, $$
the line $rx+sy+tz=0$ becomes
$$v=-\frac{(s+t)u+(r-s)x}{t-s} ,$$
and therefore the positions of the horizontal edges in the slice $x$ read again $1-2v$ for $v$ is the coordinates of the vertices of the path closest to the line, leading to:
\begin{equation}\label{w-defn}
	1+2 \left\lfloor \frac{(s+t)u+(r-s)x}{t-s} \right\rfloor .
\end{equation}
Note finally that the $v$ coordinate here is shifted by $-x/2$, hence the absolute positions of horizontal edges in the $(r,s,t)$-slanted bipartite graph are
given by
\begin{equation}\label{wofx}
w=w(x)=1-x+2 \left\lfloor \frac{(s+t)u+(r-s)x}{t-s} \right\rfloor (u,x\in \Z).
\end{equation}
This is illustrated in Fig. \ref{path249} (c), where we indicate the periodicity of the graph with a dot (which tracks the position of the $w=1$ edge of the $x=0$ slice in the other slices, modulo the $2(t+s)$ periodicity along the vertical direction). Note that the formula for $w$ in \ref{wofx} is for the case when the domain is centered at $(0,j,k)$ for some values of $j,k \in \Z$. The general case where $i\neq 0$ will requires some translation in $x$ and $v$ (see next section).	
% Let us first recall the original definition and the meaning of the parameters encoding the pinecone in the language of \cite{ProppWestB-M}

\subsubsection{Dimer graph}
Recall that the initial data domain of interest for the solution of $T_{i,j,k}$ must lie in the pyramidal cone $\mathcal C:$ $|z-k| \geq |x-i| + |y-j|$ for $(x,y,z)\in \Z^3$ with $x+y+z=0$ mod 2. The dimer graph boundaries are therefore delimited by the intersection of the initial data stepped surface $\bk$ and the cone boundary $\partial \mathcal C$.
The border of the largest domain is obtained by intersecting $\partial \mathcal C$ with the $P_0$ plane $r x+s y+t z=0$.
In the plane $x=i$, $\partial \mathcal C$ reduces to the two lines $z-k=y-j$ and $z-k=-(y-j)$. In the $(u,v)$ coordinate frame,
the former reads $v=\frac{k-j-i}{2}$ and the latter $u=\frac{k+i+j}{2}$. The first line corresponds to an upper bound on the maximum value of $v$ reached (i.e. $v_{\rm max}=\frac{k-i-j}{2}$), and the second on the maximum value of $u$
(i.e. $u_{\rm max}=\frac{k+i+j}{2}$). Recall from the previous section that the maximum value of $v$ reached in \eqref{w-defn} is $0$ from Fig. \ref{path249}(b). Thus, it is requires a translation by $\ds v_{\rm max}=\frac{k-i-j}{2}$. Once re-translated into the set of positions of horizontal edges in the dimer graph of previous section, this gives the positions:
\begin{equation}\label{wxu}
w(x,u)=1+(k-i-j)-(x-i) +2\left\lfloor \frac{(t+s)u+(r-s)x}{t-s} \right\rfloor 
=k-j+1-x+2\left\lfloor \frac{(t+s)u+(r-s)x}{t-s} \right\rfloor  .
\end{equation}
More generally, in the parallel planes $x=$const. we get bounds on the values taken by $u$.
Writing 
$k-z= \epsilon (x-i) +\eta(y-j)$ with $\epsilon, \eta \in \{1,-1\}$, and eliminating $x$ via $rx =-sy-t z$ gives:
\begin{equation}
%	\begin{aligned}
		y=\frac{(\eta j+k)t+(r-t\epsilon)x+i t \epsilon}{\eta t - s}, \qquad 
		z=\frac{(\eta j+k)s+(\eta r -\epsilon s)x +\epsilon s i}{s-t \eta}
%				y=\frac{(\eta j+k)t+(r-t\epsilon)(x-i)}{\eta t - s}, \qquad 
%		z=k+\eta j-\frac{(\eta j+k)t+(r-\epsilon\eta s)(x-i)}{t-s \eta}
\label{x,y,z boundaries}
\end{equation}
Applying the change of variables of the previous section $\ds u=\frac{x+y+z}{2}, \ds v=\frac{z-x-y}{2}$, the quantity 
$w=k-j+1-2v-x$ reads
\begin{equation}\label{wboundary}
w=k-j+1+y-z=k-j+1+\frac{(\epsilon i +\eta j+k)(s+t)+x(r(1+\eta)-\epsilon(s+t))}{\eta t -s}\qquad (\epsilon,\eta \in \{-1,1\})
%1-k-\eta j+\frac{(\eta+1)(j+k)t+\left((\eta+1)r+(t+s)\epsilon\eta\right)(x-i)}{t-\eta s}
\end{equation}

%		w&=1+y-z=1-k-\eta j+\frac{(\eta+1)(j+k)t+\left((\eta+1)r+(t+s)\epsilon\eta\right)(x-i)}{t-\eta s}
%	\end{aligned}
%	\label{x,y,z boundaries}
%\end{equation}
This gives the four lines $w=w_{\rm min}^\pm (x)$ and $w=w_{\rm max}^\pm(x)$ in the $(x,w)$ plane, delimiting the dimer graph, and eventually the four line segments
\begin{eqnarray}
&&\left\{ \begin{matrix} w_{\min}^-(x)&=&1-(x-i)  \\ 
w_{\rm max}^-(x)&=& 1+\frac{2(ri+sj+tk)+(2r+t+s)(x-i)}{t-s}\end{matrix} \right. \quad  \left(-\frac{ri+sj+tk}{r+t}\leq x-i\leq 0\right)\label{boundonew}\\
&&\left\{\begin{matrix} w_{\min}^+(x)&=&1+(x-i)  \\  
w_{\rm max}^+(x)&=&1+\frac{2(ri+sj+tk)+(2r-t-s)(x-i)}{t-s}\end{matrix}  \right .\quad  \left(0\leq x-i\leq \frac{ri+sj+tk}{t-r}\right)\label{boundtwow}
\end{eqnarray}

\begin{figure}
	\includegraphics[scale=.3]{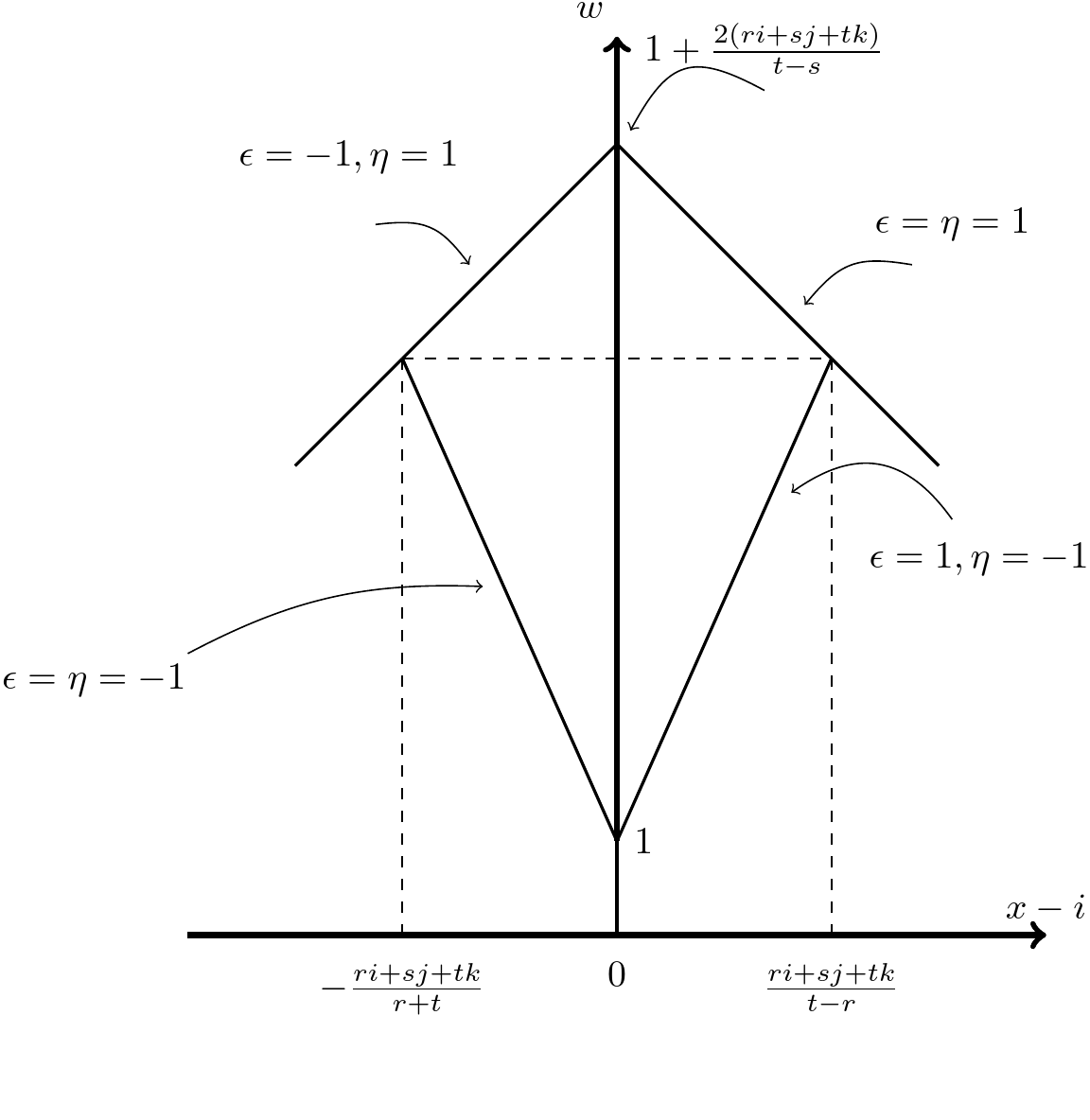}
	\caption{The domain \eqref{boundw} in the $x,w$ coordinates.}
	\label{domain in wx}
\end{figure}

Finally, the dimer graph $\mathcal G$ is determined by the function $w(x)$ \eqref{wofx} together with the conditions 
\begin{equation}\label{boundw}
 w_{\rm min}^\pm(x)\leq w(x,u)\leq w^\pm_{\rm max}(x),\quad \left(-\frac{kt+j s+ir}{r+t}\leq x-i\leq \frac{ri+sj+tk}{t-r}\right) .
 \end{equation}
 The corresponding domain is depicted in Fig. \ref{domain in wx}.
 
 \subsection{Comparison with Pinecones} The Pinecones defined in \cite{ProppWestB-M} were constructed 
 to provide combinatorial solutions of the three-term Gale-Robinson sequence:

\begin{equation}\label{G-R sequence}
 	a(\tn)a(\tn-\tm)=a(\tn-\ti)a(\tn-\tj)+a(\tn-\tk)a(\tn-\tl)
 \end{equation} 
with the initial condition $a(\tn)=1$ for $\tn = 0, 1, \dots, \tm-1$ with $\ti,\tj,\tk,\tl$ given integers satisfying $\ti+\tj=\tk+\tl=\tm$, and $\tj=\min\{\ti,\tj,\tk,\tl\}$. The parameters $\ti,\tj,\tk,\tm,\tl,\tn$ are borrowed from the notations in \cite{ProppWestB-M} and are different from our $i,j,k,m$ (we have used the tilde notation to distinguish them). For each set of parameters $\{\ti,\tj,\tk,\tl\}$, the authors construct a sequence of pinecone graphs $({\mathcal P}_n)_{n\geq 0}\equiv ({\mathcal P}(\tn;\ti,\tj,\tk,\tl))_{\tn \geq 0}$, entirely determined by the functions: \newline
 \begin{adjustbox}{scale={0.8}{0.8},center}
 $\begin{aligned}
U(\tn,R,C)&=2C+R-3-2\left\lfloor \frac{\tm C+\tk R+\ti-\tn-1}{\tj}\right\rfloor =-1-R-2\left\lfloor\frac{\ti C+(\tk-\tj)R+\tm-\tn-1}{\tj}\right\rfloor\\ 
L(\tn,R,C)&=2C+R-3-2\left\lfloor \frac{\tm C+\tl R+\ti-\tn-1}{\tj}\right\rfloor =U(\tn,-R,C+R) .
 \end{aligned}$
 \end{adjustbox}
 \newline
The graphs are drawn on a substrate of ``brick wall" lattice in $\Z^2$ made of horizontal hexagons (i.e. with all horizontal edges $(x,y)-(x+1,y)$, $x,y\in \Z$,
and every other vertical one $(x,y)-(x,y+1)$, $x,y\in Z$, $x+y=0$ mod 2), to which some central vertical edges $(x,y)-(x,y+1)$ are added at positions determined by $U(\tn,R,C)$ (in the upper part of the pinecone in the row $R=0,1,2,...$) and $L(\tn,R,C)$ (in the lower part of the pinecone in the rows $-R=0,-1,2,...$),
with $C\geq 0$ limited by the conditions that $U(\tn,R,C)>R$ and $L(\tn,R,C)>R$.
These give the respective following bounds for the positions $U(\tn,R,C)$ and $L(\tn,R,C)$:
\begin{eqnarray}
&&U\ : \ \  0\leq C\leq \left\lfloor \frac{\tn-\tm-\tk R}{\tj}\right\rfloor\label{boundU}\\
&&L\ :\ \ 0\leq C\leq  \left\lfloor \frac{\tn-\tm-\tl R}{\tj}\right\rfloor \label{boundL}
\end{eqnarray}

\begin{lemma}\label{floor_lemma}
	\begin{equation}
	\label{floor_equality}
		U(\tn,R,C)=-1-R-2\left\lfloor\frac{\ti C+(\tk-\tj)R+\tm-\tn-1}{\tj}\right\rfloor=1-R+2\left\lfloor\frac{-\ti C+(\tj-\tk)R+\tn-\tm}{\tj}\right \rfloor
	\end{equation}
	for $\tj> 0 \in \Z$
\end{lemma}
\begin{proof}
	It is sufficient to show $- \lfloor x\rfloor=\ds \lfloor -x -\frac{1}{\tj}\rfloor+1$ for $x =\frac{a}{\tj}, \tj> 1, a\in \Z$. Let $n=\lfloor x\rfloor$, i.e. such that $n\leq x < n+1$, $n \in \Z$, 
	then $-n-1-\frac{1}{\tj}< -x-\frac{1}{\tj} \leq -n-\frac{1}{\tj}<-n$, hence:
	\begin{equation*}
		-n-2\leq \lfloor-x-\frac{1}{\tj}\rfloor \leq -n-1
	\end{equation*}
	Assuming $-n-2= \lfloor-x-\frac{1}{\tj}\rfloor$, then we would have $-x-\frac{1}{\tj} <-n-1$, hence $n+1-\frac{1}{\tj}<\frac{a}{\tj}<n+1$, 
	which contradicts $a\in \Z$, as the width of the interval is strictly less that $\frac{1}{\tj}$.
Therefore we must have $-n-1= \lfloor-x-\frac{1}{\tj}\rfloor$, and the Lemma follows.
\end{proof}

Comparing Eqs. (\ref{boundU}-\ref{boundL}) to \eqref{boundw}, we find the following correspondence:

\begin{thm}\label{corresp}
The graphs ${\mathcal G}$ are identified with pinecones via the following
correspondence between our parameters $(r,s,t)$ and the parameters $\ti,\tj,\tk,\tl,\tm$ of  \cite{ProppWestB-M}:
\begin{equation}\label{from stepped surface to pinecone parameters}
	\left\{ \begin{matrix}  \ti=t+s & \tj=t-s & \tk=t+r & \tl=t-r &\tm=2t& \tn = ri+sj+tk+2t &{\rm if} \ r,s,t \ {\rm not}\ {\rm all}\ {\rm odd}\\
 \ti=\frac{t+s}{2} & \tj=\frac{t-s}{2} & \tk=\frac{t+r}{2} & \tl=\frac{t-r}{2} &\tm=t& \tn=\frac{ri+sj+tk}{2}+t&{\rm otherwise}  \end{matrix}
\right.
\end{equation}
\end{thm}
\begin{proof}
From Lemma \ref{floor_lemma}:
\begin{eqnarray*}U(\tn,R,C)&=&1-R+2\left\lfloor \frac{-\ti C+(\tj-\tk) R+\tn-\tm}{\tj} \right\rfloor \\
L(\tn,R,C)&=&1+R+2\left\lfloor \frac{-\ti C-\tl R+\tn-\tm}{\tj} \right\rfloor=1-R+2\left\lfloor \frac{-\ti C+(\tj-\tl) R+\tn-\tm}{\tj} \right\rfloor
\end{eqnarray*}
For $r,s,t$ not all odd, the first identification of parameters in the theorem (with $\ti=t+s$, $\tj-\tl=r-s$, $\tm=2t$, $\tj=t-s$) allows to identify for $x\geq i$ where $i$ is the index for the solution of the $T$-system $T_{i,j,k}$:
$$ w(x,u)=L\left(ri+sj+tk+2t,x-i,\frac{k+i+j}{2}-u\right) ,$$
corresponding to a mapping of variables $R=x-i$, $\ds C=\frac{k+i+j}{2}-u$ and $\tn=kt+js+ri+2t$.
Moreover, the bounds
$R< L(\tn,R,C)\leq L(\tn,R,0)$ turn into
$$x-i< w(x,u)\leq  1+\frac{2(ri+sj+tk)+(2r-t-s)(x-i)}{t-s} $$
which is equivalent to (\ref{boundtwow}-\ref{boundw}).
Similarly, when $x\leq i$, using $U(\tn,R,C)=L(\tn,-R,R+C)$, we find that
$$ w(x,u)=U\left(ri+sj+tk+2t,-(x-i),\frac{k+i+j}{2}-u+x-i\right) ,$$
while the bounds $R< U(\tn,R,C)\leq U(\tn,R,0)$ and \eqref{boundonew} are identical. When $r,s,t$ are all odd, we may rewrite
$$w(x,u)= k-j+1-x+2\left\lfloor \frac{\frac{t+s}{2}u+\frac{r-s}{2}x}{\frac{t-s}{2}} \right\rfloor$$
and the above identifications correspond now to the second line of \eqref{from stepped surface to pinecone parameters}.
\end{proof}

We conclude that the dimer graphs for $r,s,t$-slanted initial data $T$-system solutions are nothing but the pinecones of \cite{ProppWestB-M}, with the correspondence of Theorem \ref{corresp} above. For convenience, in the remainder of this paper, we will work in the original $(x,y,z)$ coordinates, 
and no longer refer to the $(x,u,v)$ frame. Any of our results on limit shapes can be straightforwardly translated into pinecone language  via the change of variables $\ds u=\frac{z+y+x}{2}$, $\ds v=\frac {z-x-y}{2}$, and $x$ unchanged (note that the pinecones must also be flipped to match our dual slanted graphs).

\begin{figure}
\includegraphics[scale=.35]{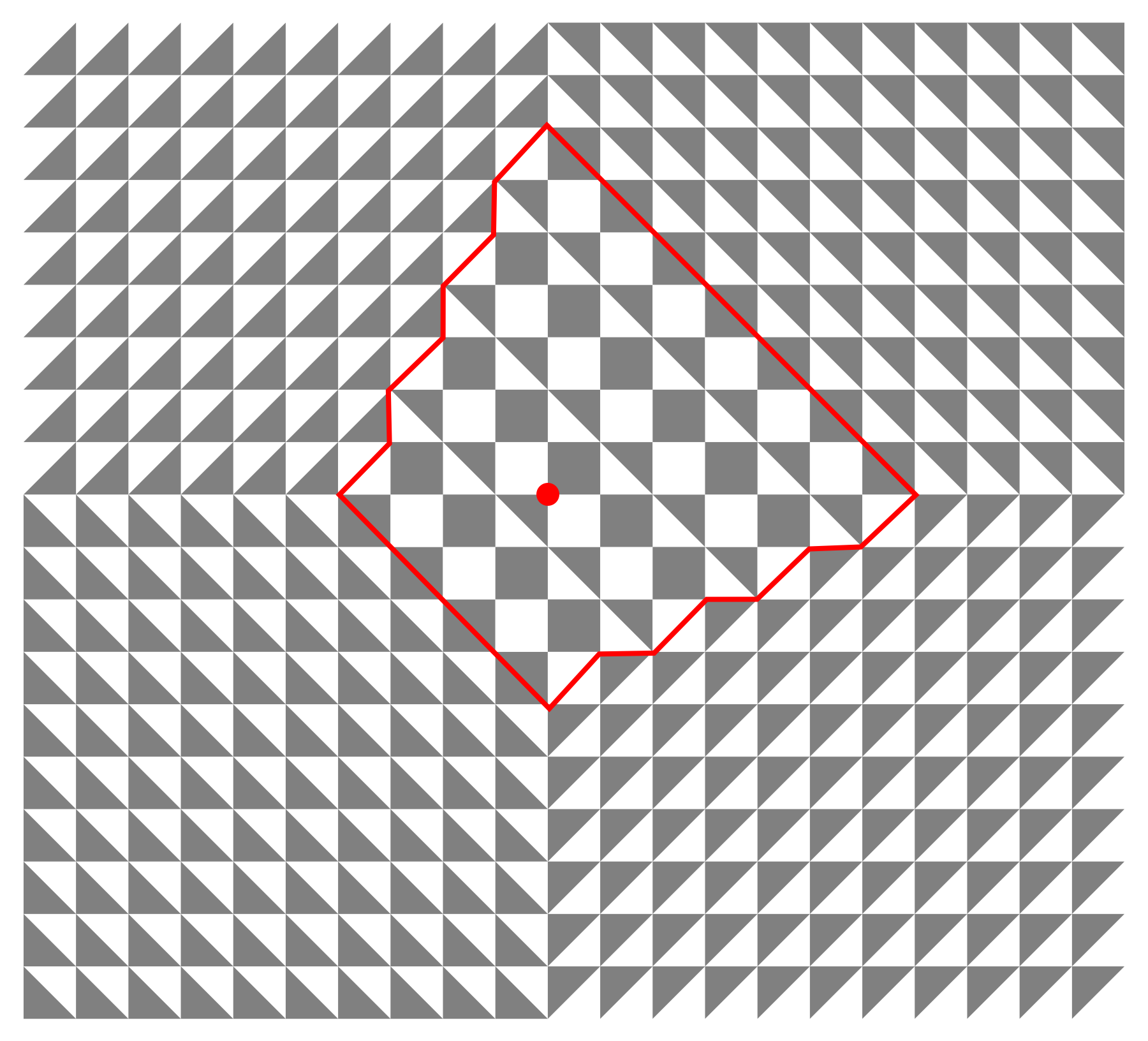}
\caption{Tessellated domain for the ($1,1,3$)-slanted stepped surface for $T_{0,0,6}$, extended by the four 
sides of the pyramid with apex at $(i,j,k)$. The red line marks the boundary between the domain and its continuation  via the four infinite planes forming the pyramid.}
\label{113_pyramid}
\end{figure}

Let us briefly recall the core phenomenon of \cite{ProppWestB-M}, which is best explained by embedding the pinecone configurations into a minimal Aztec diamond shaped domain, obtained by adding brick-wall configurations (with only hexagons and no extra vertical edges), with fixed boundary dimer configurations, which propagate throughout
the added hexagons to provide frozen configurations at/within the boundaries of the pinecone. The dimer configurations involve a ``core" of active edges that may or may not be occupied by dimers.  From the dual (stepped-surface) point of view, such brick-wall additions correspond to a continuation of the stepped surface beyond the intersection with the pyramid of apex $(i,j,k)$,  by the four plane faces of the pyramid itself (see Fig. \ref{113_pyramid} for an illustration). Indeed, the latter planes decompose into alternating Black and White triangles, with only 6-valent vertices,
thus giving rise to the hexagons of the added brick-wall in the dual graph.

We give below examples of pinecones and the corresponding function $w(x,u)$, and of the core phenomenon.

\begin{example}\label{249_wxu_structure}
Let $(r,s,t)=(2,4,9)$, and $i=1,j=2,k=3$, with $ri+sj+tk+2t=55=n$ and $\frac{k+i+j}{2}=3$. We have
$w(x,u)=2-x+2 \left\lfloor \frac{13 u-2x}{5} \right\rfloor=1-R+2 \left\lfloor \frac{37-13 C-2R}{5} \right\rfloor$, ($R=x-1,C=3-u$), leading to the successive positions for $x-1\in [-3,5]$ in \eqref{249_UV_b}
\begin{figure}[h]
	\centering
	\begin{minipage}[t]{.6\textwidth}
		\includegraphics[scale=.25]{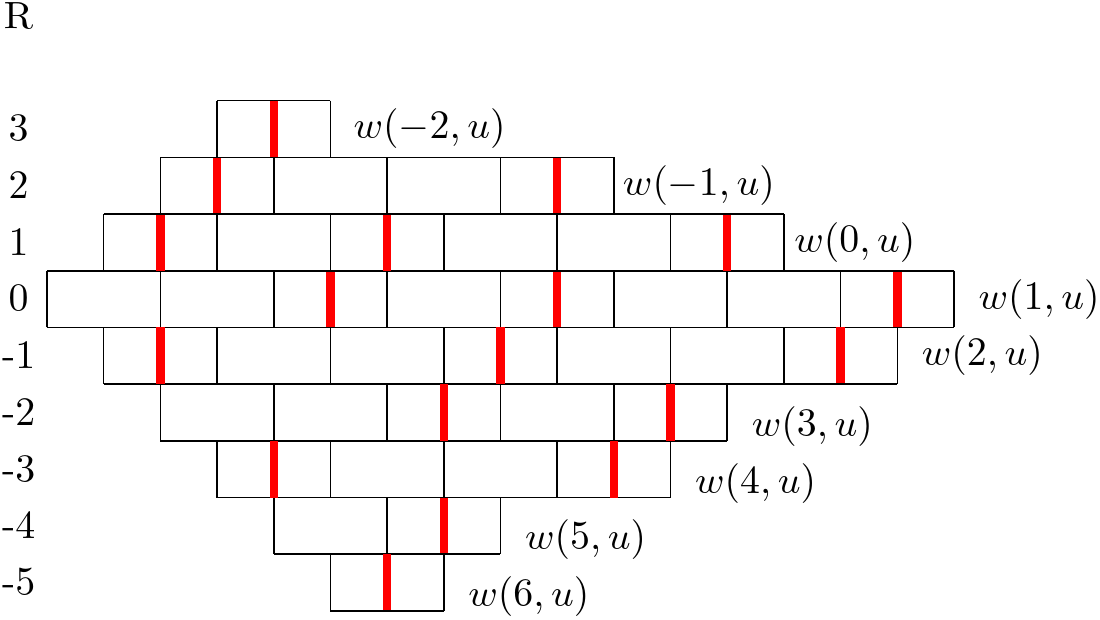}
		\caption{Dimer graph for the partition function $T_{1,2,3}$ with rows indexed by $R$ for the $(2,4,9)$-slanted stepped surface}\label{249_UV_a}
	\end{minipage}%\hskip 1cm 
	\begin{minipage}[]{.35\textwidth}\vskip -2.cm
\hskip 2 cm $\begin{matrix}
		& &4 & w(-2,u)\\
 & 3 &9 & w(-1,u) \\
 2    & 6 & 12 & w(0,u) \\
 5 & 9 & 15 & w(1,u)\\
 2 & 8 & 14 & w(2,u)\\
 & 7 & 11  & w(3,u)\\
 & 4 &10 & w(4,u)\\
 & & 7 & w(5,u)\\
 & & 6 & w(6,u)\\
\end{matrix}$
\caption{Value of $w(x,u)$ indicating the positions of extra vertical edges (in red)}\label{249_UV_b}
\end{minipage}
\end{figure}
These coincide with the values $U(55,R,C)$, ($R=3,2,1$) followed by  $L(55,R,C)$,  ($R=0,1,2,3,4,5$), which determine the pinecone on the left (Fig.  \ref{249_UV_a}).

 \end{example}

\begin{example}\label{113_wxu_structure}
	Let $(r,s,t)=(1,1,3)$, and $i=0,j=0,k=4$, with $\frac{ri+sj+tk}{2}+t=9$, and $\frac{k+i+j}{2}=2$ as in example \eqref{explaining core}. We have
$w(x,u)=5-x+4u$, ($R=x,C=2-u$) which is represented in Fig \ref{113_UV_a}.

\begin{figure}
	\centering
	\begin{minipage}[t]{.45\textwidth}
		\includegraphics[scale=.3]{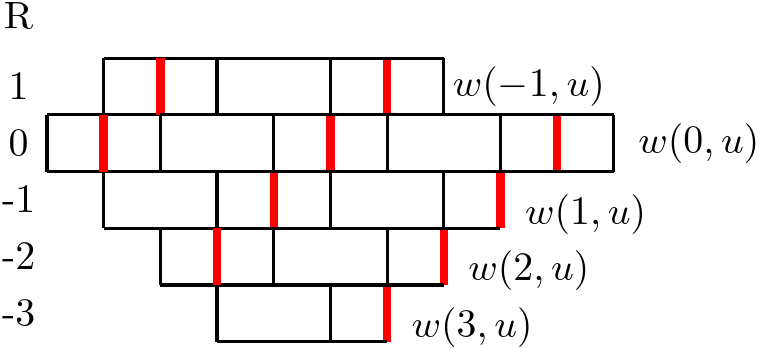}
		\caption{Dimer graph for the partition function $T_{0,0,4}$ with rows indexed by $R$ for the $(1,1,3)$-slanted stepped surface}\label{113_UV_a}
	\end{minipage}\hskip 1cm
\begin{minipage}[]{.4\textwidth}
		\vskip -2.cm 
		\hskip 2cm$\begin{matrix}
   & 2 &  6  & w(-1,u) \\
  1 & 5 & 9  & w(0,u) \\
    & 4 & 8  & w(1,u)\\
   & 3 &  7  & w(2,u)\\
   &   &  6  & w(3,u)\\
		\end{matrix}$
		\caption{Value of $w(x,u)$}\label{113_UV_b}
		\label{113_wxu_bar_loc}
	\end{minipage}

\end{figure}
These coincide with the values $U(9,R,C)$, ($R=1$) followed by  $L(9,R,C)$,  ($R=0,1,2,3$), which determine the pinecone on the left (Fig. \ref{113_wxu_bar_loc}).

 \end{example}
%\begin{figure}
%   		\includegraphics[scale=.1]{figures/112_nice_example_brickwall.png}
%   		\caption{Dimer configurations corresponding to the 4 terms of the $T_{0,-1,3}$ solution of the $(1,1,2)$-slanted initial data T-system. }
%   		\label{112_nice_example_brickwall}
%   \end{figure}

%Finally, we illustrate the core phenomenon in the pinecone language. The tessellation in Figure \ref{tessellation progress 246} depicts the core regions of the corresponding dimer domain. The domain can be extended infinitely, with a brickwall of hexagons outside of the core, Note that the hexagons don't contribute as they all receive weight $1$. This core phenomenon also coincides in the perfect matching setting in example \ref{core phenomenon example}

\begin{example} \label{core phenomenon example} We now illustrate the core phenomenon in the case of the solution
\begin{figure}
   		\includegraphics[scale=.125]{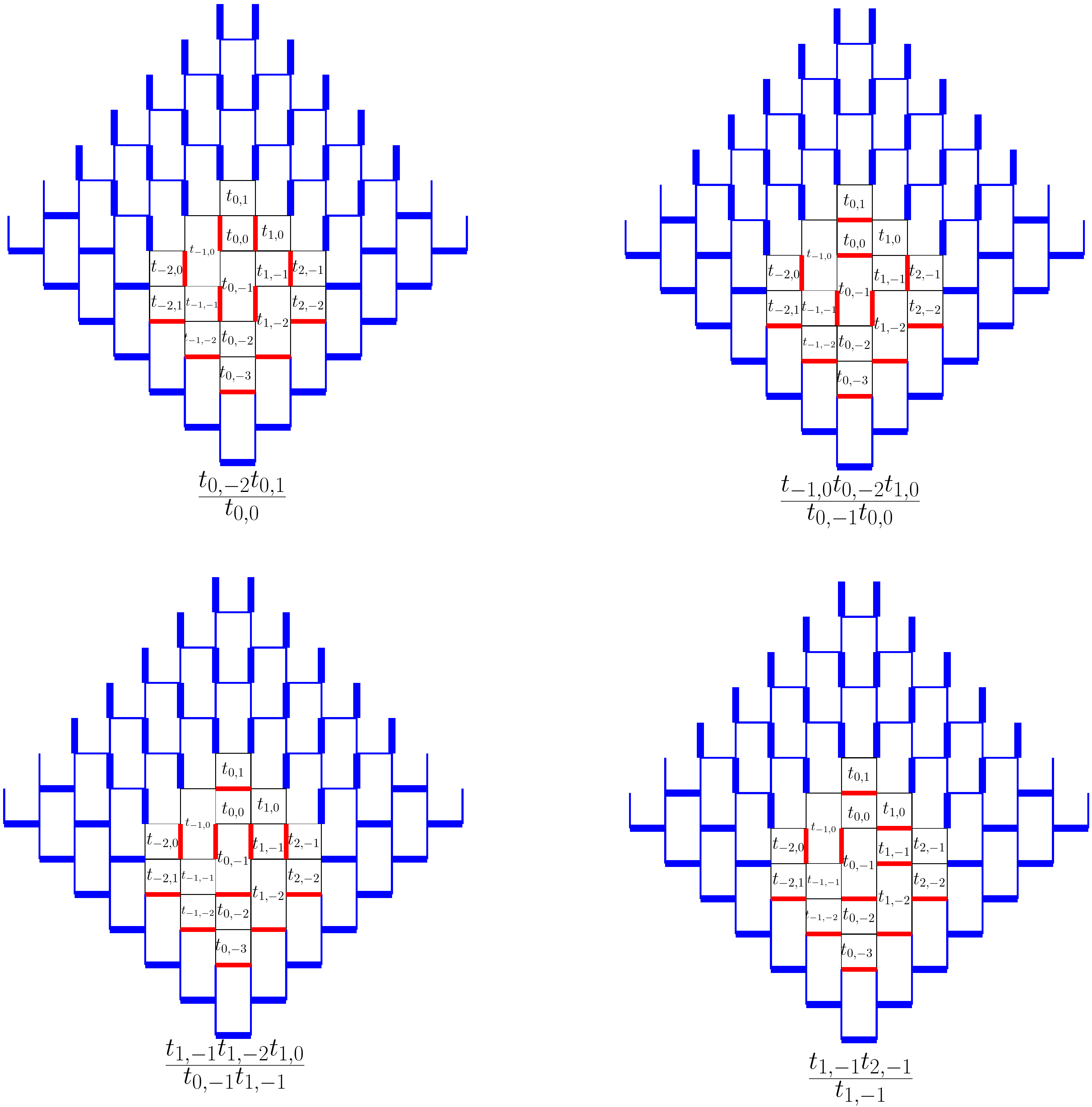}
   		\caption{Dimer configurations corresponding to the 4 terms of the $T_{0,-1,3}$ solution of the $(1,1,3)$-slanted initial data T-system. }
   		\label{113_nice_example_brickwall}
\end{figure}
$ T_{0,-1,3}$ of the $(1,1,3)$-slanted $T$-system. The solution reads:
$$\ds T_{0,-1,3} =\frac{t_{0,-2} t_{0,1}}{t_{0,0}}+\frac{t_{-1,0} t_{0,-2} t_{1,0}}{t_{0,-1}
   t_{0,0}}+\frac{t_{-1,-1} t_{1,-2} t_{1,0}}{t_{0,-1} t_{1,-1}}+\frac{t_{-1,-1}
   t_{2,-1}}{t_{1,-1}}. $$
%The differences of the model is can be computed directly via the formulation of $w(x,u)$ similar to example \ref{113_wxu_structure} and \ref{249_wxu_structure}. 
We have represented the corresponding four dimer configurations of the full dimer domain  in Figure \ref{113_nice_example_brickwall}: 
here the core is extended by brick wall hexagonal faces to an Aztec diamond shape (but could go on and 
cover the entire plane as well, as suggested by the dual graph of that of Fig. ~\ref{113_pyramid}). 
The brick wall addition is similar to Figure 5 in \cite{ProppWestB-M}, and the blue faces don't contribute to the partition function by theorem \ref{mainthm}. 
\end{example}

\section{The case of uniform slanted initial data}

\subsection{Uniform $T$-system solution}\label{uniform solution}

For fixed values of $(r,s,t)$ the simplest solution of the T-system \eqref{Tsys} corresponds to choosing uniform initial data in each initial data plane $(P_\ell)$, $\ell=0,1,...,2t-1$.
More precisely, choosing the initial values of $T$ to be $T_{i,j,k}=a_\ell$ for all $(i,j,k)\in (P_\ell)$ for some positive real numbers $a_0,a_1,...,a_{2t-1}$,
we deduce that for all $m\geq 2t$:
$$T_{i,j,k}=a_m \qquad (i,j,k)\in (P_m)$$
where $a_m$, $m\geq 2t$ are subject to the ``Gale-Robinson" recursion relation
$$ a_m \, a_{m-2t}=a_{m+r-t}\, a_{m-r-t}+a_{m+s-t}\, a_{m-s-t} $$

Among these solutions a particularly simple one consists in taking $a_\ell=\al^{\ell(\ell-1)/2}$ for $\ell=0,1,...,2t-1$, leading to
\begin{equation}\label{unisol}
a_m=\al^{m(m-1)/2},\qquad (m\in \Z),
\end{equation} 
provided $\al$ satisfies
\begin{equation}\label{alphaeq} \al^{t^2} =\al^{r^2}+\al^{s^2} .
\end{equation}
It is easy to see that this equation always admits a unique positive solution such that $\al>1$, which we pick from now on. As an example, taking $r=s=0$ and $t=1$ leads to the "flat"
initial data along two parallel planes $k=0$ and $k=1$, leading to the Aztec diamond domino tiling solution $T_{i,j,k}=2^{k(k-1)/2}$ with $\al=2$.
By a slight abuse of language we shall call the solution \eqref{unisol} the uniform solution of the $T$-system with $(r,s,t)$-slanted plane initial data.

\subsection{Density}\label{density definition}

\subsubsection{Expectation values}\label{expectation defintition}

In Sect. \ref{pinecone section}, we have interpreted the solution $T_{i,j,k}$ of the $T$-system as the partition function of some suitable $(r,s,t)$-dimer model 
with local Boltzmann weights expressed in terms of the initial data. To gain access to statistical properties of the dimer model, such
as the average number of dimers occupying the edges adjacent to a given face, we may use the dependence of $T$ on the initial data 
as follows. Pick a point $(i_0,j_0,k_0=k_{i_0,j_0})$ belonging to one of the initial data planes $(P_{ri_0+sj_0+tk_0})$
with $0\leq r i_0+s j_0+t k_0<2t$). Assume it corresponds in the dimer graph to the center of a $2v$-valent face. As the local contribution for this face to the partition function 
is $(t_{i_0,j_0})^{v-1-N_{i_0,j_0}(\cD)}$, we may write
\begin{equation}\label{averageD}
\rho^{(i_0,j_0,k_0)}_{i,j,k}:=\frac{1}{T_{i,j,k}}\, t_{i_0,j_0} \partial_{t_{i_0,j_0}} (T_{i,j,k}) =\langle v-1- N_{i_0,j_0}(\cD) \rangle_{i,j,k} 
\end{equation}
where $\langle f \rangle_{i,j,k}$ stands for the statistical average of the function $f$ over the dimer configurations $\cD$ for the $(i,j,k)$ 
dimer model, and where $k_0=k_{i_0,j_0}$ indicates the time variable along the initial data surface. 
We refer to the function $\rho$ as the (local) density of dimers at position $(i_0,j_0,k_0)$ in the $(i,j,k)$ dimer model.
\begin{example}
	\label{T004_example}
	We compute explicitly all values of $\rho_{0,0,4}$ at various sources $(i_0,j_0)$ with uniform initial data \eqref{alphaeq}. 
	\begin{equation}
		\begin{array}{cccccc}
		& & &\vdots& & \\
		&\rho_{0,0,4}^{(-3,5)}&\rho_{0,0,4}^{(-2,5)}&\cdots&\rho_{0,0,4}^{(5,5)}&\\
		&\rho_{0,0,4}^{(-3,4)}&\rho_{0,0,4}^{(-2,4)}&\cdots&\rho_{0,0,4}^{(5,4)}&\\
		\cdots&\vdots&\ddots&\ddots&\vdots&\cdots\\
		&\rho_{0,0,4}^{(-3,-3)}&\rho_{0,0,4}^{(-2,-3)}&\cdots&\rho_{0,0,4}^{(5,-3)}&\\
		& & &\vdots& & \\
\end{array}=\begin{array}{ccccccccc}
 0 & 0 & 0 & 0 & 0 & 0 & 0 & 0 & 0 \\
 0 & 0 & 0 & \frac{1}{16} & 0 & 0 & 0 & 0 & 0 \\
 0 & 0 & \frac{1}{16} & -\frac{1}{8} & \frac{1}{4} & 0 & 0 & 0 & 0 \\
 0 & 0 & \frac{3}{8} & 0 & -\frac{3}{8} & \frac{3}{8} & 0 & 0 & 0 \\
 0 & \frac{1}{4} & -\frac{1}{2} & \frac{1}{8} & -\frac{1}{8} & -\frac{3}{8} & \frac{1}{4} & 0 & 0 \\
 0 & \frac{1}{4} & -\frac{1}{4} & \textcolor{red}{0} & \frac{1}{8} & 0 & -\frac{1}{8} & \frac{1}{16} & 0 \\
 0 & 0 & \frac{1}{2} & -\frac{1}{4} & -\frac{1}{2} & \frac{3}{8} & \frac{1}{16} & 0 & 0 \\
 0 & 0 & 0 & \frac{1}{4} & \frac{1}{4} & 0 & 0 & 0 & 0 \\
 0 & 0 & 0 & 0 & 0 & 0 & 0 & 0 & 0 \\
\end{array}
	\end{equation}
	where the $0$'s at the boundary extend to infinity as these initial data points do not contribute to the density $\rho_{0,0,4}$. The density profile is shown in Figure \ref{density_profile_113_004}
	\begin{figure}
		\centering
		\includegraphics[scale=0.5]{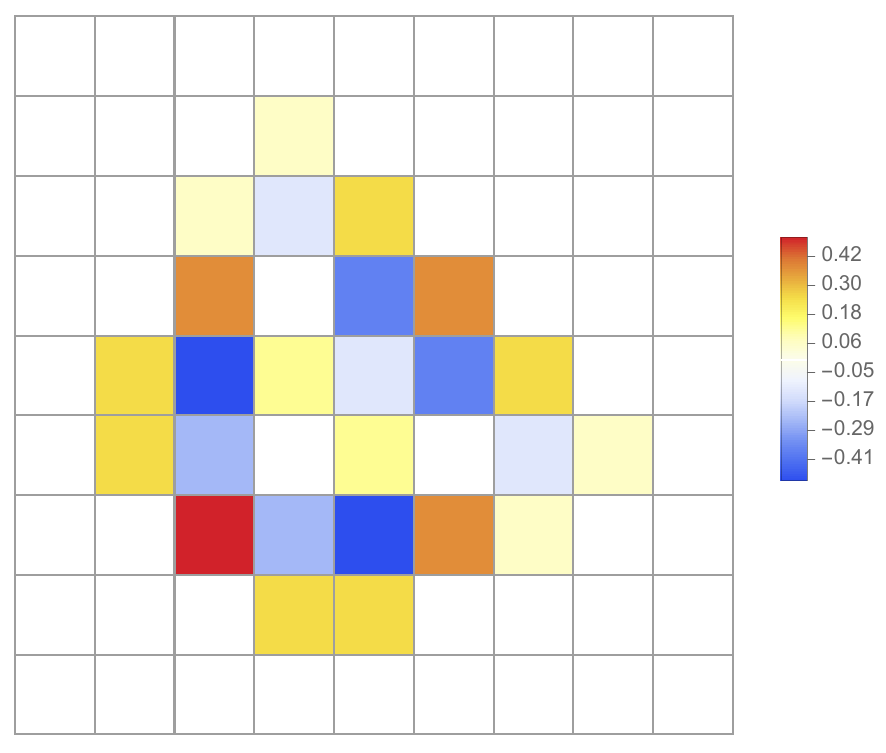}
		\caption{Density profile for $\rho_{0,0,4}^{(i_0,j_0)}$ where $i_0,j_0=-3, \cdots,5$}
		\label{density_profile_113_004}
	\end{figure}
\end{example}
An interesting property of this density is that it is an {\it order parameter} for the crystalline/liquid phases of the model, namely $\rho^{(i_0,j_0,k_0)}_{i,j,k}$
vanishes identically in the crystal phase, while it fluctuates and becomes non-zero in the liquid regions. Indeed, as we shall see below, 
the crystalline phase is characterized by
the presence of exactly $v-1$ dimers around each $2v$-valent face (1 for squares, 2 for hexagons),
leading to a vanishing local density by \eqref{averageD}.

Another property is translation invariance, namely that
\begin{equation}
\label{transinv}
\rho^{(i_0,j_0,k_0)+(x,y,z)}_{i+x,j+y,k+z}=\rho^{(i_0,j_0,k_0)}_{i,j,k}\Big\vert_{t_{a,b}\to t_{a+x,b+y}}, \quad x,y,z\in \Z,\ x+y+z=0\,  {\rm mod}\,  2,
\end{equation}
for all translations by $(x,y,z)$ that leave the initial data surface $\bk$ invariant. In the case of $(r,s,t)$ parallel initial data planes, only translations
such that $rx+sy+tz=0$ are allowed.
The latter property is key to allow us to use the explicit value of $\rho^{(i_0,j_0,k_0)}_{i,j,k}$, for varying $(i,j)$ and fixed $(i_0,j_0,k_0=k_{i_0,j_0})$ to browse 
through the local densities of the dimer model. Indeed, instead of interpreting this quantity as a local density of the $(i,j,k)$ dimer model, we may use translational invariance \eqref{transinv} to reinterpret it as the 
local density at some varying point $\sim(i_0-i,j_0-j)$ in a dimer model whose graph is centered at a fixed point $\sim(0,0)$ close to the origin.

The $T$-system relation allows to derive a linear recursion relation for $\rho$, by simply differentiating w.r.t. $t_{i_0,j_0,k_0}$:
\begin{equation}\label{recurho}
\rho^{(i_0,j_0,k_0)}_{i,j,k+1}+\rho^{(i_0,j_0,k_0)}_{i,j,k-1}= L_{i,j,k}\, (\rho^{(i_0,j_0,k_0)}_{i+1,j,k}+\rho^{(i_0,j_0,k_0)}_{i-1,j,k})+R_{i,j,k}\, (\rho^{(i_0,j_0,k_0)}_{i,j+1,k}+\rho^{(i_0,j_0,k_0)}_{i,j-1,k}), 
\end{equation}
where
\begin{equation}\label{defLR}
L_{i,j,k}:=\frac{T_{i+1,j,k}\, T_{i-1,j,k}}{T_{i,j,k+1}T_{i,j,k-1}}, \qquad R_{i,j,k}=\frac{T_{i,j+1,k}\, T_{i,j-1,k}}{T_{i,j,k+1}T_{i,j,k-1}}=1-L_{i,j,k} .
\end{equation}
$\rho$ is further determined by the initial conditions $\rho^{(i_0,j_0,k_0)}_{i,j,k} =\delta_{i,i_0}\delta_{j,j_0}\delta_{k,k_0}$ along the initial data surface $\bk$. 
%planes $(P_\ell)$, $\ell=0,1,...,2t-1$. 

The density $\rho$ can be explicitly computed whenever the solution $T_{i,j,k}$ of the $T$-system is explicit. This is done in the next sections for $(r,s,t)$-slanted initial data planes planes $(P_\ell)$, $\ell=0,1,...,2t-1$.

\subsubsection{The density of the uniform case}\label{density in the uniform case}

In the uniform case, we have the solution $T_{i,j,k}=\al^{m(m-1)/2}$ where $m=ri+sj+tk$, leading to the coefficients
$$L_{i,j,k}=\al^{r^2-t^2}, \qquad R_{i,j,k}=\al^{s^2-t^2}, $$
independent of $i,j,k$, while $\al>1$ is the solution of \eqref{alphaeq}.
Let us define the function $\mu(i,j,k)=ri+sj+tk$. It will be convenient to gather solutions of \eqref{recurho} into generating functions
$$\rho^{(i_0,j_0,k_0)}(x,y,z):= \sum_{i,j,k\in \Z\atop \mu(i,j,k)\geq 0} \rho_{i,j,k}^{(i_0,j_0,k_0)} \, x^i\, y^j\, z^k .$$
As a first example, taking $(i_0,j_0,k_0)=(0,0,0)$,
and using the recursion relation \eqref{recurho}, multiplying by $x^iy^jz^k$ and summing over $i,j,k \in \Z$ gives
$$1-\rho^{(0,0,0)}(x,y,z)=\frac{z^2}{1+z^2-z\al^{r^2-t^2}\Big(x+\frac{1}{x}\Big)-z\al^{s^2-t^2}\Big(y+\frac{1}{y}\Big)} ,$$
easily derived by noting that $\rho_{0,0,2}=-1$. For later use, we define
\begin{equation}\label{dfunc}
D_{r,s,t}(x,y,z):=1+z^2-z\al^{r^2-t^2}\Big(x+\frac{1}{x}\Big)-z\al^{s^2-t^2}\Big(y+\frac{1}{y}\Big) .
\end{equation}
This denominator will govern the arctic phenomenon for the pinecones.

More generally, we define the refined densities for $m=0,1,...,2t-1$:
$$\rho^{(i_0,j_0,k_0)}_m(x,y,z):=\sum_{i,j\in \Z,\,k\in \Z_+\atop \mu(i,j,k)\geq 0,\ \mu(i,j,k)=m \, [2t]} \rho_{i,j,k}^{(i_0,j_0,k_0)} \, x^i\, y^j\, z^k .$$
All these functions can be obtained from the density $\rho^{(i_0,j_0,k_0)}$ via the following:

\begin{lemma}\label{reflem}
Setting $\omega=e^{i\frac{\pi}{t}}$, we have the identity
$$\rho^{(i_0,j_0,k_0)}_m(x,y,z)=\frac{1}{2t} \sum_{\ell=0}^{2t-1} \omega^{-m\ell} \, \rho^{(i_0,j_0,k_0)}(x\omega^{r\ell},y\omega^{s\ell},z\omega^{t\ell})$$
\end{lemma}

\begin{proof} We compute 
\begin{align*}
	\frac{1}{2t} \sum_{\ell=0}^{2t-1} \omega^{-m\ell} \, \rho^{(i_0,j_0,k_0)}(x\omega^{r\ell},y\omega^{s\ell},z\omega^{t\ell})&=\frac{1}{2t} \sum_{\ell=0}^{2t-1} \omega^{-m\ell} \left( \sum_{i,j,k}\rho_{i,j,k}^{(i_0,j_0,k_0)}x^iy^jz^k\omega^{ri\ell+sj\ell+tk\ell}\right)\\
	&=\frac{1}{2t} \sum_{\ell=0}^{2t-1} \left( \sum_{i,j,k}\rho_{i,j,k}^{(i_0,j_0,k_0)}x^iy^jz^k\omega^{\ell(ri+sj+tk-m)}\right)\\
	&=\rho^{(i_0,j_0,k_0)}_m(x,y,z)
\end{align*}
where in the last line we have used the identity $\delta_{x,m\, [2t]}=\frac{1}{2t}\sum_{\ell=0}^{2t-1} \omega^{\ell(x-m)}$
\end{proof}

Note that if $r,s,t$ are all odd integers, then $\mu(i,j,k)$ only takes even integer values, due to the condition $i+j+k=0\, [2]$, hence only 
even $m$'s contribute to the refined densities. 
This simplifies the study of solutions to \eqref{recurho}, which we postpone to the end of this section.

Assume $r,s,t$ are not all odd.
In this case, we may find a triple of integers $u,v,w\in \Z^2$ such that $r u+s v+t w=1$ and $u+v+w=0\, [2]$. Then $\mu(mu,mv,mw)=m$. 
Assume that $\mu(i,j,k)\geq 0$ and $\mu(i,j,k)=m$ $[2t]$.
Consider the integer $k_{i,j}$ such that $\mu(i,j,k_{i,j})=m$ (and therefore $k-k_{i,j}\in 2\Z$), so that we have
$\mu(mu-i,mv-j,mw-k_{i,j})=0$.
%and for any $(i,j,k)$ such that $\mu(i,j,k)=m\, [2t]$, we may write $\mu(i-m\al,j-m\beta,\epsilon-m\gamma)=0 \, [2t] $ for $\epsilon= k\, [2]$.
We may use the translational invariance \eqref{transinv} for the translation vector $(mu-i,mv-j,mw-k_{i,j})$ to rewrite
$$\rho_{i,j,k}^{(i_0,j_0,k_0)}=\rho_{mu,mv,mw+k-k_{i,j}}^{(mu+i_0-i,mv+j_0-j,mw+k_0-k_{i,j})} ,$$
where the values of the initial data $t_{a,b}=t_{a+mu,b+mv}$ are unchanged, as the translation is parallel to the initial data planes.
This equation allows us to  re-interpret the generating function $\rho^{(i_0,j_0,k_0)}_m(x,y,z)$ as that of the local dimer densities of the dimer model for
%$T_{m\al,m\beta,m\gamma+k-\epsilon}$. 
$T_{mu,mv,mw+k-k_{i,j}}$.
Here $k$ governs the size of the dimer graph and the coordinates $i,j$ allow to explore its faces at positions
%$(m\al+i_0-i,m\beta+j_0-j)$. 
$(mu+i_0-i,mv+j_0-j)$. 

More precisely, we may rewrite the generating function:
\begin{eqnarray*}
\rho^{(i_0,j_0,k_0)}_m(x,y,z) &=&\sum_{i,j,k\in \Z\atop \mu(i,j,k)\geq0,\ \mu(i,j,k)=m\, [2t]} \rho_{mu,mv,mw+k-k_{i,j}}^{(mu+i_0-i,mv+j_0-j,mw+k_0-k_{i,j})}\, x^i y^j z^k\\
&=&  x^{mu+i_0}\, y^{mv+j_0}\,\sum_{i',j'\in \Z,\ k'\in \Z_+ }\rho_{mu,mv,mw+2k'}^{(i',j',k_0')} x^{-i'} \, y^{-j'}\, z^{k_{i,j}+2k'} 
\end{eqnarray*}
where $i'=mu+i_0-i$, $j'=mv+j_0-j$, and $k_0'$ is such that $p=\mu(i_0,j_0,k_0)=\mu(i',j',k_0')$, and where we have expressed $k=k_{i,j}+2k'$, so that the summation is over
integers $k'$ such that $m+2t k'\geq 0$, i.e. $k'\geq 0$ as $0\leq m<2t$.
Using finally $k_{i,j}=(m-ri-sj)/t=(m-mru-msv-ri_0-sj_0+ri'+sj')/t$, we arrive at
$$\rho^{(i_0,j_0,k_0)}_m(x,y,z)=x^{mu+i_0}\, y^{mv+j_0}\,z^{mw+k_0-p/t} \,\sum_{i',j'\in \Z,\ k'\in \Z_+}\rho_{mu,mv,mw+2k'}^{(i',j',k_0')} (z^{r/t}x^{-1})^{i'} \, (z^{s/t}y^{-1})^{j'}\, z^{2k'}$$
This implies that the generating function $\tilde \rho_{p,m} (x, y, z):=\sum_{i',j',k'}  \rho_{mu,mv,mw+2k'}^{(i',j',k_0')} x^{i'}y^{j'}z^{2k'}$
is expressed as
\begin{equation}
\label{goodrho}
\tilde \rho_{p,m} (x, y, z) =x^{mu+i_0}y^{mv+j_0} z^{-m/t} \, \rho^{(i_0,j_0,k_0)}_m(z^{r/t}x^{-1},z^{s/t}y^{-1},z) .
\end{equation}

This generating function however only explores points with 
fixed value of  $\mu(i_0,j_0,k_0)=p$,
and we need to consider values of $(i_0,j_0,k_0)$
pertaining to the different planes $(P_0),(P_1)$,...,$(P_{2t-1})$ to explore all the faces of the dimer graph. To this end, we must compute the
generating function $\rho^{(i_0,j_0,k_0)}(x,y,z)$ for all values $\mu(i_0,j_0,k_0)=p\in [0,2t-1]$. We have the following:\hfill\newline
\begin{thm}\label{rhopthm}
For all triples $(i_0,j_0,k_0)$ such that $\mu(i_0,j_0,k_0)=p\in [0,2t-1]$, we have $\rho^{(i_0,j_0,k_0)}(x,y,z)=x^{i_0}y^{j_0}z^{k_0}\, \rho_p(x,y,z)$,
where for $r\leq s<t$:
$$\rho_p(x,y,z)=\frac{1}{D(x,y,z)}\left\{ \begin{matrix} 
1-z\al^{r^2-t^2}\Big(x+\frac{1}{x}\Big)-z\al^{s^2-t^2}\Big(y+\frac{1}{y}\Big) & (0\leq p<t-s) \\
1-z\al^{r^2-t^2}\Big(x+\frac{1}{x}\Big)-z\al^{s^2-t^2}\Big(\frac{1}{y}\Big) & (t-s\leq p<t-r) \\
1-z\al^{r^2-t^2}\Big(\frac{1}{x}\Big)-z\al^{s^2-t^2}\Big(\frac{1}{y}\Big) & (t-r\leq p<t+r) \\
1-z\al^{r^2-t^2}\Big(\frac{1}{x}\Big) & (t+r\leq p<t+s) \\
1 & (t+s\leq p<2t)
\end{matrix} \right. ,$$
and for $s\leq r<t$:
$$\rho_p(x,y,z)=\frac{1}{D(x,y,z)}\left\{ \begin{matrix} 
1-z\al^{r^2-t^2}\Big(x+\frac{1}{x}\Big)-z\al^{s^2-t^2}\Big(y+\frac{1}{y}\Big) & (0\leq p<t-r) \\
1-z\al^{r^2-t^2}\Big(\frac{1}{x}\Big)-z\al^{s^2-t^2}\Big(y+\frac{1}{y}\Big) & (t-r\leq p<t-s) \\
1-z\al^{r^2-t^2}\Big(\frac{1}{x}\Big)-z\al^{s^2-t^2}\Big(\frac{1}{y}\Big) & (t-s\leq p<t+s) \\
1-z\al^{r^s-t^2}\Big(\frac{1}{y}\Big) & (t+s\leq p<t+r) \\
1 & (t+r\leq p<2t)
\end{matrix} \right. ,$$
with $D(x,y,z)$ as in \eqref{dfunc}.
\end{thm}
\begin{proof}
Using the recursion relation \eqref{recurho}, we see that the initial data at $(i_0,j_0,k_0)$ propagates to the following points
(with $k-k_0=1,2$):
\begin{itemize}
\item $(i_0,j_0,k_0+2)$ with $\mu=2t+p\geq 2t$ for all $p$; 
\item $(i_0+1,j_0,k_0+1)$ with $\mu=t+r+p\geq 2t$ for $p\geq t-r$;
\item $(i_0,j_0+1,k_0+1)$ with $\mu=t+s+p\geq 2t$ for $p\geq t-s$;
\item $(i_0-1,j_0,k_0+1)$ with $\mu=t-r+p\geq 2t$ for $p\geq t+r$;
\item $(i_0,j_0-1,k_0+1)$ with $\mu=t-s+p\geq 2t$ for $p\geq t+s$;
\end{itemize}
These govern the numerators in the above formulas for the densities $\rho_p$.
\end{proof}

Let us now consider the case when $r,s,t$ are all odd integers. In that case, we may find a triple of integers $(u,v,w)\in \Z^2$ such that $ru+sv+tw=2$,
and $u+v+w=0\, [2]$. As $\mu(i,j,k)$ is even, we write $\mu(i,j,k)=2m$ $[2t]$, and $\mu(i,j,k_{i,j})=2m$
and apply again the same translation invariance. The net result is an even version of \eqref{goodrho}:
\begin{equation}
\label{evengoodrho}
\tilde \rho_{2p,2m} (x, y, z) =x^{mu+i_0}y^{mv+j_0} z^{-2m/t} \, \rho^{(i_0,j_0,k_0)}_{2m}(z^{r/t}x^{-1},z^{s/t}y^{-1},z) .
\end{equation}
where $\mu(mu+i_0-i,mv+j_0-j,mw+k_0-k_{i,j})=\mu(i_0,j_0,k_0)=2p$.
%\begin{equation}
%\frac{1}{2t} \sum_{\ell=0}^{2t-1} x^{p\al}y^{p\beta}z^{p\gamma} \omega^{\ell(p-m)} \rho_{2p}(x\omega^{\ell r},y\omega^{\ell s},z\omega^{\ell t})\Big\vert_{x^i y^jz^k} ,
%\end{equation}
%with $\rho_{2p}$ as in Theorem \ref{rhopthm}, by noting that $\mu(p\al,p\beta,p\gamma)=2p$. Again, to include the local densities on all the faces of the dimer graph for 
%$(m\al,m\beta,m\gamma+k-\epsilon)$, we must consider this quantity for $p=0,1,...,t-1$ and take arbitrary coordinates $i,j\in \Z$, with $i+j=\epsilon\, [2]$.

In all cases, combining the results of Lemma \ref{reflem} and Theorem \ref{rhopthm}, we now have access to the large $i',j',k'$ asymptotics of the local densities  
$\rho_{mu,mv,mw+2k'}^{(i',j',k_0')}$ which are governed by the singularities of their generating function $\tilde \rho_{p,m}(x,y,z)$, namely
the zeroes of their common denominator as $x,y,z$ approach $1$. In all cases the denominator vanishes like
\begin{equation}\label{denom}
\Delta_{r,s,t}(x,y,z):=D_{r,s,t}(z^{r/t}x^{-1},z^{s/t}y^{-1},z) ,\end{equation}
as $x,y,z$ approach $1$. Note that in the scaling limit when $i,j,k\to \infty$ with
$i/k,j/k$ finite, the ``center" $(mu,mv)$ of the dimer domain, which depends on $m\in [0,2t-1]$ scales uniformly to the origin.

\subsection{Arctic phenomenon}
	
\subsubsection{Asymptotics of the density function and arctic curve}\label{asymptotic of the density}

As shown in the previous section the singularities of the density are governed by the zeros of the function $\Delta(x,y,z)$ \eqref{denom}.
We now apply the method of multivariate generating functions by \cite{PW2004} \cite{PW2005} for conical singularities, letting $x \mapsto e^{\epsilon x}$, $y \mapsto e^{\epsilon y}$ and $z \mapsto e^{-\epsilon(ux+ vy)}$ and expanding at leading order in $\epsilon$, we find
$$ \Delta_{r,s,t}(e^{\epsilon x},e^{\epsilon y},e^{-\epsilon (ux+ vy)})=\epsilon^2\, H_{r,s,t}(x,y,z)  +O(\epsilon^4)$$
for some explicit polynomial $H$ of $x,y$ (we drop the subscript $r,s,t$ when there is no ambiguity).
Further imposing that $H(x,y)=\partial_x H(x,y)=\partial_y H(x,y)=0$ has a non-trivial solution in $x,y$ leads to the vanishing condition of the Hessian determinant: 
$\left\vert \begin{matrix}\partial_x^2 H &\partial_x\partial_y H\\  \partial_x \partial_yH&\partial_y^2 H\end{matrix}\right\vert=0$, and finally to the (dual) arctic curve:
\begin{equation}\label{arcticu}
%(1-A)(t^2-A r^2) u^2+A(t^2-(1-A)s^2)v^2-A (1-A)(2r s u v+2r t u+2s t v+t^2)=0
(1-A)\,t^2u^2+A\,t^2v^2-A(1-A)\,(r u+s v+t)^2 =0
\end{equation}
where 
\begin{equation}\label{adef}
A=A_{r,s,t}:=\al^{r^2-t^2},\qquad 1-A=\al^{s^2-t^2}
\end{equation} 
It is easy to show that $t^2-A r^2>0$ and $t^2-(1-A)s^2>0$, while $0<A<1$ generically, so that the curve \eqref{arcticu} is always an ellipse.

Let us also derive the scaling limit of the pinecone domain, centered at the origin. In the original $(i,j,k)$ coordinates, it is the intersection of the pyramid $|x|+|y|=|k-z|$
and of the slanted initial data planes $P_m$: $r x+s y+t z=m$ $m=0,1,...,2t-1$. After rescaling by $k$, setting $u=x/k$, $v=y/k$ and $w=z/k$, we get for $k\to \infty$:
$|u|+|v|=|1-w|$ and $r u +s v +t w =0$. The resulting 4 equations $t(1 \pm u\pm v)+ru+s v=0$ give rise to 4 segments\footnote{The four corresponding segments are the images of the four segments (\ref{boundonew}-\ref{boundtwow}) of the pinecone formulation.}:
\begin{equation}\label{scaling limit}
	\begin{aligned}
 v&=-\frac{t}{t+s}-\frac{t+r}{t+s}\, u , \qquad v=-\frac{t}{s-t}-\frac{t+r}{s-t}\, u \qquad  \left(u\in\left[-\frac{t}{t+r},0\right]\right)\\
 v&=-\frac{t}{t+s}+\frac{t-r}{t+s}\, u, \qquad v=-\frac{t}{s-t}+\frac{t-r}{s-t}\, u \qquad \left(u\in \left[0,\frac{t}{t-r}\right]\right) .
	\end{aligned}
\end{equation}
We summarize the results into the following:
\begin{thm}
The limit shape of typical large size (r,s,t)-pinecone domino tilings associated to the solution of the $T$-system with uniform initial data $t_{i,j}=\al^{m(m-1)/2}$ on each slanted plane $m=ri+sj+tk = 0,1, \cdots, 2t-1$, is the ellipse \eqref{arcticu} 
inscribed in the scaling domain \eqref{scaling limit}. This ``arctic" ellipse separates a liquid phase (center) from four frozen crystalline phases (corners).
\end{thm}

\subsubsection{Examples}
\begin{figure}[H]
\begin{center}
\begin{minipage}{0.33\textwidth}
        \centering
        \includegraphics[scale=.5]{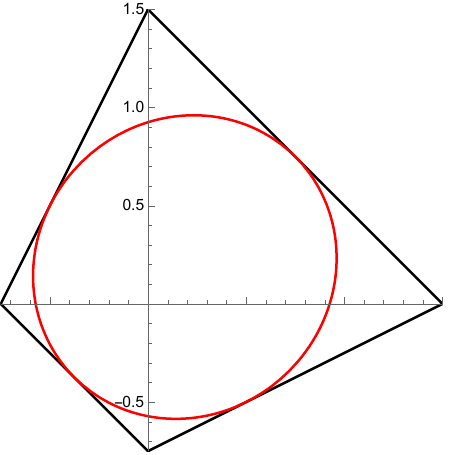}  % second figure itself
        %\caption{second figure}
    \end{minipage}
\begin{minipage}{0.33\textwidth}
        \centering
        \includegraphics[scale=.5]{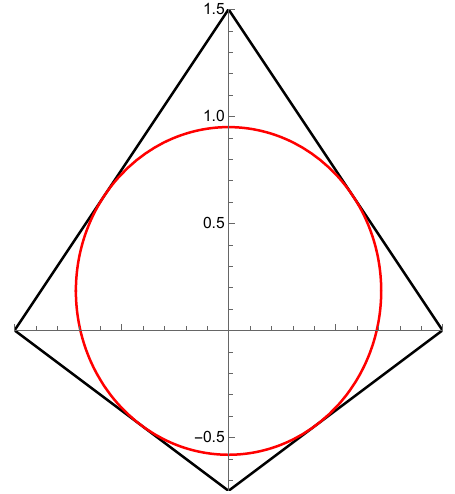} % first figure itself
        %\caption{first figure}
    \end{minipage}\hfill
    \begin{minipage}{0.33\textwidth}
        \centering
        \includegraphics[scale=.5]{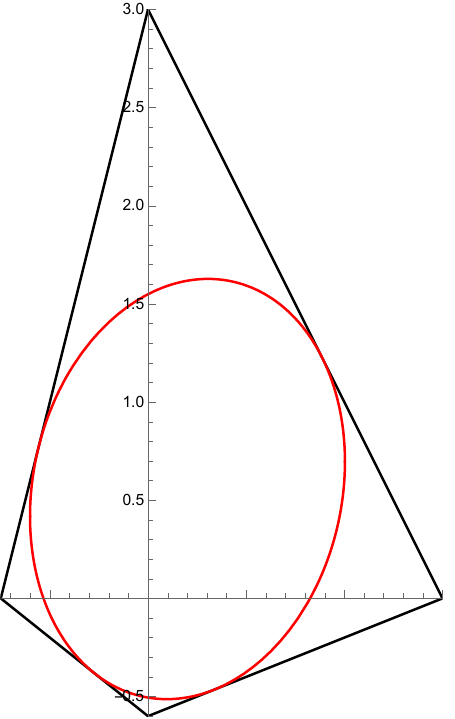}  % second figure itself   
        %\caption{second figure}
    \end{minipage}
\end{center}
\caption{\small Arctic curves for $r,s,t$-pinecones, together with the scaled domain. Left: $(r,s,t)=(1,1,3)$, center: $(r,s,t)=(0,1,3)$, right $(r,s,t)=(1,2,3)$ .}
\label{fig:arcticex}
\end{figure}

\begin{example}{\bf Case $t=1$, $r=s=0$.} In this case, we have $\al=2$ and $A=1-A=\frac{1}{2}$. This is the case of ``flat" initial data planes $k_{i,j}={\rm Mod}(i+j,2)\in \{0,1\}$, for which the pinecone dimer configurations reduce to domino tilings of the Aztec diamond
of size $k$. The corresponding arctic curve is the celebrated arctic circle $u^2+v^2-\frac{1}{2}=0$, inscribed in the rescaled domain, the square $|u|+|v|=1$.
\end{example}

\begin{example}{\bf Cases $r=s<t$.}\label{examplers} In this case, we have $\al=2^{\frac{1}{t^2-r^2}}$ and $A=1-A=\frac{1}{2}$ again. The arctic ellipse takes the simple form:
$$ t^2 (u-v)^2 +(t^2-r^2)\left(u+v-\frac{r t}{t^2-r^2}\right)^2 =\frac{t^4}{t^2-r^2} .$$
This curve is displayed in Fig. \ref{fig:arcticex} (left) for $r=s=1$ and $t=3$.
\end{example}

\begin{example}{\bf Cases $r=0<s<t$.}  We have $\al^{t^2}=1+\al^{s^2}$, $A=\al^{-t^2}$, and the arctic ellipse reads:
$$ (1-A)t^2 u^2 +A(t^2-(1-A)s^2)\left( v-\frac{(1-A)s t}{t^2-(1-A)s^2}\right)^2= \frac{A(1-A)t^4}{t^2-(1-A)s^2} . $$
This curve is displayed in Fig. \ref{fig:arcticex} (center) for $r=0$, $s=1$, and $t=3$.\newline
%\color{red} [xxx change this, and explain the figure] \color{black}
\end{example}

Recall that the local densities $\rho_{mu,mv,mw+2k}^{(x,y,k_0')}$ measure of the expectation value, within the statistical ensemble of dimer configurations of the pinecone domain of size $k$, of the observable  $v_{x,y}/2-1-\mathcal N_{x,y}(\mathcal D)$, where $\mathcal N_{x,y}(\mathcal D)$ is the number of dimers occupying the edges around the face $(x,y)$ of the domain $\mathcal D$, and $v_{x,y}$ is the valency of the face $(x,y)$. 
The 4 corners of the scaled domain have vanishing density, indicating a "frozen configuration", where each face is occupied by 1 (resp. 2) edge(s) for square faces (reps. hexagonal faces), resulting in ${v}_{x,y}/2-1-\mathcal N_{x,y}=0$ in both cases. 

For completeness, we have computed numerically the value of the local density $\rho_{mu,mv,mw+2k}^{(i,j,k_0')}$, 
which is the
coefficient of $x^iy^jz^k$ of the density generating function $\tilde \rho_{p,m}(x,y,z)$ of eqs. \eqref{goodrho} and \eqref{evengoodrho}, in domains centered at the origin, for large $k$, as functions of $(i/k,j/k)$.  
The corresponding plots for the three cases $(r,s,t)=(1,2,3),(1,1,3),(1,0,3)$ (with darker shades for larger values of $\rho_{i,j,k}$) are represented in Fig.~\ref{mathematica_simulation}, for large values of $k$. 
The domain of non-zero values of $\rho$ show the arctic ellipses, outside of which $\rho\to 0$ (white zone delimited
by the segments \eqref{scaling limit}).
%We display in Fig.~\ref{mathematica_simulation} the density profile for the value of $\rho$ up to extracting the coefficient of $\rho(x,y,z)$
\begin{figure}
\begin{center}
\begin{subfigure}{0.33\textwidth}
        \centering
        \includegraphics[width=4.cm]{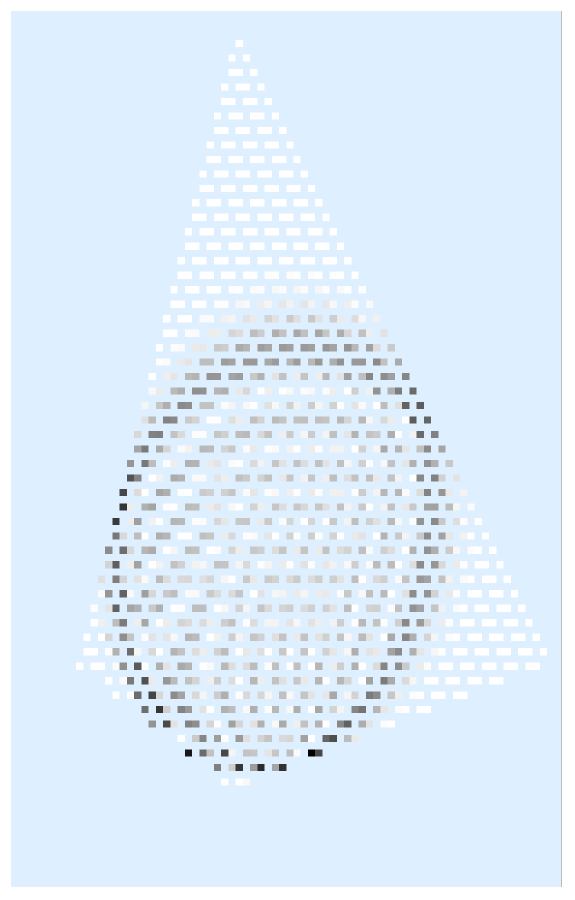}  % second figure itself
        \caption{\tiny $(r,s,t)=(1,2,3), k=90$}
    \end{subfigure}
\begin{subfigure}{0.33\textwidth}
        \centering
        \includegraphics[width=4.cm]{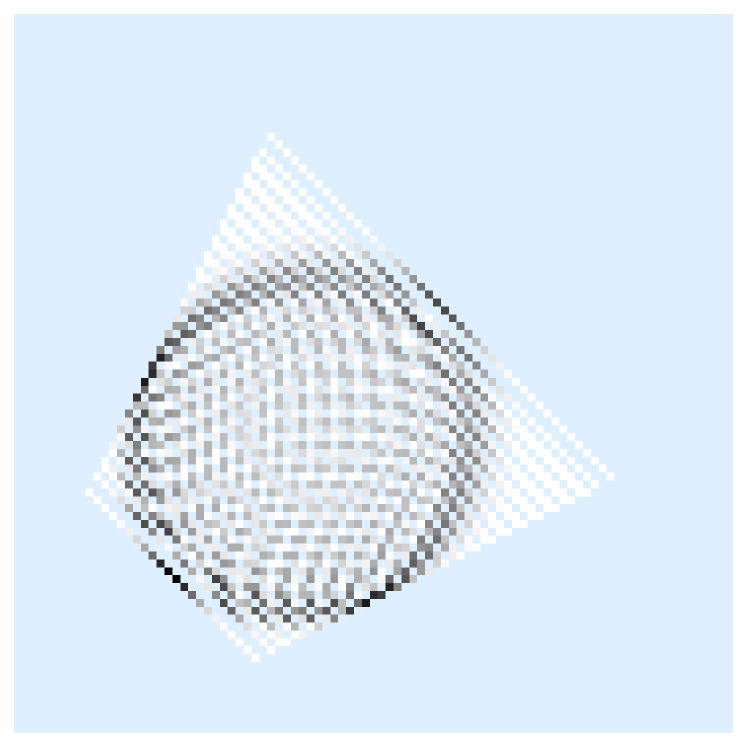} % first figure itself
        \caption{\tiny $(r,s,t)=(1,1,3), k=90$}
    \end{subfigure}\hfill
    \begin{subfigure}{0.33\textwidth}
        \centering
        \includegraphics[width=4.cm]{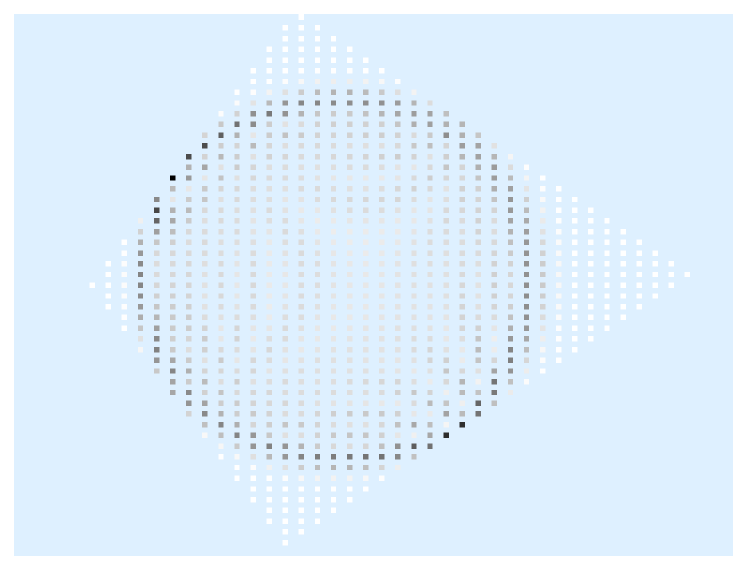}  % second figure itself   
        \caption{\tiny $(r,s,t)=(1,0,3),k=150$}
    \end{subfigure}
\end{center}
\caption{\small Density profile for $r,s,t$-pinecones obtained by extracting the coefficients of $x^iy^jz^{k}$ in 
$\tilde \rho_{p,m}(x,y,z)$ and plotting the result as a function of $i/k,j/k$ for fixed large $k$, $i,j$ in a suitable range around the origin. For better statistics, graphs for a few neighboring values of $k$ are superposed. The lighter color corresponds to the smaller value of $\rho_{i,j,k}$, and the light blue color indicates when  $\rho_{i,j,k}$ is strictly $0$.}
\label{mathematica_simulation}
\end{figure}
%In particular, we see clearly the four ``corner" domains where $\rho\to 0$ (in light blue), delimited
%by the segments \eqref{scaling limit}. 
The corresponding white-colored corners correspond to fundamental crystalline states as mentioned above.
There are four distinct such states, each corresponding to a (N,S,E,W) corner similar to the case of the Aztec Diamond in \cite{DiFrancesco1}, characterized by an occupation number $\mathcal N_{x,y}(\mathcal D)=1$ (resp. $2$) on each square (resp. hexagonal) face $(x,y)$  (see Fig. \ref{explain_frozen} for an illustration). Thus, away from the corners, the dimer model has non-trivial entropy with $\rho\neq 0$, indicating the liquid phase (darker shades) in Fig \ref{mathematica_simulation}.

%The physical explanation for the arctic curve in the uniform initial data case is that the dimer configurations contributing to the partition function $T_{i,j,k}$ tend to be in a fundamental crystalline state in the vicinity of the corners of the scaled domain, separated from the "liquid" region by the 4 line segments in \ref{scaling limit}. There are four distinct such states, each corresponding to a (N,S,E,W) corner similar to the case of the Aztec Diamond in \cite{DiFrancesco1}, characterized by an occupation number $\mathcal N_{x,y}(\mathcal D)=1 \ \text{or}\ 2$ on each face $(x,y)$. Away from the corner, the dimer model has non-trivial entropy. 
% \begin{center}
% 		\tikzfig{tikz/frozen_explain}
% \end{center}
\begin{figure}
\centering
\input{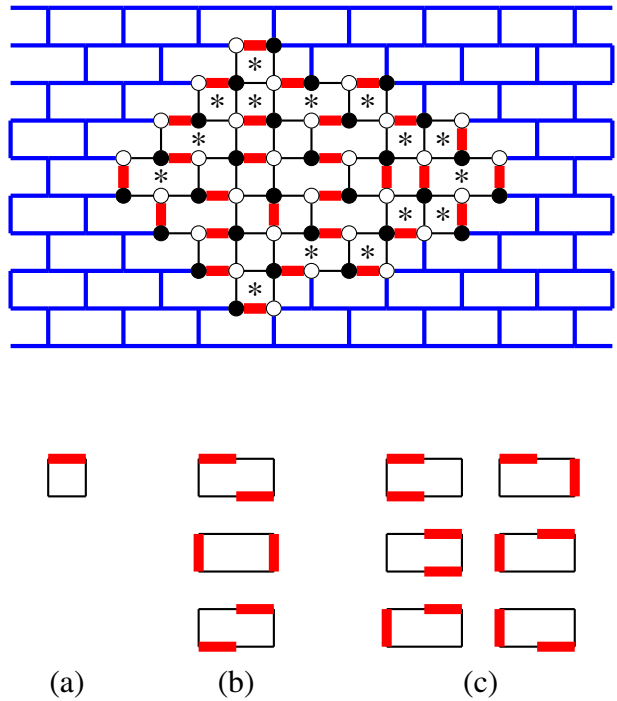}
\caption{The four frozen dimer configurations (up to rotation) in the four corners of the $(r,s,t)=(1,0,3)$-pinecone dual graph. In a larger case, one should see that these frozen facets concentrate at the $4$ corners of the dual graph. The blue regions of brickwall configurations stay frozen outside of the "active" zone }
\label{explain_frozen}
\end{figure}
\begin{remark}\label{delete_crystalline_remark}
We expect the four corners of the quadrangular scaling domain of $(r,s,t)$ pinecones to be in a crystalline phase, where each square face is singly occupied and each hexagonal face is doubly occupied (see Fig. \ref{explain_frozen}, where we also listed the various (a) square (up to rotation)  and (b),(c) hexagonal face configurations).
In general, the $(r,s,t)$ pinecones can be drawn on a square lattice, with missing vertical edges corresponding to hexagons. The crystal phases of the Aztec diamond graph
(with only square faces) are either of the 4 ``brick wall" configurations (2 horizontal and 2 vertical obtained from each other by respectively a horizontal/vertical translation by one unit) depicted up to translation in Fig. \ref{crystal_hex}  (left). Three of these four give rise to  admissible crystal phases of the pinecone drawn on the same lattice (see Fig. \ref{crystal_hex}  (right) for the case (1,1,3), where each row is a succession of a sequence of two squares followed by one horizontal hexagon, and successive rows are shifted by one unit to the right). Indeed the condition of each square being singly occupied and each hexagon being doubly occupied is obviously satisfied. In general, we expect the bulk of the crystal phases to be in analogues of the three abovementioned states, governed by either of the three hexagon configurations of Fig \ref{explain_frozen} (b). Outside of the scaled domain, every faces are hexagonal and oriented in such ways that the contributing weight is alway $0$.
	\begin{figure}[H]
		\centering
		\includegraphics[scale=.25]{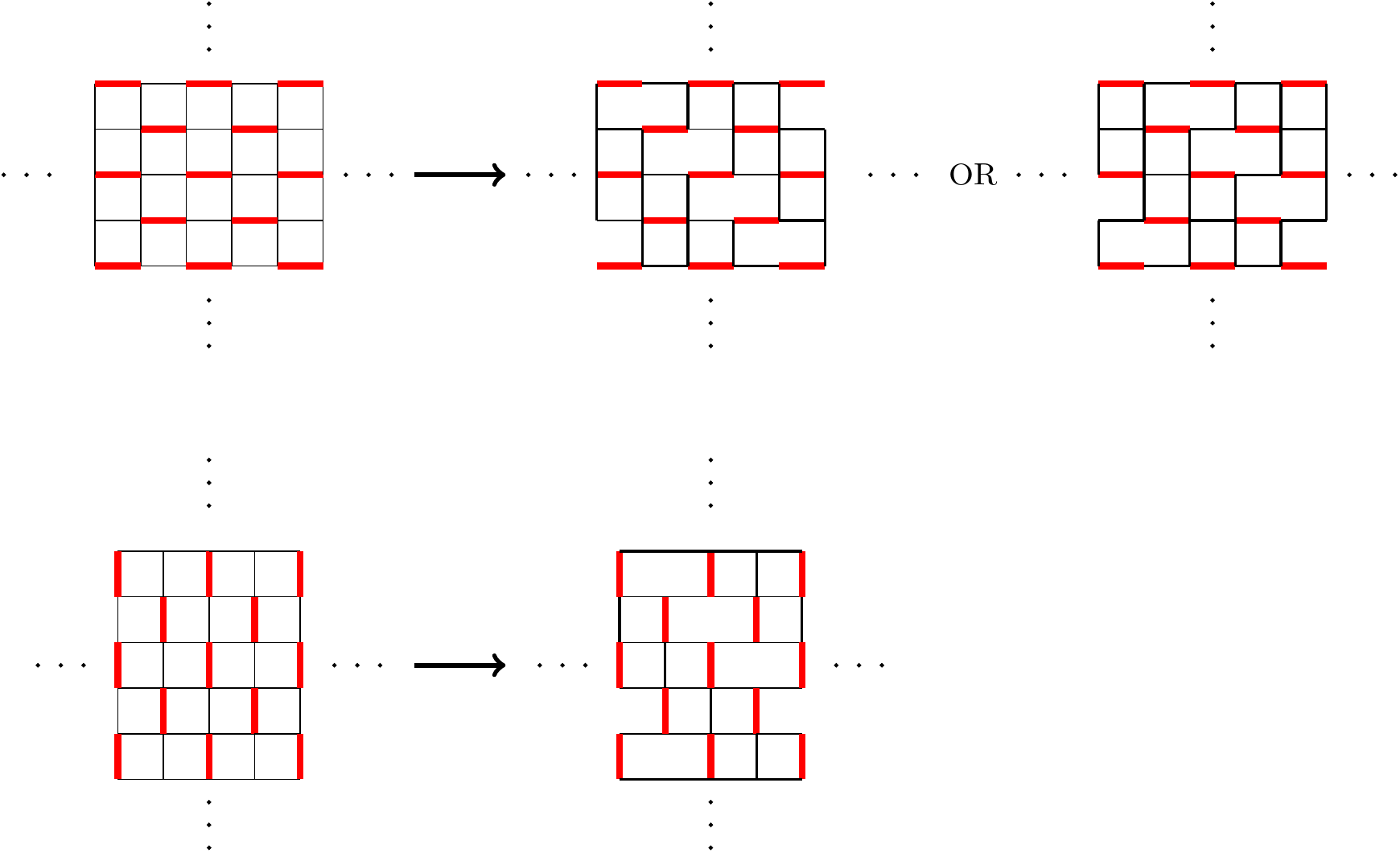}
		\caption{From horizontal/vertical ``brick wall" crystalline phases on the square lattice (left) to that on pinecones (right), illustrated here in the case (1,1,3), where the corresponding stepped surfaces are drawn on a square lattice with missing vertical edges (middle of hexagons). Out of the four possible brick wall crystals on the square lattice, only three are compatible with the hexagon arrangement.}
		\label{crystal_hex} 
	\end{figure}
\end{remark}

\section{The case of slanted $2 \times 2$-periodic Initial Data: solution and periodicities}
In this section, we will focus on a specific initial data on the stack of slanted $(r,s,t)=(r,r,t)$-planes, which has periodicity two in both $i$ and $j$ direction, namely, on each slanted plane $P_m$, $m \in \{0,1, \cdots, 2t-1\}$:
\begin{equation}\label{twobytwo}
t_{i,j}=\al^{m(m-1)/2}\, \times \, \begin{cases}
	a \hskip 0.5cm (i=0,j=0 \ mod \ 2) \\
	b \hskip 0.5cm (i=0,j=1 \ mod \ 2 )\\
	c \hskip 0.5cm (i=1,j=0 \ mod \ 2) \\
	d \hskip 0.5cm (i=1,j=1 \ mod \ 2) \\
\end{cases}
\end{equation}
such that $ri+sj+tk=m$, $i+j+k=0 \ mod \ 2$, and with  $\al=2^{\frac{1}{t^2-r^2}}$ as in Example \ref{examplers}.
This initial data determines entirely the solution of the $T$-system, and its built-in periodicity induces a periodicity property of the quantities $L_{i,j,k}$ and $R_{i,j,k}$
which allows to derive the density function exactly, as we show in the sections below. 

\subsection{The case of $r=s$ and $t$ in $2\Z+1$}
\label{sec:oddodd}

When $r=s$ and $t$ are odd coprime integers, the condition $r(i+j)+t k=\mu$, together with the fact that $i+j+k=0$ mod 2, forces $\mu$ to be an even integer. Indeed,
the quantity $ r(i+j)+t k= r(i+j+k)+(t-r) k $ must be even as both $t-r$ and $i+j+k$ are even. In particular, 
only even planes $P_{2m}$, $m\in [0,t-1]$, contain initial data points, and the solution is defined on even planes as well. 
Writing $r(i+j+k)+(t-r) k=2m$, we see that $k$ is constrained by the relation $\frac{t-r}{2}\, k =m$ mod $r$. As $r$ and $t$ are coprime this is easily solved as
$k=\theta\, m$ mod $r$, where $\theta\in [0,r-1]$ denotes the inverse of $\frac{t-r}{2}$ modulo $r$.

This suggests to apply a change of variables $(i,j,k)\to (i,k,m)$, with $m=r\,\frac{i+j+k)}{2}+\frac{t-r}{2}\, k$, and $k=\theta \, m$ mod $r$, which allows to recover $j=(m-\frac{t-r}{2}\, k)/r$.
Accordingly, we write $T_{i,j,k}=T^{2m}_{i,k}$. For these new variables, the initial data $t'_{i,k}$ is naturally indexed by pairs $i,k\in \Z$, in bijection with the initial data $t_{i,j}$. Indeed, from the discussion of Section 2.1, the stepped surface of initial data is $(i,j,k_{i,j})$ with $k_{i,j}$ as in \eqref{kdef},
and we have the initial data assignments $t_{i,j}=T_{i,j,k_{i,j}}=T_{i,k_{i,j}}^{2m_{i,j}}=t'_{i,k_{i,j}}$ \eqref{init},  where 
$2m_{i,j}={\rm Mod}(r(i+j)+t\,{\rm Mod}(i+j,2),2t)$. In the particular case of $2\times 2$ periodic initial data \eqref{twobytwo}, with $t_{i+2,j}=t_{i,j+2}=t_{i,j}$,
we have the following.

\begin{lemma}\label{periodic lemma}
For every $i,k$, the following periodicity relations hold for all initial data planes $m=0,1,...,t-1$: 
\begin{align*}
T^{2m}_{i+2,k}=T^{2m}_{i,k}\\
T^{2m}_{i,k+2r}=T^{2m}_{i,k}
\end{align*}
\end{lemma}
\begin{proof}
The first relation follows from the fact that $(i+2,j,k)$ and $(i,j+2,k)$ belong to the same plane $2m=r(i+j+2)+t k$,
therefore share the same value of $t_{i+2,j}=t_{i,j+2}=t_{i,j}$ and therefore $T^{2m}_{i+2,k}=T^{2m}_{i,k}$.
For the second relation, note that the points $(i+2t,j,k)$ and $(i,j,k+2r)$ belong to the same plane with
$2m=r(i+j)+t k+2 r t$, therefore share the same value of $t_{i+2t,j}=t_{i,j}$, hence $T^{2m}_{i+2t,k}=T^{2m}_{i,k+2r}=T_{i,k}^{2m}$.

\end{proof}

Recall that for fixed $m$, $k$ must satisfy $k=\theta m$ mod $r$. From Lemma \ref{periodic lemma}, $T^{2m}_{i,k}$
depends only on $k$ mod $2r$, which takes only two values $k_0$ and $k_1$, where 
$k_0={\rm Mod}(\theta m,r)\in [0,r-1]$ and $k_1=k_0+r \in [r,2r-1]$. Moreover, as $r$ is odd, $k_0$ and $k_1$ have opposite parities, and we deduce that $T^{2m}_{i,k}=t'_{i,k}$ only depends on $k$ mod 2  (as well as $i$ mod 2 from the Lemma).
This results in the correspondence of initial data: $t'_{i, Mod(i+j,2)}=t_{i,j}$, namely on each plane $P_{2m}$:

 	\begin{equation}\label{initdataik}
 	\begin{matrix}
 		t_{0,0}=t'_{0,0}=\al^{m(2m-1)}\, a\\
 		t_{0,1}=t'_{0,1}=\al^{m(2m-1)}\, b\\
 		t_{1,0}=t'_{1,1}=\al^{m(2m-1)}\, c\\
 		t_{1,1}=t'_{1,0}=\al^{m(2m-1)}\, d 	
 	\end{matrix}
 	\end{equation}

Using the change of variables $T_{i,j,k}\to T_{i,k}^{2m}$, we may rewrite the quantities $L$ of \eqref{defLR}
as
$$L_{i,j,k}=L^{2m}_{i,k}=\frac{T^{2m+r-t}_{i+1,k-1}T^{2m-r-t}_{i-1,k-1}}{T^{2m-2t}_{i,k-2}T^{2m}_{i,k}}$$

The next theorem is the key of our  study of the arctic phenomenon for $2\times 2$-periodic initial data.

\begin{thm}\label{ratiothm} 
The solution of the $T$-system with $2 \times 2$-periodic initial data has the following periodicity properties:
\begin{enumerate}
\item $L^{2m}_{i,k}$, for $k=\theta m$ mod $r$, depends only on $i,k$ modulo $2$.
\item $L^{2m+4t}_{i,k}=L^{2m}_{i,k}$.
\end{enumerate}
 \end{thm}
\begin{proof}
The proof uses the uniqueness of the solution $T^{2m}_{i,k}$ of the $T$-system subject to the initial data. Let us define recursively a new variable $\theta^{2m}_{i,k}$ by:
\begin{eqnarray}\theta^{2m+r+t}_{i+1,k+1}&=&\theta^{2m}_{i,k}\, f_{i,k}(m) \qquad (i,k\in \Z;m\geq 0)\label{rec}\\
\theta^{2m}_{i,k}&=&t'_{i,k}, \qquad \left(i,k\in \Z; m\in \left[0, \frac{t+r-2}{2}\right]\right), \label{initofrec}
\end{eqnarray}
where the functions $f_{i,k}(m)$ depend only on $i,k$ mod $2$ with
 	\begin{equation}\label{quantityf}
 		\begin{aligned}
 			f_{0,0}(m)&=\al^{\frac{r+t}{2}(4m+r+t-1)}\,\frac{\left(\frac{a^2+d^2}{2}\right)^{\eta(m)}\times \left(\frac{b^2+c^2}{2}\right)^{\eta(m-\frac{t-r}{2})}}{a^{\mu'(m)}b^{\mu(m-\frac{t-r}{2})}c^{\mu'(m-\frac{t-r}{2})}d^{\mu(m)}}\\
 			f_{1,0}(m)&=\al^{\frac{r+t}{2}(4m+r+t-1)}\,\frac{\left(\frac{a^2+d^2}{2}\right)^{\eta(m)}\times \left(\frac{b^2+c^2}{2}\right)^{\eta(m-\frac{t-r}{2})}}{a^{\mu(m)}b^{\mu'(m-\frac{t-r}{2})}c^{\mu(m-\frac{t-r}{2})}d^{\mu'(m)}}=f_{0,0}(m)\big|_{a \leftrightarrow d, b \leftrightarrow c}\\
 			f_{0,1}(m)&=\al^{\frac{r+t}{2}(4m+r+t-1)}\,\frac{\left(\frac{a^2+d^2}{2}\right)^{\eta(m-\frac{t-r}{2})}\times \left(\frac{b^2+c^2}{2}\right)^{\eta(m)}}{a^{\mu(m-\frac{t-r}{2})}b^{\mu'(m)}c^{\mu(m)}d^{\mu'(m-\frac{t-r}{2})}}=f_{0,0}(m)\big|_{a \leftrightarrow b, c \leftrightarrow d}\\
 			f_{1,1}(m)&=\al^{\frac{r+t}{2}(4m+r+t-1)}\,\frac{\left(\frac{a^2+d^2}{2}\right)^{\eta(m-\frac{t-r}{2})}\times \left(\frac{b^2+c^2}{2}\right)^{\eta(m)}}{a^{\mu'(m-\frac{t-r}{2})}b^{\mu(m)}c^{\mu'(m)}d^{\mu(m-\frac{t-r}{2})}}=f_{0,0}(m)\big|_{a \leftrightarrow c, b \leftrightarrow d}
 		\end{aligned}
 	\end{equation}
and 
 		\begin{align*}\label{eta_mu}
 			\eta(m)&=\left\lfloor\frac{m}{t}\right\rfloor+\left\lfloor\frac{m-\frac{t-r}{2}}{t-r}\right\rfloor-\left\lfloor\frac{m-(t-r)}{t-r}\right\rfloor\\
 			\mu(m)&=2\left\lfloor\frac{m+t}{2t}\right\rfloor\\
 			\mu'(m)&=2\left\lfloor\frac{m}{2t}\right\rfloor+1 .
 		\end{align*}
Our aim is to show that $\theta_{i,k}^{2m}=T_{i,k}^{2m}$ for all $i,k\in \Z$ and $m\geq 0$. 
By construction the initial values for $0\leq m\leq \frac{t+r}{2}-1$
coincide. To show the identity between $\theta_{i,k}^{2m}$ and $T_{i,k}^{2m}$, we must show that they agree on the remaining initial data planes, and that moreover they obey the same $T$-system relations. This is the content of the following lemma:
 	\begin{lemma}
 	\begin{enumerate}
 		\item $\theta^{2m}_{i,k}=T^{2m}_{i,k}$ for all $\ds \frac{r+t}{2}\leq m\leq t-1$
 		\item The 2 quantities
 		\begin{equation*}
 			{\mathcal L}^{2m}_{i, k}=\frac{\theta^{2m+r-t}_{i+1,k-1}\theta^{2m-r-t}_{i-1,k-1}}{\theta^{2m-2t}_{i,k-2}\theta^{2m}_{i,k}}
			 \ \ \ {and}\ \ \ 
 			 {\mathcal R}^{2m}_{i,k}=\frac{\theta^{2m+r-t}_{i,k-1}\theta^{2m-r-t}_{i,k-1}}{\theta^{2m-2t}_{i,k-2}\theta^{2m}_{i,k}}
 		\end{equation*}
 		can be expressed as ratios of $f$'s, which satisfy ${\mathcal L}^{2m}_{i, k}+{\mathcal R}^{2m}_{i,k}=1$.
 	\end{enumerate}
 	\end{lemma}
	\begin{proof}

First note that as both the initial data $t'_{i,k}$ and the transition functions $f_{i,k}(m)$ only depend on $i,k$ mod $2$, so does  $\theta_{i,k}^{2m}$ for all $m\geq 0$.
It is sufficient to check (1) in the case where $i,k = 0 \ [2]$, as we can reach the other cases by a permutation of the variables $(a,b,c,d)$ (as is clear from \eqref{quantityf}). For $0\leq m<\frac{t-r}{2}$, we have:
 		\begin{align*}
 			\theta^{2m+r+t}_{0,0}=\theta^{2m}_{1,1}\frac{\left(\frac{a^2+d^2}{2}\right)^{\eta(m)}\times \left(\frac{b^2+c^2}{2}\right)^{\eta(m-\frac{t-r}{2})}}{a^{\mu'(m)}b^{\mu(m-\frac{t-r}{2})}c^{\mu'(m-\frac{t-r}{2})}d^{\mu(m)}}=c\times\frac{\left(\frac{a^2+d^2}{2}\right)^{\eta(m)}\times \left(\frac{b^2+c^2}{2}\right)^{\eta(m-\frac{t-r}{2})}}{a^{\mu'(m)}b^{\mu(m-\frac{t-r}{2})}c^{\mu'(m-\frac{t-r}{2})}d^{\mu(m)}}
 		\end{align*}
 		We must compare this to $T^{2m+r+t}_{0,0}=a$ for $0\leq m \leq \ds \frac{t-r-2}{2}$. The desired identification follows from the relations $\ds \eta(m)=\eta(m-\frac{t-r}{2})=\mu(m-\frac{t-r}{2})=\mu(m)=0$, $\mu'(m-\frac{t-r}{2})=1$ and $\mu'(m)=-1$, all valid for $0\leq m \leq \ds \frac{t-r-2}{2}$. \hfill\newline
To show (2), 
first note that we have from \eqref{rec}:
$${\mathcal L}^{2m}_{i, k}=\frac{f_{i,k-2}(m-t)}{f_{i-1,k-1}(m-\frac{r+t}{2})}, $$
which depends only on $i,k$ mod 2.
The quantity ${\mathcal R}^{2m}_{i,k}$ can be easily related to ${\mathcal L}^{2m}_{i, k}$ by noting that it corresponds to an interchange of the roles of $i$ and $j$ in the original variables, which is implemented by the interchange $b\leftrightarrow c$ in the initial data. This gives:
$${\mathcal R}^{2m}_{i,k}={\mathcal L}^{2m}_{i, k}\Big\vert_{b\leftrightarrow c},$$
also depending only on $i$ and $k$ mod 2.
We may restrict ourselves to the case $i=k=0$ mod 2 as all the other parities of $i,k$ are obtained by permuting $a,b,c,d$. 

We compute
\begin{equation}\label{LR00}
{\mathcal L}^{2m}_{0,0}=\frac{f_{0,0}(m-t)}{f_{1,1}(m-\frac{r+t}{2})}=\frac{\left(\frac{b^2+c^2}{2}\right)^{\varphi(m-\frac{t+r}{2})}}{2\, b^{\psi(m-\frac{t+r}{2})}c^{\psi'(m-\frac{t+r}{2})}}, \quad 
 {\mathcal R}^{2m}_{0,0}={\mathcal L}^{2m}_{0,0}\Big\vert_{b \leftrightarrow c}=\frac{\left(\frac{b^2+c^2}{2}\right)^{\varphi(m-\frac{t+r}{2})}}{2\, b^{\psi'(m-\frac{t+r}{2})}c^{\psi(m-\frac{t+r}{2})}}, \end{equation}
expressed in terms of the functions
\begin{eqnarray*}
\varphi(m)=\eta(m+r-t)-\eta(m),\ \ \psi(m)=\mu(m+r-t)-\mu(m), \ \ \psi'(m)=\mu'(m+r-t)-\mu'(m).
\end{eqnarray*}	
Notice that these can only take finitely many values. Specifically, we note that $\varphi(m+t)=\varphi(m)$, and 
$$\varphi(m)=\left\{ \begin{matrix} -1 &{\rm for}& m\in [0,t-r-1]\, {\rm mod}\, t\\ 0 &{\rm for}&  m\in [t-r,t-1]\, {\rm mod}\, t\end{matrix}\right. $$
We also have $\psi(m)=\psi'(m+t)$, $\psi'(m+2t)=\psi'(m)$, and
\begin{eqnarray*}\psi'(m)&=&\left\{ 
\begin{matrix} 0  &{\rm for}& m\in [t-r,2t-1]\, {\rm mod}\, 2t\\ -2 & {\rm for}& m\in [0,t-r-1]\, {\rm mod}\, 2t\end{matrix}\right.\\ 
\psi(m)&=&\psi'(m+t)=\left\{ 
\begin{matrix} 0  &{\rm for}& m\in [-r,t-1]\, {\rm mod}\, 2t\\ -2 & {\rm for}& m\in [-t,-r-1]\, {\rm mod}\, 2t\end{matrix}\right.
\end{eqnarray*} 

We deduce that if $m\in [t-r,t-1] $ mod $2t$, then $\psi(m)=\psi'(m)=0$, and similarly for $m\in [2t-r,2t-1] $ mod $2t$,
while in both cases $\varphi(m)=0$.
On the other hand, if
$m\in [0,t-r-1] $ mod $2t$, then $\psi'(m)=-2$, $\psi(m)=0$, whereas if $m\in [t,2t-r-1] $ mod $2t$,
then $\psi'(m)=0$, $\psi(m)=-2$, while in both cases $\varphi(m)=-1$. This leads finally to:
\begin{eqnarray*}{\mathcal L}^{2m}_{0,0}&=&{\mathcal R}^{2m}_{0,0}=\frac{1}{2}, \ \ {\rm for}\ \ m\in \left[ \frac{t-3r}{2},\frac{t-r}{2}-1\right]\ {\rm mod}  \ t\\
{\mathcal L}^{2m}_{0,0}&=&\frac{c^2}{b^2+c^2},\  {\mathcal R}^{2m}_{0,0}= \frac{b^2}{b^2+c^2}, \ \ {\rm for}\ \ m\in \left[-\frac{t+r}{2}, \frac{t-3r}{2}-1\right]\ {\rm mod}  \ 2t\\
{\mathcal L}^{2m}_{0,0}&=&\frac{b^2}{b^2+c^2},\  {\mathcal R}^{2m}_{0,0}= \frac{c^2}{b^2+c^2}, \ \ {\rm for}\ \ m\in \left[\frac{t-r}{2}, \frac{3t-3r}{2}-1\right]\ {\rm mod}  \ 2t\\
\end{eqnarray*}
%	On the other hand, if $\eta(m-\frac{3}{2}t+\frac{1}{2}r)-\eta(m-\frac{1}{2}t-\frac{1}{2}r)=-1$ then either $$\ds \mu(m-\frac{3}{2}t+\frac{1}{2}r)-\mu(m-\frac{1}{2}t-\frac{1}{2}r)=-2, \hskip 1cm  \mu'(m-\frac{3}{2}t+\frac{1}{2}r)-\mu'(m-\frac{1}{2}t-\frac{1}{2}r)=0$$ or
%	$$ \ds \mu'(m-\frac{3}{2}t+\frac{1}{2}r)-\mu'(m-\frac{1}{2}t-\frac{1}{2}r)=-2, \hskip 1cm \ds \mu(m-\frac{3}{2}t+\frac{1}{2}r)-\mu(m-\frac{1}{2}t-\frac{1}{2}r)=0$$ 
In all cases this gives ${\mathcal L}^{2m}_{0,0}+{\mathcal R}^{2m}_{0,0}=1$, which is equivalent to the $T$-system relation for $i,k=0$ mod $2$.
\end{proof}
We conclude that the variables $\theta^{2m}_{i,k}$ and $T^{2m}_{i,k}$ are identical, as they obey the same $T$-system relations with the same initial data. The statement (1) of Theorem \ref{ratiothm} follows, as $L_{i,k}^{2m}={\mathcal L}_{i,k}^{2m}$ only depends on $i,k$ mod $2$.
The statement (2) follows from the fact that the functions $\varphi,\psi,\psi'$ are $2t$-periodic in $m$.
Indeed, restricting again to the case $i=k=0$ mod 2, the periodicity property 
${\mathcal L}_{0,0}^{2m+4t}={\mathcal L}_{0,0}^{2m}$ follows immediately from eq. \eqref{LR00}.
%$$L^{2m}_{0,0}=\frac{f_{0,0}(m)}{f_{1,1}(m+\frac{t-r}{2})}=\frac{\left(\frac{b^2+c^2}{2}\right)^{\varphi(m+\frac{t-r}{2})}}{2\,b^{\psi(m+\frac{t-r}{2})}c^{\psi'(m+\frac{t-r}{2})}}$$
% 	Notice that:
% 	\begin{align*}
%		& \eta(m-\frac{1}{2}t+\frac{1}{2}r)-\eta(m+\frac{1}{2}t-\frac{1}{2}r)\\
%		&=\left\lfloor\frac{m-\frac{1}{2}t+\frac{1}{2}r}{t}\right\rfloor+\left\lfloor\frac{m-t+r}{t-r}\right\rfloor-\left\lfloor\frac{m-\frac{3}{2}t+\frac{3}{2}r}{t-r}\right\rfloor\\
%		&-\left(\left\lfloor\frac{m+\frac{1}{2}t-\frac{1}{2}r}{t}\right\rfloor+\left\lfloor\frac{m}{t-r}\right\rfloor+\left\lfloor\frac{m-\frac{1}{2}t+\frac{1}{2}r}{t-r}\right\rfloor\right)\\
%		&=\left\lfloor\frac{m-\frac{3}{2}t+\frac{1}{2}r}{t}+1\right\rfloor+\left\lfloor\frac{m-2t+r}{t-r}+\frac{t}{t-r}\right\rfloor-\left\lfloor\frac{m-\frac{5}{2}t+\frac{3}{2}r}{t-r}+\frac{t}{t-r}\right\rfloor\\
%		&-\left(\left\lfloor\frac{m-\frac{1}{2}t-\frac{1}{2}r}{t}+1\right\rfloor+\left\lfloor\frac{m-t}{t-r}+\frac{t}{t-r}\right\rfloor-\left\lfloor\frac{m-\frac{3}{2}t+\frac{1}{2}r}{t-r}+\frac{t}{t-r}\right\rfloor\right)\\
%		&=\eta(m-\frac{3}{2}t+\frac{1}{2}r)-\eta(m-\frac{1}{2}t-\frac{1}{2}r)
%	\end{align*}
%	We can prove similarly for $\mu$ and $\mu'$ thus concluding $L^{2m}_{0,0}=L^{2m+2t}_{0,0}$. The computation for periodicity $(i,k)=(0,1),(1,0),(1,1)$ also follows by the change of variable of $f$
%% 	\end{enumerate} 
 	\end{proof}

% 	\begin{remark}
% 		Theorem \ref{ratiothm} also suggests that we can intepret the $T$-system equipped with a slanted initial data with parameters $(r,r,t)$, $r,t$ odd, in terms of a recurrence relation. Specifically, by considering the product: 
% 		\begin{equation}\label{recurrence relation on T}
% 		\begin{aligned}
% 			\frac{T^{2m+2r+2t}_{0 [2],0 [2]}}{T^{2m+r+t}_{1 [2],1 [2]}}\frac{T^{2m+r+t}_{1 [2],1 [2]}}{T^{2m}_{0 [2],0 [2]}}&=\frac{T^{2m+2r+2t}_{0 [2],0 [2]}}{T^{2m}_{0 [2],0 [2]}}\\
% 			&=\frac{\left(\frac{a^2+d^2}{2}\right)^{\eta(m+\frac{t+r}{2})+\eta(m-\frac{t-r}{2})}\times \left(\frac{b^2+c^2}{2}\right)^{\eta(m+r)+\eta(m)}}{a^{\mu'(m+\frac{t+r}{2})+\mu'(m-\frac{t-r}{2})}b^{\mu(m+r)+\mu(m)}c^{\mu'(m+r)+\mu'(m)}d^{\mu(m+\frac{t+r}{2})+\mu(m-\frac{t-r}{2})}}
% 		\end{aligned}	
% 		\end{equation}
% 	\end{remark}
% 

%%%%%%%%%%%%%%%%%%%%%%%%%%%

\subsection{The case of $r=s$ and $t$ of opposite parity}
\label{sec:evenodd}

Most of the results in this section are proved identically to those of Section \ref{sec:oddodd}.
In the case when $r=s$ and $t$ have opposite parity, all integer values of $m=r(i+j+k)+(t-r)k$ contribute. 
We now have a unique solution
$k=\xi m$ mod $2r$ where $\xi$ is the inverse of $(t-r)$ mod $2r$. The periodicity Lemma 
\ref{periodic lemma} now becomes:
\begin{lemma}
For every $i,k$, the following periodicity relations hold for all initial data planes $m=0,1,...,2t-1$: 
\begin{align*}
T^{m}_{i+2,k}=T^{m}_{i,k}\\
T^{m}_{i,k+2r}=T^{m}_{i,k}
\end{align*}
\end{lemma}
Note that as a consequence of $2r$-periodicity in $k$, and the fact that $k$ is fixed mod $2r$ (to the value $\xi m$ mod $2r$), we may drop the index $k$ from the initial data. Finally Theorem \ref{ratiothm} becomes:

\begin{thm}\label{ratioperthm} 
The solution of the $T$-system with $2 \times 2$-periodic initial data has the following periodicity properties:
\begin{enumerate}
\item $L^{m}_{i,k}$, for $k=\xi m$ mod $2r$, depends only on $i$ modulo $2$, and not on $k$.
\item $L^{m+4t}_{i,k}=L^{m}_{i,k}$.
\end{enumerate}
 \end{thm}
 
The technical details of the proof are somewhat cumbersome and will be given elsewhere.

\section{The case of slanted $2 \times 2$-periodic Initial Data: Arctic Phenomenon}
Throughout this section we restrict to the case when $0\leq r=s<t$, with $r,t$ coprime, and to the 2x2 periodic $(r,r,t)$-slanted initial data \eqref{twobytwo}.
\subsection{General case: deriving the density function}

Using the change of variables $(i,j,k)\to (i,k,m)$, the local density of dimers at $(i_0,j_0,k_0)$ in the domain centered at $(i,j)$ can be written $\rho_{i,j,k}=\rho^m_{i,k}$, and satisfies the equation
\eqref{recurho}, namely
\begin{equation}\label{recurhop}
\rho^{m}_{i,k}+\rho^{m-2t}_{i,k-2}= {\mathcal L}^m_{i,k}\, (\rho^{m-t+r}_{i+1,k-1}+\rho^{m-t-r}_{i-1,k-1})+{\mathcal R}^m_{i,k}\, (\rho^{m-t+s}_{i,k-1}+\rho^{m-t-s}_{i,k-1}), 
\end{equation}
subject to the initial conditions $\rho^m_{i,k} =\delta_{i,i_0}\delta_{k,k_0}\delta_{m,ri_0+sj_0+tk_0}$. 

In the previous sections, we have established periodicity properties of the coefficients $L_{i,j,k}={\mathcal L}^m_{i,k}$ 
and $R_{i,j,k}={\mathcal L}^m_{i,k}$ in the variables $i,k,m$. 

Assume that $L_{i,j,k}$ 
is periodic along some lattice $\Lambda\subset\Z^3$,
then $\rho_{i,j,k}=\rho_{i,k}^m$ can be computed explicitly by the method of \cite{DiFrancesco1} (See section 3.2
in particular), which consists of splitting the 
generating function $\rho(x,y,z)=\sum_{i,j,k\in \Z, ri+sj+tk\geq 0} \rho_{i,j,k}x^i y^j z^k$ into pieces corresponding to equivalence classes of points $(i,j,k)$ modulo the periodicity lattice $\Lambda$.
The results of previous sections display naturally the lattice $\Lambda$ in the variables $(i,k,m)$. For short, we write
${\mathcal L}_{i,k}^m={\mathcal L}_{\beta}$, $\rho_{i,k}^m=\rho_\beta$, etc. with $\beta=(i,k,m)$. The  periodicity property is ${\mathcal L}_{\beta+\Lambda}={\mathcal L}_\beta$.
We have the density generating function 
$$\rho(x,y,z)=\sum_{i,k\in \Z,\ m\in \Z_+}  \rho^m_{i,k}x^i y^{\frac{m-tk-ri}{s}} z^k=\sum_{\beta} \rho_\beta w_\beta(x,y,z), $$
where $w_\beta(x,y,z)=x^{\beta_1}y^{\frac{\beta_3-t \beta_2-r \beta_1}{s}}z^{\beta_2}$ are additive weights,
namely satisfying $w_{\beta} w_{\gamma}=w_{\beta+\gamma}$. The recursion relation \eqref{recurhop} reads:
\begin{equation}\label{recurhop2}
\rho_\beta+\rho_{\beta-(0,2,2t)}= {\mathcal L}_\beta\, (\rho_{\beta+(1,-1,r-t)}+\rho_{\beta-(1,1,t+r)})+{\mathcal R}_\beta\, (\rho_{\beta+(0,-1,s-t)}+\rho_{\beta-(0,1,t+s)}) .
\end{equation}
The above splitting amounts to write
$\rho(x,y,z)=\sum_{\gamma\in  F} \rho^{(\gamma)}(x,y,z) $ where $F$ is a fundamental domain for $\Lambda$, and
$$\rho^{(\gamma)}(x,y,z)=\sum_{\beta\in \gamma+\Lambda,\ \beta_3\geq 0}  \rho_\beta\, w_\beta(x,y,z) .$$
This allows to rewrite \eqref{recurhop2} in terms of generating functions with $\beta$ summed over $\gamma+\Lambda$.
For all $\gamma\in F$, we get:
\begin{eqnarray}\label{recurhopsys}
\rho^{(\gamma)}+w_{(0,2,2t)}\rho^{(\gamma-(0,2,2t))}&=& \epsilon_\gamma+{\mathcal L}_\gamma\, (w_{(1,-1,r-t)}\rho^{(\gamma+(1,-1,r-t))}+w_{-(1,1,t+r)}\rho^{(\gamma-(1,1,t+r))})\nonumber \\
&&\qquad+{\mathcal R}_\gamma\, (w_{(0,-1,s-t)}\rho^{(\gamma+(0,-1,s-t))}+w_{-(0,1,t+s)}\rho^{(\gamma-(0,1,t+s))}), 
\end{eqnarray}
where all superscripts are understood modulo $\Lambda$, and represented by elements of the fundamental domain $F$.
The term $\epsilon_\gamma$ corresponds to the initial condition $\rho_{i,j,k}=\delta_{i,i_0}\delta_{j,j_0}\delta_{k,k_0}$
along the plane $P_{m_0}$, with $m_0=ri_0+sj_0+tk_0$, namely $\epsilon_\gamma=\delta_{\gamma,(m_0,i_0,k_0)}$.
This gives a linear system of $|F|$ equations for the functions $\rho^{(\gamma)}=\rho^{(\gamma)}(x,y,z)$, $\gamma\in F$,
which can be solved by Cramer's rules. The common feature to all $\rho^{(\gamma)}(x,y,z)$ is that they are rational functions of $(x,y,z)$ with common denominator $D(x,y,z)$ given by the $|F|\times |F|$ determinant of the system.

\subsection{Singularity loci}
By  the same argument as in Section \ref{asymptotic of the density}, we deduce that the singularities of the actual dimer density generating function
come from the function 
$$\Delta(x,y,z)=D(z^{r/t}x^{-1},z^{s/t}y^{-1},z) .$$
We finally have to apply ACSV \cite{PW2004,pemantle_wilson_2013,PW2005,BP1,DiFrancesco1} to the vicinity of the point $(x,y,z)=(1,1,1)$ to derive the arctic curve of the model. 

\subsection{Elimination}\label{Euclidean section}\hfill \newline
For $r>1$, similarly to Sect.3.1.1, we expand 
\begin{equation}
\label{Hdef}
\Delta(e^{\epsilon x},e^{\epsilon y},e^{-\epsilon (ux+vy)})=\epsilon^\theta (H(x,y,u,v)+O(\epsilon))
\end{equation} 
at leading order in $\epsilon$,
leading to polynomials $H(x,y,u,v)$ generically of higher degree $\theta>2$. 
Explicit calculations up to $t=8$ lead us to conjecture:
\begin{conjecture}\label{thetadef}
 	\begin{equation}
\theta=\theta_{r,r,t}:= 2\big(t+1+\,{\rm Mod}(r,2)\big).
\end{equation}
 \end{conjecture}

%The following method is our non-rigorous approach to obtain a function of in $(u,v)$.\hfill \newline 
To find the (dual) arctic curve, we must eliminate the $x,y$ variables from the system $H(x,y,u,v)=\partial_x H(x,y,u,v)=\partial_y H(x,y,u,v)=0$.
 	To that effect, we perform the euclidean division of the polynomial $H(x,y,u,v)$ by $\partial_x H(x,y,u,v)$,
	both  considered as polynomials of $x$. 
	%Notice that if we accept the condition $H(x,y)=\partial_x H(x,y)=\partial_y H(x,y)=0$, then the Euclidean algorithm on $H(x,y)$ gives
	This gives: $H(x,y)=q_1(x)(\partial_x H(x,y))+R_1(x)$, which we iterate in the form $R_{i-1}(x)=q_{i+1}(x)R_{i}(x)+R_{i+1}(x)$ for $i\geq 1$, with $R_0(x)=\partial_x H(x,y)$. 
% 	Thus, $R_1(x)=0$ is a polynomial in $x$ with degree lower than $\partial_x H(x,y)$. We continue the process, by writing: 
% 	$$\partial_x H(x,y)=q_2(x)R_1(x)+R_2(x)$$ 
The process is iterated until we reach the "constant" (say $R_m(y,u,v)$) term, which will be a polynomial in $y,u,v$ where $y$ can be factored since the starting polynomial $H(x,y,u,v)$ is homogenous in $x$ and $y$. 
%Notice that the polynomial $H(x,y)$ is homogeneous in $x,y$ also implies this the process works similarly if we take $\partial_y$ and considering the Euclidean algorithm starting with $H(x,y)$ as polynomial in $y$. We will have all of the example in the section \hfill\newline
After removing the $y$ dependent factor, we end up with a polynomial of $u,v$ which determines the arctic curve.
%\subsubsection{Identifying the arctic curves}
Note that in this elimination process, there are instances when at some  $i$-th iteration of the Euclidean algorithm, the remainder $R_i(x)$ already factors out some polynomial in $u,v$ independent of $x,y$. However, such polynomials in general are either linear or of lower degree than the one of interest, reached only at the last step. 
%Thus, we pick the factor of the highest degree in $u,v$ and continue the process.

\subsection{Symmetries}
\label{sec:sym}

As we only consider the cases with $r=s$, we note that the $(r,r,t)$-slanted initial data planes are invariant under the translation 
$(i,j,k)\to (i+1,j-1,k)=(i,j,k)+(1,-1,0)$. However, the initial data assignments \eqref{twobytwo} are not invariant: this translation corresponds to a permutation $(a,b,c,d)\to (d,c,b,a)$. In the scaling limit of large $k$ and finite $i/k,j/k$, 
the dimer partition functions $T_{i,j,k}$ and $T_{i+1,j-1,k}$ become undistinguishable: in particular,
such a (bounded) translation does not affect the limit shape, therefore we expect the limiting arctic curve to be invariant
under the permutation $(a,b,c,d)\to (d,c,b,a)$.

We may repeat the argument with translations by $(1,-2,1)$ and $(-2,1,1)$, which do not leave the (r,r,t)-slanted planes invariant but map them on uniformly close ones in the scaling limit. As a consequence, we expect the limiting arctic curve to be invariant under the permutations $(a,b,c,d)\to (c,d,a,b)$ and $(a,b,c,d)\to(b,a,d,c)$ as well.

\subsection{Examples}\label{toroidal example section}

For both cases of Sections \ref{sec:oddodd} and \ref{sec:evenodd}, the fundamental domain of $\Lambda$ has  $|F|=8t$ points.
This number is quite large in general, and we choose to give a few meaningful examples in the following sections.
Throughout the reamaining sections, we use the following two parameters:
\begin{equation}\label{sigtau}
\tau=\frac{a^2}{a^2+d^2}\ ,\qquad \sigma=\frac{b^2}{b^2+c^2} .
\end{equation}
The arctic curves are determined using the method described in Sect. \ref{Euclidean section}. We use the notation $H_{rst}(x,y,u,v)$ for the generic $(r,s,t)$ case.
The symmetries described in Sect. \ref{sec:sym} imply that the arctic curves are invariant under each of the three transformations:  
$$(\sigma,\tau)\to (1-\sigma,1-\tau), \qquad (\sigma,\tau)\to (\tau,\sigma),\qquad (\sigma,\tau)\to (1-\tau,1-\sigma). $$

%\begin{remark}
%	As we have seen in previous section, determine the generating function that yields the arctic curves requires one to compute the two quantities $L_{i,j,k}$ and $R_{i,j,k}$. We want to make a remark on the symmetry of $\sigma$ and $\tau$ for the case $(r,r,t)$-slanted $T$-system. Specifically, the change in the four quantities $\sigma, \tau, 1-\sigma, 1-\tau$ (which dictates $L_{i,j,k}$ and $R_{i,j,k}$) corresponding to exactly an intechange of pairs of $a,b,c,d \in \R$. Specifically,
%	\begin{equation}
%		\begin{matrix}
%			(\sigma, \tau) \mapsto (1-\sigma, 1-\tau) \sim a\leftrightarrow d,b\leftrightarrow c\\
%			(\sigma, \tau) \mapsto (1-\tau, 1-\sigma) \sim a\leftrightarrow c,b\leftrightarrow d\\	
%			(\sigma, \tau) \mapsto (\tau, \sigma) \sim a\leftrightarrow b,d\leftrightarrow c\\	
%		\end{matrix}
%	\end{equation}
%	and other combinations. This correspondence of symmetries in initial data and $\sigma, \tau$, thus, yields the symmetries and other interesting properties when one derives arctic curves. We will provide detail in the next section
%	
%%%%%%%%%%%%%%%%%%%%%%%%%%%%%%%
%
%\end{remark}

\subsubsection{\textbf{The case $r=s=1$, $t$  odd in general}}\label{odd_odd_section}\hfill\newline
For simplicity, the coefficients of the linear system \eqref{recurhopsys} may be organized into quadruples
$Q_m=(\mathcal L^{2m}_{0,0},\mathcal L^{2m}_{1,0},\mathcal L^{2m}_{0,1},\mathcal L^{2m}_{1,1})$ for $m \in [0,1,\cdots,2t-1]$. The periodicity lattice $\Lambda$ for the triples $(i,k,m)$ for $\mathcal L_{i,k}^{2m}$'s and $\mathcal R_{i,k}^{2m}$'s is generated by $(2,0,0),(0,2,0),(0,0,2t)$
with $F=[0,1]\times [0,1]\times [0,2t-1]$.

%	The two quantities $$\ds \mathcal L^{2m}_{i , k}=\frac{T^{2m+r-t}_{i+1,k-1}T^{2m-r-t}_{i-1,k-1}}{T^{2m-2t}_{i,k-2}T^{2m}_{i,k}}$$ 
%	$$\ds \mathcal R^{2m}_{i,k}=\frac{T^{2m+r-t}_{i,k-1}T^{2m-r-t}_{i,k-1}}{T^{2m-2t}_{i,k-2}T^{2m}_{i,k}}=1-\mathcal L^{2m}_{i,k}$$ dictates the asymptotic behavior of the system, which gives rise to the arctic phenomenon separating frozen regions using ACSV \cite{PW2004,pemantle_wilson_2013,PW2005,BP1,DiFrancesco1}. By the $2 \times 2$ periodic property, it is sufficient to investigate the 2 quantities up to modulo $2$ in the $(i,k)$ direction. 
%	
%	\noindent In order to study the arctic phenomenon, we once again use the density generating function $\rho^{i_0,j_0,k_0}(x,y,z)$. Similar to section $\bf{(3.2)}$ in \cite{DiFrancesco1}, we want to find a system of linear equations for all the periodiciy of the quantities $L_{i,j,k}$ and $R_{i,j,k}$. Because of the complexity of the slanted initial data, we will focus on the specific case where $(r,s,t)=(1,1,3)$ for detail computation.The computation for other cases will be in the appendix section. \newline
We mainly discuss the case $(r=s=1,t=3)$ in full detail, as higher odd values of $t$ display similar behaviors.
On $F$, the coefficients of the system \eqref{recurhopsys} read:
\begin{equation}\label{(r,s,t)=113 lattice}
\begin{matrix}
			Q_0=(\sigma,1-\sigma,\tau,1-\tau)&
			Q_1=(\frac{1}{2},\frac{1}{2}, \frac{1}{2},\frac{1}{2})&
			Q_2=(1-\sigma,\sigma, 1-\tau,\tau)\\
			Q_3=(1-\sigma,\sigma, 1-\tau,\tau)&
			Q_4=(\frac{1}{2},\frac{1}{2}, \frac{1}{2},\frac{1}{2})&
			Q_5=(\sigma,1-\sigma, \tau,1-\tau)\\
\end{matrix}
\end{equation}
in terms of the parameters \eqref{sigtau}. 
\begin{remark}
	Notice from the vector $Q_m$ that $\mathcal L_{0,0}^{2m}=1-\mathcal L_{1,0}^{2m}$ and $\mathcal L^{2m}_{0,1}=1- \mathcal L^{2m}_{1,1}$. This is simply a consequence of \eqref{quantityf} and the fact that $\mathcal L_{i,k}^{2m}=\mathcal R_{i,k}^{2m}|_{a \leftrightarrow d, b\leftrightarrow c}$. For example, the case: 
	\begin{equation}\label{symmetry in L}
	 	\begin{aligned}
	 		\mathcal L^{2m}_{0,0}&=\frac{f_{0,0}(m-t)}{f_{1,1}(m-\frac{t+r}{2})}=\frac{f_{0,0}(m-t)}{f_{0,0}(m-\frac{t+r}{2})|_{a \leftrightarrow c, b\leftrightarrow d}}\\
	 		1-\mathcal L_{1,0}^{2m}&=\mathcal R_{1,0}^{2m}=\mathcal L_{1,0}^{2m}|_{b \leftrightarrow c, a\leftrightarrow d}=\frac{f_{0,0}(m-t)}{f_{0,0}(m-\frac{t+r}{2})|_{a \leftrightarrow c, b\leftrightarrow d}}
	 	\end{aligned}
	 \end{equation}

\end{remark}
The coefficient matrix of the linear system \eqref{(r,s,t)=113 lattice} can be found in \cite{mathfiles} (file ``Supplementary Material"). Following the steps of ACSV like in Section \ref{asymptotic of the density}, we find that the leading expansion \eqref{thetadef} of the scaled determinant of the system leads to $\theta=\theta_{1,1,3}=10$ and: 
\begin{equation}\label{113 general determinant}
	\begin{aligned}
		H_{1,1,3}(x,y,u,v)&=1024 \,(x + 2 u x - y + 2 v y) (-x + 2 u x + y + 2 v y)\\
		&\times (2 u x+2 v y+2 \sigma  \tau  x-\sigma  x-\tau  x-x-2 \sigma  \tau  y+\sigma  y+\tau  y-2 y)\\
		&\times(4 u x+4 v y+2 \sigma  \tau  x-\sigma  x-\tau  x+2 x-2 \sigma  \tau  y+\sigma  y+\tau  y+y)\\
		&\times H_{1,1,3}^*(x,y,u,v)
	\end{aligned}
\end{equation} 
$ H_{1,1,3}^*(x,y,u,v)$ is the factor of interest with the highest degree, which yields the arctic curves in our case. The full expression is cumbersome (See file ``Supplementary Material" in \cite{mathfiles}). We provide some initial examples where we pick somewhat arbitrary values of $\sigma$ and $\tau$. This provides us with 2 inner regions inside the "initial" ellipses from the uniform case. However, in this case, the initial polynomial to which we apply the elimination procedure of Section \ref{Euclidean section} is generically of degree 6, 
while the final (dual) arctic curve is of degree $14$.
\begin{figure}[H]
\begin{subfigure}{0.2\textwidth}
	\includegraphics[scale=0.25]{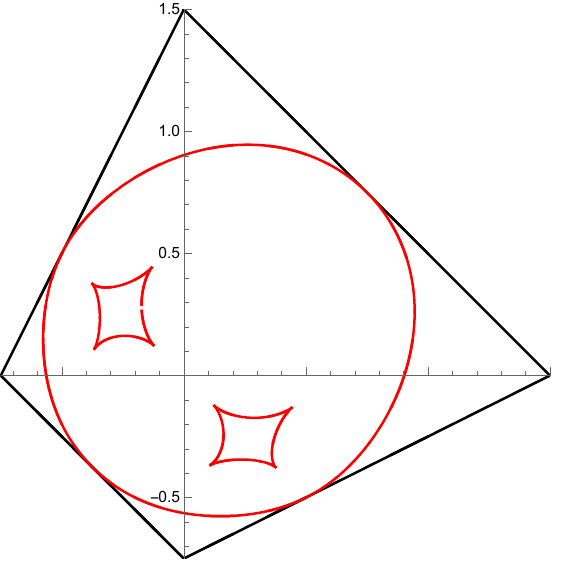}
	\subcaption[t]{\tiny $\sigma=1/4, \tau=1/2 $}
\end{subfigure}\hskip .5cm
\begin{subfigure}{0.2\textwidth}
	\includegraphics[scale=0.25]{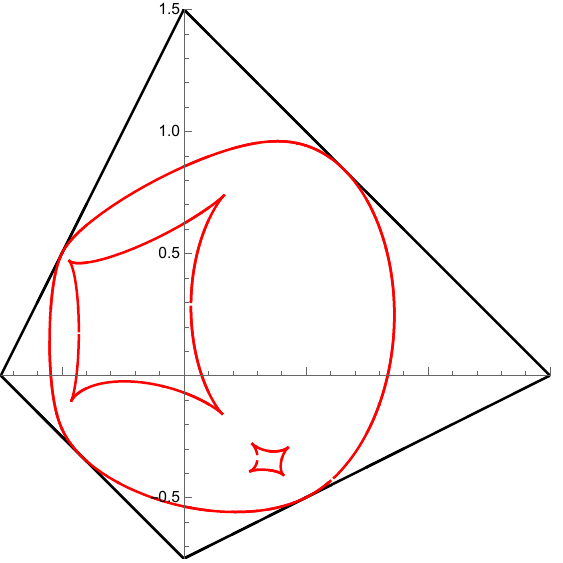}
	\subcaption[t]{\tiny $\sigma=1/4, \tau=1/8 $}
\end{subfigure}\hskip .5cm
\begin{subfigure}{0.2\textwidth}
	\includegraphics[scale=0.25]{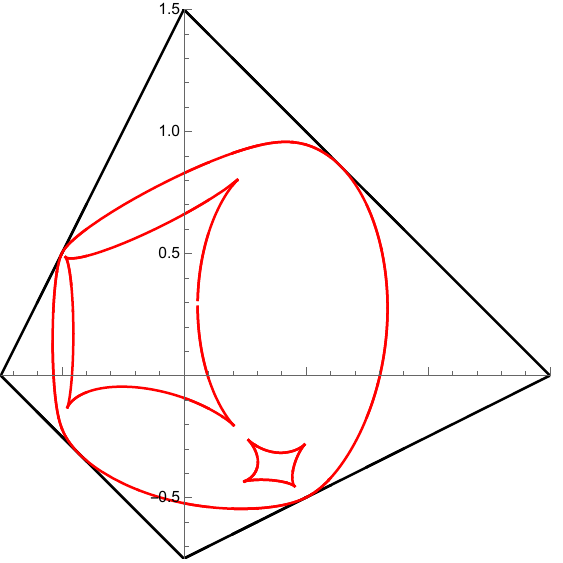}
	\subcaption[t]{\tiny $\sigma=1/4, \tau=1/16 $}
\end{subfigure}\hskip .5cm
\begin{subfigure}{0.2\textwidth}
	\includegraphics[scale=0.25]{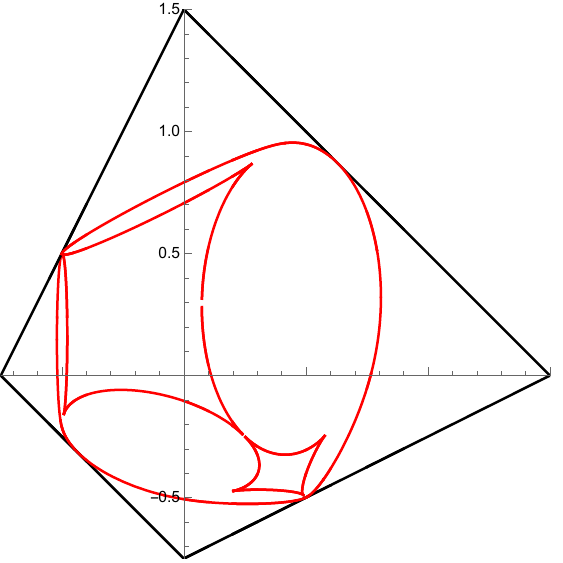}
	\subcaption[t]{\tiny $\sigma=1/4, \tau=1/64 $}
\end{subfigure}
	\caption{\small Arctic curves for $(r,s,t)=(1,1,3)$, together with the scaled domain with fixed $\sigma=1/4$.} 
\label{fig:generic113}
\end{figure}
Through investigating different values of $\sigma$ and $\tau$, we observe interesting collapses of inner regions when for instance $\tau$ becomes small (see e.g. Fig.~ \ref{fig:generic113} (D)), which is the motivation for the next section.

\noindent{$\bullet$ \textbf{$\tau=0$ case.}}\hfill\newline
For $\tau=0$, the function $H_{1,1,3}(x,y,u,v)$ factors into a product of some linear and two quadratic polynomials, denoted $H_1(x,y,u,v)$ and $H_2(x,y,u,v)$. Imposing the condition of vanishing of the Hessian like in section \ref{asymptotic of the density} on each of the latter results in two tangent ellipses. The full phase separation also includes segments obtained by including the other factors.

We have:
	\begin{equation}\label{eq: H(x,y),tau=0}
		\resizebox{.95\hsize}{!}{$\begin{aligned}
		H(x,y,u,v)|_{\tau=0}&=1024(-2 u x-2 v y+x-y)^2 (-2 u x-2 v y-x+y)^2 (-2 u x-2 v y+\sigma  x+x-\sigma  y+2 y)
		(-4 u x-4 v y+\sigma  x-2x-\sigma  y-y) \\
		 &\times H_1(x,y,u,v) \, H_2(x,y,u,v)\\
   		H_1(x,y,u,v)&=
%(4 u x+4 v y+\sigma  x-2 x-\sigma  y-y) 
   (4 u^2 x^2-8 u v x y+2 \sigma  u x^2-4 u x^2-2 \sigma  u x y-2 u x y+4 v^2 y^2
+2 \sigma  v x y-4 v x y-2
   \sigma  v y^2-2 v y^2-3 \sigma  x^2+x^2+x y+3 \sigma  y^2-2 y^2)\\
   		H_2(x,y,u,v)&=\left(8 u^2 x^2+16 u v x y+6 \sigma  u x^2-6 \sigma  u
   x y+6 u x y+8 v^2 y^2+6 \sigma  v x y-6 \sigma  v y^2+6 v y^2+3 \sigma  x^2-2 x^2+x y-3 \sigma  y^2+y^2\right)
	\end{aligned}$
	}
	\end{equation}
	which gives the 2 elliptic pieces of arctic curves:
	\begin{equation}\label{eq: H(x,y),tau=1/4}
		\resizebox{.9\hsize}{!}{$\begin{aligned}
   		P_1(u,v)&=-(4 - 10 u - 14 v + 32 u v)^2 + (2 - 28 u + 32 u^2) (-10 - 20 v + 
    32 v^2)\\
   		P_2(u,v)&=-(4 + 18 u + 6 v + 64 u v)^2 + (-10 + 12 u + 64 u^2) (2 + 36 v + 
    64 v^2)
	\end{aligned}$
	}
	\end{equation}
These encompass the liquid phase (see Fig: \ref{fig: first degenerate curve}).
	\begin{figure}
		\centering
		\includegraphics[scale=.35]{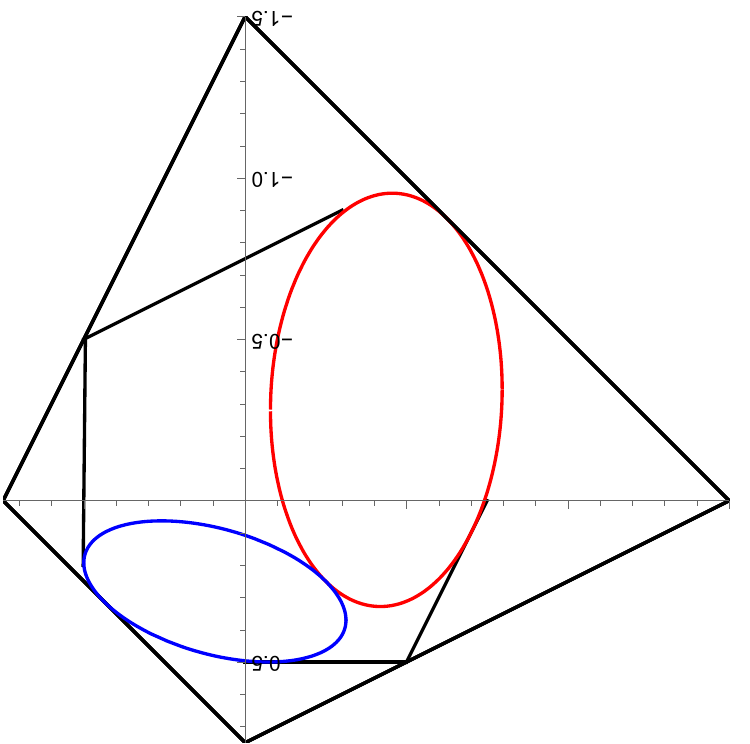}
		\caption{\small Artic Curve for $(r,s,t)=(1,1,3)$ for $\tau=0$ and $\sigma=\frac{1}{4}$ }
		\label{fig: first degenerate curve}
	\end{figure}

In addition to the two tangent ellipses, we find segments that are tangent to the ellipses. We argue that these are the degenerate limits when $\tau\to 0$ of the smooth arctic curve with $\tau>0$. 
This can be visualized by comparing Fig. \ref{fig: first degenerate curve} to the last plots $(C),(D)$ of Fig. \ref{fig:generic113},  upon interchanging the roles of $\sigma$ and $\tau$.

We end this section by providing another 4 arctic curves corresponding to $(r,s,t)=(1,1,3)$ and $(r,s,t)=(1,1,5)$ with $\tau=0$ and some choices of $\sigma$ in Fig:\ref{fig:tau=0_(11t)}.
\begin{figure}
\begin{subfigure}{0.33\textwidth}
        \centering
        \includegraphics[width=3.cm]{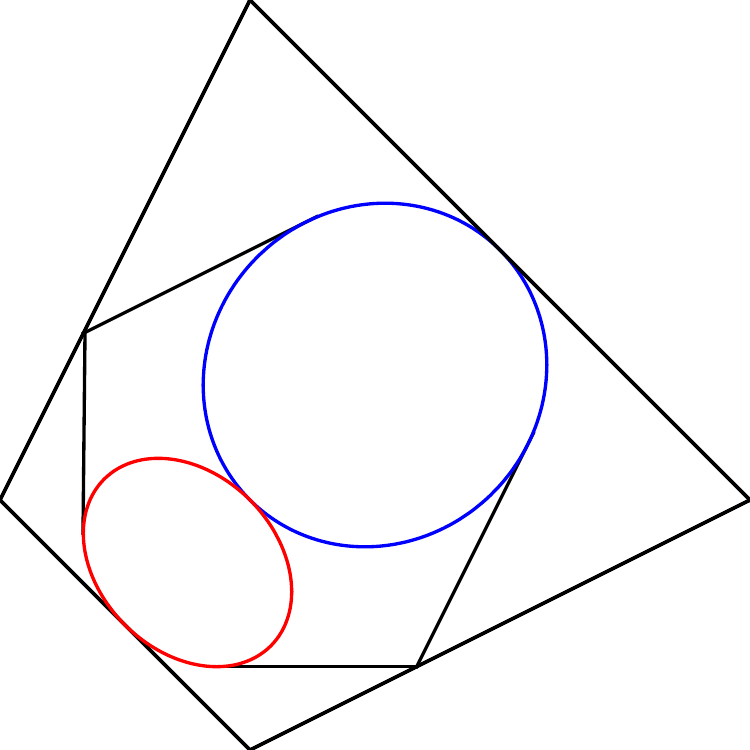} % second figure itself
        %\caption{second figure}
    \end{subfigure}
\begin{subfigure}{0.33\textwidth}
        \centering
        \includegraphics[width=3.cm]{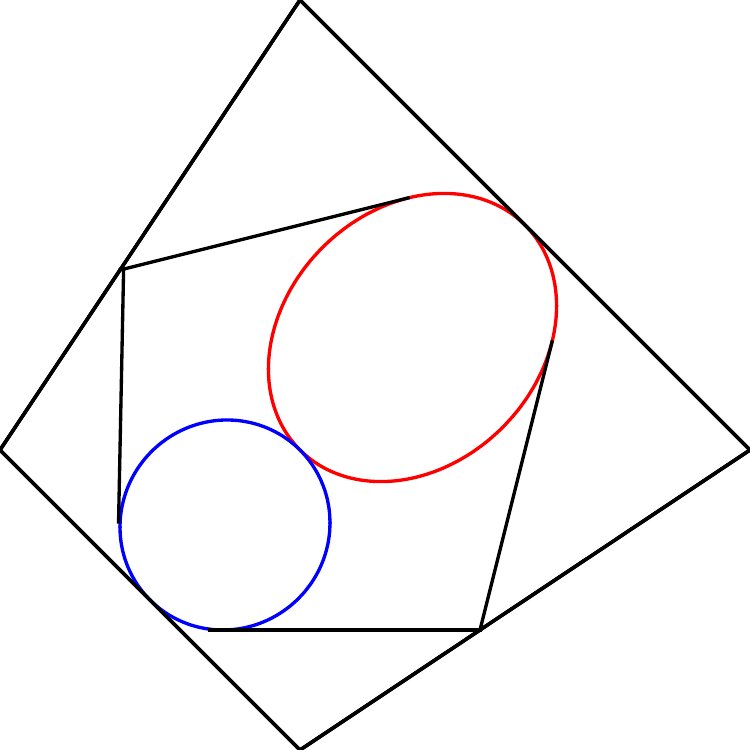} % first figure itself
        %\caption{first figure}
    \end{subfigure}\hfill
    \begin{subfigure}{0.33\textwidth}
        \centering
        \includegraphics[width=3.cm]{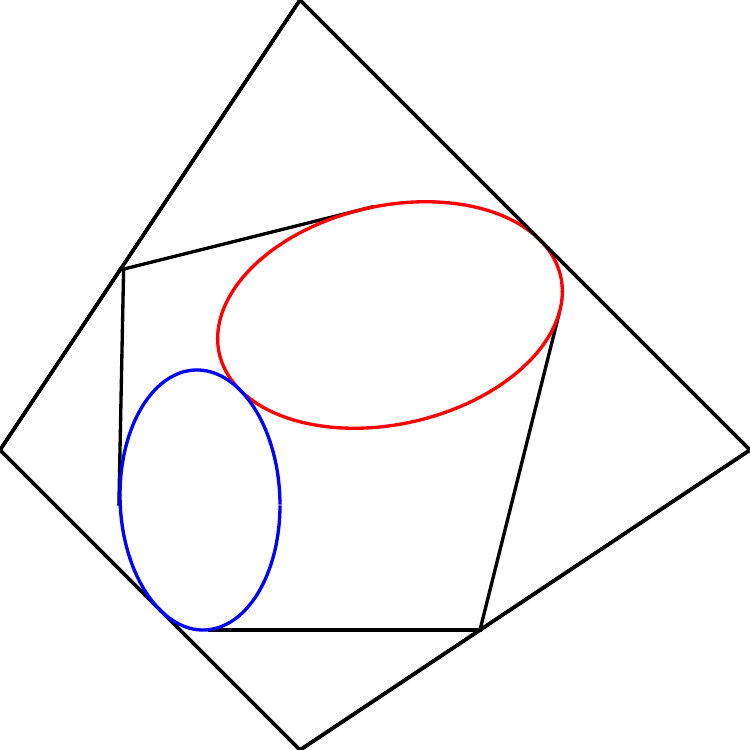} % second figure itself
        %\caption{second figure}
    \end{subfigure}
\caption{\small Arctic curves for $r,s,t$-pinecones, together with the scaled domain for fixed $\tau=0$. Left: $(r,s,t)=(1,1,3), \sigma=\frac{1}{2}$, \, Center: $(r,s,t)=(1,1,5),\sigma=\frac{1}{2}$, \, Right: $(r,s,t)=(1,1,5), \sigma=\frac{2}{3}$ .}
\label{fig:tau=0_(11t)}
\end{figure}

\noindent{$\bullet$ \textbf{$\sigma=\tau$ case.} }\hfill\newline
Another interesting case is when $\sigma=\tau$. 
The leading term $H(x,y,u,v)$ at the leading order $\lambda^{10}$ reads: 
\begin{equation}\label{eq: H(x,y),tau=sigma}
		\begin{aligned}
		H(x,y,u,v)|_{\tau \to \sigma}&=1024 (2 u x+x+2 v y-y) (2 u x-x+2 v y+y) \\
		&\times (-x + 2 u x - 2 y + 2 v y - 2 x \sigma + 2 y \sigma + 
  2 x\sigma^2 - 2 y \sigma^2)\\
   &\times (2 x + 4 u x + y + 4 v y - 2 x \sigma + 2 y \sigma + 
  2 x \sigma^2 - 2 y \sigma^2)\\
   &\times (-x + 2 u x + y + 2 v y + 2 x \sigma - 2 y \sigma - 
  2 x \sigma^2 + 2 y \sigma^2)^2 \\
   &\times H^*(x,y,u,v)
	\end{aligned}
	\end{equation}
where $H^*(x,y,u,v)$ is the highest order polynomial factor of interest 
(see file ``Supplementary Material" in \cite{mathfiles} for an explicit expression).
%\eqref{H_star_113_sigma=tau})	.
Note that $H^*(x,y,u,v)$ has degree $4$ in $x,y$, hence we must use the 
elimination process of (\ref{Euclidean section}).
	
	We display in Fig:\ref{tau=sigma} the resulting arctic curves, for $1/2>\tau\geq 0$.
	Observe the development of inner curves from the interior of a bigger curve as $\sigma=\tau$ decreases. 
\begin{figure}
	\begin{subfigure}{.2\textwidth}
		\centering 
		\includegraphics[scale=0.16]{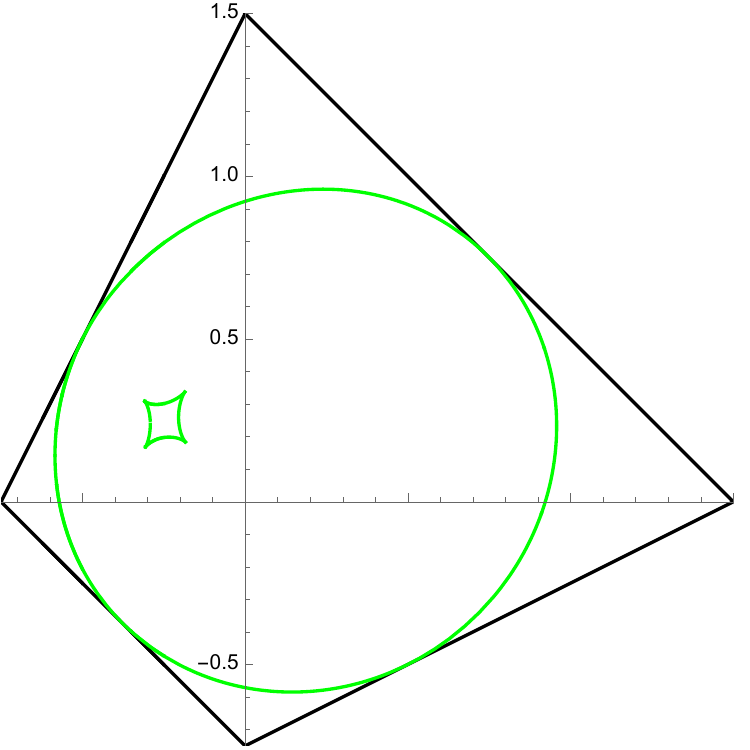}
		\subcaption[t]{\tiny$\sigma=\tau=7/16$}
	\end{subfigure}
	\begin{subfigure}{.2\textwidth}
		\centering
		\includegraphics[scale=0.16]{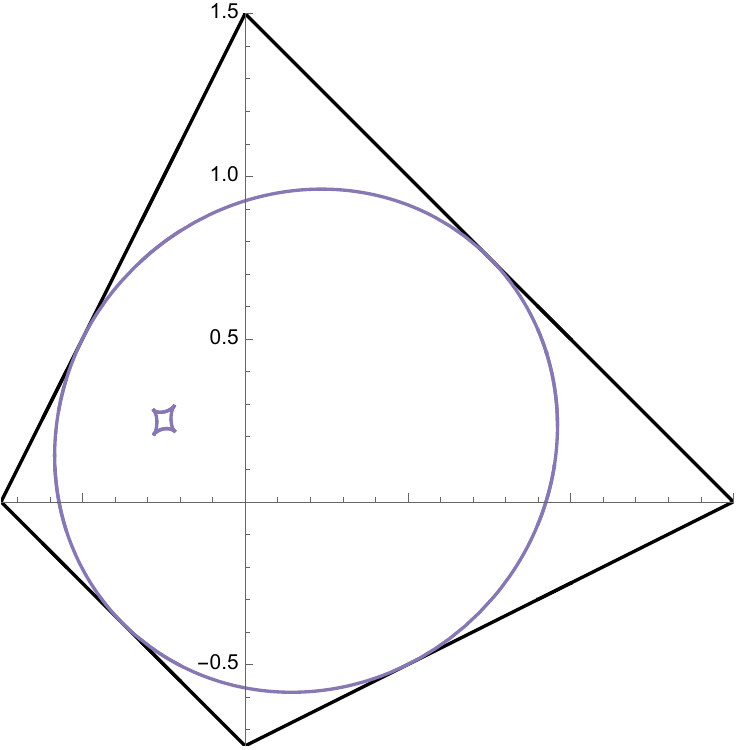}
		\subcaption[t]{\tiny$\sigma=\tau=15/32$}
	\end{subfigure}
	\begin{subfigure}{.2\textwidth}
		\centering
		\includegraphics[scale=0.16]{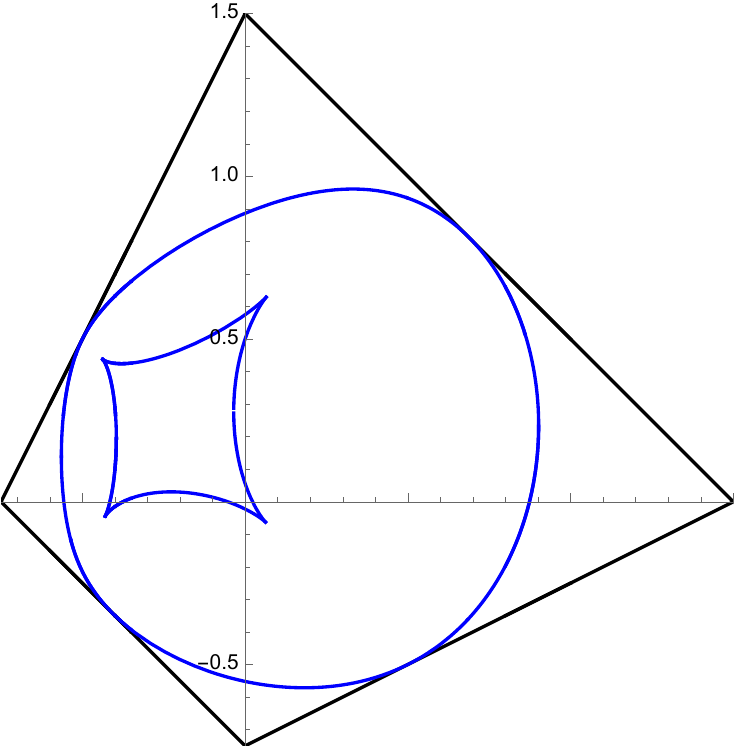}
		\subcaption[t]{\tiny$\sigma=\tau=1/4$}
	\end{subfigure}
	\begin{subfigure}{.2\textwidth}
		\centering
		\includegraphics[scale=0.16]{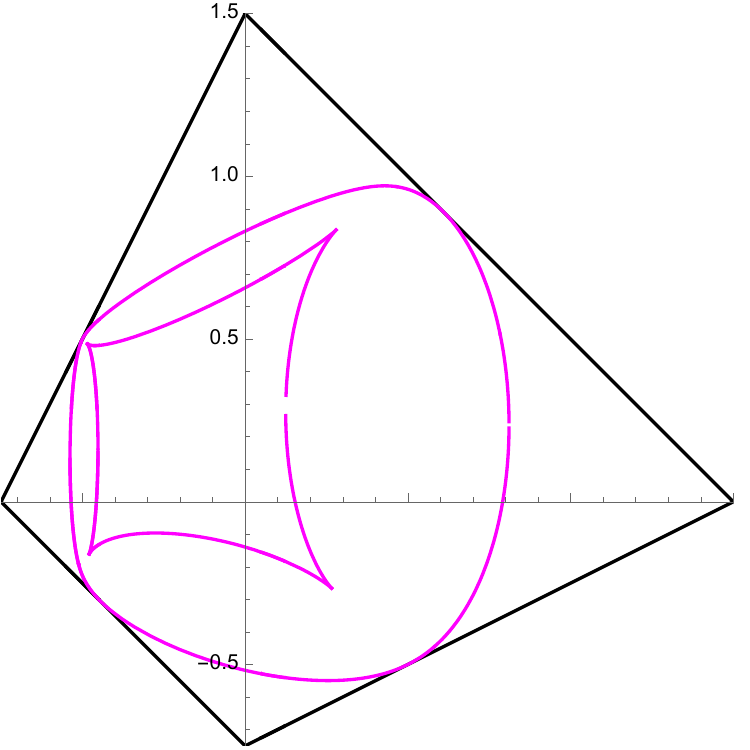}
		\subcaption[t]{\tiny$\sigma=\tau=1/8$}
	\end{subfigure}\\
	\begin{subfigure}{.2\textwidth}
		\centering
		\includegraphics[scale=0.16]{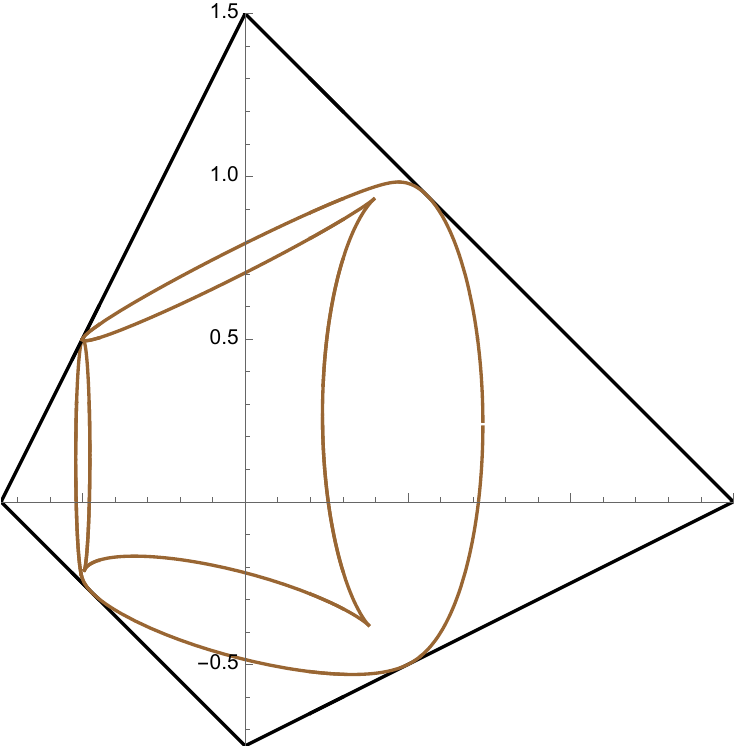}
		\subcaption[t]{\tiny$\sigma=\tau=1/16$}
	\end{subfigure}
	\begin{subfigure}{.2 \textwidth}
		\centering
		\includegraphics[scale=0.16]{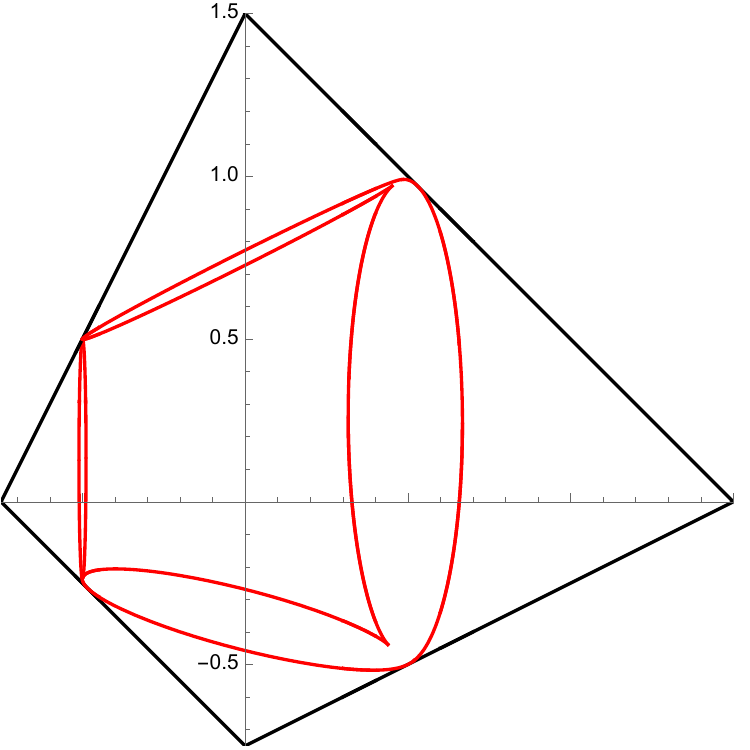}
		\subcaption[t]{\tiny$\sigma=\tau=1/32$}
	\end{subfigure}
	\begin{subfigure}{.2\textwidth}
		\centering
		\includegraphics[scale=0.25]{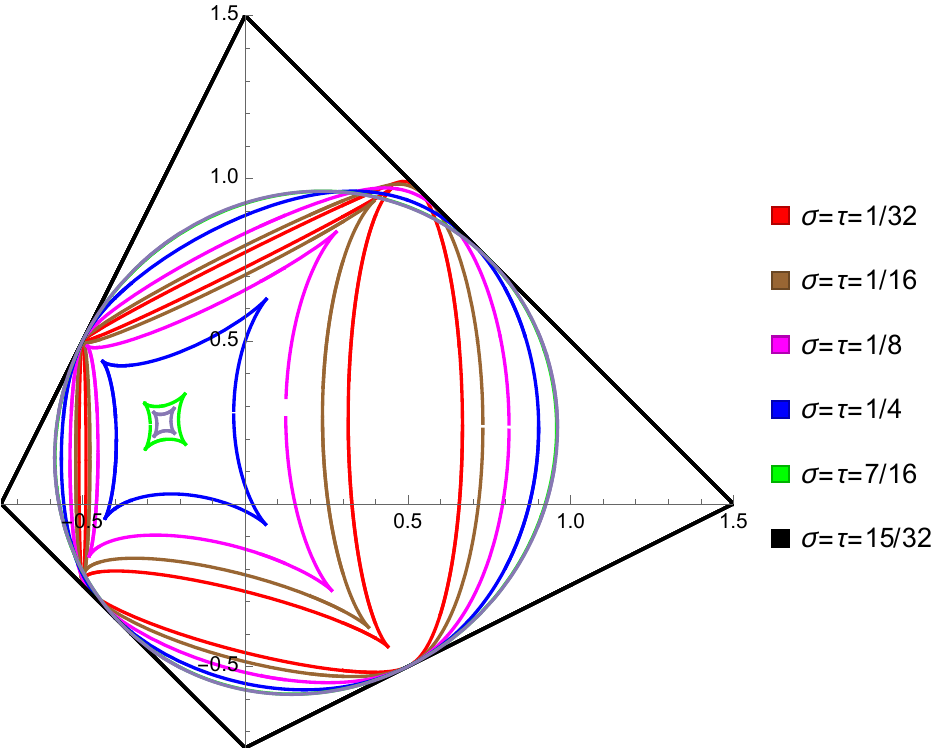}
		\subcaption[t]{Curves (A-F)}
	\end{subfigure}\hskip 1cm
	\begin{subfigure}{.2\textwidth}
		\centering
		\includegraphics[scale=0.16]{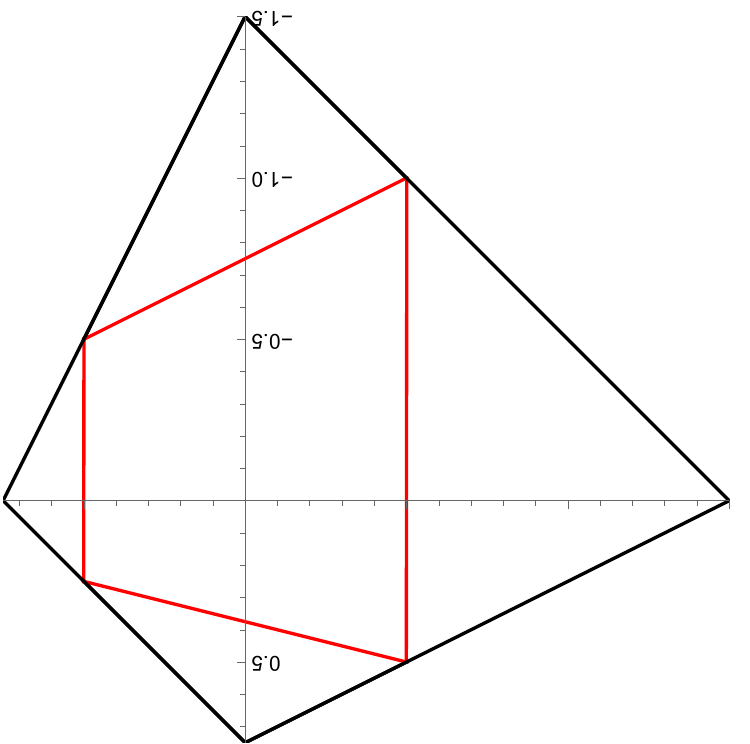}
		\subcaption[t]{\tiny$\sigma=\tau=0$}
	\end{subfigure}
	\caption{\small Arctic Curves for $(r,s,t)=(1,1,3)$ and $\sigma=\tau$}
	\label{tau=sigma}
\end{figure}

An interesting feature of these curves is that when $\tau=\sigma=0$, the curve degenerates into a polygon. We will provide a combinatorial intepretation of this phenomenon in the discussion section. To conclude this section, we notice that a similar behavior happens when $\sigma=1-\tau$ (see Fig: \ref{tau=1-sigma}).  
\begin{figure}
	\begin{subfigure}{.3 \textwidth}
		\centering
		\includegraphics[scale=0.2]{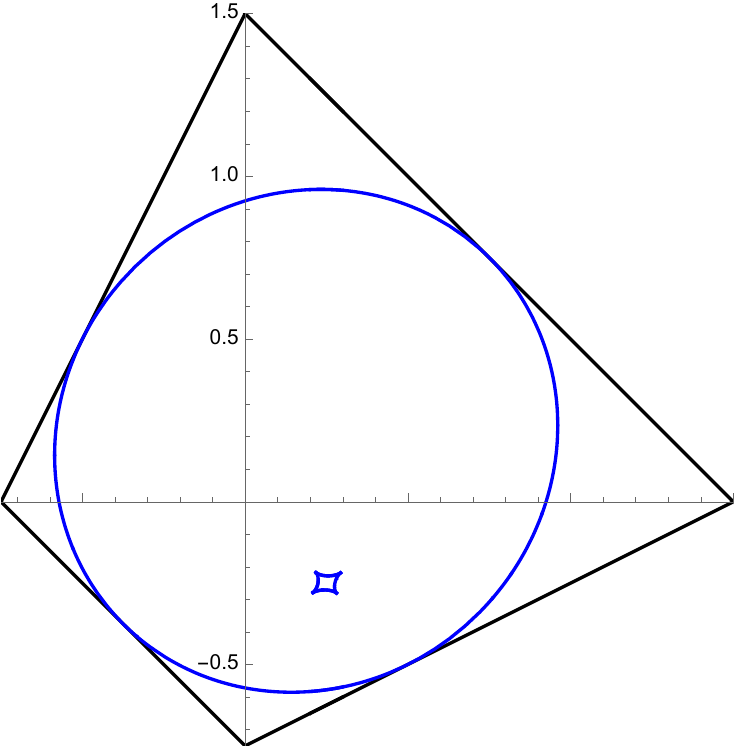}
		\subcaption[t]{\small $\sigma=1-\tau=17/32$}
	\end{subfigure}
	\begin{subfigure}{.3\textwidth}
		\centering
		\includegraphics[scale=0.2]{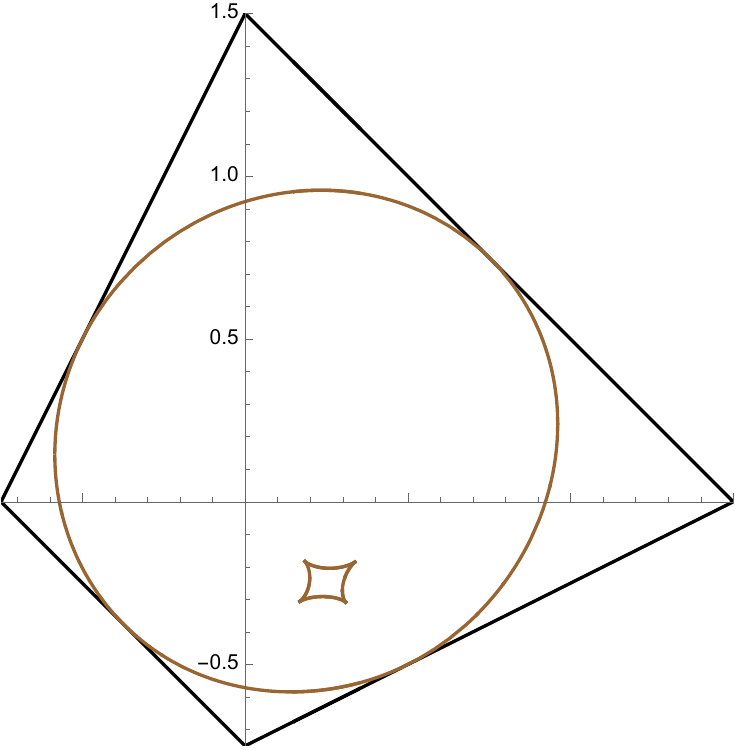}
		\subcaption[t]{\small $\sigma=1-\tau=9/16$}
	\end{subfigure}
	\begin{subfigure}{.3\textwidth}
		\centering
		\includegraphics[scale=0.2]{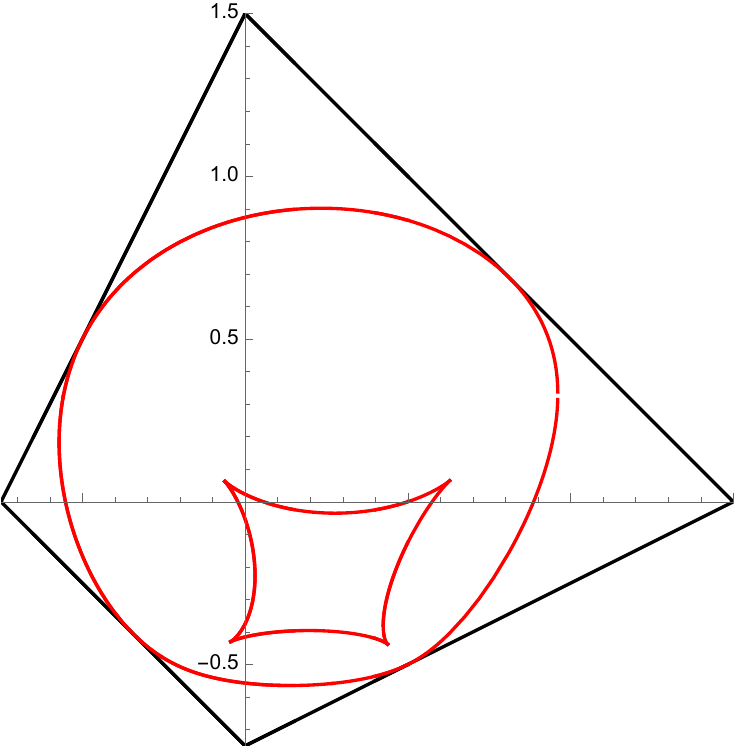}
		\subcaption[t]{\small $\sigma=1-\tau=3/4$}
	\end{subfigure}
	\begin{subfigure}{.3\textwidth}
		\centering
		\includegraphics[scale=0.2]{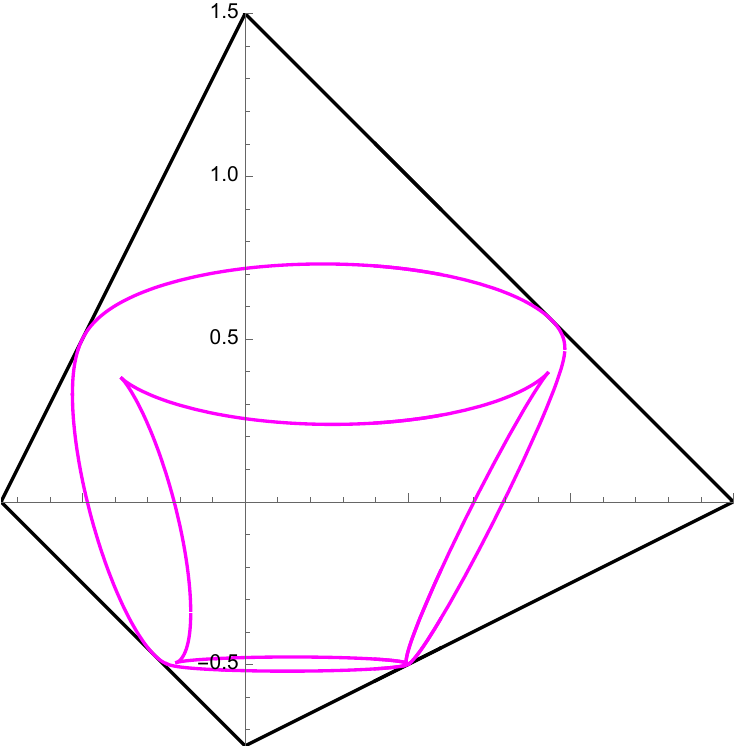}
		\subcaption[t]{\small $\sigma=1-\tau=15/16$}
	\end{subfigure}
	\begin{subfigure}{.3\textwidth}
		\centering
		\includegraphics[scale=0.225]{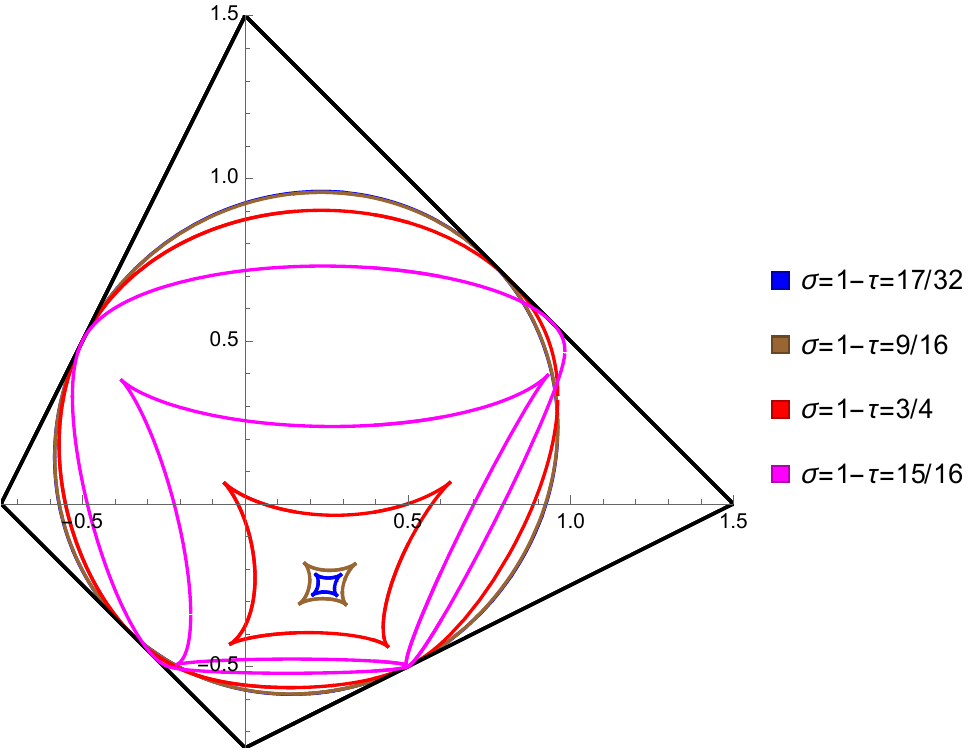}
		\subcaption[t]{Curves (A-D)}
	\end{subfigure}
	\begin{subfigure}{.3\textwidth}
		\centering
		\includegraphics[scale=0.2]{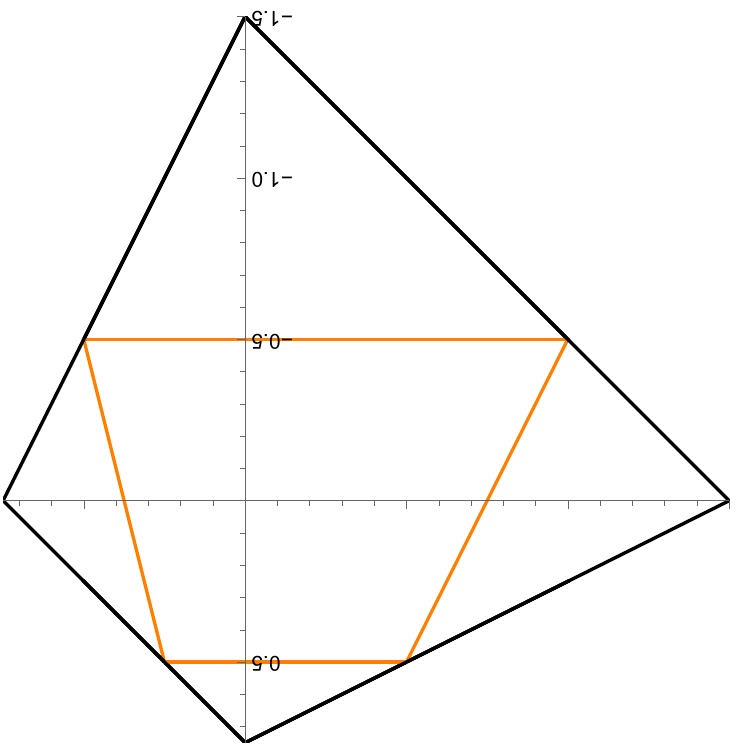}
		\subcaption[t]{\small $\sigma=1-\tau=0$}
	\end{subfigure}
	\caption{Arctic Curves for $(r,s,t)=(1,1,3)$ and $\sigma=1-\tau$}
	\label{tau=1-sigma}
	\end{figure}

\subsubsection{\textbf{The case $r,t$ not both odd}}\hfill\newline
\noindent{$\bullet$ \textbf{$\sigma,\tau$ general }}\hfill\newline
We fix the case of interest to be $(r,s,t)=(2,2,3)$. Recall that theorem \ref{ratioperthm}, the periodicity of $\mathcal L$ and $\mathcal R$ only dependent on $i$ modulo $2$. The periodicity lattice $\Lambda$ for the $\mathcal L$'s and $\mathcal R$'s is now generated only by $(2,0,0),(0,0,4t)$, i.e. $F=\{(i,m)\}=[0,1]\times [0,4t-1]$

The coefficients of the linear system \eqref{recurhopsys} may be organized into
$Q_m=(\mathcal L^{m}_{0},\mathcal L^{m}_{1})$ for $m \in [0,1,\cdots,4t-1]$ which also gives $|F|=24$, as $F=[0,1]\times [0,4t-1]$. The difference with section (\ref{odd_odd_section}) where both $r,t$ are odd is that each $Q_i$ contains only two values since $\mathcal L$ depends only on $i$ modulo $2$ but there are twice as many values of $m$ to be considered.
On $F$, the coefficients read
\begin{equation}
\begin{matrix}
Q_0=(\sigma,1-\sigma) &
Q_1=(\frac{1}{2},\frac{1}{2})&
Q_2=(\frac{1}{2},\frac{1}{2})&
Q_3=(\frac{1}{2},\frac{1}{2})&
Q_4=(\frac{1}{2},\frac{1}{2})&
Q_5=(1-\tau,\tau)\\
Q_6=(1-\sigma, \sigma)&
Q_7=(\frac{1}{2},\frac{1}{2})&
Q_8=(\frac{1}{2},\frac{1}{2})&
Q_9=(\frac{1}{2},\frac{1}{2})&
Q_{10}=(\frac{1}{2},\frac{1}{2})&
Q_{11}=(\tau,1-\tau)
\end{matrix}
\label{223_lattice}
\end{equation}
in terms of the variables $\sigma,\tau$ of \eqref{sigtau}.

The coefficient matrix for the linear system \eqref{223_lattice} is shown in \cite{mathfiles} (file ``Supplementary Material"). We apply a similar technique as previous section and obtain the following data. Notice that the choice $\sigma=\tau=\frac{1}{2}$ reduces to the uniform case, and we recover an ellipse with no inner region: 
$$P(u,v)|_{\sigma=\tau=1/2}=\left(20 u^2+24 u-18\right) \left(20 v^2+24 v-18\right)-(20 u v+12 u+12 v)^2$$
For general $\sigma$ and $\tau$, the expansion is up to order $\theta=8$, and the quantity $H_{223}(x,y,u,v)$ is a homogenous polynomial of order $8$ in $x,y$,
which necessitates seven iterations of the eliminating process of $H_{223}(x,y,u,v)$ and its derivative w.r.t. $x$ to obtain the final arctic curve $P(u,v)$. We end this section with some explicit examples of $(2,2,3)$-slanted $2 \times 2$ toroidal initial data. However, the detail of the computation for these cases is cumbersome and will be available only upon request. 
\begin{figure}
	\begin{subfigure}{0.2\textwidth}
	\centering
	\includegraphics[scale=0.25]{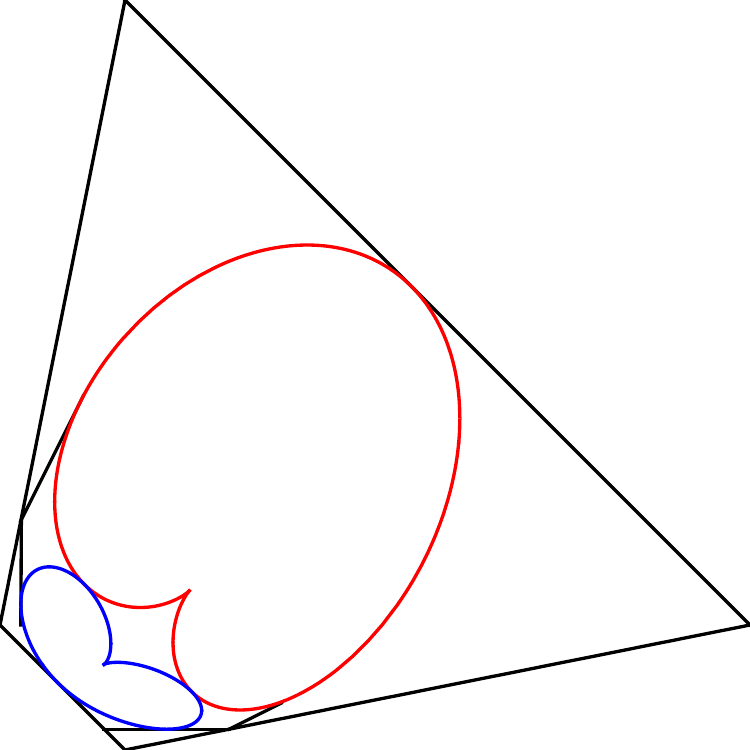}
	\subcaption[t]{\tiny $\tau=0, \sigma=1/4 $}
	\end{subfigure}\hskip 1cm
	\begin{subfigure}{0.2\textwidth}
	\centering
	\includegraphics[scale=0.25]{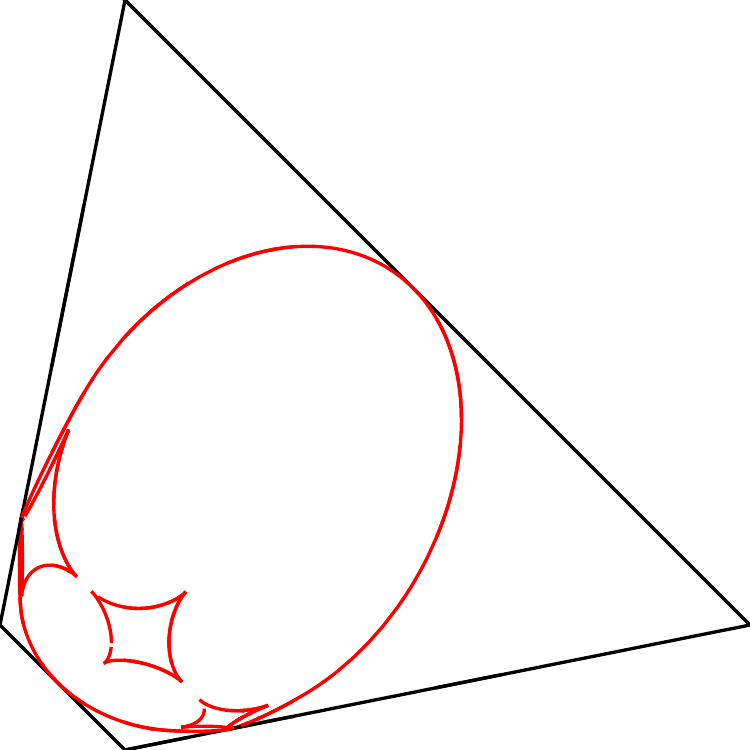}
	\subcaption[t]{\tiny $\tau=1/32,\sigma=1/4 $}
	\end{subfigure}\hskip 1cm 
	\begin{subfigure}{0.2\textwidth}
	\centering
	\includegraphics[scale=0.25]{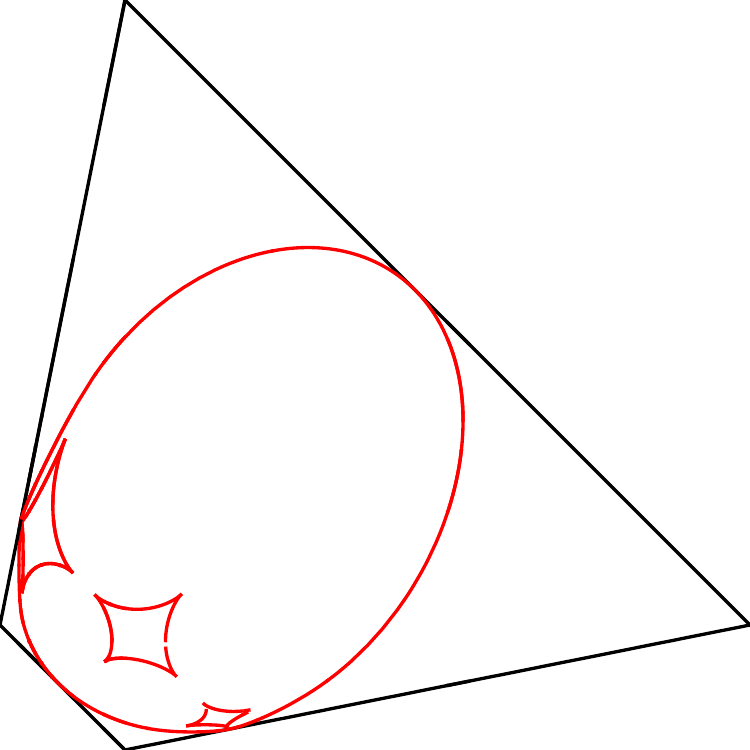}
	\subcaption[t]{\tiny $\tau=1/16,\sigma=1/4$}
	\end{subfigure}\\
	\begin{subfigure}{0.2\textwidth}
	\centering
	\includegraphics[scale=0.25]{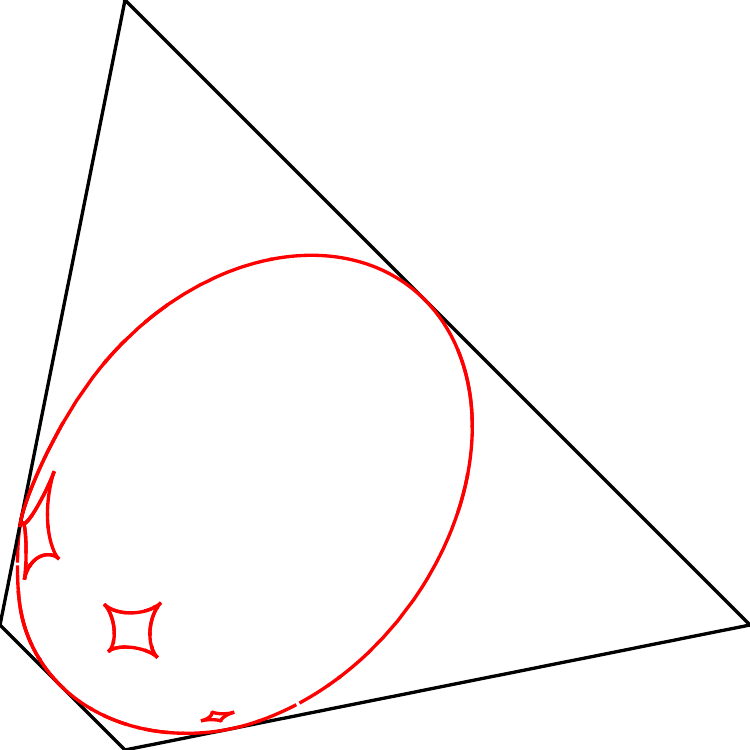}
	\subcaption[t]{\tiny $\tau=1/4,\sigma=1/4 $}
	\end{subfigure}\hskip .5cm
	\begin{subfigure}{0.2\textwidth}
	\centering
	\includegraphics[scale=0.25]{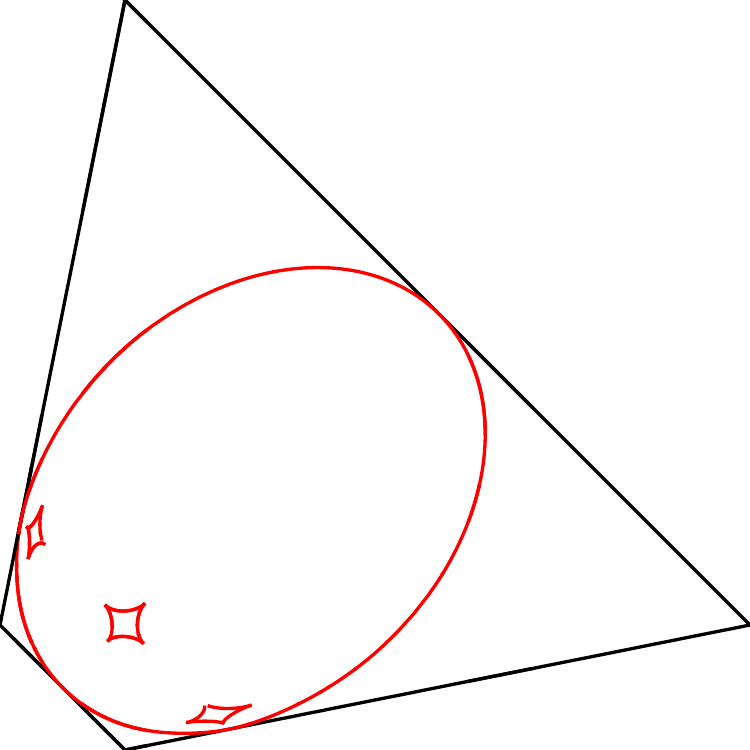}
	\subcaption[t]{\tiny $\tau=17/32,\sigma=1/4 $}
	\end{subfigure}\hskip .5cm
	\begin{subfigure}{0.2\textwidth}
	\centering
	\includegraphics[scale=0.25]{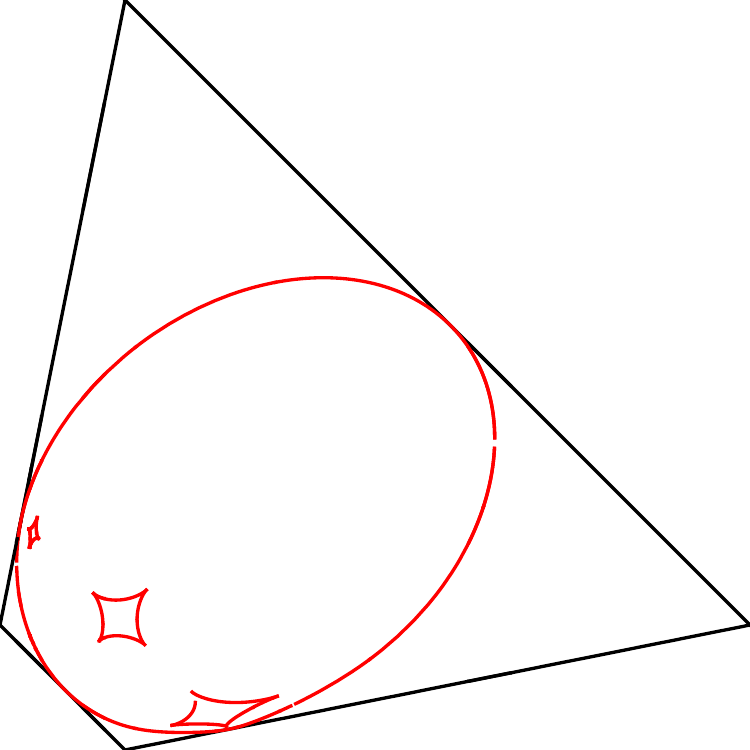}
	\subcaption[t]{\tiny $\tau=3/4,\sigma=1/4 $}
	\end{subfigure}\hskip.5cm
	\begin{subfigure}{0.2\textwidth}
	\centering
	\includegraphics[scale=0.25]{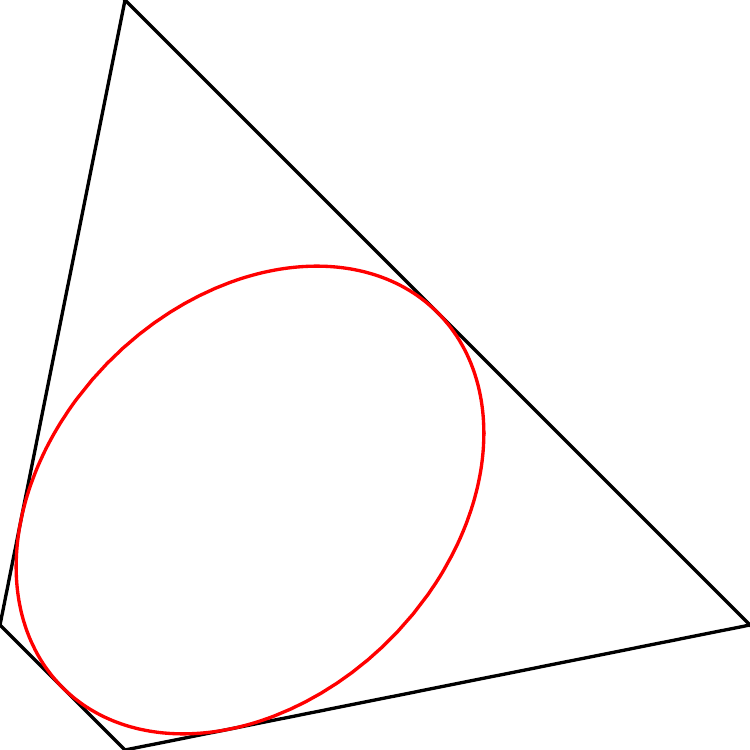}
	\subcaption[t]{\tiny $\tau=1/2,\sigma=1/2$}
	\end{subfigure}
	\caption{Some arctic curves for $(r,s,t)=(2,2,3)$}
	\label{general 223 library}
\end{figure}
%%%%%%

\noindent{$\bullet$ \textbf{$\tau=0$ case}}\hfill \newline
The polynomial $H(x,y,u,v)|_{\tau=0}$ reads:\newline
\begin{adjustbox}{scale={0.9}{0.9},center}
$\label{tau=0}
\begin{aligned}
	&H(x,y,u,v)|_{\tau=0}=16 (2 u x+2 v y+x-y) (2 u x+2 v y-x+y)\\
						&\times (32 u^3 x^3+96 u^2 v x^2 y+48 u^2 x^3+48 u^2 x^2 y+96 u v^2 x y^2+96 u v x^2 y+96 u v x y^2-24\sigma  u x^3 \\
   						&-2 u x^3+28 u x^2 y+24 \sigma  u x y^2-26 u x y^2+32 v^3 y^3+48 v^2 x y^2+48 v^2 y^3-24 \sigma  v x^2 y-2 v x^2 y\\
   						&+28 v x y^2+24 \sigma  v
   y^3-26 v y^3+4 \sigma  x^3-5 x^3-12 \sigma  x^2 y+9 x^2 y+12 \sigma  x y^2-3 x y^2-4 \sigma  y^3-y^3)\\ 
   &\times (160 u^3 x^3+480 u^2 v x^2 y-64
   \sigma  u^2 x^3-16 u^2 x^3+64 \sigma  u^2 x^2 y-80 u^2 x^2 y+480 u v^2 x y^2\\
   &-128 \sigma  u v x^2 y-32 u v x^2 y+128 \sigma  u v x y^2-160 u v x y^2+24
   \sigma  u x^3-34 u x^3+44 u x^2 y\\
   &-24 \sigma  u x y^2-10 u x y^2+160 v^3 y^3-64 \sigma  v^2 x y^2-16 v^2 x y^2+64 \sigma  v^2 y^3-80 v^2 y^3+24 \sigma 
   v x^2 y\\
   &-34 v x^2 y+44 v x y^2-24 \sigma  v y^3-10 v y^3+4 \sigma  x^3+x^3-12 \sigma  x^2 y+3 x^2 y+12 \sigma  x y^2-9 x y^2\\
   &-4 \sigma  y^3+5 y^3)
\end{aligned}
$
\end{adjustbox}
Note that in this case, the polynomial factors into two polynomials, each of order $3$ in $x,y$. This results in two higher degree curves, which delimit two tangent regions like in the previous section (see Fig.~\ref{223_tau=0} for an illustration).
\begin{figure}[H]
	\begin{subfigure}{0.2\textwidth}
	\centering
	\includegraphics[scale=0.3]{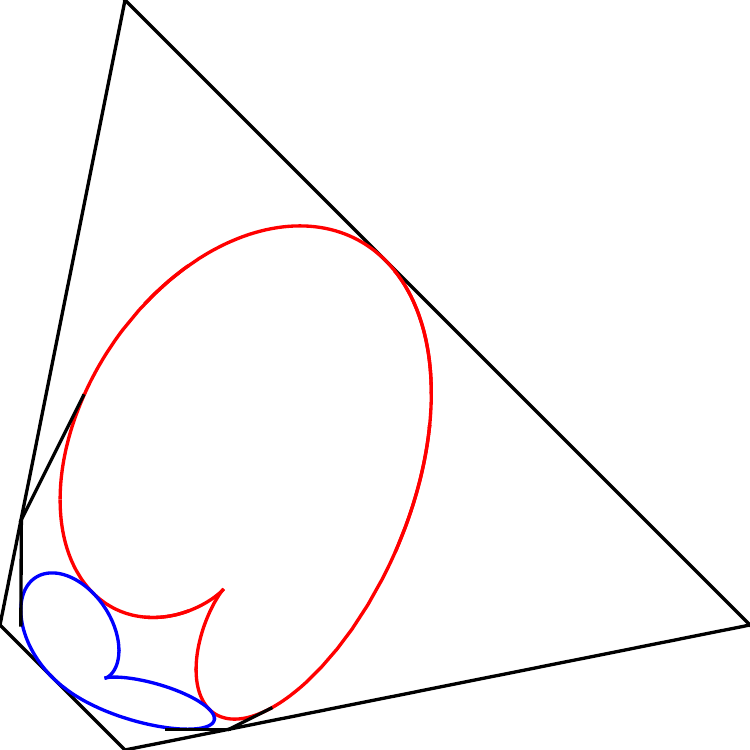}
	\subcaption[t]{\tiny $\tau=\sigma=0 $}
	\end{subfigure}\hskip 1cm
	\begin{subfigure}{0.2\textwidth}
	\centering
	\includegraphics[scale=0.3]{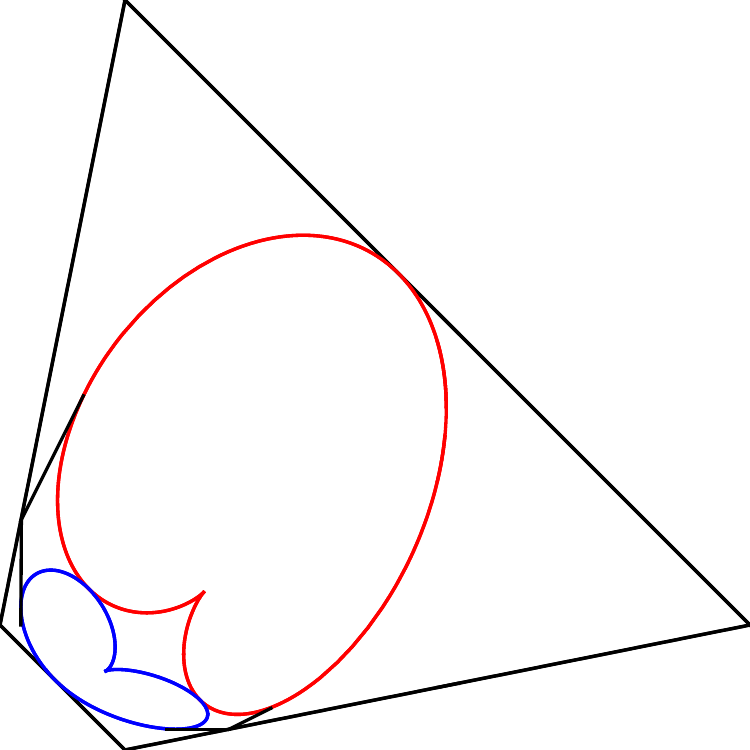}
	\subcaption[t]{\tiny $\tau=0,\sigma=1/8 $}
	\end{subfigure}\hskip 1cm 
	\begin{subfigure}{0.2\textwidth}
	\centering
	\includegraphics[scale=0.3]{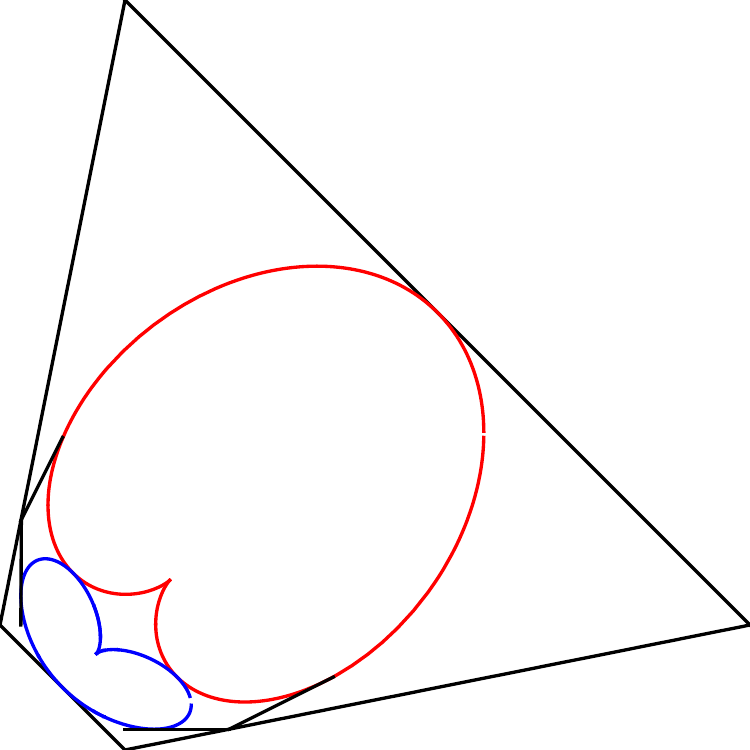}
	\subcaption[t]{\tiny $\tau=0,\sigma=1/2 $}
	\end{subfigure}\hskip 1cm
	\begin{subfigure}{0.2\textwidth}
	\centering
	\includegraphics[scale=0.3]{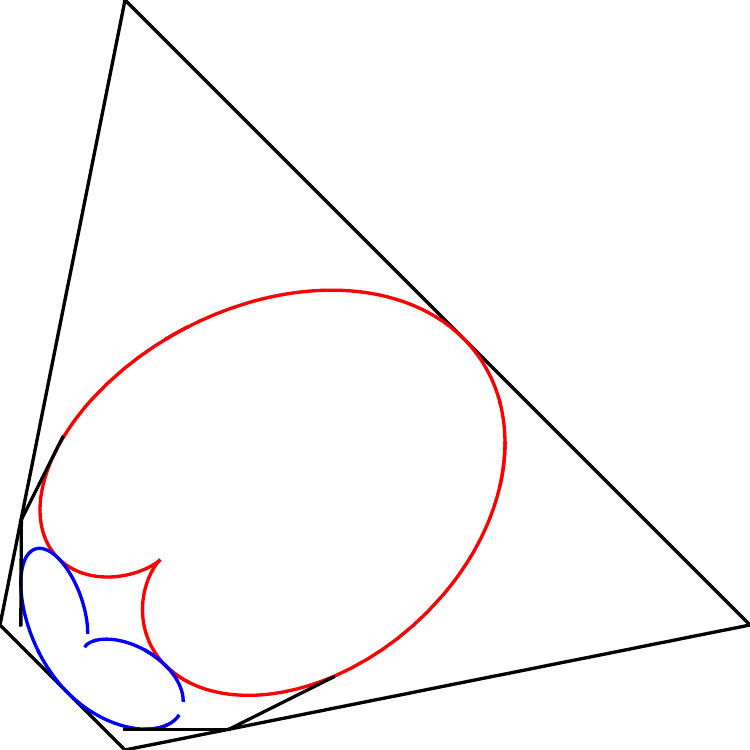}
	\subcaption[t]{\tiny $\tau=0,\sigma=3/4 $}
	\end{subfigure}
	\caption{Arctic curves for $(r,s,t)=(2,2,3)$ with fixed $\tau=0$}
	\label{223_tau=0}
\end{figure}

However, one different feature for this case is that when $\tau=\sigma=0$, the arctic curve no longer degenerates into a polygon.

\noindent{$\bullet$ \textbf{$\tau=\sigma$ case}}\hfill \newline
When $\tau=\sigma$, the leading coefficient $H_{223}(x,y,u,v)$ is a homogenous polynomial of degree $8$ in $x,y$ and the curve $P(u,v)$ is of degree $20$ in $u,v$ (see Fig.~\ref{random_223} for some illustration). Note the symmetry $\tau \to 1-\tau$ as expected from Sect. \ref{sec:sym}. 
%We will provide this computation in the Appendix B. 
%Therefore, we are not expecting any reduction in the number of inner regions.
\begin{figure}[H]
	\begin{subfigure}{0.17\textwidth}
	\centering
	\includegraphics[scale=0.2]{figures/223_toroidal/223_1_4_1_4.pdf}
	\subcaption[t]{\small $ \tau=1/4$}
	\end{subfigure}\hskip .3cm
	\begin{subfigure}{0.17\textwidth}
	\centering
	\includegraphics[scale=0.2]{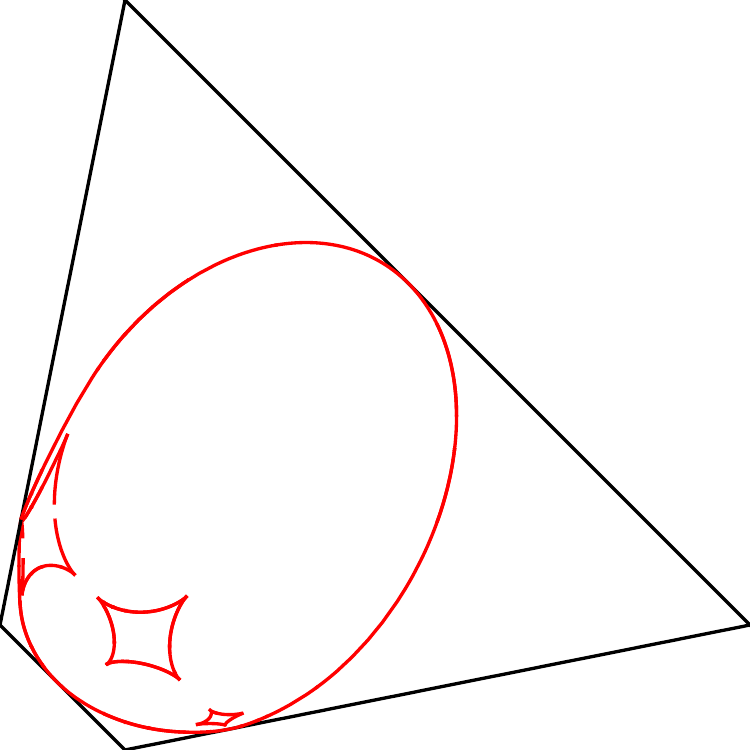}
	\subcaption[t]{\small $\tau=1/8$}
	\end{subfigure}\hskip .3cm
	\begin{subfigure}{0.17\textwidth}
	\centering
	\includegraphics[scale=0.2]{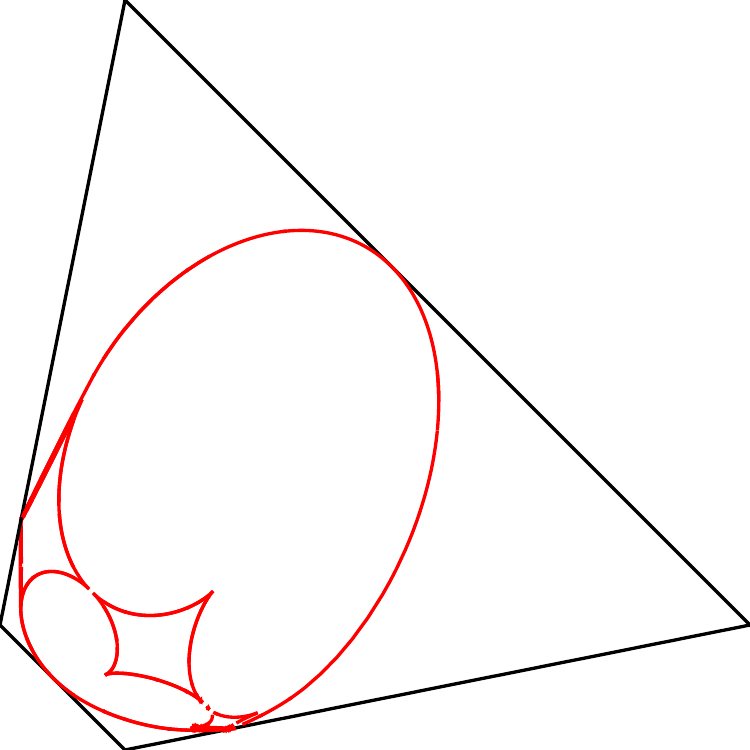}
	\subcaption[t]{\small $\tau=1/32$}
	\end{subfigure}\hskip .3cm
	\begin{subfigure}{0.17\textwidth}
	\centering
	\includegraphics[scale=0.2]{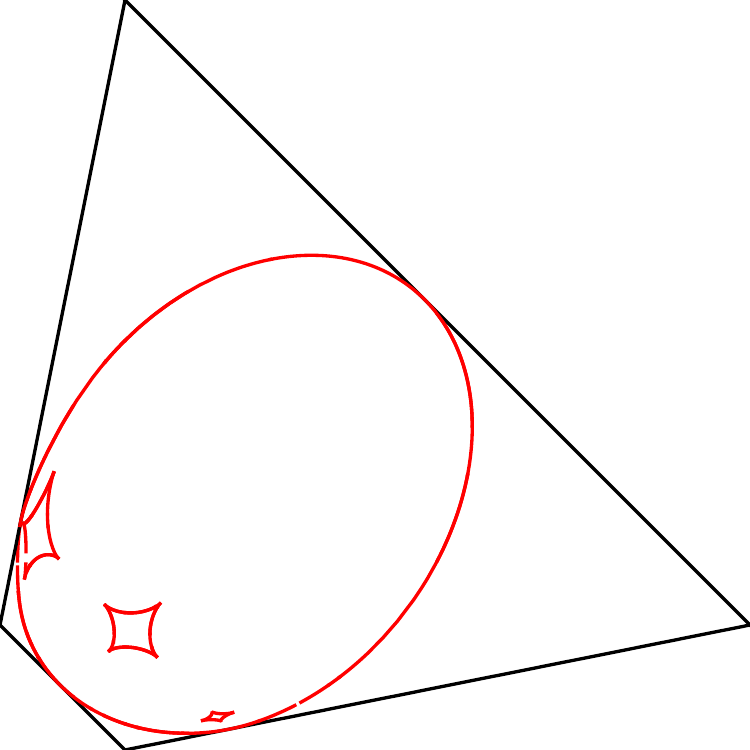}
	\subcaption[t]{\small $\tau=3/4$}
	\end{subfigure}\hskip .3cm
	\begin{subfigure}{0.17\textwidth}
	\centering
	\includegraphics[scale=0.2]{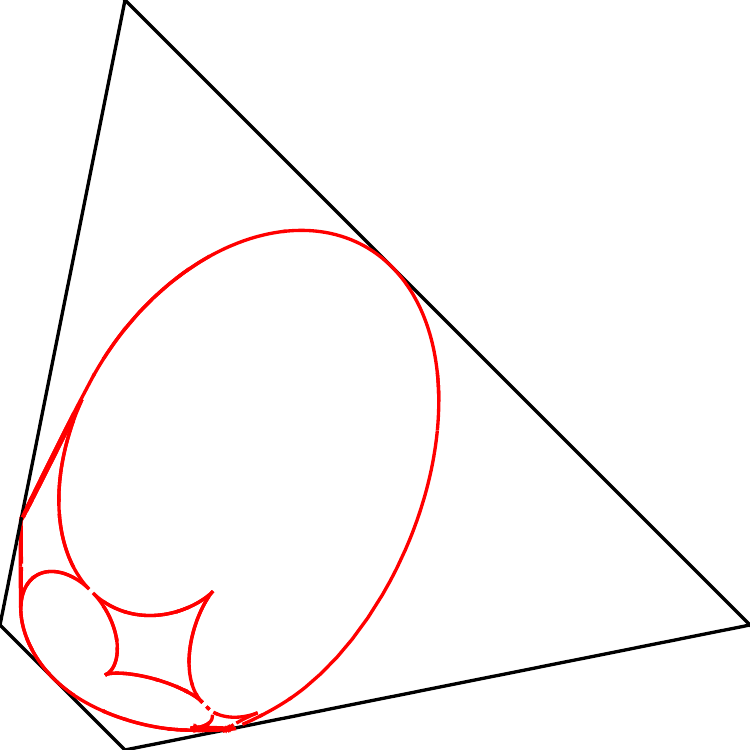}
	\subcaption[t]{\small $\tau=31/32$}
	\end{subfigure}
	\caption{Arctic curves for $(r,s,t)=(2,2,3)$ with some sample values of $\sigma=\tau$}
	\label{random_223}
\end{figure}

\subsubsection{\textbf{The case $r>1$, $t$  odd}}\hfill\newline
We fix the case of interest to be $(r,s,t)=(3,3,5)$. For the most general values of $\sigma,\tau$, the order of the expansion is $\theta=\theta_{3,3,5}=14$, but our computational capability does not provide credible resolution for the arctic curves. As the value of $\theta$ increases, we expect more inner regions within the scaled domain.
%	\begin{center}
%		\includegraphics[scale=.3]{plots/335_1:4_1:2.pdf}
%	\end{center}
However, as before, the calculation simplifies in special cases such as $\tau=0$ or $\tau=\sigma$ described below.

\noindent{$\bullet$ \textbf{$\tau=0$ and $\sigma$ arbitrary}} \hfill\newline
When $\tau=0$ and for all $\sigma$, the coefficient $H(x,y,u,v)$ is of the form:
\begin{equation}
% \resizebox{.8\textwidth}{!}{$
	\begin{aligned}
	H(x,y,u,v)|_{\tau=0}&= 4096 (2 u x+x+2 v y-y)^2 (2 u x-x+2 v y+y)^2\\
	&\times(2 u x+\sigma  x+2 x+2 v y+3 y-y \sigma ) (8 u x+\sigma  x-3 x+8 v y-2 y-y \sigma )\\
	&\times  H_1(x,y,u,v) \times H_2(x,y,u,v)
   \end{aligned}
	\label{eq:335_tau=0}
\end{equation}
where $H_1$ and $H_2$ are two factors of degree $4$ as polynomials of $x,y$. We obtain a similar situation as in previous section (see Fig.~\ref{fig:335_tau0}), where the arctic curve consists of two tangent components corresponding to the two factors  $H_1$ and $H_2$.
	\begin{figure}
		\begin{subfigure}{0.2\textwidth}
			\centering
			\includegraphics[width=2.cm]{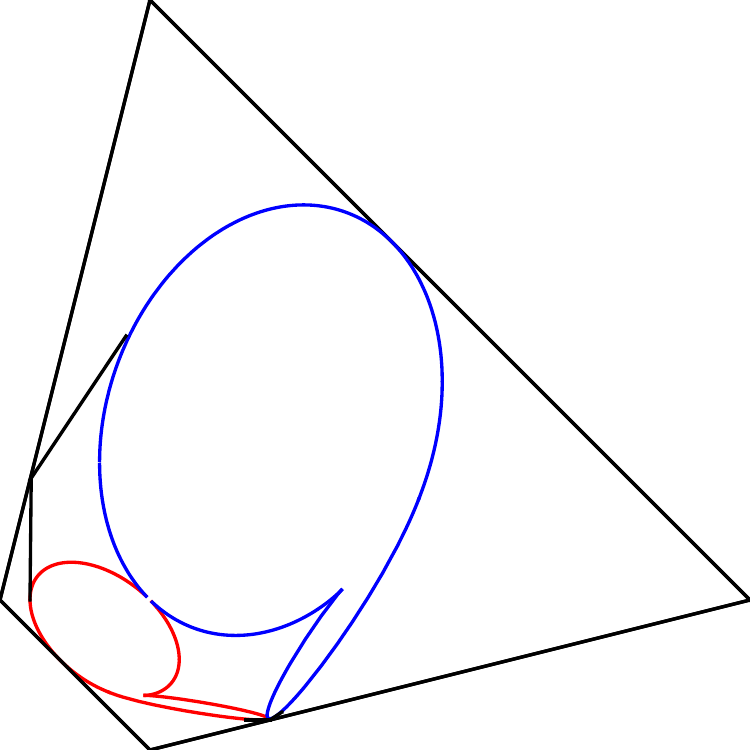}
			\subcaption[t]{$\tau=0,\sigma=1/64$}
		\end{subfigure}\hskip .5cm
		\begin{subfigure}{0.2\textwidth}
			\centering
			\includegraphics[width=2.cm]{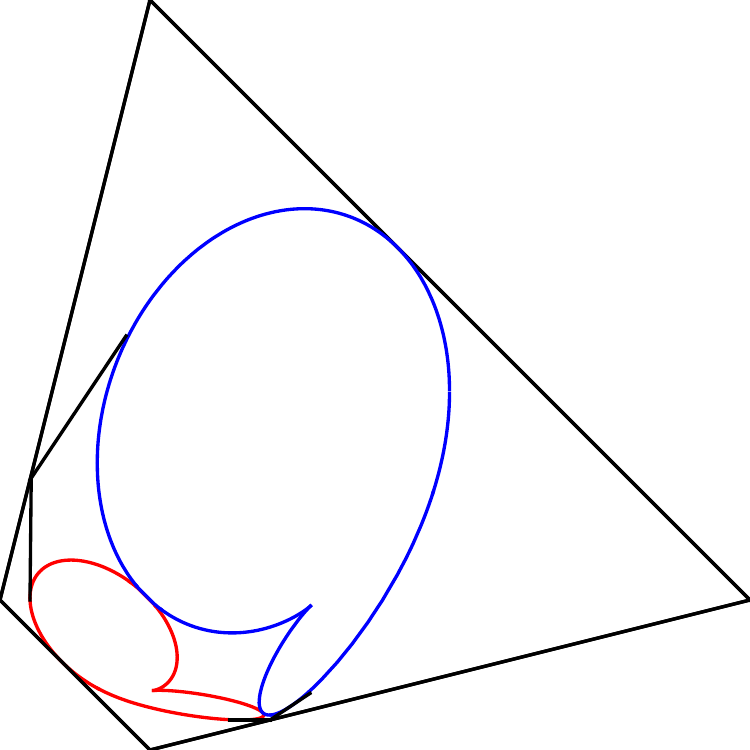}
			\subcaption[t]{$\tau=0,\sigma=1/16$}
		\end{subfigure}\hskip .5cm
		\begin{subfigure}{0.2\textwidth}
			\centering
			\includegraphics[width=2.cm]{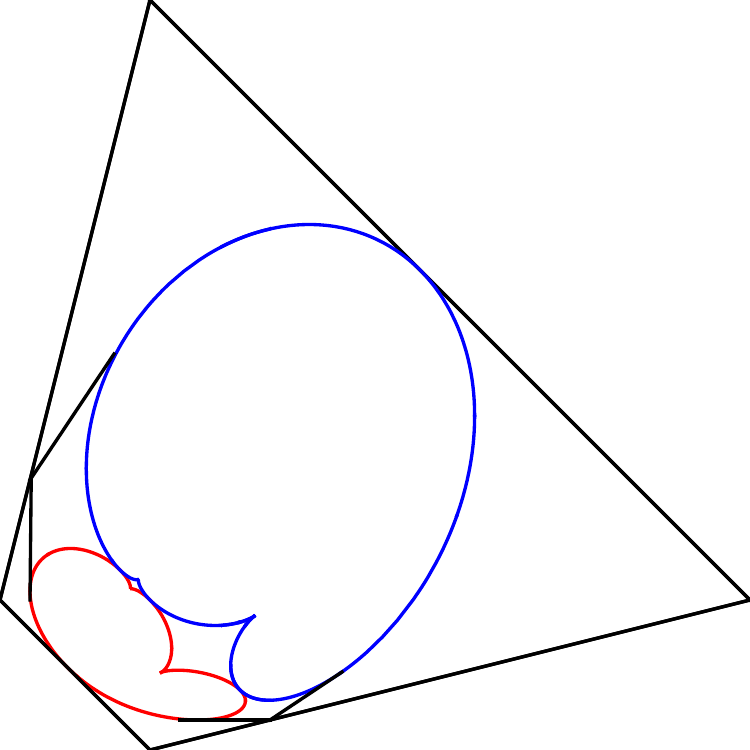}
			\subcaption[t]{$\tau=0,\sigma=1/4$}
		\end{subfigure}
		\begin{subfigure}{0.2\textwidth}
			\centering
			\includegraphics[width=2.cm]{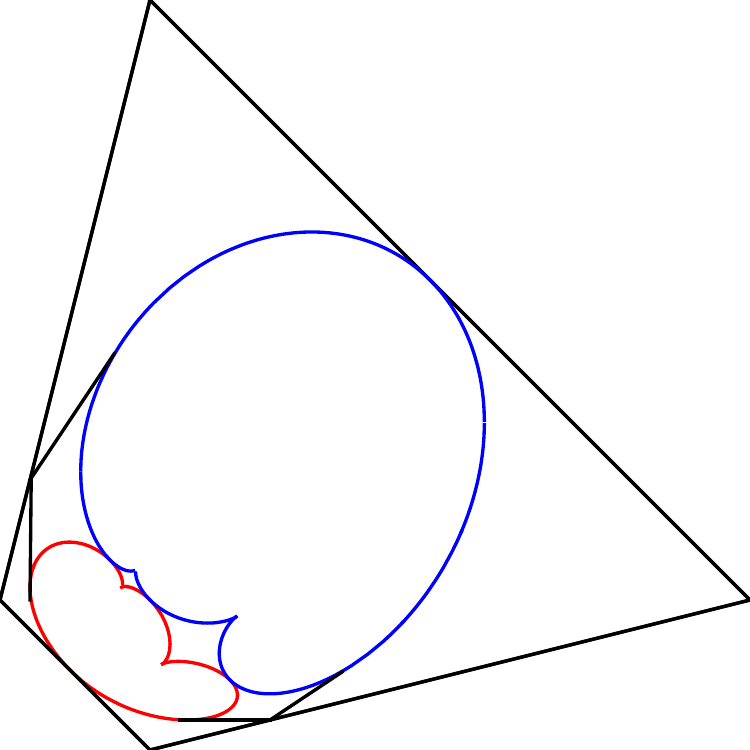}
			\subcaption[t]{$\tau=0,\sigma=1/3$}
		\end{subfigure}\\
		\begin{subfigure}{0.2\textwidth}
			\centering
			\includegraphics[width=2.cm]{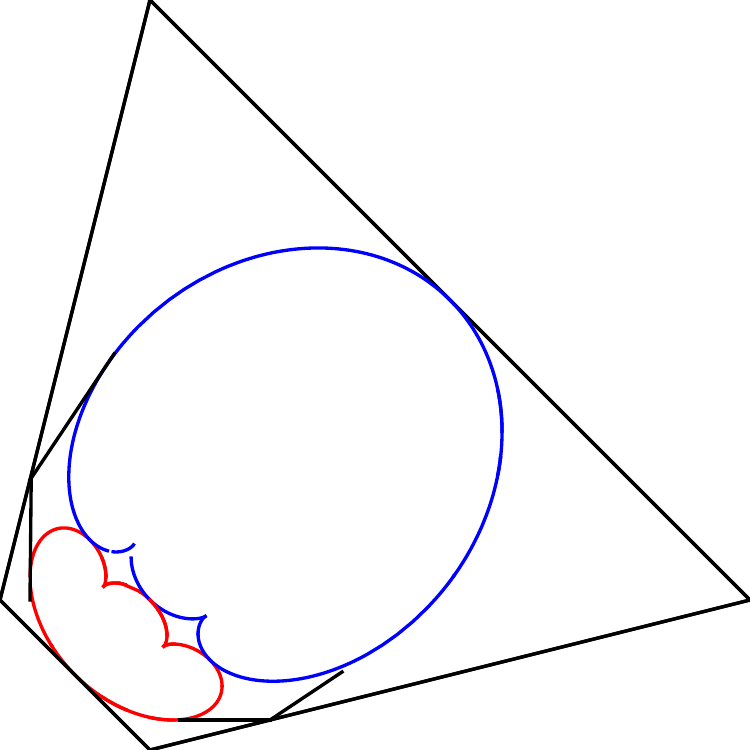}
			\subcaption[t]{$\tau=0,\sigma=1/2$}
		\end{subfigure}\hskip .5cm
		\begin{subfigure}{0.2\textwidth}
			\centering
			\includegraphics[width=2.cm]{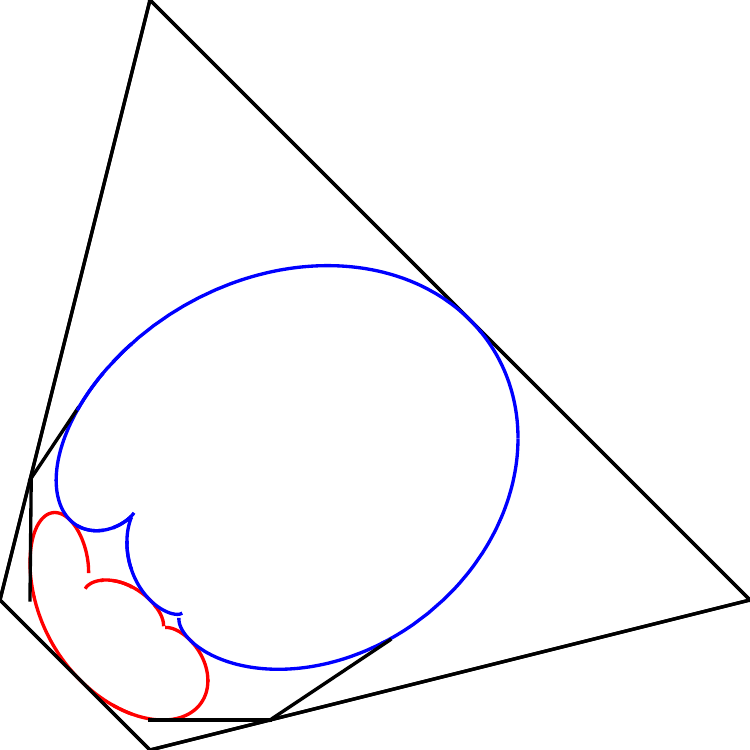}
			\subcaption[t]{$\tau=0,\sigma=2/3$}
		\end{subfigure}\hskip .5cm
		\begin{subfigure}{0.2\textwidth}
			\centering
			\includegraphics[width=2.cm]{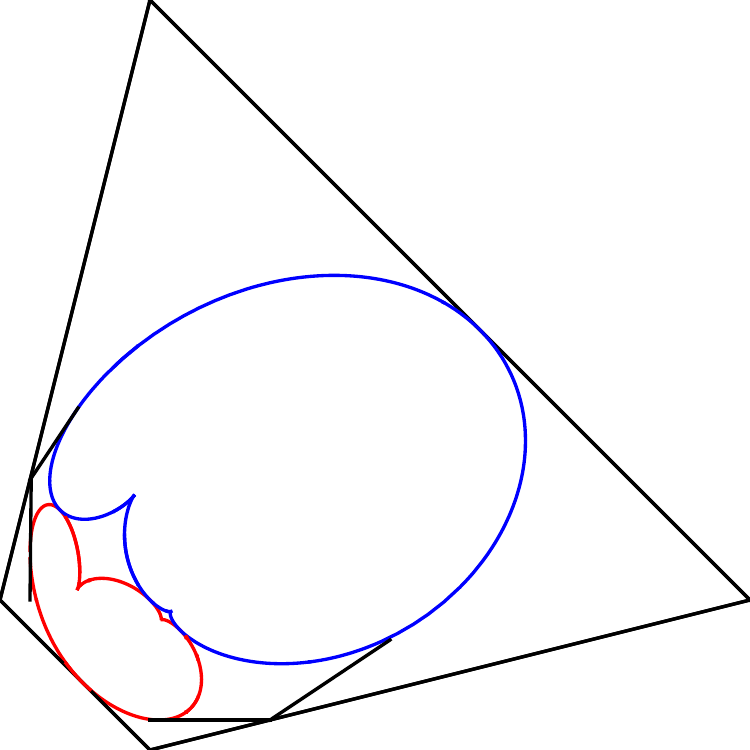}
			\subcaption[t]{$\tau=0,\sigma=3/4$}
		\end{subfigure}
		\caption{2 pieces curve of $(r,s,t)=(3,3,5), \tau=0$ }
		\label{fig:335_tau0}
	\end{figure}

\noindent{$\bullet$ \textbf{$\tau=\sigma$ arbitrary}} \hfill\newline
%We provide the coefficient polynomial in the appendix A. 
In this case, the factor of interest $H^*(x,y,u,v)$ in the coefficient $H(x,y,u,v)$ is of degree $6$ in $x,y$ (see \cite{mathfiles} for details), resulting in two inner regions (see Fig: \ref{fig:335_tau=sigma}).
\begin{figure}
		\begin{subfigure}{0.2\textwidth}
			\centering
			\includegraphics[width=2.5cm]{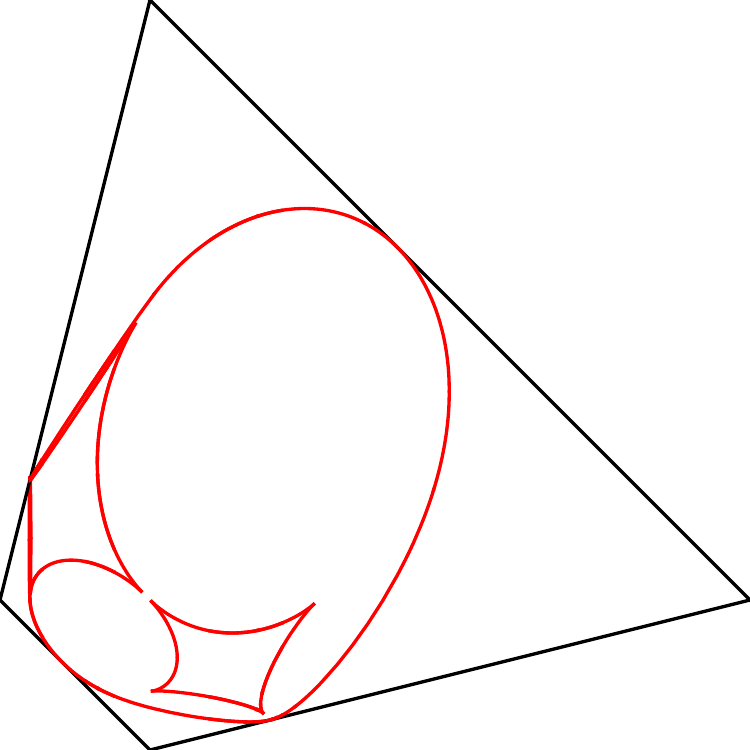}
			\subcaption[t]{$\tau=\sigma=1/32$}
		\end{subfigure}\hskip .5cm
		\begin{subfigure}{0.2\textwidth}
			\centering
			\includegraphics[width=2.5cm]{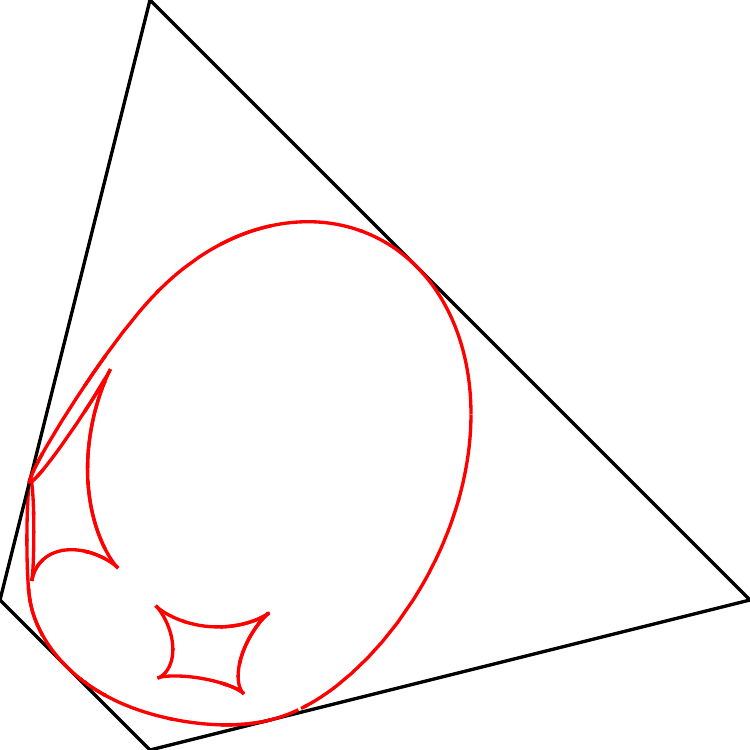}
			\subcaption[t]{$\tau=\sigma=1/8$}
		\end{subfigure}\hskip .5cm
		\begin{subfigure}{0.2\textwidth}
			\centering
			\includegraphics[width=2.5cm]{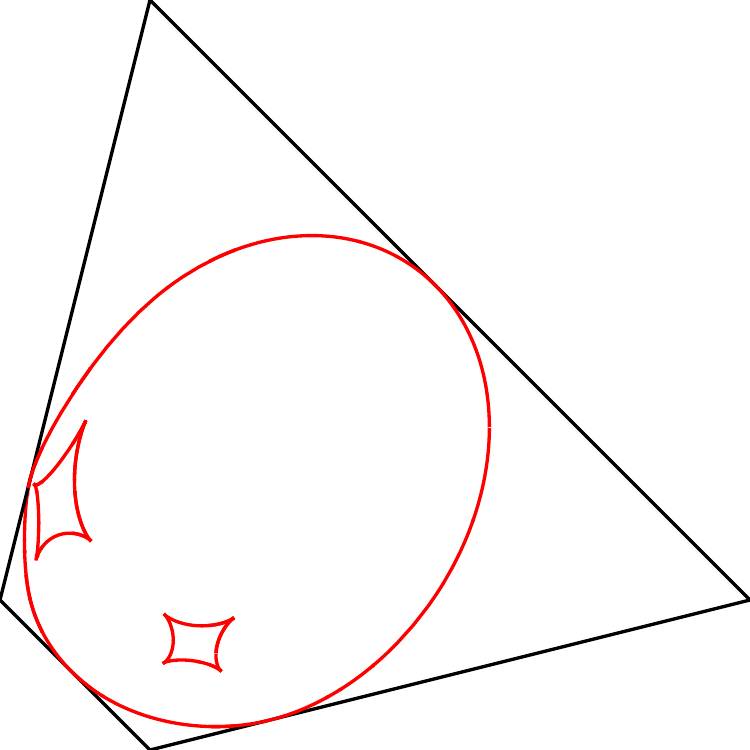}
			\subcaption[t]{$\tau=\sigma=1/4$}
		\end{subfigure}
		\begin{subfigure}{0.2\textwidth}
			\centering
			\includegraphics[width=2.5cm]{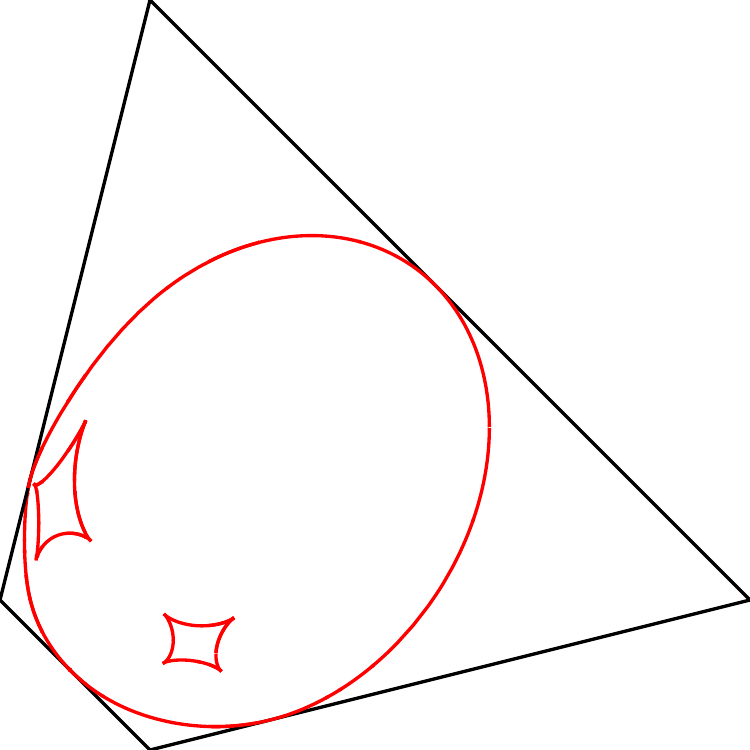}
			\subcaption[t]{$\tau=\sigma=3/4$}
		\end{subfigure}\\
		\caption{2 pieces curve of $(r,s,t)=(3,3,5), \tau=\sigma$ }
		\label{fig:335_tau=sigma}
	\end{figure}
%%%%%%%%%%%%%%%%
\section{A Holographic principle for $T$-system limit shapes}\label{sec:holo}

\subsection{General principle}\label{sec:hologen}
In this paper we have studied new exact solutions of the $T$-system, with $(r,s,t)$-initial data specified along parallel planes
prependicular to some direction $(r,s,t)$. The ``flat" case studied in \cite{DiFrancesco1} 
corresponds to $r=s=0$, $t=1$, and used different solutions of the $T$-system with various periodicities. 
In this section, starting from a solution of the $T$-system for some $(r,s,t)$-initial,
we re-interpret the {\it same} solution as corresponding to some different initial data along some $(\tilde r,\tilde s,\tilde t)$-planes. The latter is simply dictated by the exact solution of the former, however simple $(r,s,t)$ initial data (such as the uniform case studied in Section \ref{uniform solution}) naturally lead to highly non-uniform and more complicated initial data
in arbitrary $(\tilde r,\tilde s,\tilde t)$-planes. In particular, when $(\tilde r,\tilde s,\tilde t)=(0,0,1)$ our 2x2 periodic solutions give rise to new solutions of the $T$-system with ``flat " initial data in the planes $k=0,1$, which were not considered in \cite{DiFrancesco1}. However the holographic principle described below allows to derive arctic 
curves for those as well.

Using the new $(\tilde r,\tilde s,\tilde t)$ initial data settings, the solution of the $T$-system is interpreted as partition function for dimers on $(\tilde r,\tilde s,\tilde t)$ pinecones, whose limit shape is governed by the {\it same} equations as the original $(r,s,t)$ setting. In particular, the determinant of the linear system for the density $\rho$, 
which is a function $D_{r,s,t}(x,y,z)$ of $(r,s,t)$, 
remains the same. However, due to the new interpretation, we must apply the rescaling \eqref{denom} with the
new values $(\tilde r,\tilde s,\tilde t)$ instead of $(r,s,t)$, therefore the singularities of the density generating function $\rho(x,y,z)$ are governed by the function:
\begin{equation}
\label{holo}
\Delta_{r,s,t}^{\tilde r,\tilde s,\tilde t}(x,y,z)=D_{r,s,t}(z^\frac{\tilde r}{\tilde t}/x,z^\frac{\tilde s}{\tilde t}/y,z)
\end{equation}

This amounts to simply changing the ``point of view" 
on the same solution of the $T$-system, namely considering it from the perspective of another direction $(\tilde r,\tilde s,\tilde t)\neq (r,s,t)$, providing a sort of holographic view on the former result. 

In the next sections, we illustrate this holographic principle for uniform and 2x2 periodic cases.

\subsection{The uniform case}\label{sec:holo-uniform}

To illustrate the holographic concept, let us first consider the simplest case of uniform $(r,s,t)$-plane initial data
viewed from the ``flat" perspective with $(\tilde r,\tilde s,\tilde t)=(0,0,1)$.
Specifically, we re-interpret the solution $T_{i,j,k}$ of the $T$-system with $(r,s,t)$-slanted uniform initial data in section \ref{uniform solution} as new ``flat" initial data the solution for $k=0,1$. This new initial data reads $\tilde t_{i,j}=T_{i,j,{\rm Mod}(i+j,2)}$, with:
\begin{align*}
		T_{i,j,0}&=\alpha^{\frac{(ri+sj)(ri+sj-1)}{2}}\qquad (i+j=0\, {\rm mod}\, 2)\\
		T_{i,j,1}&=\alpha^{\frac{(ri+sj+t)(ri+sj+t-1)}{2}}\qquad (i+j=1\, {\rm mod}\, 2)
\end{align*}
with $\al>1$ as in \eqref{alphaeq}.
This non-uniform initial data $\tilde t_{i,j}$ on the flat stepped surface $k_{i,j}=i+j$ mod 2 depends explicitly on $(i,j,k)$
via the quantity $ri+sj+tk$ and on $(r,s,t)$ via $\al$ as well. 
Applying the rescaling \eqref{holo} with $(\tilde r,\tilde s,\tilde t)=(0,0,1)$, and repeating the derivation of the dual curve as in Section \ref{asymptotic of the density}, we find that the arctic curve \eqref{arcticu} is replaced with:
\begin{equation}\label{nonuniform_flat}
		(1-A)u^2+Av^2-A(1-A)=0
\end{equation}
where $A=A_{r,s,t}$ as in \eqref{adef}. Note that the result amounts to substituting  $(r,s,t)\to (\tilde r,\tilde s,\tilde t)=(0,0,1)$ in \eqref{arcticu}, while keeping the value of $A=A_{r,s,t}$ unchanged. 
The above is a family of ellipses parameterized by $A$, inscribed in the Aztec square $|u|+|v|=1$.
%Be aware that this is not the $r,s,t$ in the new non-uniform initial data (which we already makes clear by having $\tilde r, \tilde s, \tilde t$), but rather the $r,s,t$ encoded the scaled domain from the initial stepped surface. 
Note that when $r=s$, \eqref{nonuniform_flat} reduces to the artic circle, as $a=1-a=\frac{1}{2}$. 
%By the simlar process, we would have a family of ellipes inscribed in the rescaled domain, the square $|u|+|v|=1$, when $\tilde r\neq \tilde s$. 
We provide computational evidence for this observation in Fig. \ref{tildeuniform} for a few values of $(r,s,t)$. 
\begin{figure}[H]
	\begin{subfigure}{0.2\textwidth}
	\centering
	\includegraphics[scale=0.2]{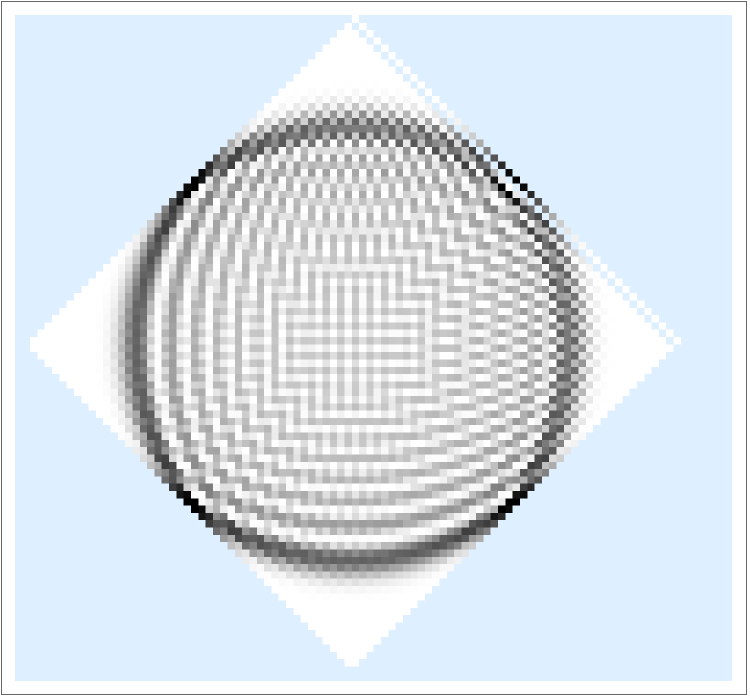}
	\subcaption[t]{\tiny $(r,s,t)=(1,1,3)$}
	\end{subfigure}\hskip 0.5cm
	\begin{subfigure}{0.2\textwidth}
	\centering
	\includegraphics[scale=0.2]{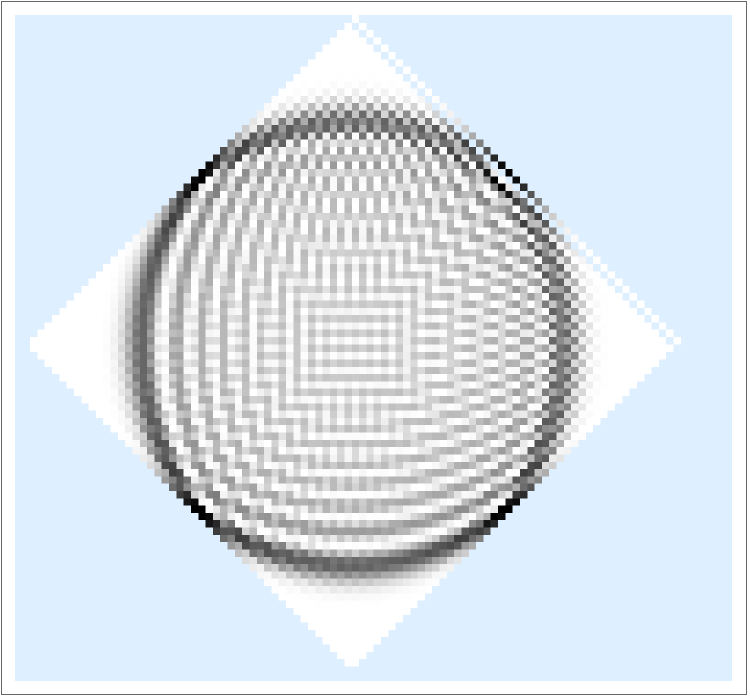}
	\subcaption[t]{ \tiny $(r,s,t)=(1,2,5)$}
	\end{subfigure}\hskip .5cm 
	\begin{subfigure}{0.2\textwidth}
	\centering
	\includegraphics[scale=0.2]{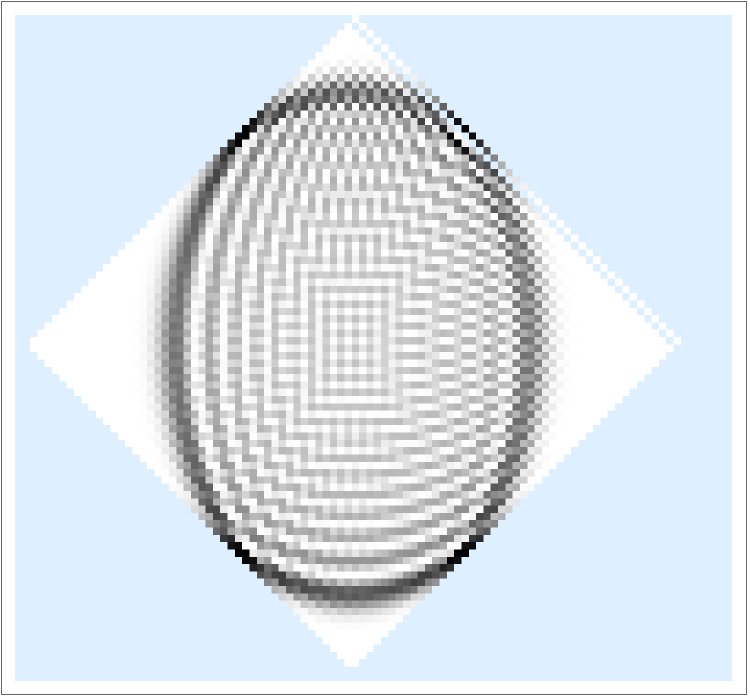}
	\subcaption[t]{\tiny $(r,s,t)=(1,7,9)$}
	\end{subfigure}\hskip .5cm
	\begin{subfigure}{0.2\textwidth}
	\centering
	\includegraphics[scale=0.2]{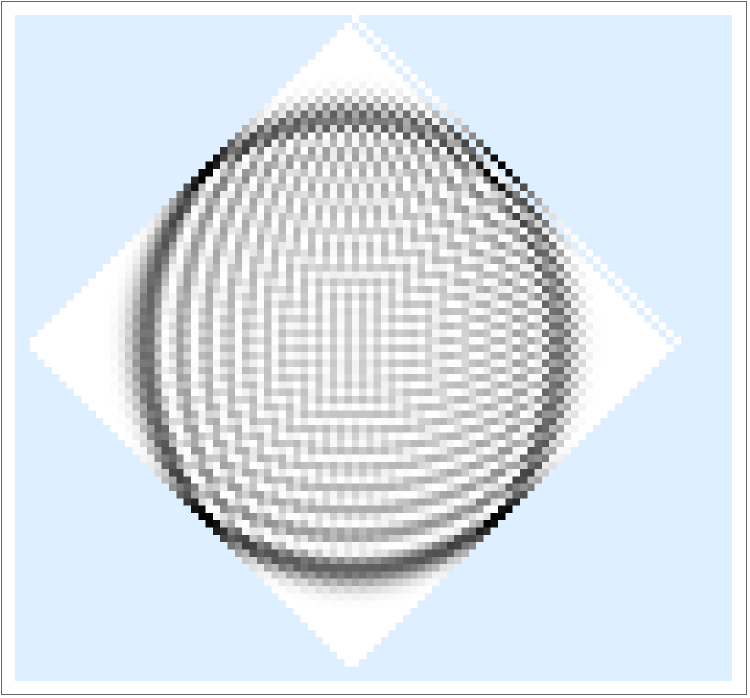}
	\subcaption[t]{\tiny $(r,s,t)=(1,7,15) $}
	\end{subfigure}\\\vskip 0.5cm
	\begin{subfigure}{0.2\textwidth}
	\centering
	\includegraphics[scale=0.2]{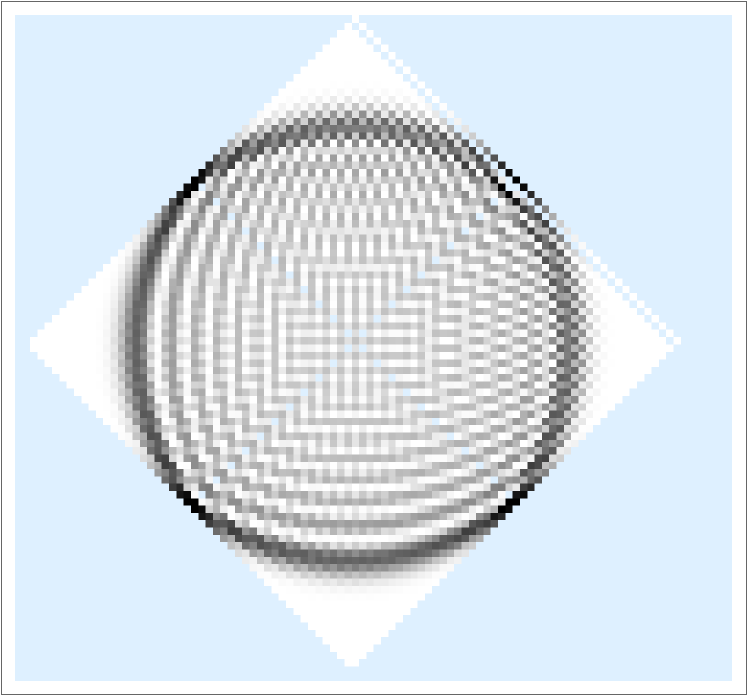}
	\subcaption[t]{\tiny $(r,s,t)=(2,2,9)$}
	\end{subfigure}\hskip .5cm
	\begin{subfigure}{0.2\textwidth}
	\centering
	\includegraphics[scale=0.2]{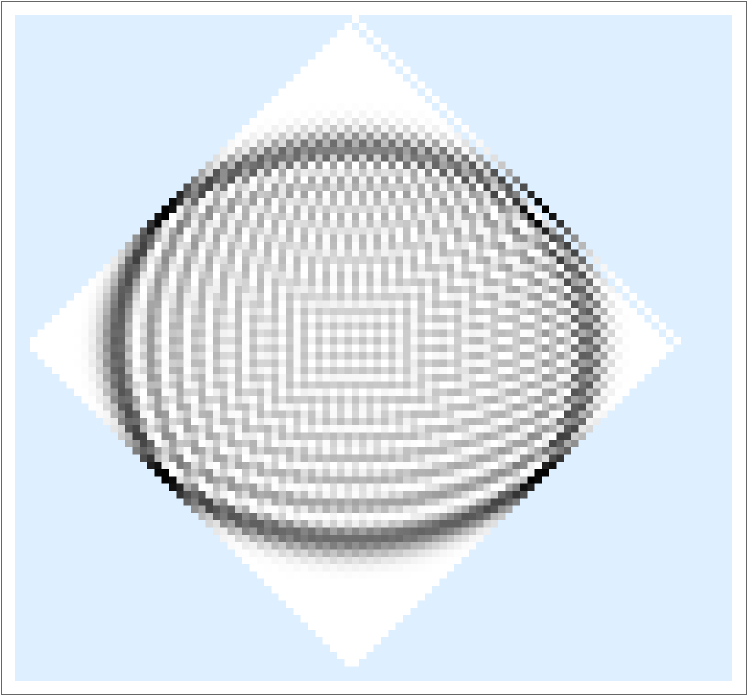}
	\subcaption[t]{\tiny $(r,s,t)=(5,0,8)$}
	\end{subfigure}\hskip .5cm
	\begin{subfigure}{0.2\textwidth}
	\centering
	\includegraphics[scale=0.2]{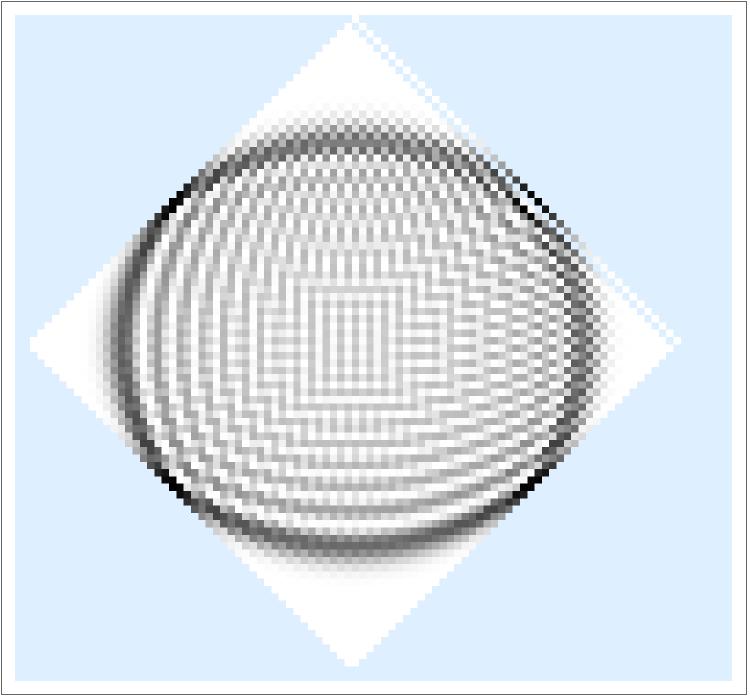}
	\subcaption[t]{\tiny $(r,s,t)=(7,0,13)$}
	\end{subfigure}\hskip .5cm
	\begin{subfigure}{0.2\textwidth}
	\centering
	\includegraphics[scale=0.2]{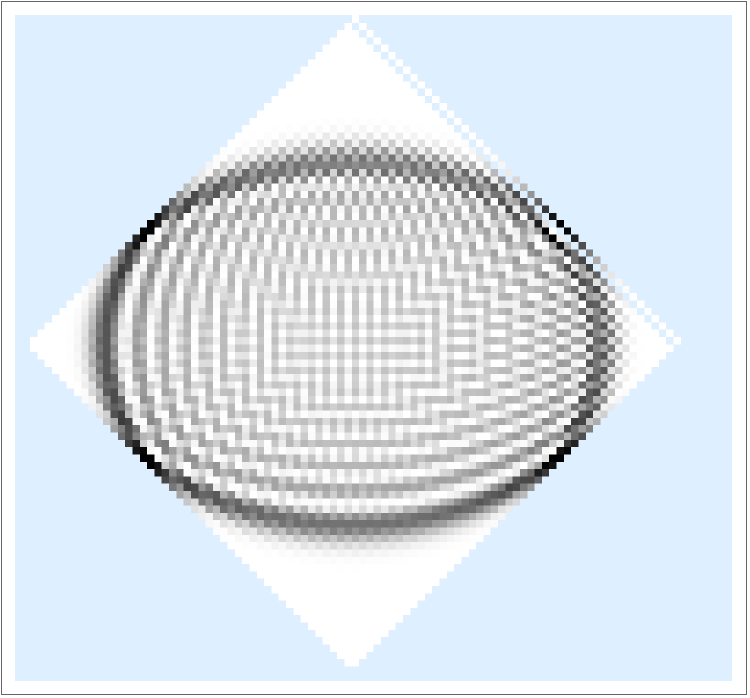}
	\subcaption[t]{\tiny $(r,s,t)=(3,0,4)$}
	\end{subfigure}
	\caption{\tiny{Density profile for given $(r,s,t)$-slanted values viewed from the $(0,0,1)$ perspective on the scaled domain $|u|+|v|=1$. We have represented the coefficient $\rho_{i,j,k}$ of the generating series $\rho(x,y,z)$ of order $k=45$ and $k=46$, written as a function of $(i/k,j/k)$, for $i \in \{-45, \cdots, 45\}$ and $j \in \{-45, \cdots, 52\}$} }
\label{tildeuniform}
\end{figure}

More generally, let us examine the $(r,s,t)$-uniform solution from the generic $(\tilde r,\tilde s,\tilde t)$ point of view. Recall the value \eqref{dfunc} of the denominator $D_{r,s,t}(x,y,z)$ of the density generating function $\rho^{(0,0,0)}(x,y,z)$ for the uniform $(r,s,t)$-slanted case.
%, the denominator of the density generating function is of the form: 
%$$D_{r,s,t}(x,y,z):=1+z^2-z\alpha^{r^2-t^2}\Big(x+\frac{1}{x}\Big)-z\alpha^{s^2-t^2}\Big(y+\frac{1}{y}\Big)$$ 
%where $\al$ satisfying the equation $\al^{t^2}=\al^{r^2}+\al^{s^2}$. 
Using the new rescaling \eqref{holo}, we get: 
$$\Delta_{r,s,t}^{\tilde r, \tilde s, \tilde t}(x,y,z):=D_{r,s,t}(z^{\tilde r/\tilde{t}}x^{-1},z^{\tilde s/\tilde{t}}y^{-1},z), $$
and expanding at leading order in $\lambda$, we find:
$$ \Delta_{r,s,t}^{\tilde r, \tilde s, \tilde t}(e^{\lambda x},e^{\lambda y},e^{-\lambda(ux+ vy)})= \lambda^2\, H(x,y,u,v)  +O(\lambda^4)$$
We again have an explicit polynomial $H(x,y,u,v)$ for which the vanishing Hessian condition leads to the following family of ellipses:
% \color{red}
% [xxx   probably the same as \eqref{arcticu} but with a unchanged while $(r,s,t)\to (\tilde r,\tilde s,\tilde t)$.]
% \color{black}
\begin{equation}\label{arctictttilde}
(1-A)\,{\tilde t}^2\,u^2+A\,{\tilde t}^2\,v^2-A \,(1-A)\,({\tilde r}u+{\tilde s}  v+{\tilde t})^2=0
\end{equation}
with $A=A_{r,s,t}$ as in \eqref{adef}.

\subsection{$2\times 2$-toroidal case}
We now illustrate the holographic principle in the case of $(r,r,t)$-$2 \times 2$ periodic initial data as in \eqref{initdataik}, with the solution as in Section \ref{periodic lemma}, viewed from the ``flat" perspective with 
$(\tilde r,\tilde s,\tilde t)=(0,0,1)$. 
As before, the new initial data is simply $\tilde t_{i,j}=T_{i,j,{\rm Mod}(i+j,2)}$, where $T_{i,j,0}$ and $T_{i,j,1}$ are the solutions of the $2 \times 2$-toroidal, $(r,r,t)$-slanted $T$-system. 

For the remainder of this section, let us restrict to the case $(r,r,t)=(1,1,3)$. 
As explained above, the singularities of the density function are governed by the rescaled determinant of the 
system $\Delta_{113}^{0,0,1}(x,y,z):=D_{113}(x^{-1},y^{-1},z)$ with $D_{1,1,3}=D$ as in Sect. \ref{odd_odd_section}. 
This leads to new dual arctic curves inscribed in the Aztec square $|x|+|y|=1$. We now follow the sequence of results of Section \ref{toroidal example section}, which we reinterpret in the $(\tilde r,\tilde s,\tilde t)=(0,0,1)$ setting.
\subsubsection{\textbf{case $\tau =0$.}}\hfill\newline
For $\tau=0$, we again obtain the same type of inner curves but now with the scaling domain $|u|+|v|=1$ of the Aztec Diamond (see Fig.~\ref{flat_nonuniform_tau0}). The leading order terms at $\lambda^{10}$ and depending on $\sigma,x,y$ takes the form: 
$$\resizebox{.95\hsize}{!}{$\begin{aligned}
&	H(x,y)=1024 (2 u x+2 v y+x-y)^2 (2 u x+2 v y-x+y)^2 \\
&\times(3 u x+3 v y-\sigma  x-x+\sigma  y-2 y) (3 u x+3 v y-\sigma  x+2 x+\sigma  y+y) \\
&\times\left(6 u^2 x^2+12 u v x y+4 \sigma  u x^2+u x^2-4 \sigma  u x y+5 u x y+6
   v^2 y^2+4 \sigma  v x y+v x y-4 \sigma  v y^2+5 v y^2+3 \sigma  x^2-2 x^2+x y-3 \sigma  y^2+y^2\right) \\
 &\times \left(6 u^2 x^2+12 u v x y+4 \sigma  u x^2-5 u x^2-4 \sigma  u x y-u x y+6 v^2 y^2+4 \sigma  v x y-5 v x y-4 \sigma  v y^2-v y^2-3 \sigma  x^2+x^2+x y+3 \sigma  y^2-2 y^2\right)
\end{aligned}$} $$ 
By the same techniques as before, we obtain similar curves as in Section  \ref{tau=0}, but with the scaled domain for $(\tilde r,\tilde s,\tilde t)=(0,0,1)$.
\begin{figure}
\begin{subfigure}{0.2\textwidth}
	\centering
    \includegraphics[scale=.2]{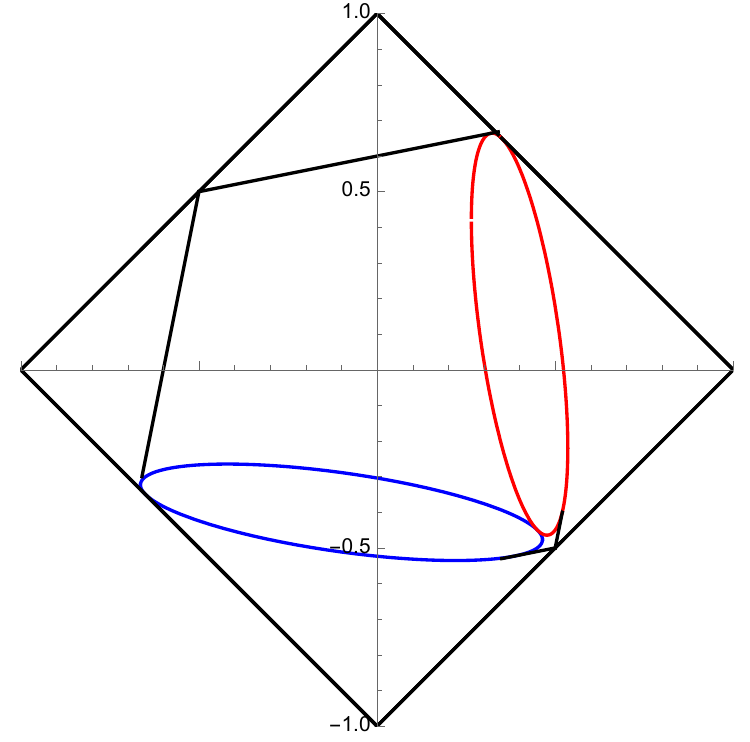}
    \subcaption[t]{\tiny $\tau=0$, $\sigma=1/20$} 
\end{subfigure}\hskip 2cm
\begin{subfigure}{0.2\textwidth}
		\centering
        \includegraphics[scale=.2]{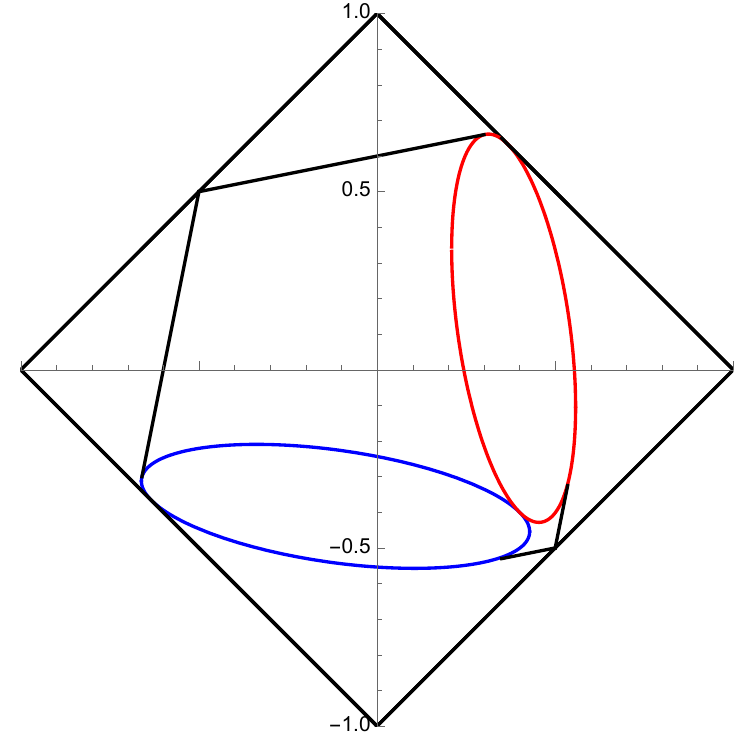}
        \subcaption[t]{\tiny $\tau=0$, $\sigma=1/10$} 
\end{subfigure}\hskip 2cm
\begin{subfigure}{0.2\textwidth}
		\centering
        \includegraphics[scale=.2]{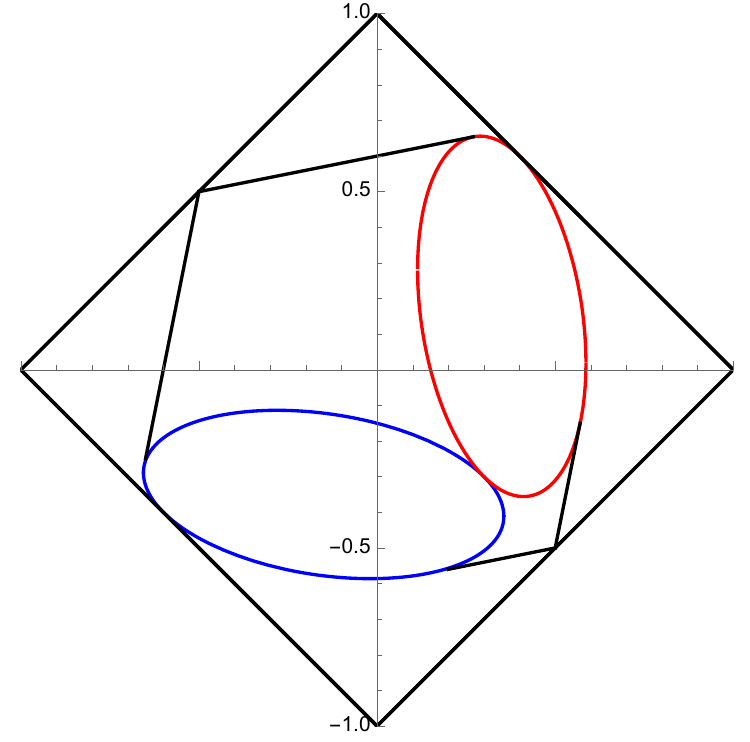}
        \subcaption[t]{\tiny $\tau=0$, $\sigma=1/5$} 
 \end{subfigure}\hskip 2cm
 \begin{subfigure}{0.2\textwidth}
 		\centering
        \includegraphics[scale=.2]{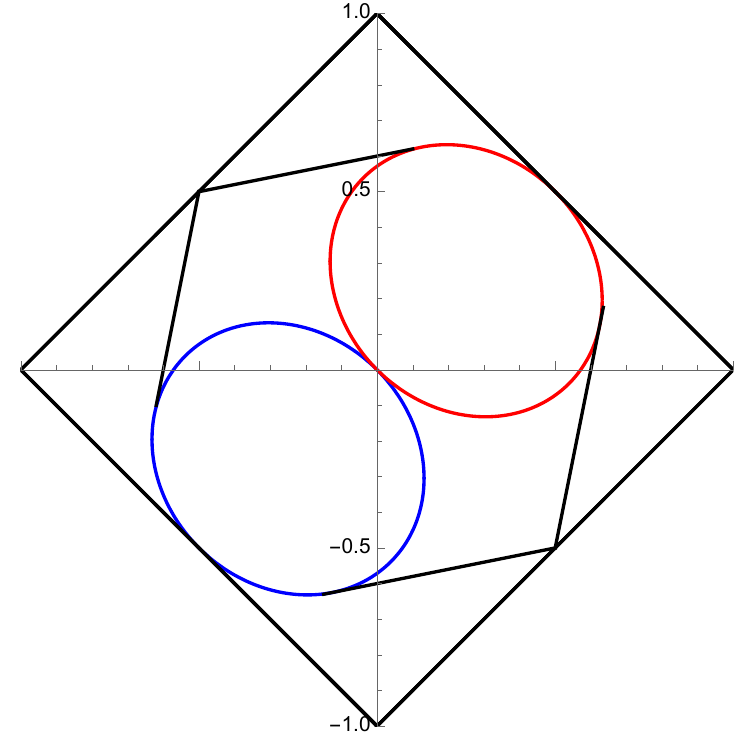}
        \subcaption[t]{\tiny$\tau=0$, $\sigma=1/2$} 
 \end{subfigure}\hskip 2cm
 \begin{subfigure}{0.2\textwidth}
 		\centering
        \includegraphics[scale=.2]{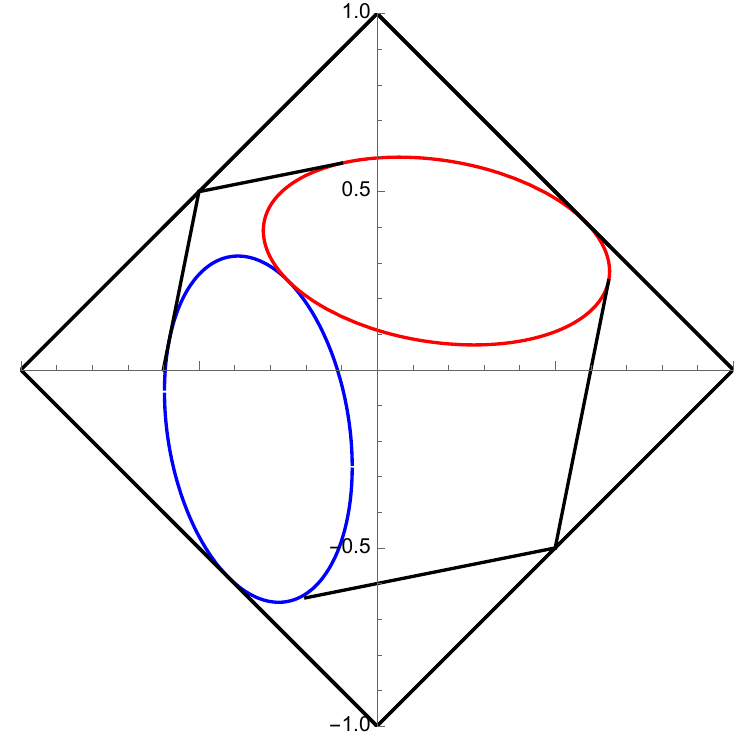}
        \subcaption[t]{\tiny$\tau=0$, $\sigma=3/4$} 
 \end{subfigure}
 	\caption{Arctic curves library for non-uniform data, fixed $\tau=0$ and $\sigma$ varied}
  \label{flat_nonuniform_tau0}
\end{figure}

\subsubsection{\textbf{case $\tau =\sigma$.}}\hfill\newline
We repeat the above with the case $\sigma=\tau$ in Fig \ref{flat_nonuniform_sigma=tau} and the case $\sigma \neq\tau$ in Fig \ref{flat_nonuniform_sigma not tau}. Explicit forms of arctic curves are available upon request.
\begin{figure}
\begin{subfigure}{0.2\textwidth}
	\centering
    \includegraphics[scale=.2]{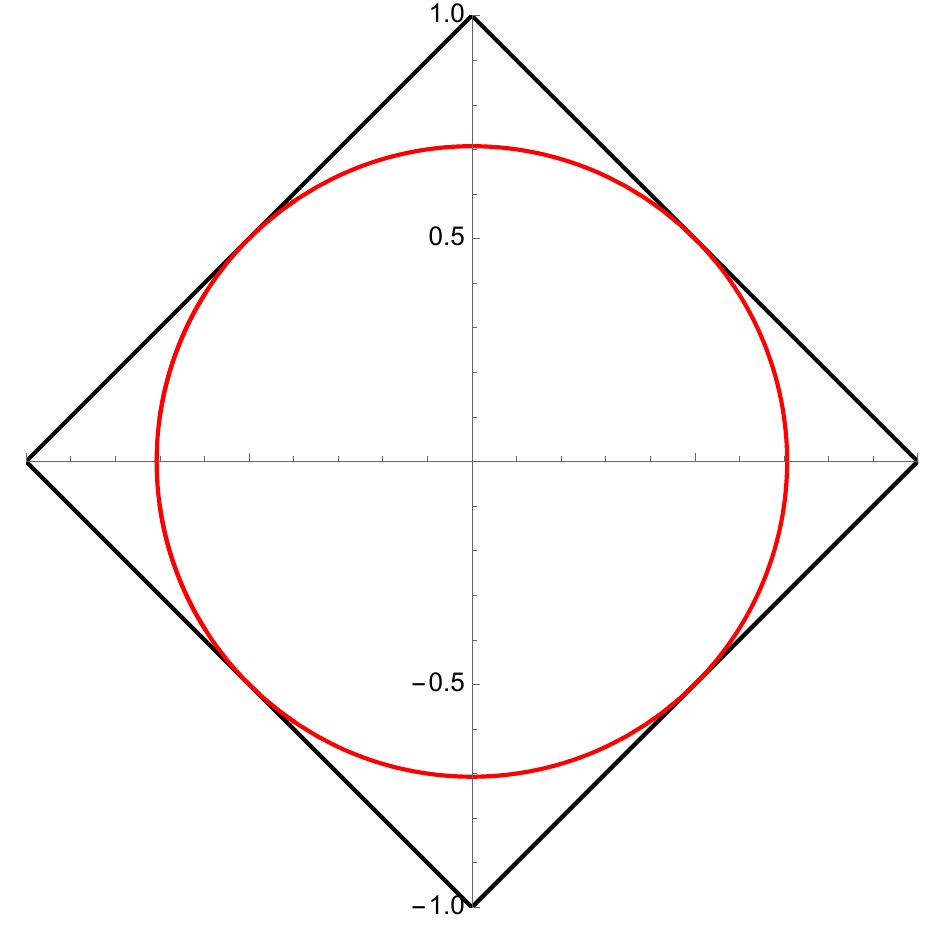}
    \subcaption[t]{\tiny $\tau=\sigma=1/2$} 
\end{subfigure}\hskip 2cm
\begin{subfigure}{0.2\textwidth}
		\centering
        \includegraphics[scale=.3]{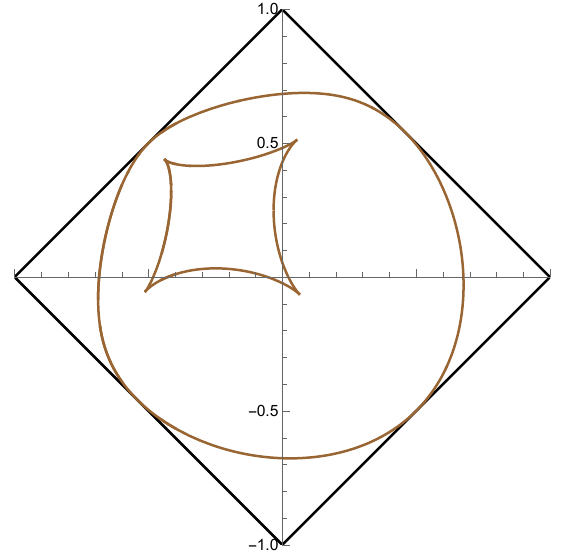}
        \subcaption[t]{\tiny $\tau=\sigma=1/4$} 
\end{subfigure}\hskip 2cm
\begin{subfigure}{0.2\textwidth}
		\centering
        \includegraphics[scale=.3]{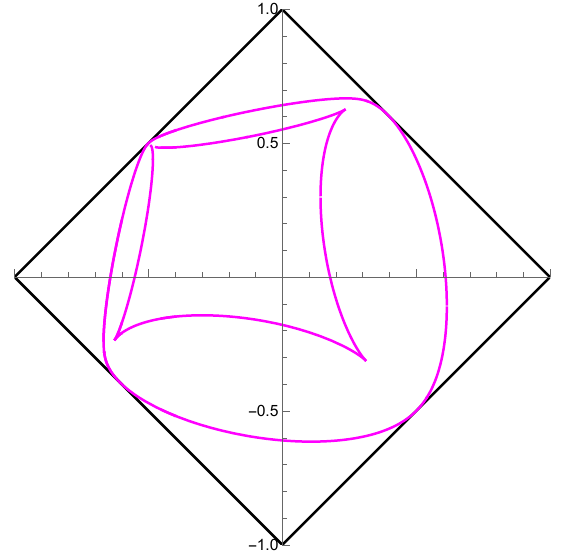}
        \subcaption[t]{\tiny $\tau=\sigma=1/10$} 
 \end{subfigure}\\
 \begin{subfigure}{.2\textwidth}
 		\centering
        \includegraphics[scale=.2]{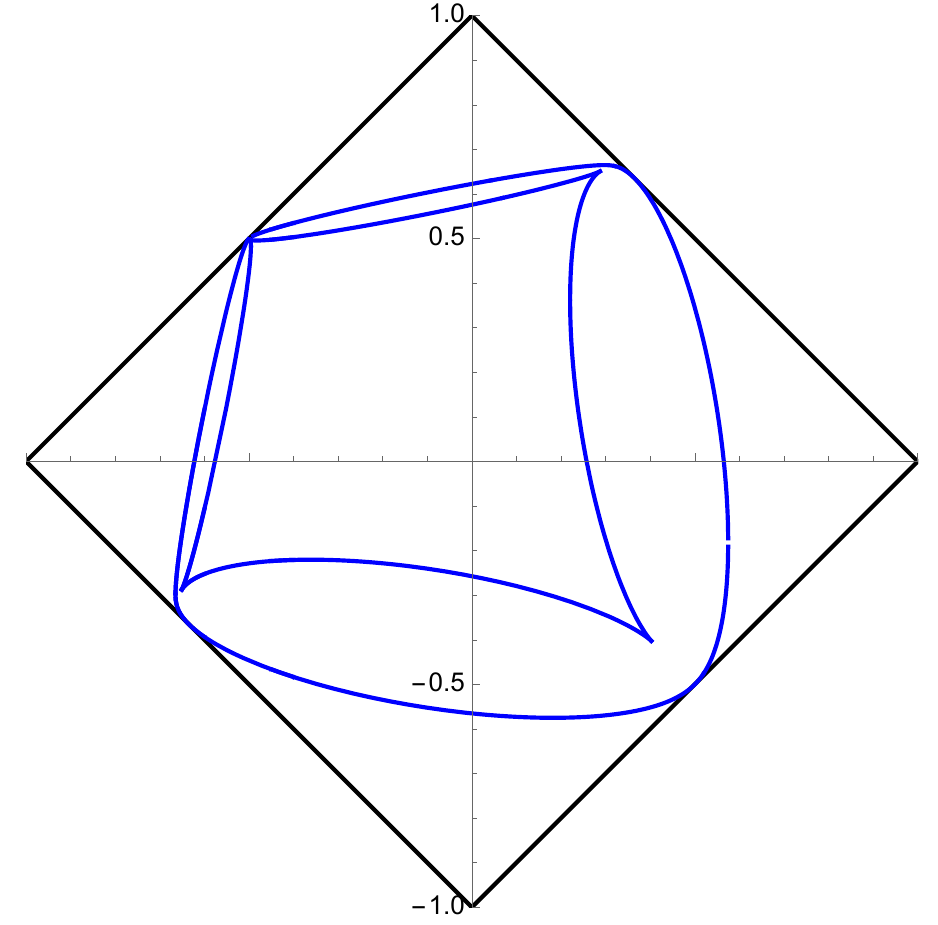}
        \subcaption[t]{\tiny$\tau=\sigma=1/20$} 
 \end{subfigure}\hskip 2cm
 \begin{subfigure}{.2\textwidth}
 		\centering
        \includegraphics[scale=.2]{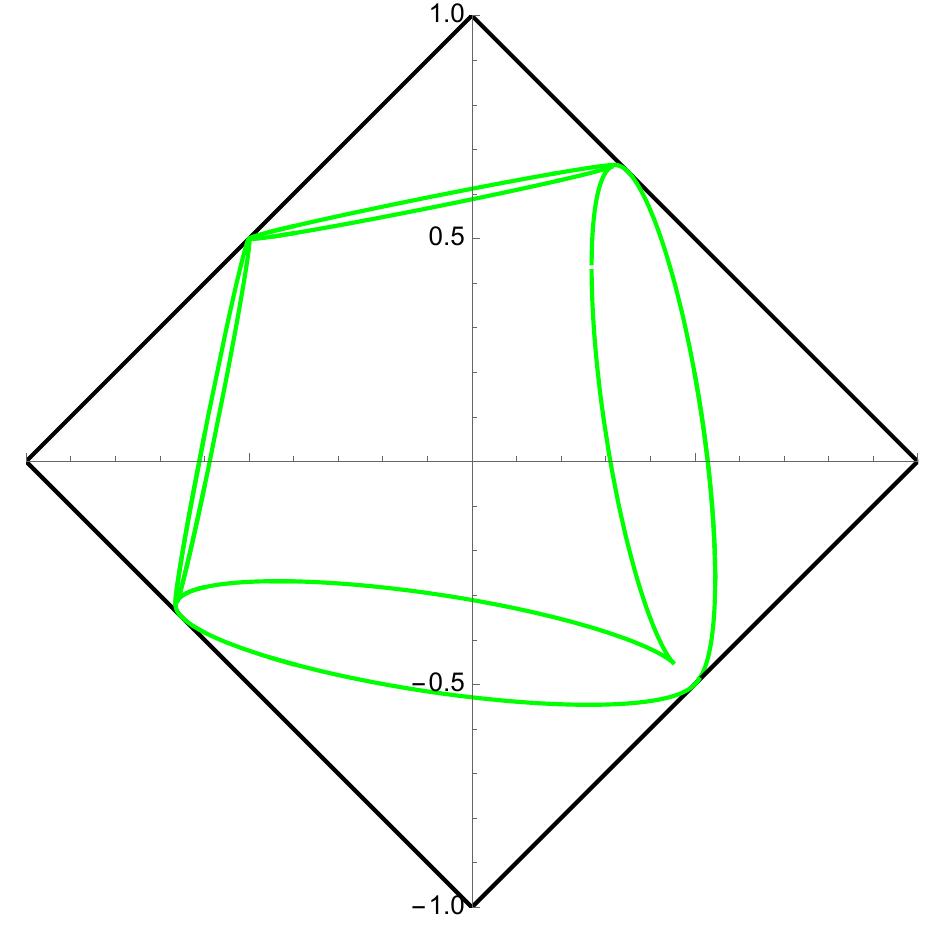}
        \subcaption[t]{\tiny$\tau=\sigma=1/40$} 
 \end{subfigure}\hskip 2cm
 \begin{subfigure}{.2\textwidth}
 	\centering
 	\includegraphics[scale=.2]{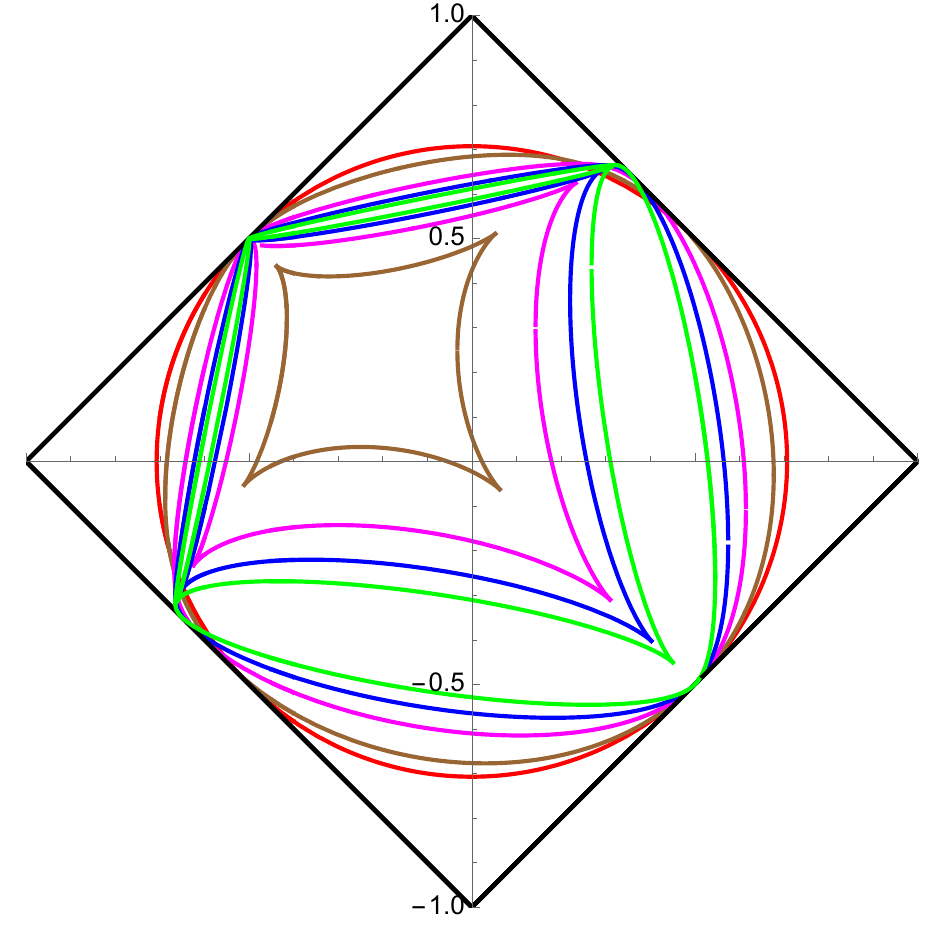}
    \subcaption[t]{Curves (A-E)} 
 \end{subfigure}
 % \begin{subfigure}
 % 	\includegraphics[scale=.3]{plots/nonuniform_113_tau=sigma_mergeall.pdf} \end{subfigure}
 	\caption{Arctic curves library for non-uniform data, $\tau=\sigma$ varied}
  \label{flat_nonuniform_sigma=tau}
\end{figure}
\subsubsection{\textbf{case $\tau=1/2$, $\sigma$ arbitrary.}}\hfill\newline
The case  $\tau=1/2$ and $\sigma$ arbitrary is represented in Fig.~\ref{flat_nonuniform_sigma not tau}.
\begin{figure}[H]
\begin{subfigure}{0.2\textwidth}
	\centering
    \includegraphics[scale=.25]{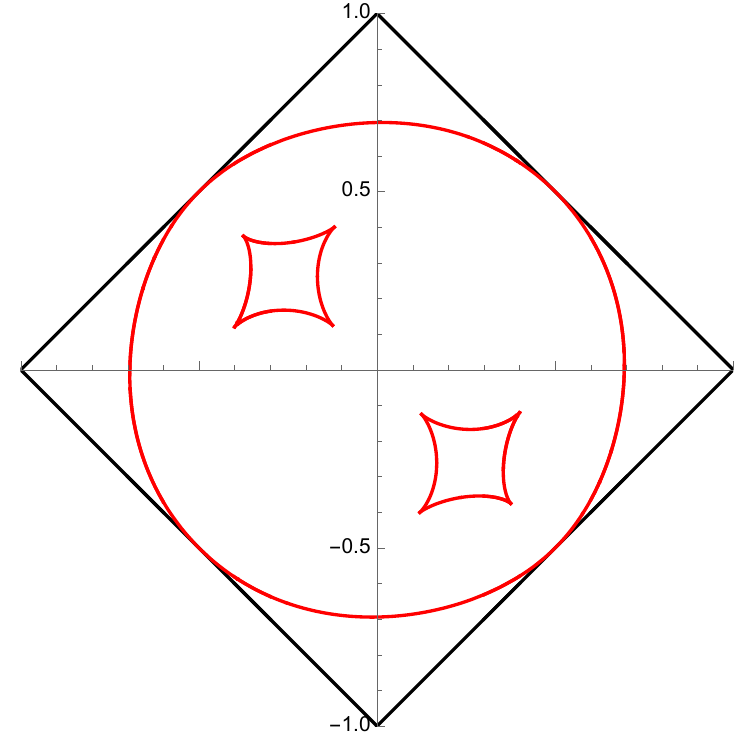}
    \subcaption[t]{\tiny $\tau=1/2$, $\sigma=1/4$} 
\end{subfigure}\hskip 1cm
\begin{subfigure}{0.2\textwidth}
		\centering
        \includegraphics[scale=.25]{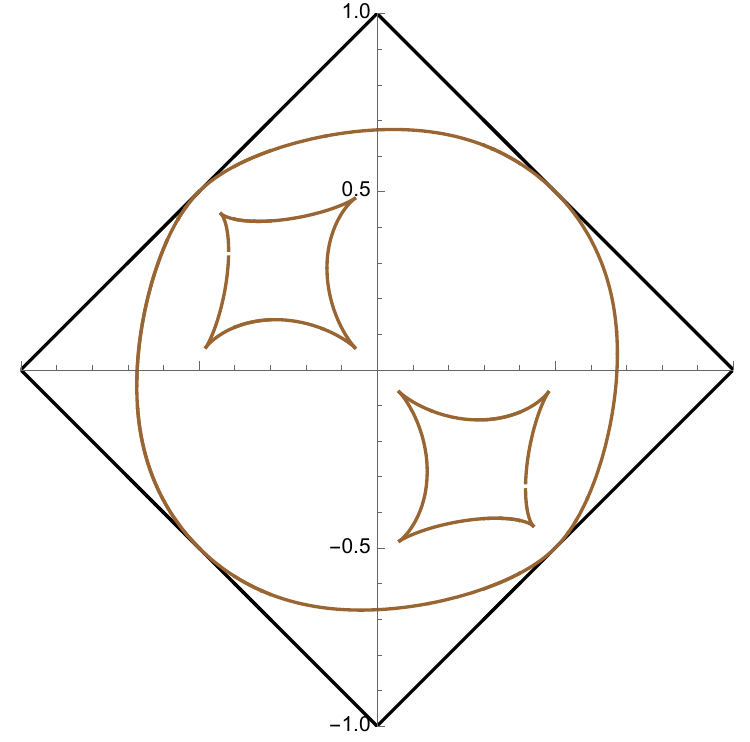}
        \subcaption[t]{\tiny $\tau=1/2$, $\sigma=1/8$} 
\end{subfigure}\hskip 1cm
\begin{subfigure}{0.2\textwidth}
		\centering
        \includegraphics[scale=.2]{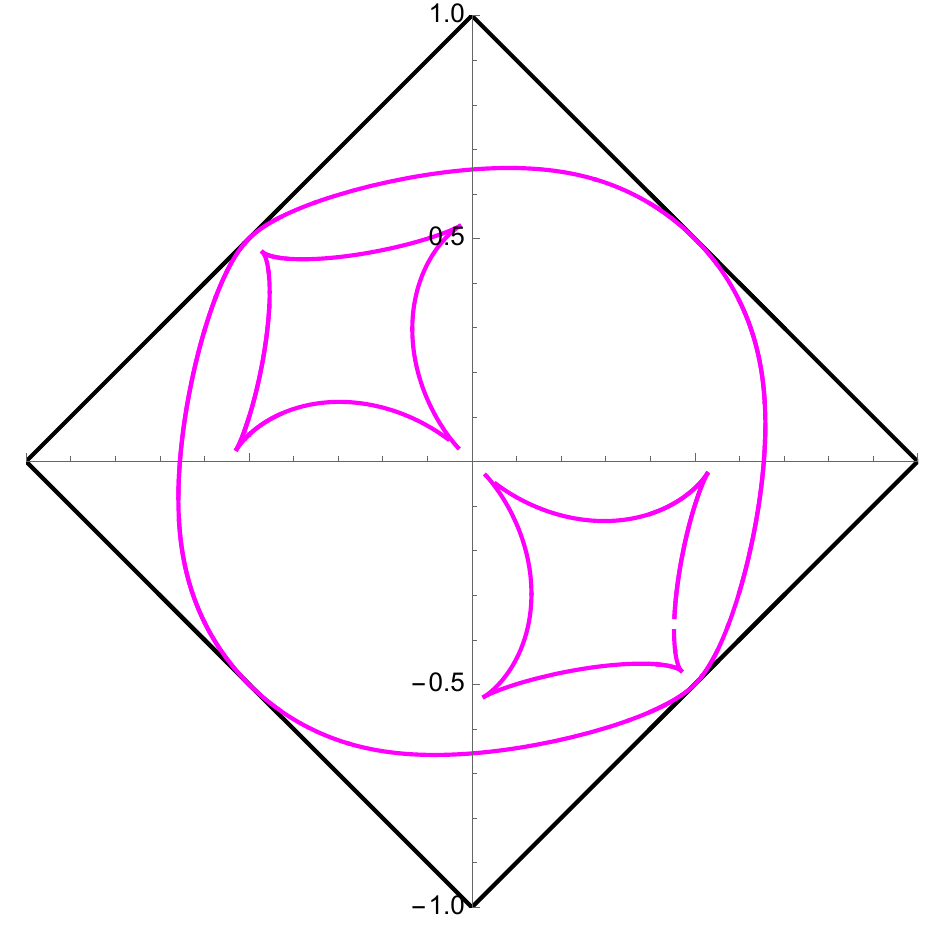}
        \subcaption[t]{\tiny $\tau=1/2$, $\sigma=1/16$} 
 \end{subfigure}\\
 \begin{subfigure}{0.3\textwidth}
 		\centering
        \hskip .65cm\includegraphics[scale=.25]{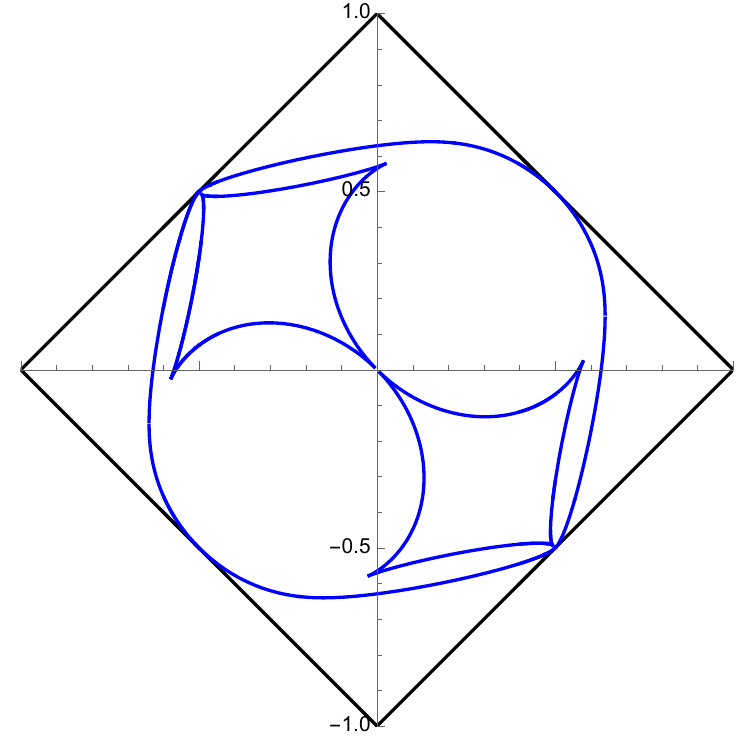}
        \subcaption[t]{\tiny$\tau=1/2$, $\sigma=1/64$} 
 \end{subfigure}
 \begin{subfigure}{0.5\textwidth}
 	\centering
 	\hskip 2cm\includegraphics[scale=.3]{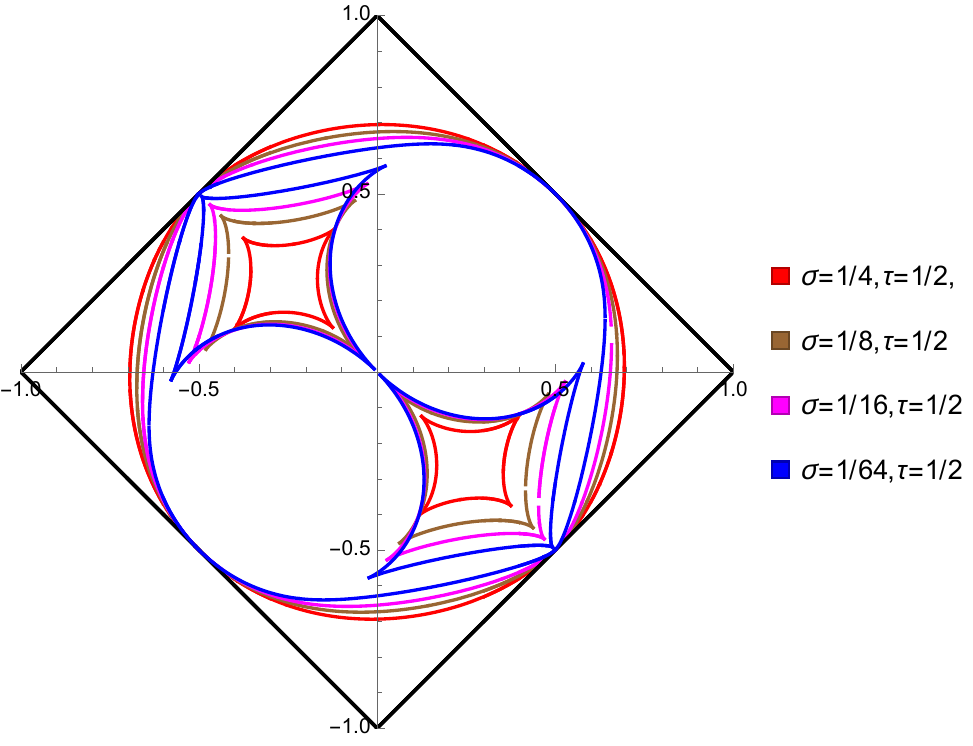}
    \subcaption[t]{Curves (A-D)} 
 \end{subfigure}
 % \begin{subfigure}
 % 	\includegraphics[scale=.3]{plots/nonuniform_113_tau=sigma_mergeall.pdf} \end{subfigure}
 \caption{Arctic curves library for non-uniform data, fixed $\tau=1/2$ and $\sigma$ varied}
  \label{flat_nonuniform_sigma not tau}
\end{figure}

\subsubsection{\textbf{Generic holography.}}\hfill\newline

\begin{remark}
While our $2\times 2$ toroidal solutions only hold in the $(r,r,t)$ cases,  there is no restriction on the ``point of view" direction $(\tilde r,\tilde s,\tilde t)$. Indeed, we can take any triple $(\tilde r,\tilde s,\tilde t)$ encoding the desired stepped surface where the new initial data $\tilde t_{i,j}$ lies. 
\end{remark}

To illustrate the remark above, we show in Fig.~\ref{Fig:nonuniform library non flat pic} several arctic curves from different $(\tilde r,\tilde s,\tilde t)$-stepped surfaces points of view, taking for initial data the $2 \times 2$-toroidal $(1,1,3)$-slanted $T$-system with 
$\sigma=\frac{1}{2}$ and $\tau=\frac{1}{4}$. 

\begin{figure}[H]
\begin{subfigure}{0.2\textwidth}
	\centering
    \includegraphics[scale=.2]{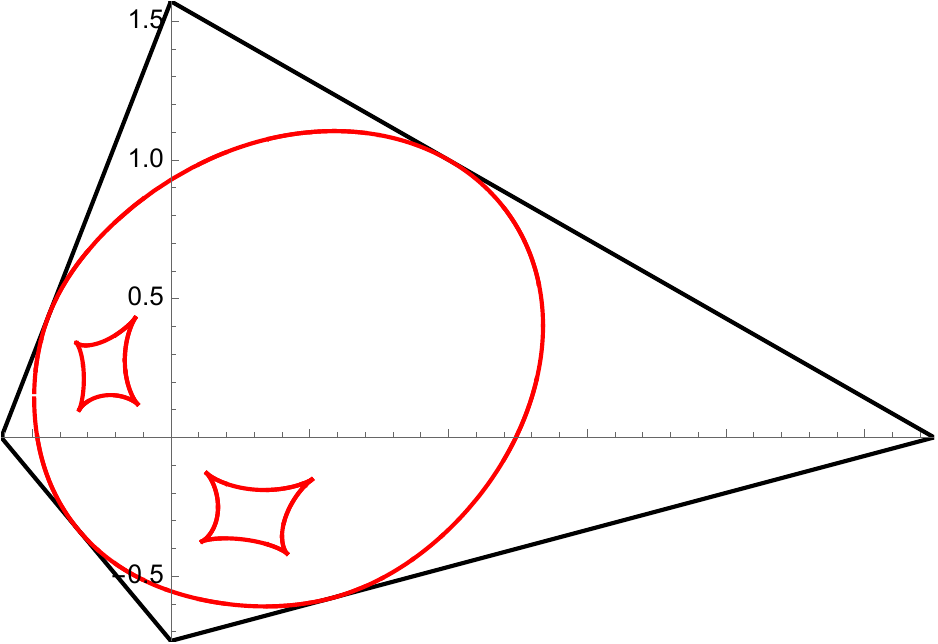}
    \subcaption[t]{\tiny $(\tilde r,\tilde s,\tilde t)=(7,4,11)$} 
\end{subfigure}\hskip .5cm
\begin{subfigure}{0.2\textwidth}
		\centering
        \includegraphics[scale=.2]{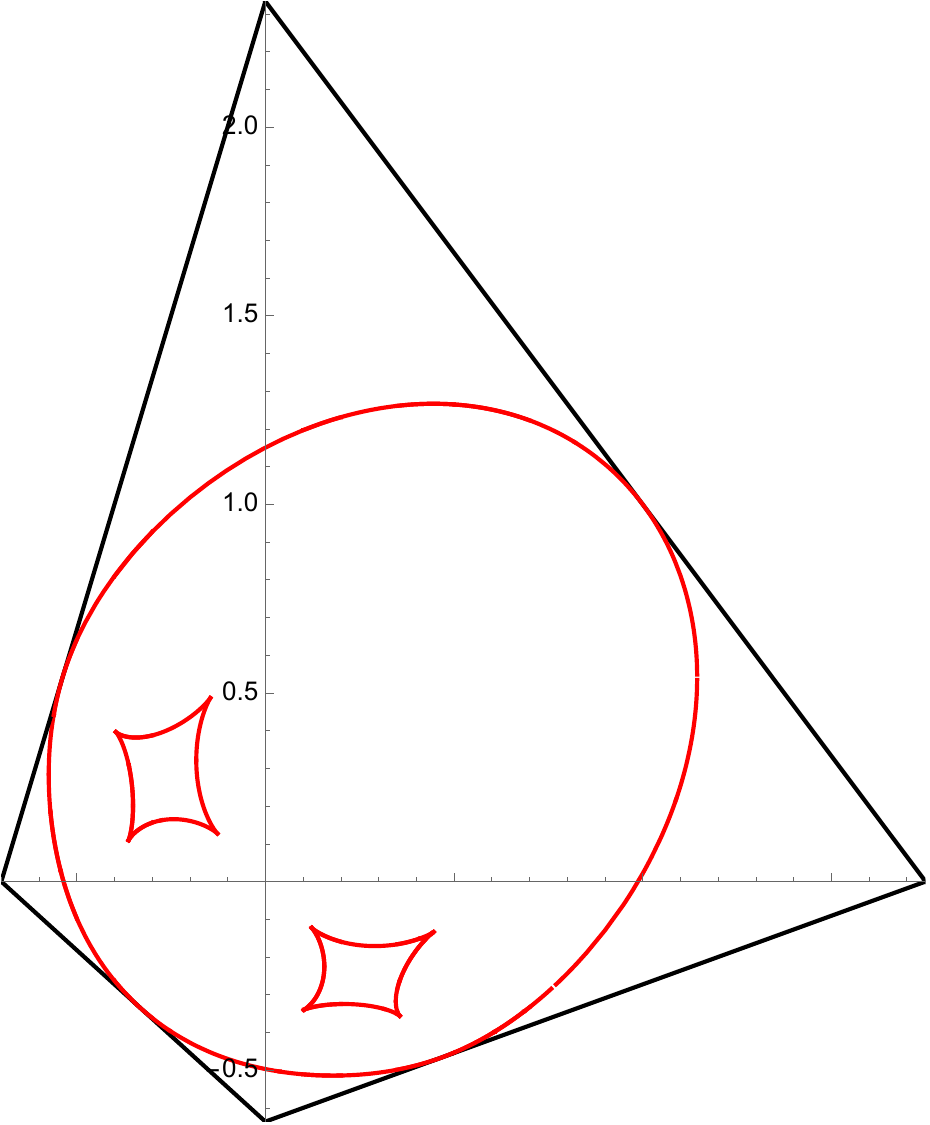}
        \subcaption[t]{\tiny $(\tilde r,\tilde s,\tilde t)=(3,4,7)$} 
\end{subfigure}\hskip .5cm
\begin{subfigure}{0.2\textwidth}
		\centering
        \includegraphics[scale=.2]{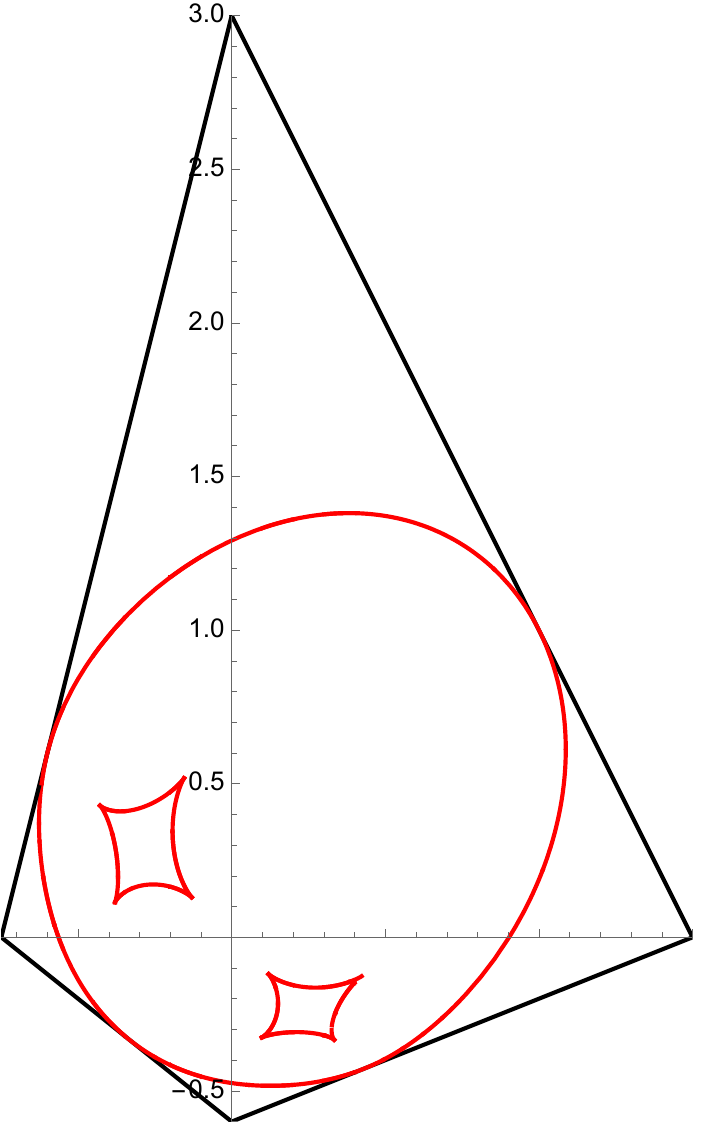}
        \subcaption[t]{\tiny $(\tilde r,\tilde s,\tilde t)=(1,2,3)$}
 \end{subfigure}\hskip .5cm
 \begin{subfigure}{0.2\textwidth}
 		\centering
        \includegraphics[scale=.2]{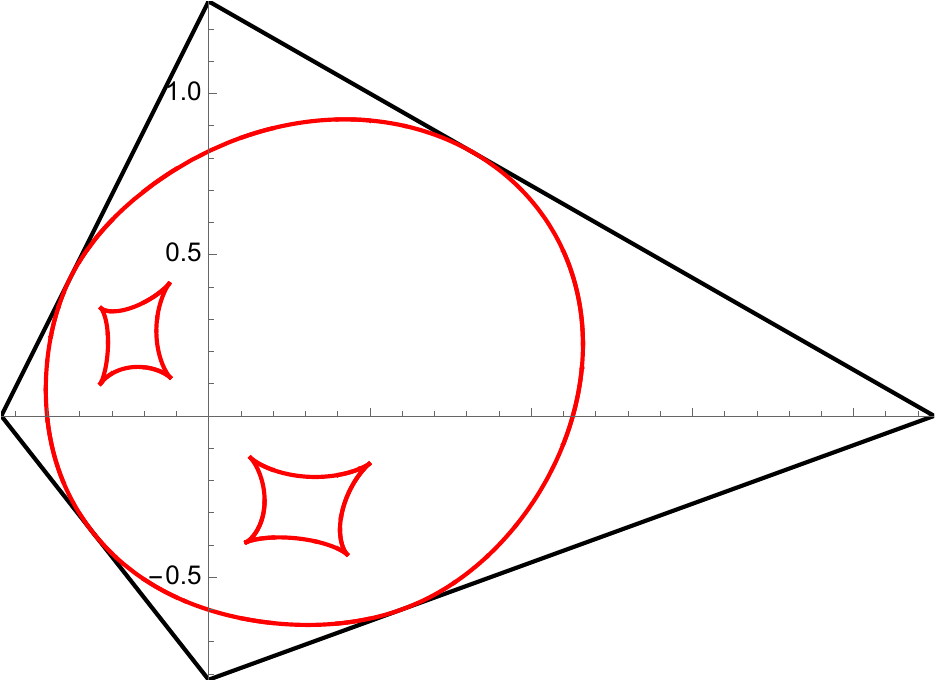}
        \subcaption[t]{\tiny $(\tilde r,\tilde s,\tilde t)=(5,2,9)$}
 \end{subfigure}
 \caption{Several $(\tilde r,\tilde s,\tilde t)$ views of arctic curves of the $(1,1,3)$ slanted 2x2 periodic solution with $\tau=1/4$, $\sigma=1/2$.}
\label{Fig:nonuniform library non flat pic}
 \end{figure}
 \subsection{$m$-toroidal holography}
 In this section, we apply the holographic principle to the m-toroidal solutions of the $T$-system with flat iniital data with $(r,s,t)=(0,0,1)$ (see Section 4 in Ref.~\cite{DiFrancesco1}). The m-toroidal initial data on the flat initial data plane $(r,s,t)=(0,0,1)$ are prescribed to be:
 \begin{equation}\label{m-toroidal data}
 	\begin{array}{cc}
 		a_i=T_{i+1,-i,0}=t_{i+1,-i} & b_i=T_{i+2,-i+1,0}=t_{i+2,-i+1}\\
 		c_i=T_{i,-i,1}=t_{i,-i}     & d_i=T_{i+1,-i+1,1}=t_{i+1,-i+1} 
 	\end{array}
 \end{equation}
 for $i \in \Z$, with the additional restriction: 
 \begin{equation*}
 	t_{i+m,j-m}=t_{i,j} \hskip .5cm \text{and} \hskip .5cm t_{i+2,j+2}=t_{i,j} \hskip .5 cm (i,j \in \Z)
 \end{equation*}
 The $T$-system with this initial data was solved explicitly (see Theorem 4.2 of \cite{DiFrancesco1} for details). 
For our interest, we want an explicit description of the quantities $L_{i,j,k}$ and $R_{i,j,k}$ in the recursion relation for the density $\rho_{i,j,k}$. They obey the following periodicities:
 \begin{equation*}
 	\begin{array}{ccc}
 		L_{i+2,j+2,k}=L_{i,j,k} & L_{i+m,j-m,k}=L_{i,j,k} & L_{i+1,j+1,k+2}=L_{i,j,k}\\
 		R_{i+2,j+2,k}=R_{i,j,k} & R_{i+m,j-m,k}=R_{i,j,k} & R_{i+1,j+1,k+2}=R_{i,j,k}
 	\end{array}
 \end{equation*}
Therefore, the density function can split as above, modulo the periodicity lattice $\Lambda\subset \Z^3$ generated by the vectors $(2,2,0), (m,-m,0)$ and $(1,1,2)$, and obeys a linear system. Let $\ds \lambda_i=\frac{a_ib_{i-1}}{a_{i-1}b_i+a_ib_{i-1}}$ and $\ds \mu_i=L_{i+2,-i+1,1}=\frac{c_{i+1}d_i}{c_id_{i+1}+c_{i+1}d_i}$. These quantities are not indepedent and satisfy the relation: 
$$\prod_{i=0}^{m-1}(\frac{1}{\lambda_i}-1)=\prod_{i=0}^{m-1}(\frac{1}{\mu_i}-1)=1$$
The solution of the general system of $m$-periodic densities is the rational function in $x,y,z$ with denominator the determinant of the following block matrix: 
\begin{equation}
C=	\left(\begin{array}{c|c|c|c}
		z^{-1}I	&zI &-M(x,y) &-\bar{M}(y^{-1},x^{-1})\\
		\hline
		zI&z^{-1}I&-M(y^{-1},x^{-1}) &-\bar{M}(x,y)\\
		\hline
		-P(x,y)&-\bar{P}(y^{-1},x^{-1})&z^{-1}I&zI\\
		\hline
		-P(y^{-1},x^{-1})&-\bar{P}(x,y)&zI&z^{-1}I
	\end{array}\right)
\end{equation}
where
\begin{equation*}
	P(x,y)=	\left(\begin{array}{cccccc}
		\frac{\mu_0}{x}	&0 &0 &\cdots&0&\frac{1-\mu_0}{y}\\
		\frac{1-\mu_1}{y}	&\frac{\mu_1}{x}&0 &\cdots&0&0\\
		0&\frac{1-\mu_2}{y}&\frac{\mu_2}{x}&\ddots&0&0\\
		\vdots&\ddots&\ddots&\ddots&\ddots&\vdots\\
		0&\cdots&\cdots&0&\frac{1-\mu_{m-1}}{y}&\frac{\mu_{m-1}}{x}
	\end{array}\right)
\end{equation*}
and 
\begin{equation*}
	M(x,y)=	\left(\begin{array}{cccccc}
		\frac{1-\lambda_0}{y}	&\frac{\lambda_0}{x} &0 &\cdots&0&0\\
		0&\frac{1-\lambda_1}{y}&\frac{\lambda_1}{x} &\cdots&0&0\\
		0&0&\frac{1-\lambda_2}{y}&\frac{\lambda_2}{x}&\ddots&0\\
		\vdots&\ddots&\ddots&\ddots&\ddots&\frac{\lambda_{m-2}}{x}\\
		\frac{\lambda_{m-1}}{x}&\cdots&\cdots&0&0&\frac{1-\lambda_{m-1}}{y}
	\end{array}\right)
\end{equation*}
and $\bar{P}$ is $P$ where $\mu_i$ and $1-\mu_i$ interchanged and $\bar M$ is $M$ where $\mu_i$ and $1-\mu_i$ interchanged. 

By the holographic principle, we now wish to view the exact solution of the $T$-system with m-toroidal ``flat" initial data 
from an $(\tilde r,\tilde s,\tilde t)$-slanted perspective.
%The word "holographic" means that we want to view the asymptotic behavior of an initial data on any slanted surface as the solution of some simpler initial data $T$-system where $L_{i,j,k}$ and $R_{i,j,k}$ takes finitely many values independent of $i,j,k$. In this case, we want to fix $(\tilde{r},\tilde s, \tilde t)=(0,0,1)$ with $m$-toroidal initial data satisfying the restriction in \eqref{m-toroidal data}. 
Setting $D_{0,0,1}(x,y,z)=\det(C)$, the denominator of the density in the new persective reads:
$\Delta_{0,0,1}^{\tilde r,\tilde s,\tilde t}(x,y,z):=D_{0,0,1}(z^{\tilde r/\tilde t}x^{-1},z^{\tilde s/\tilde t}y^{-1},z)$. 
The corresponding dual curves give the limit shapes of large $(\tilde r,\tilde s,\tilde t)$ pinecones 
corresponding to the m-toroidal solutions.
In Figs.~\ref{123_3toroidal}--\ref{235_3toroidal}, we list these new arctic curves for the same parameters as in section 4.2 of \cite{DiFrancesco1} 
for $(\tilde r,\tilde s,\tilde t)=(1,2,3), (1,0,3)$ and $(2,3,5)$ respectively. 
\begin{figure}[H]
	\begin{subfigure}{0.2\textwidth}
		\includegraphics[scale=.2]{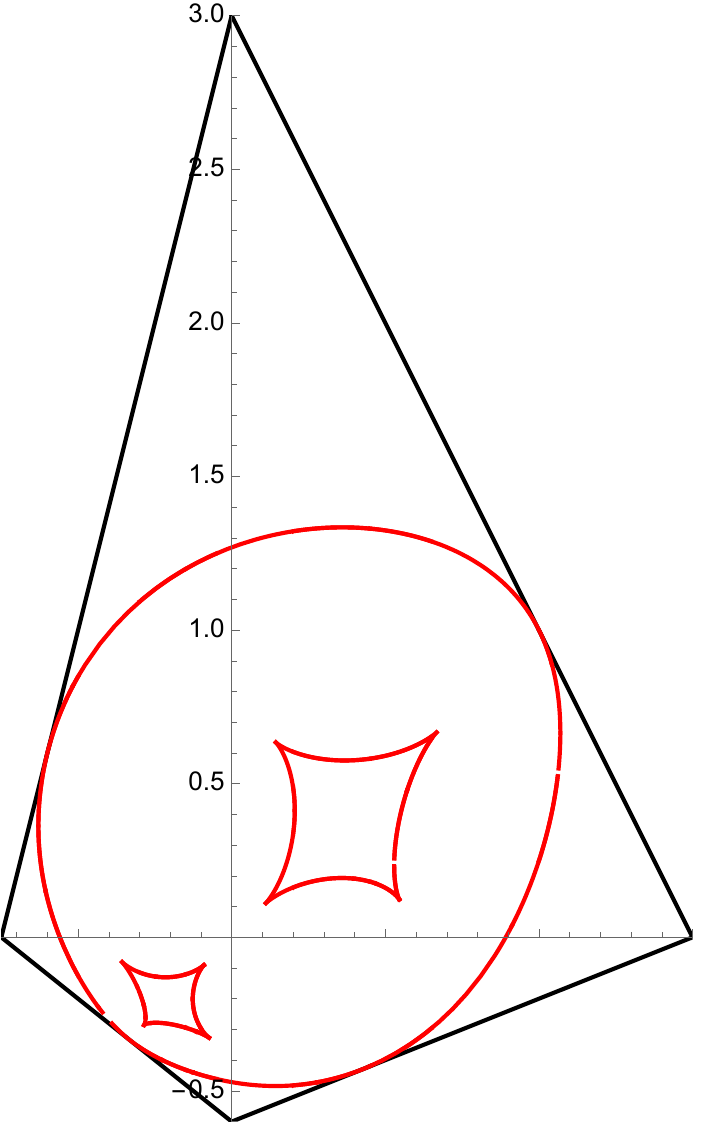}
		\caption{$\lambda_1=1/5$}
	\end{subfigure}
	\begin{subfigure}{0.2\textwidth}
		\includegraphics[scale=.2]{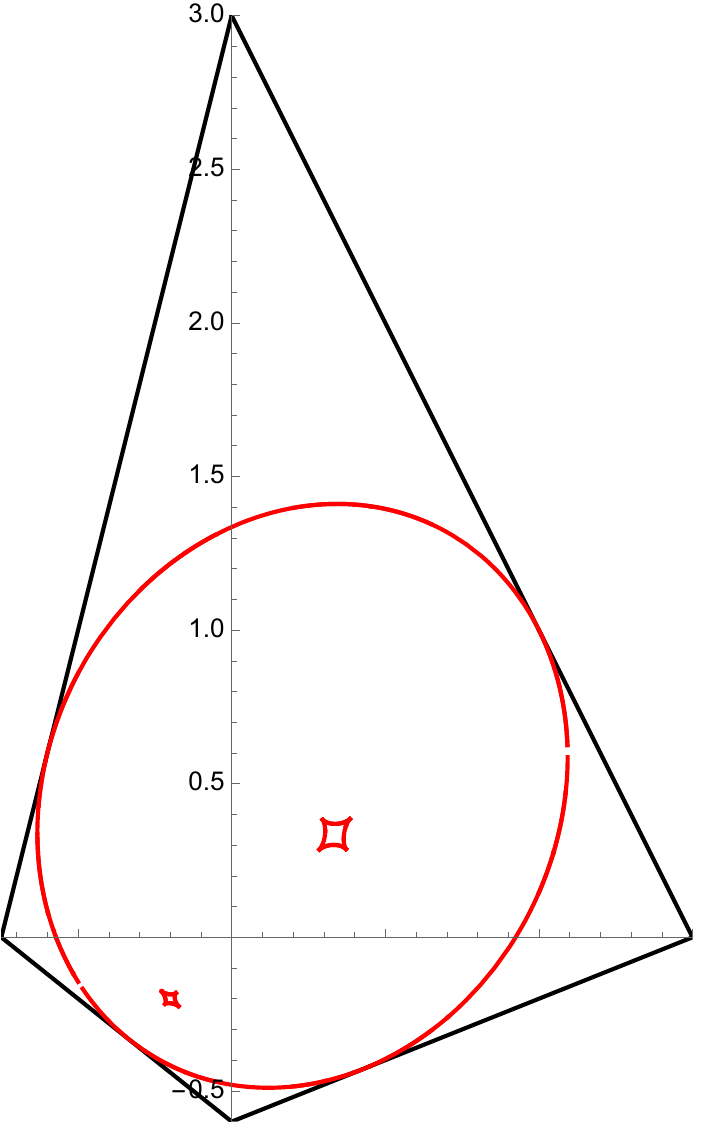}
		\caption{$\lambda_1=4/9$}
	\end{subfigure}
	\begin{subfigure}{0.2\textwidth}
		\includegraphics[scale=.2]{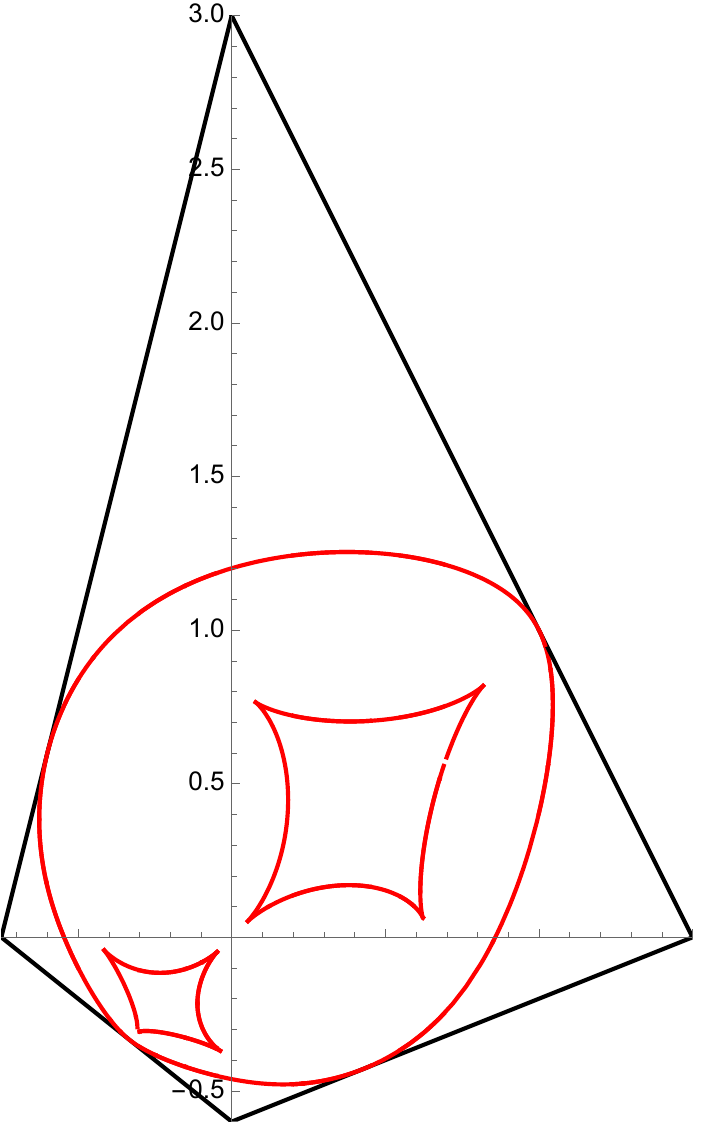}
		\caption{$\lambda_1=9/10$}
	\end{subfigure}
	\begin{subfigure}{0.2\textwidth}
		\includegraphics[scale=.2]{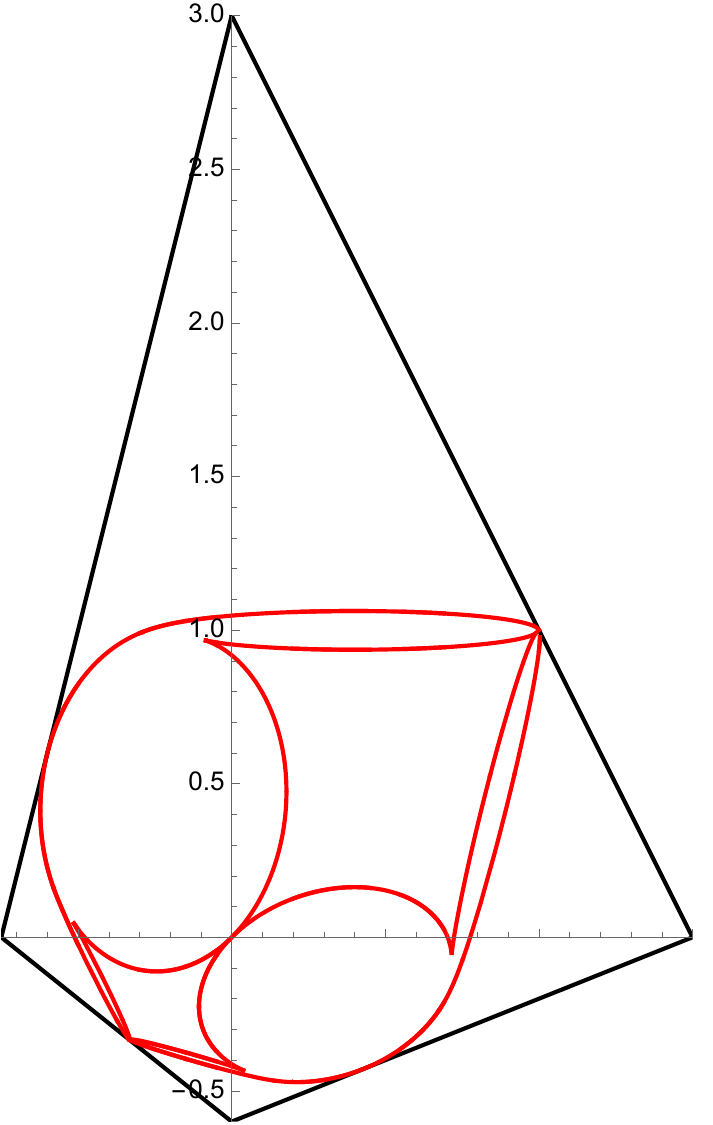}
		\caption{$\lambda_1=200/201 $}
	\end{subfigure}
	\caption{Arctic curves for the $3$-toroidal initial data corresponding to different values $\lambda_1$, where $\lambda_0 = 1/2$, $\lambda_2 = 1 - \lambda_1$ and $\mu_0 = \mu_1 = \mu_2 = 1/2$. View from $(\tilde r,\tilde s,\tilde t)=(1,2,3)$ perspective.}
\label{123_3toroidal}
\end{figure}

\begin{figure}[H]
	\begin{subfigure}{0.2\textwidth}
		\includegraphics[scale=.2]{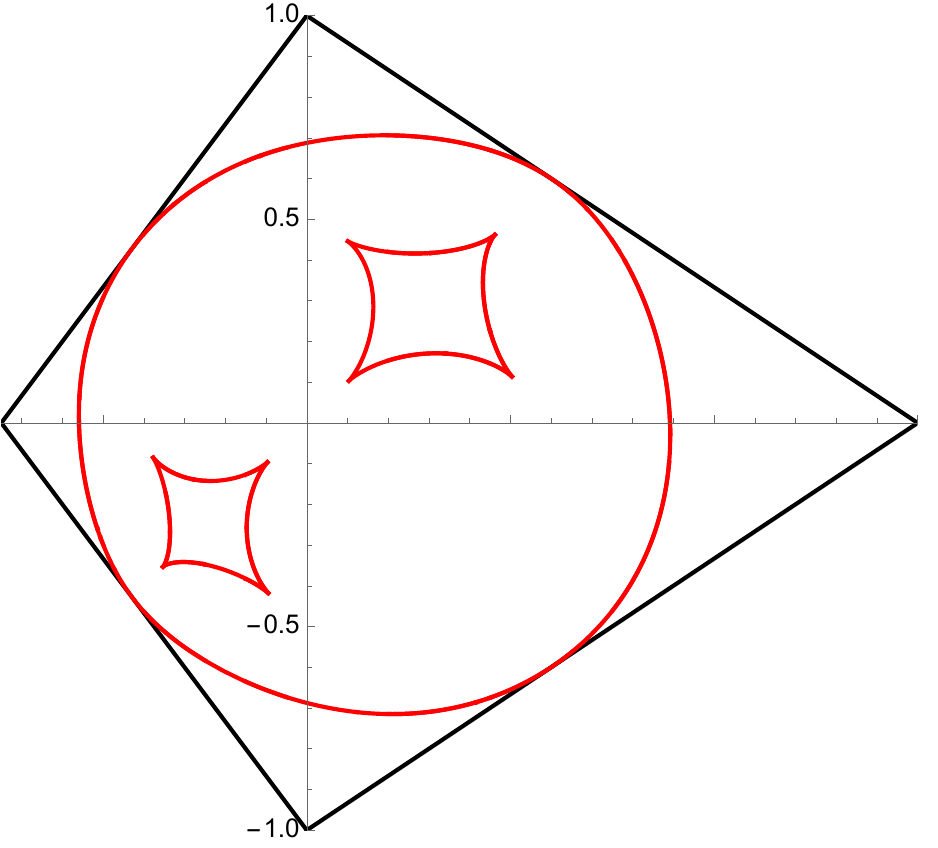}
		\caption{$\lambda_1=1/5$}
	\end{subfigure}
	\begin{subfigure}{0.2\textwidth}
		\includegraphics[scale=.2]{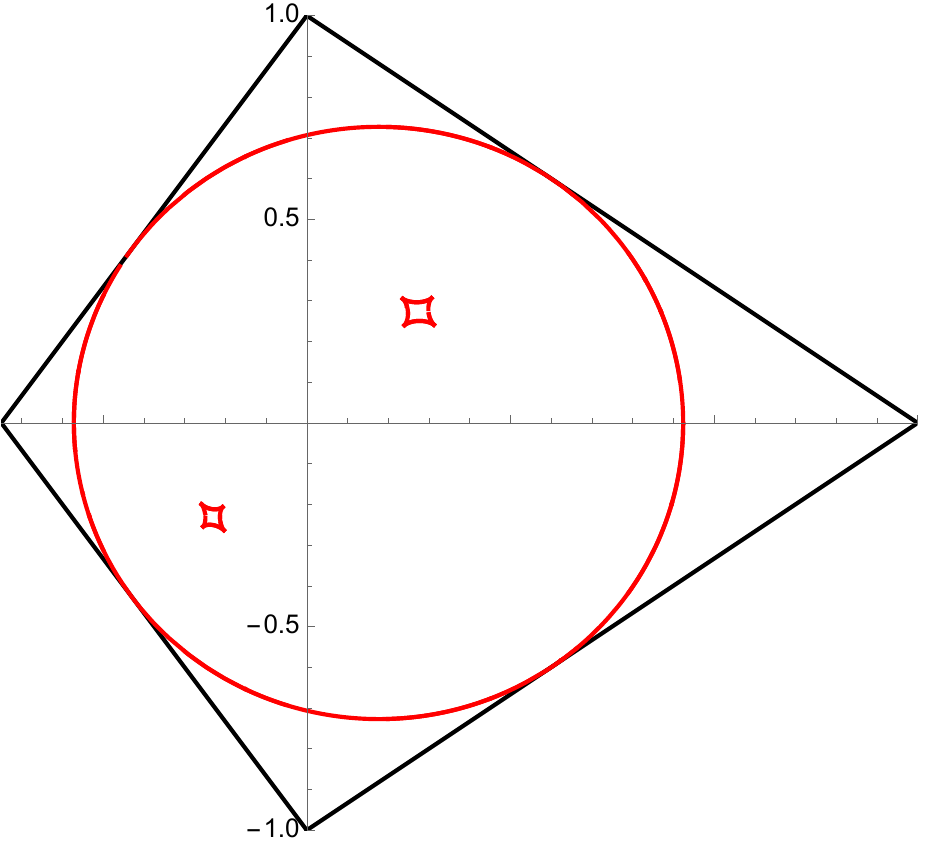}
		\caption{$\lambda_1=4/9$}
	\end{subfigure}
	\begin{subfigure}{0.2\textwidth}
		\includegraphics[scale=.2]{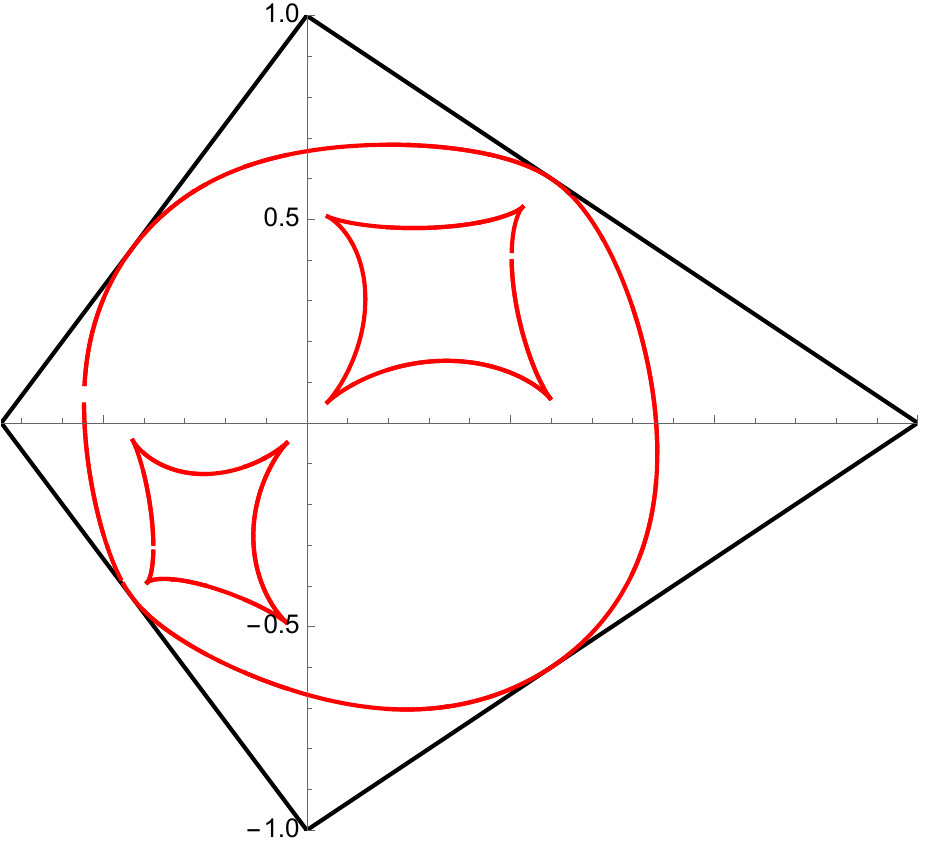}
		\caption{$\lambda_1=9/10$}
	\end{subfigure}
	\begin{subfigure}{0.2\textwidth}
		\includegraphics[scale=.2]{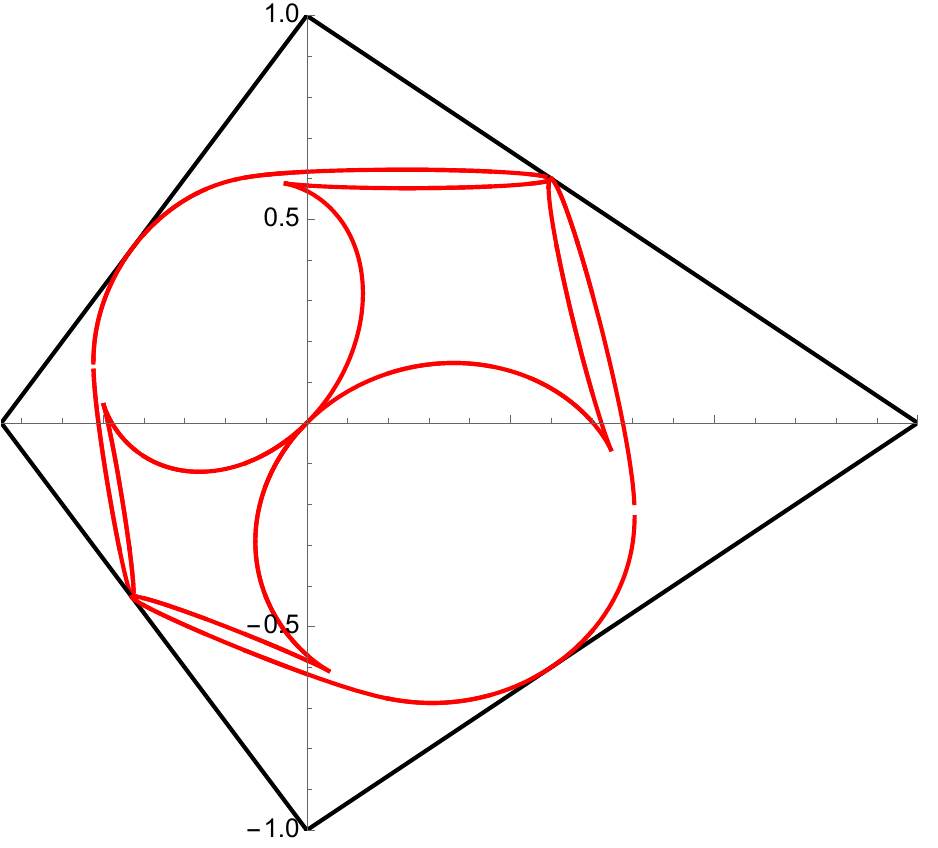}
		\caption{$\lambda_1=200/201 $}
	\end{subfigure}
	\label{103_3toroidal}
	\caption{Arctic curves for the $3$-toroidal initial data corresponding to different values $\lambda_1$, where $\lambda_0 = 1/2$, $\lambda_2 = 1 - \lambda_1$ and $\mu_0 = \mu_1 = \mu_2 = 1/2$. View from $(\tilde r,\tilde s,\tilde t)=(1,0,3)$ perspective.}
\end{figure}

\begin{figure}[H]
	\begin{subfigure}{0.2\textwidth}
		\includegraphics[scale=.2]{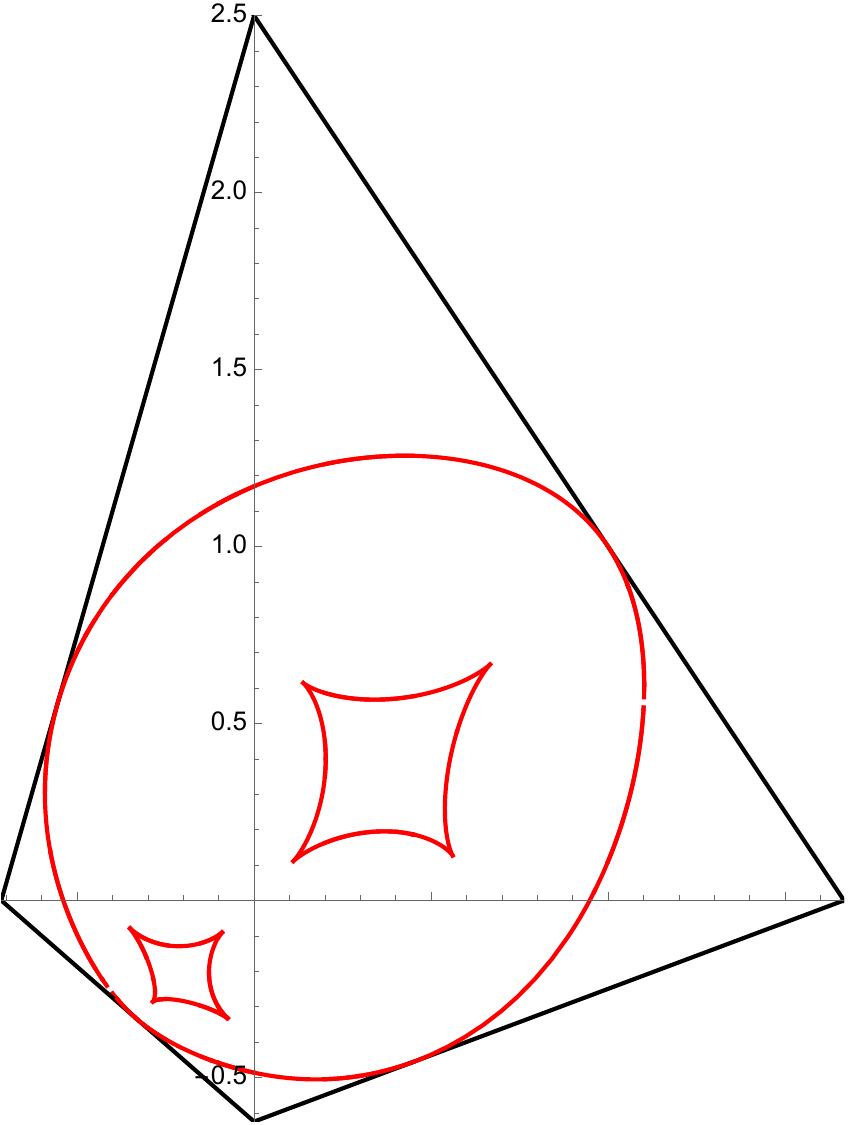}
		\caption{$\lambda_1=1/5$}
	\end{subfigure}
	\begin{subfigure}{0.2\textwidth}
		\includegraphics[scale=.2]{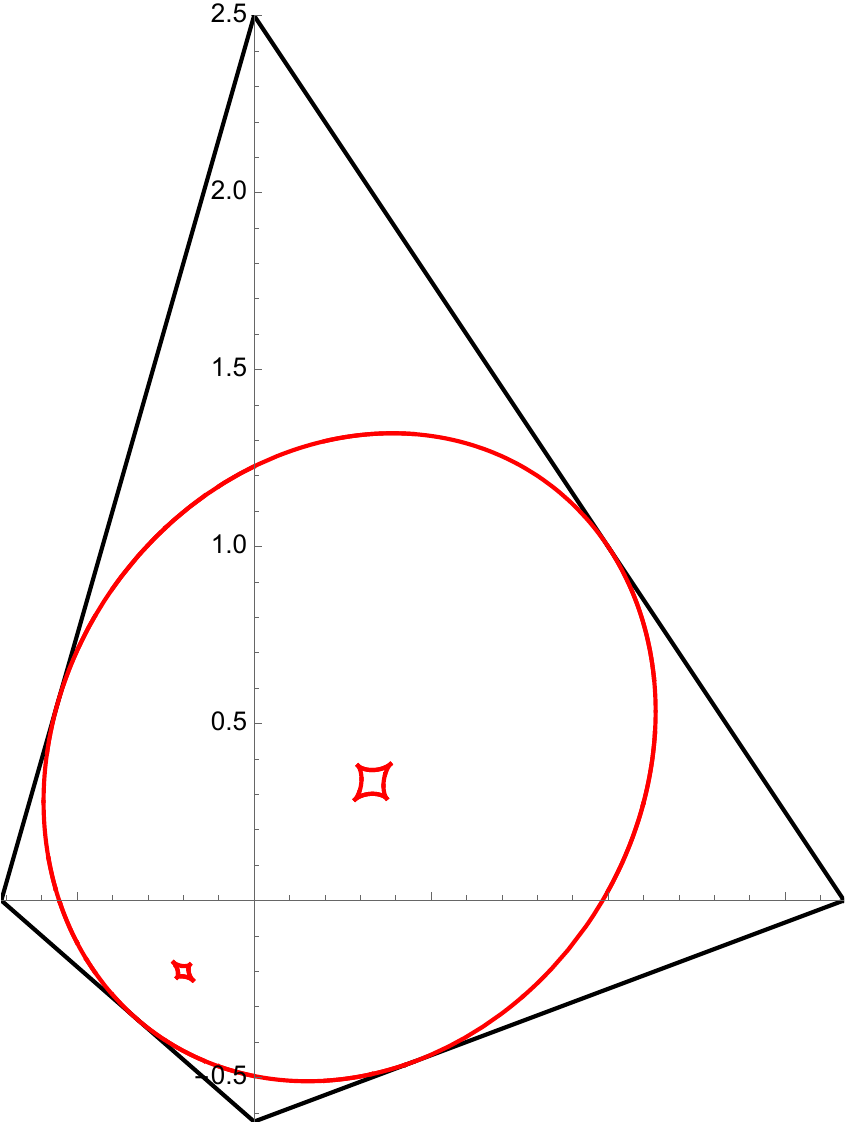}
		\caption{$\lambda_1=4/9$}
	\end{subfigure}
	\begin{subfigure}{0.2\textwidth}
		\includegraphics[scale=.2]{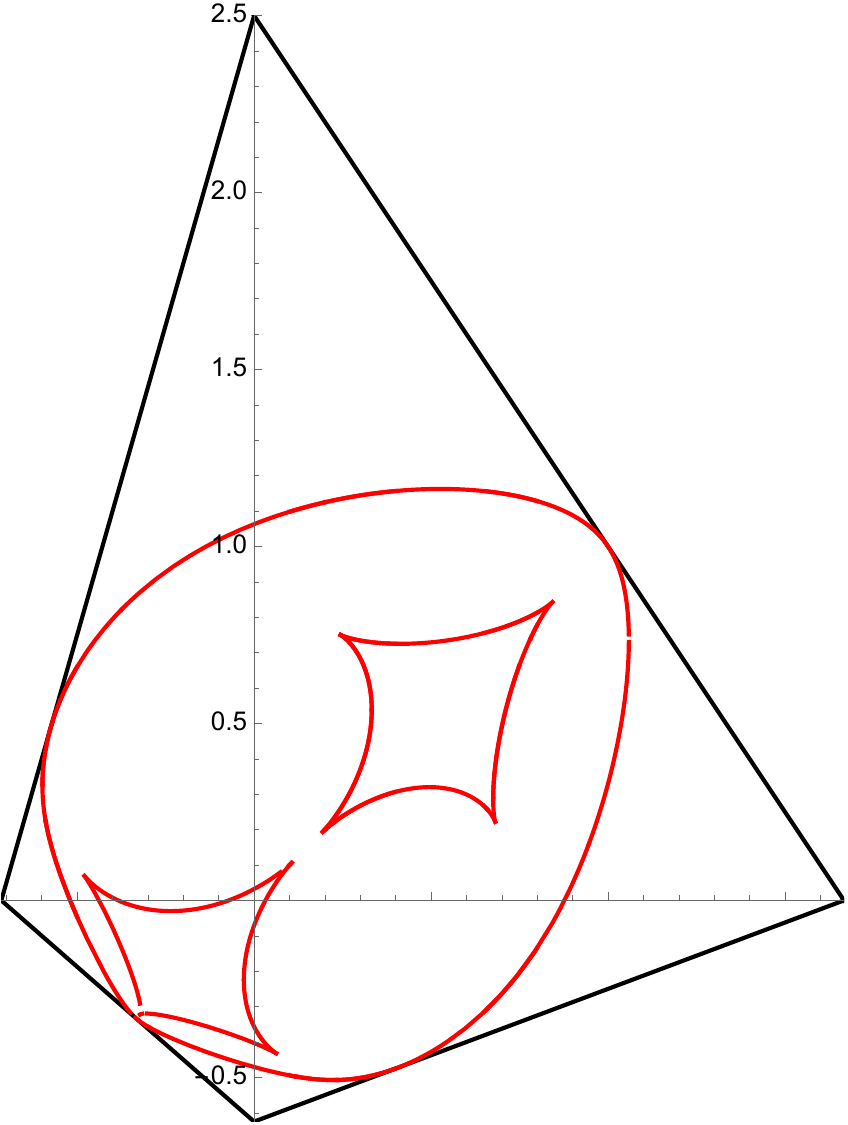}
		\caption{$\lambda_1=9/10$}
	\end{subfigure}
	\begin{subfigure}{0.2\textwidth}
		\includegraphics[scale=.2]{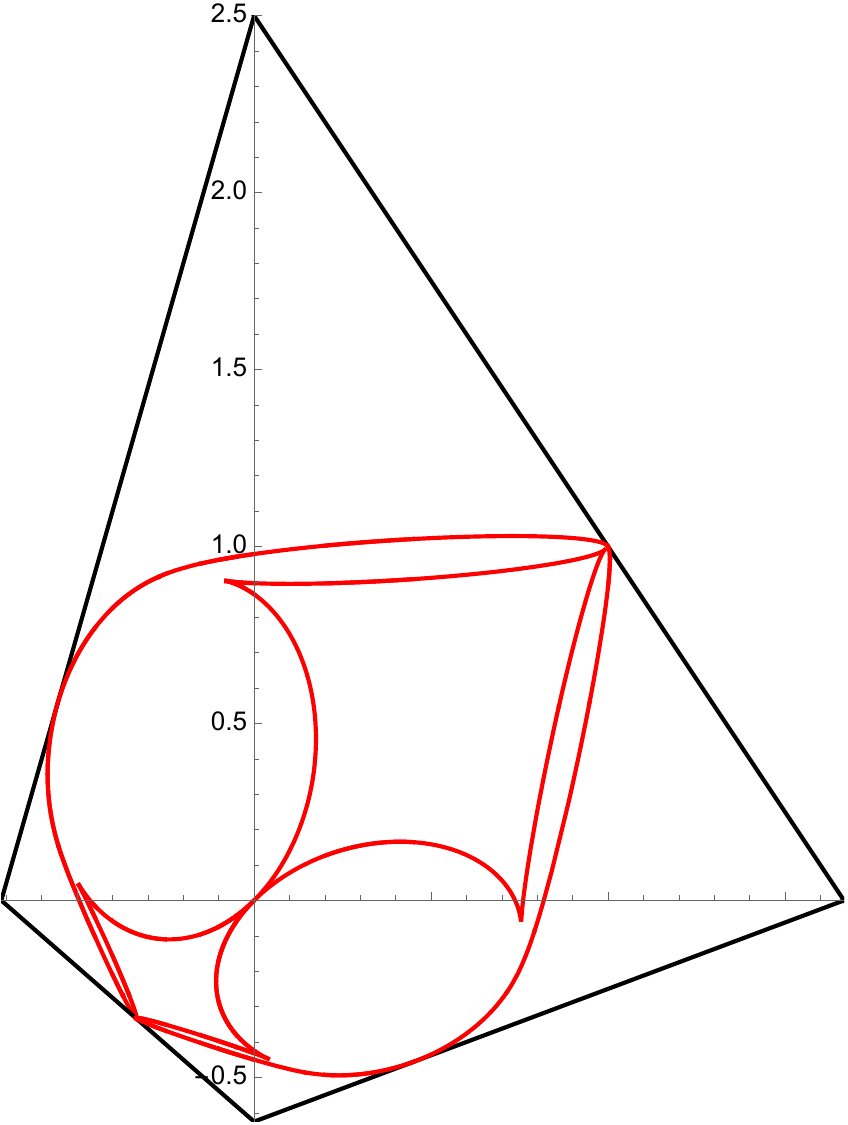}
		\caption{$\lambda_1=200/201 $}
	\end{subfigure}
	\caption{Arctic curves for the $3$-toroidal initial data corresponding to different values $\lambda_1$, where $\lambda_0 = 1/2$, $\lambda_2 = 1 - \lambda_1$ and $\mu_0 = \mu_1 = \mu_2 = 1/2$. View from $(\tilde r,\tilde s,\tilde t)=(2,3,5)$ perspective.}
\label{235_3toroidal}
\end{figure}

%%%%%%%%%%%%%%%%%%%
\section{Discussion/Conclusion}

\subsection{The ``facet" or ``pinned" phase}
In this paper, we have investigated the limit shape of large typical dimers configurations on $r,s,t$-pinecones, in the cases of uniform (initial data plane-dependent) and 2x2 periodic slanted plane initial data. Whereas the uniform case only displays a liquid region separated from frozen corners by an arctic ellipse, the periodic case shows the emergence of a new ``facet" phase already observed in the case of the domino tilings of large Aztec diamonds with 2x2 periodic weights \cite{DiFrancesco1}. In this work, the phase was investigated in the limit when $a\to 0$, where the liquid phase disappears,
and shown to be ``pinned" on the sublattice of square faces with initial data weight $a$.

%\subsection{Explanation of the polygonal behavior when $\sigma,\tau \to 0$}
We argue that a similar structure holds in the $2\times 2$ slanted case considered in this paper. 
Let us first restrict to the case $(r,s,t)=(1,1,3)$.
To investigate this new facet phase, we note that the limit $\sigma,\tau\to 0$ of Figs. \ref{tau=sigma} suppresses the liquid phase to leave us with
only a central facet phase separated from the frozen corners by a quadrangular arctic separation. 
(The same phenomenon occurs in all $(1,1,t)$ cases for odd $t$.). We may therefore concentrate on the $\sigma,\tau\to 0$ limit.

Let us consider as an example the tessellation domain for the case $(r,s,t)=(1,1,3)$ where we only include the active region in Fig: \ref{113_combinatorial_object}, along with its dual graph: 
\begin{figure}[H]
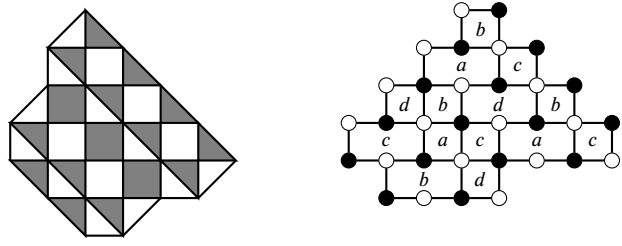

			\ctikzfig{figures/113_skeleton}
			\vspace{-4cm}
			\caption{(a): Tessellation domain of $(r,s,t)=(1,1,3)$, (b): Dual graph to the tessellation with the initial values $a,b,c,d$ indicated in each face.}
			\label{113_combinatorial_object}
\end{figure}
	\noindent As explained above, the facet phase is maximal for $\sigma=\frac{a^2}{a^2+d^2}\to 0$ and $\tau=\frac{b^2}{b^2+c^2} \to 0$, obtained by sending $a,b\to 0$ while $c,d$ remain finite and positive. From the defintion of the partition function, the contribution of the local weight at face $(x,y)$ to the partition function is $\ds t_{x,y}^{v_{x,y}/2-1-\mathcal N_{x,y}}$. Thus, as $a,b \to 0$, the contribution of maximally occupied dimer configurations around the $a, b$ faces dominates the partition function $T_{i,j,k}$, expressed as Laurent polynomial of initial data $t_{x,y}$. As $i,j,k \to \infty$, the dominating configurations are those corresponding to Laurent monomial terms with highest total degree in $a,b$ in the denominator. We illustrate this with two examples 
for the case $(r,s,t)=(1,1,3)$ via the explicit  $2\times 2$ periodic solutions $T_{0,0,4}$ and $T_{1,1,4}$. 
\subsubsection{$T_{0,0,4}$ and $T_{1,1,4}$ with $\sigma, \tau \to 0 $}
%The explicit solution $T_{0,0,4}$ as Laurent polynomial of initial data is displayed in  Fig. \ref{exactT_004} of Appendix A. 
Applying the initial data \eqref{twobytwo}, we find four dominant terms in the explicit solution $T_{0,0,4}$ as $a,b,\to 0$ namely $T_{0,0,4}$ is up to a numerical factor: 
		\begin{equation}\label{dominant term T004}
			\frac{t_{-1,0} t_{-1,3} t_{0,-2} t_{1,0} t_{1,3}}{t_{0,-1} t_{0,0} t_{0,2} t_{0,3}}+\frac{t_{-1,-1} t_{-1,3} t_{1,-2} t_{1,0} t_{1,3}}{t_{0,-1} t_{0,2} t_{0,3}
   t_{1,-1}}+\frac{t_{-1,0} t_{-1,2} t_{0,-2} t_{1,0} t_{1,1} t_{1,3}}{t_{0,-1} t_{0,0} t_{0,1} t_{0,2} t_{1,2}}+\frac{t_{-1,-1} t_{-1,2} t_{1,-2} t_{1,0} t_{1,1}
   t_{1,3}}{t_{0,-1} t_{0,1} t_{0,2} t_{1,-1} t_{1,2}} .
		\end{equation}
Similarly, $T_{1,1,4}$ is dominated by the following 8 terms: 
		\begin{equation}\label{dominant term T114}
		\begin{aligned}
		&\frac{t_{-1,0} t_{-1,2} t_{1,1}^2 t_{1,4} t_{2,-1} t_{3,1} t_{3,4}}{t_{0,0} t_{0,1} t_{0,2} t_{2,0} t_{2,1} t_{2,3} t_{2,4}}+\frac{t_{-1,0} t_{-1,3} t_{1,1}
   t_{1,2} t_{1,4} t_{2,-1} t_{3,1} t_{3,4}}{t_{0,0} t_{0,2} t_{0,3} t_{2,0} t_{2,1} t_{2,3} t_{2,4}}+\frac{t_{-1,0} t_{-1,2} t_{1,0} t_{1,1} t_{1,4} t_{3,-1} t_{3,1}
   t_{3,4}}{t_{0,0} t_{0,1} t_{0,2} t_{2,0} t_{2,3} t_{2,4} t_{3,0}}\\
   &+\frac{t_{-1,0} t_{-1,3} t_{1,0} t_{1,2} t_{1,4} t_{3,-1} t_{3,1} t_{3,4}}{t_{0,0} t_{0,2} t_{0,3}
   t_{2,0} t_{2,3} t_{2,4} t_{3,0}}+\frac{t_{-1,0} t_{-1,2} t_{1,1}^2 t_{1,3} t_{2,-1} t_{3,1} t_{3,2} t_{3,4}}{t_{0,0} t_{0,1} t_{0,2} t_{2,0} t_{2,1} t_{2,2} t_{2,3}
   t_{3,3}}+\frac{t_{-1,0} t_{-1,3} t_{1,1} t_{1,2} t_{1,3} t_{2,-1} t_{3,1} t_{3,2} t_{3,4}}{t_{0,0} t_{0,2} t_{0,3} t_{2,0} t_{2,1} t_{2,2} t_{2,3}
   t_{3,3}}\\
   &+\frac{t_{-1,0} t_{-1,2} t_{1,0} t_{1,1} t_{1,3} t_{3,-1} t_{3,1} t_{3,2} t_{3,4}}{t_{0,0} t_{0,1} t_{0,2} t_{2,0} t_{2,2} t_{2,3} t_{3,0}
   t_{3,3}}+\frac{t_{-1,0} t_{-1,3} t_{1,0} t_{1,2} t_{1,3} t_{3,-1} t_{3,1} t_{3,2} t_{3,4}}{t_{0,0} t_{0,2} t_{0,3} t_{2,0} t_{2,2} t_{2,3} t_{3,0} t_{3,3}} .
		\end{aligned}
		\end{equation}
In each of these contributions, it is easy to track the maximally occupied faces $(i,j)$, as they contribute $t_{i, j}^{-1}$. The four (resp. eight)  terms in (\ref{dominant term T004}-\ref{dominant term T114}) correspond to the following dimer configurations: 
		\begin{figure}[H]
			\centering
			\includegraphics[scale=.2]{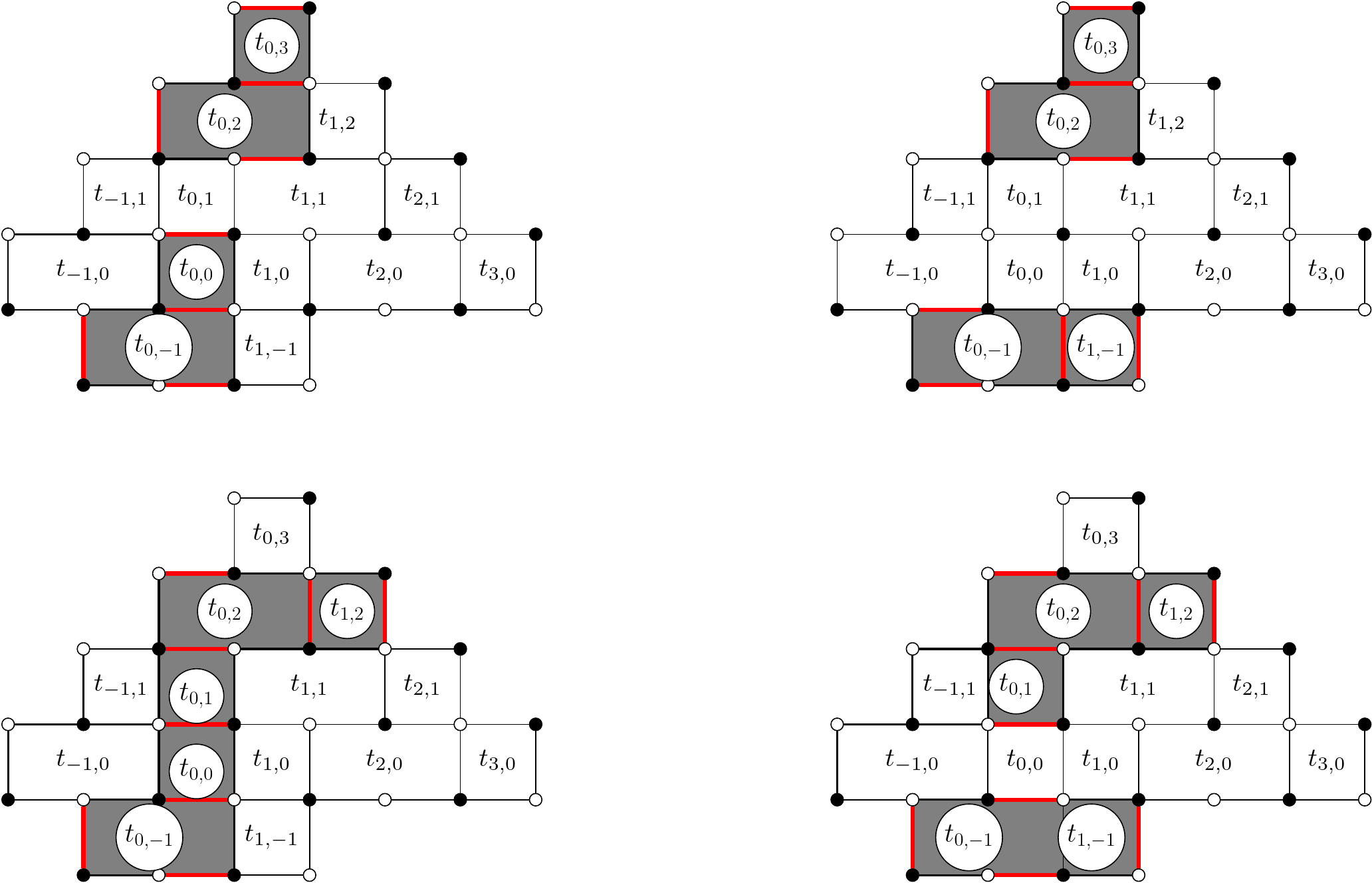}
			\caption{Dimer configurations corresponding to \eqref{dominant term T004} for the dominant terms in the (1,1,3)-slanted solution $T_{0,0,4}$. 
			The maximally occupied faces are shaded.}
			\label{dominating domino T004}
		\end{figure}
		\begin{figure}[H]
			\centering
			\includegraphics[scale=.125]{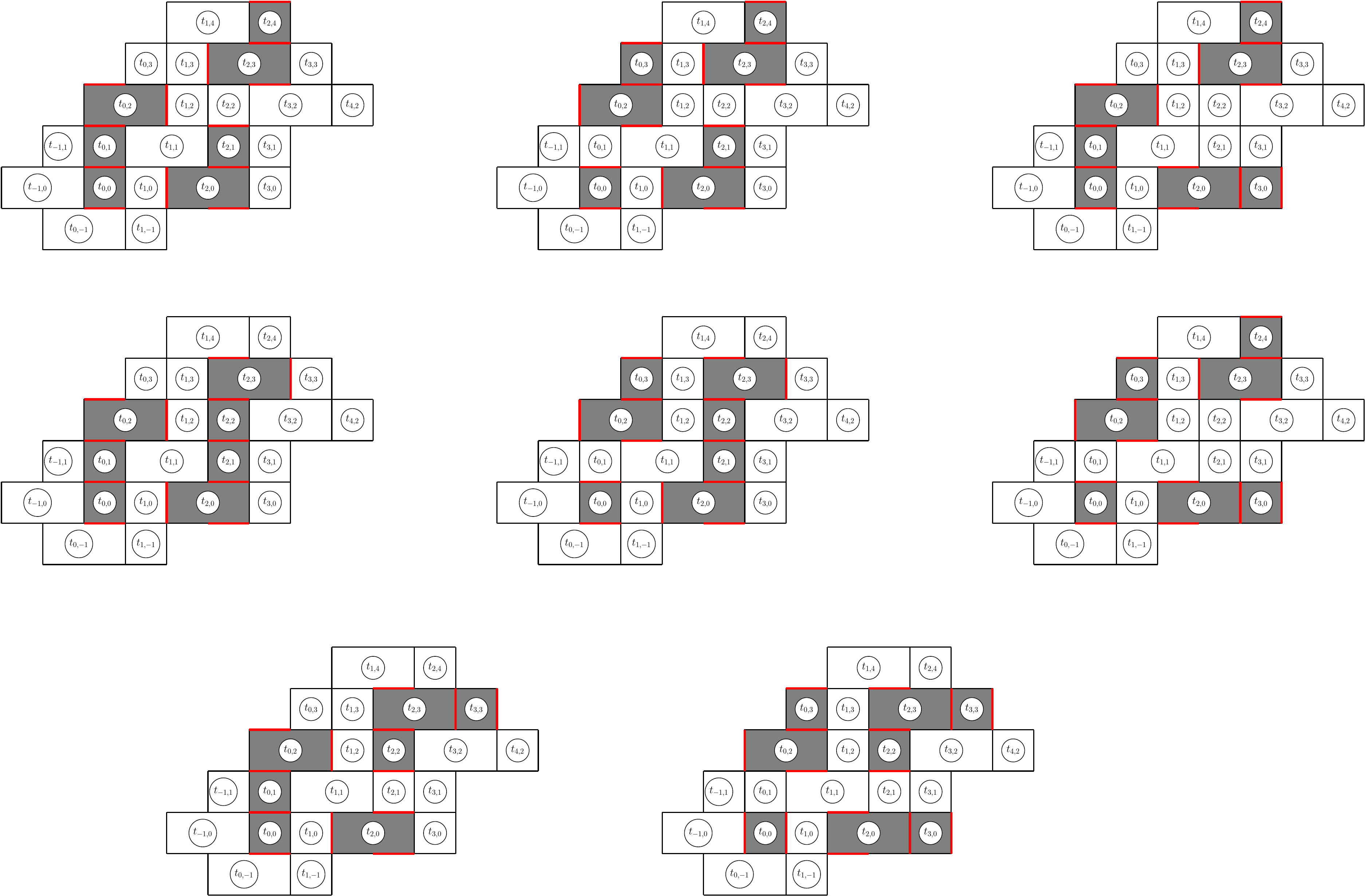}
			\caption{Dimer configurations corresponding to \eqref{dominant term T114} for the dominant terms in the (1,1,3)-slanted solution $T_{1,1,4}$, with shaded maximally occupied faces.}
			\label{dominating domino T114}
		\end{figure}
The structure of the slanted planes gives a sequence of square and hexagonal faces on the pinecone. Upon examining Figs. \ref{dominating domino T004} and \ref{dominating domino T114}, we observe that some specific hexagonal $a$ and $b$ type faces are always maximally occupied by three dimers, each with two independent ``pinned" equally probable configurations, while their surroundings vary. 
Alternatively, the dominant terms listed in \eqref{dominant term T004} and \eqref{dominant term T114} share some particular terms in the denominator that correspond to these pinned hexagonal faces. We argue that this structure generalizes to arbitrary size for $\sigma, \tau \to 0$. To see how, we display
in Fig. \ref{Density profile toroidal small cases} below some sample densities  $\rho^{(i_0,j_0,k_0)_{i,j,k}}$ say for $m=26$ ($\rho_{1,1,8}$) and $m=30$ ($\rho_{0,0,10}$): 
\begin{figure}
\begin{center}
\begin{tabular}{cc}
\includegraphics[scale=.3]{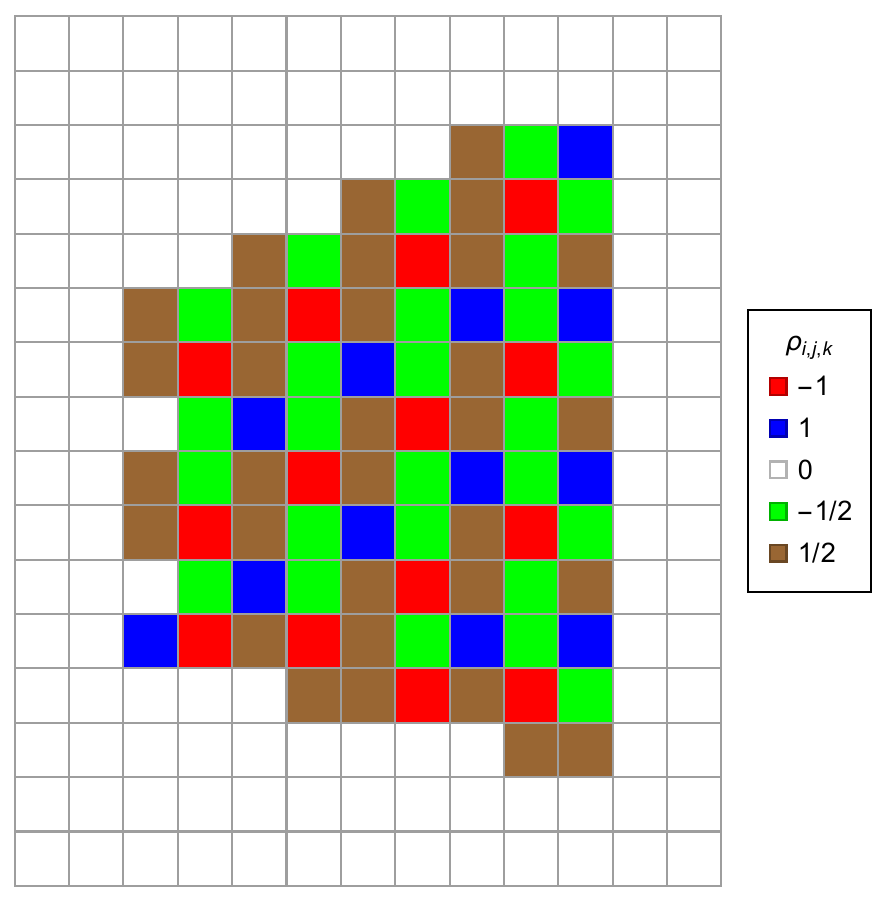}& \begin{adjustbox}{scale={0.45}{0.45},raise=12ex} $\left(
\begin{array}{ccccccccccccc}
 0 & 0 & 0 & 0 & 0 & 0 & 0 & 0 & 0 & 0 & 0 & 0 & 0 \\
 0 & 0 & 0 & 0 & 0 & 0 & 0 & 0 & 0 & 0 & 0 & 0 & 0 \\
 0 & 0 & 0 & 0 & 0 & 0 & 0 & 0 & \frac{1}{2} & -\frac{1}{2} & 1 & 0 & 0 \\
 0 & 0 & 0 & 0 & 0 & 0 & \frac{1}{2} & -\frac{1}{2} & \frac{1}{2} & -1 & -\frac{1}{2} & 0 & 0 \\
 0 & 0 & 0 & 0 & \frac{1}{2} & -\frac{1}{2} & \frac{1}{2} & -1 & \frac{1}{2} & -\frac{1}{2} & \frac{1}{2} & 0 & 0 \\
 0 & 0 & \frac{1}{2} & -\frac{1}{2} & \frac{1}{2} & -1 & \frac{1}{2} & -\frac{1}{2} & 1 & -\frac{1}{2} & 1 & 0 & 0 \\
 0 & 0 & \frac{1}{2} & -1 & \frac{1}{2} & -\frac{1}{2} & 1 & -\frac{1}{2} & \frac{1}{2} & -1 & -\frac{1}{2} & 0 & 0 \\
 0 & 0 & 0 & -\frac{1}{2} & 1 & -\frac{1}{2} & \frac{1}{2} & -1 & \frac{1}{2} & -\frac{1}{2} & \frac{1}{2} & 0 & 0 \\
 0 & 0 & \frac{1}{2} & -\frac{1}{2} & \frac{1}{2} & -1 & \frac{1}{2} & -\frac{1}{2} & 1 & -\frac{1}{2} & 1 & 0 & 0 \\
 0 & 0 & \frac{1}{2} & -1 & \frac{1}{2} & -\frac{1}{2} & 1 & -\frac{1}{2} & \frac{1}{2} & -1 & -\frac{1}{2} & 0 & 0 \\
 0 & 0 & 0 & -\frac{1}{2} & 1 & -\frac{1}{2} & \frac{1}{2} & -1 & \frac{1}{2} & -\frac{1}{2} & \frac{1}{2} & 0 & 0 \\
 0 & 0 & 1 & -1 & \frac{1}{2} & -1 & \frac{1}{2} & -\frac{1}{2} & 1 & -\frac{1}{2} & 1 & 0 & 0 \\
 0 & 0 & 0 & 0 & 0 & \frac{1}{2} & \frac{1}{2} & -1 & \frac{1}{2} & -1 & -\frac{1}{2} & 0 & 0 \\
 0 & 0 & 0 & 0 & 0 & 0 & 0 & 0 & 0 & \frac{1}{2} & \frac{1}{2} & 0 & 0 \\
 0 & 0 & 0 & 0 & 0 & 0 & 0 & 0 & 0 & 0 & 0 & 0 & 0 \\
 0 & 0 & 0 & 0 & 0 & 0 & 0 & 0 & 0 & 0 & 0 & 0 & 0 \\
\end{array}
\right)$\end{adjustbox}\\
\includegraphics[scale=.31]{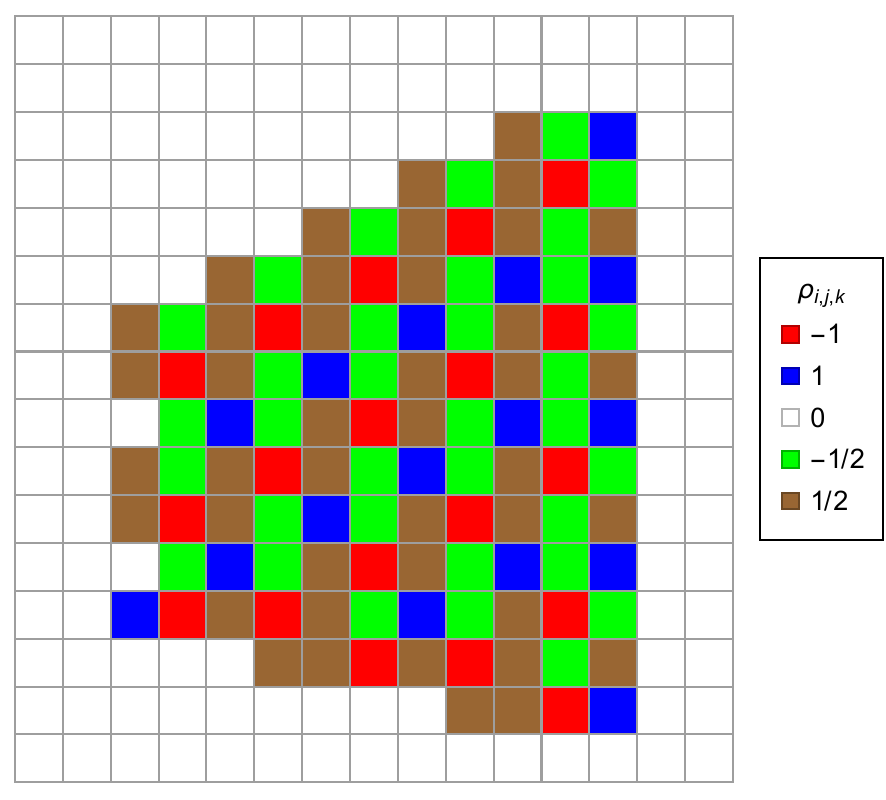}
& \begin{adjustbox}{scale={0.45}{0.45},raise=10ex} $\left(
\begin{array}{ccccccccccccccc}
 0 & 0 & 0 & 0 & 0 & 0 & 0 & 0 & 0 & 0 & 0 & 0 & 0 & 0 & 0 \\
 0 & 0 & 0 & 0 & 0 & 0 & 0 & 0 & 0 & 0 & 0 & 0 & 0 & 0 & 0 \\
 0 & 0 & 0 & 0 & 0 & 0 & 0 & 0 & 0 & 0 & \frac{1}{2} & -\frac{1}{2} & 1 & 0 & 0 \\
 0 & 0 & 0 & 0 & 0 & 0 & 0 & 0 & \frac{1}{2} & -\frac{1}{2} & \frac{1}{2} & -1 & -\frac{1}{2} & 0 & 0 \\
 0 & 0 & 0 & 0 & 0 & 0 & \frac{1}{2} & -\frac{1}{2} & \frac{1}{2} & -1 & \frac{1}{2} & -\frac{1}{2} & \frac{1}{2} & 0 & 0 \\
 0 & 0 & 0 & 0 & \frac{1}{2} & -\frac{1}{2} & \frac{1}{2} & -1 & \frac{1}{2} & -\frac{1}{2} & 1 & -\frac{1}{2} & 1 & 0 & 0 \\
 0 & 0 & \frac{1}{2} & -\frac{1}{2} & \frac{1}{2} & -1 & \frac{1}{2} & -\frac{1}{2} & 1 & -\frac{1}{2} & \frac{1}{2} & -1 & -\frac{1}{2} & 0 & 0 \\
 0 & 0 & \frac{1}{2} & -1 & \frac{1}{2} & -\frac{1}{2} & 1 & -\frac{1}{2} & \frac{1}{2} & -1 & \frac{1}{2} & -\frac{1}{2} & \frac{1}{2} & 0 & 0 \\
 0 & 0 & 0 & -\frac{1}{2} & 1 & -\frac{1}{2} & \frac{1}{2} & -1 & \frac{1}{2} & -\frac{1}{2} & 1 & -\frac{1}{2} & 1 & 0 & 0 \\
 0 & 0 & \frac{1}{2} & -\frac{1}{2} & \frac{1}{2} & -1 & \frac{1}{2} & -\frac{1}{2} & 1 & -\frac{1}{2} & \frac{1}{2} & -1 & -\frac{1}{2} & 0 & 0 \\
 0 & 0 & \frac{1}{2} & -1 & \frac{1}{2} & -\frac{1}{2} & 1 & -\frac{1}{2} & \frac{1}{2} & -1 & \frac{1}{2} & -\frac{1}{2} & \frac{1}{2} & 0 & 0 \\
 0 & 0 & 0 & -\frac{1}{2} & 1 & -\frac{1}{2} & \frac{1}{2} & -1 & \frac{1}{2} & -\frac{1}{2} & 1 & -\frac{1}{2} & 1 & 0 & 0 \\
 0 & 0 & 1 & -1 & \frac{1}{2} & -1 & \frac{1}{2} & -\frac{1}{2} & 1 & -\frac{1}{2} & \frac{1}{2} & -1 & -\frac{1}{2} & 0 & 0 \\
 0 & 0 & 0 & 0 & 0 & \frac{1}{2} & \frac{1}{2} & -1 & \frac{1}{2} & -1 & \frac{1}{2} & -\frac{1}{2} & \frac{1}{2} & 0 & 0 \\
 0 & 0 & 0 & 0 & 0 & 0 & 0 & 0 & 0 & \frac{1}{2} & \frac{1}{2} & -1 & 1 & 0 & 0 \\
 0 & 0 & 0 & 0 & 0 & 0 & 0 & 0 & 0 & 0 & 0 & 0 & 0 & 0 & 0 \\
\end{array}
\right)$\end{adjustbox}
\end{tabular}
\end{center}
	\caption{Local density $\rho_{1,1,8}$ (top) and $\rho_{0,0,10}$ (bottom) at all points $(i_0,j_0)$. }
	\label{Density profile toroidal small cases}		
	\end{figure}
First note that the local density only takes values
$-1,-\frac{1}{2},0,\frac{1}{2},1$.		
The value $-1$ corresponds to maximally occupied faces, among which hexagons form a sublattice. The hexagons correspond to the red faces (value $-1$) which alternate with blue faces (also hexagons, but with value $1$) along diagonal lines with direction $(1,-1)$, spaced by 3 units.
Once we fix the configuration of these maximally occupied hexagonal faces, 
there is a unique configuration of their surroundings, and averaging over the two possible configurations of each such face produces the factors 
$\frac{1}{2}(1+0)=\frac{1}{2}$ or $\frac{1}{2}(-1+0))=-\frac12$ i.e. the green and brown faces, while the blue hexagons correspond to an average over 4 configurations determined by the choices of the two pinned adjacent hexagons: $\frac14(0+1+1+2)=1$ i.e. the blue faces. 
In other words, we have the following local structure in Fig \ref{figpinned}. 
	\begin{figure}
		\includegraphics[scale=.2]{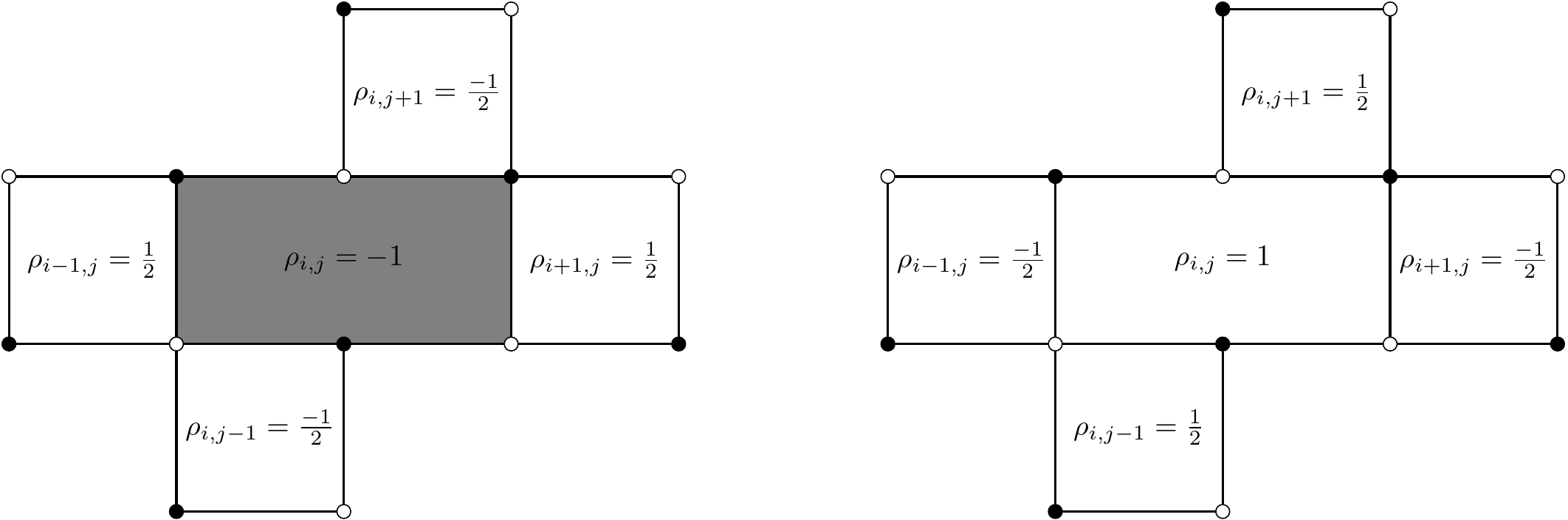}
		\caption{Local density configuration around the maximally- (left) and minimally (right) occupied hexagonal face $(i,j)$ of the dimer model in the thermodynamic limit $k\to\infty$. }
		\label{figpinned}
	\end{figure}
From Fig.\ref{figpinned} left, we see that the squares at the left and right of the pinned hexagon are occupied by 1 or $0$ dimer
(average $\frac{1}{2}$), while those on top and bottom are occupied by $1$ or $2$ (average $-\frac{1}{2}$), and it is easy to reconstruct the 
unique configuration for each choice of the pinned hexagon configurations.
%The two faces $(i,j+1)$ and $(i,j-1)$ adjacent to the maximally occupied face $(i,j)$ are always covered by at least $1$ dimer.
%	\begin{remark}
%		We want to make a remark that the tiling configurations in Fig. \ref{dominating domino T004} are completely fixed once the faces carrying weights $a,b$ are pinned with maximum number of dimers. This means that in the scaling limits, the liquid region is completely determined by the sequence of pinned faces. In \cite{DiFrancesco1}, we observed the same phenomena for the Aztec diamond dimer configuration with some $m$-toroidal weights approaching $0$. This solidifies our expectation that there is a meaningful connection between the slanted and the $m$-toroidal flat inital data $T$-system 
%	\end{remark}

In summary, like in the Aztec case of \cite{DiFrancesco1}, the facet phase observed here is pinned on a particular sublattice 
(here of hexagons of type $a,b$), but has a non-zero entropy of $2$ per pinned hexagon. This explains the fact that the partition function has always $2^{n_H}$ contributions, where $n_H$ is the number of pinned hexagons ($n_H=$2 and 3 in the examples of Figs. \ref{dominating domino T004} and \ref{dominating domino T114}).
We argue that this is the general structure of the facet phase occurring in general as bubbles inside the liquid zones.
The same structure holds for more general slanted planes $(1,1,t)$ for odd $t$. However, in addition to a sublattice of pinned hexagons,
there are additional frozen domains where strips of $t-3$ consecutive squares are maximally occupied by dimers, 
next to shifted strips of $t-3$ consecutive empty squares in alternance, between rows of pinned hexagons. 
For example, the $(1,1,5)$-slanted density profile for $\rho_{1,0,9}$ reads:
\begin{figure}[H]
	\begin{center}
		\begin{tabular}{cc}
 		\includegraphics[scale=.3]{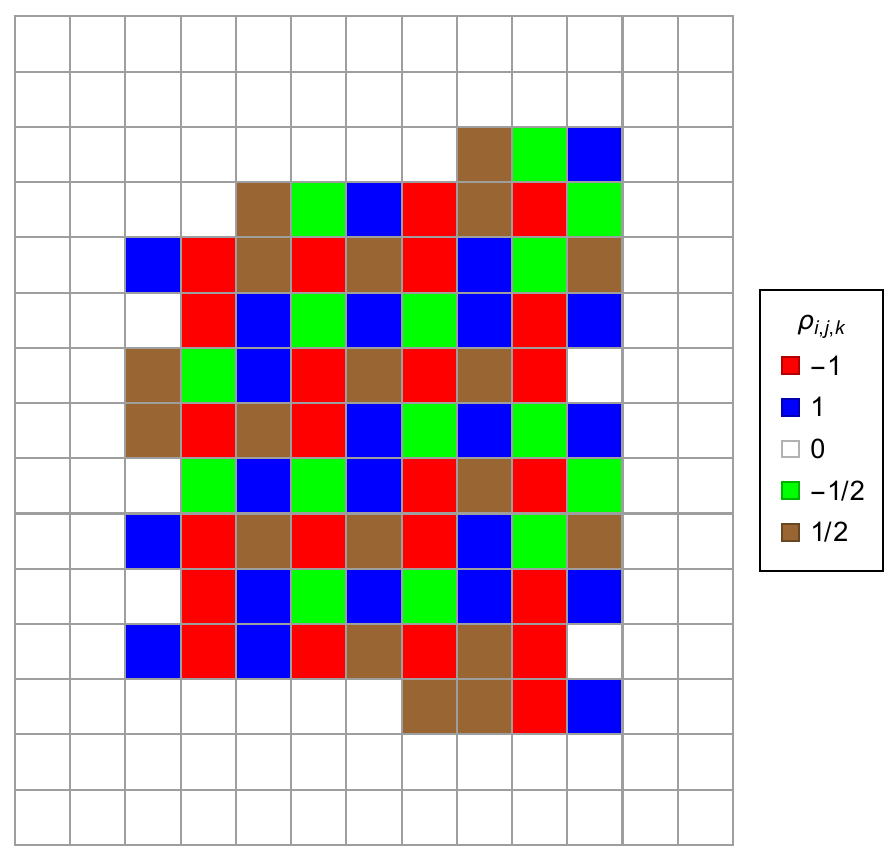}&\begin{adjustbox}{scale={0.45}{0.45},raise=10ex} $\left(
\begin{array}{ccccccccccccc}
 0 & 0 & 0 & 0 & 0 & 0 & 0 & 0 & 0 & 0 & 0 & 0 & 0 \\
 0 & 0 & 0 & 0 & 0 & 0 & 0 & 0 & 0 & 0 & 0 & 0 & 0 \\
 0 & 0 & 0 & 0 & 0 & 0 & 0 & 0 & \frac{1}{2} & -\frac{1}{2} & 1 & 0 & 0 \\
 0 & 0 & 0 & 0 & \frac{1}{2} & -\frac{1}{2} & 1 & -1 & \frac{1}{2} & -1 & -\frac{1}{2} & 0 & 0 \\
 0 & 0 & 1 & -1 & \frac{1}{2} & -1 & \frac{1}{2} & -1 & 1 & -\frac{1}{2} & \frac{1}{2} & 0 & 0 \\
 0 & 0 & 0 & -1 & 1 & -\frac{1}{2} & 1 & -\frac{1}{2} & 1 & -1 & 1 & 0 & 0 \\
 0 & 0 & \frac{1}{2} & -\frac{1}{2} & 1 & -1 & \frac{1}{2} & -1 & \frac{1}{2} & -1 & 0 & 0 & 0 \\
 0 & 0 & \frac{1}{2} & -1 & \frac{1}{2} & -1 & 1 & -\frac{1}{2} & 1 & -\frac{1}{2} & 1 & 0 & 0 \\
 0 & 0 & 0 & -\frac{1}{2} & 1 & -\frac{1}{2} & 1 & -1 & \frac{1}{2} & -1 & -\frac{1}{2} & 0 & 0 \\
 0 & 0 & 1 & -1 & \frac{1}{2} & -1 & \frac{1}{2} & -1 & 1 & -\frac{1}{2} & \frac{1}{2} & 0 & 0 \\
 0 & 0 & 0 & -1 & 1 & -\frac{1}{2} & 1 & -\frac{1}{2} & 1 & -1 & 1 & 0 & 0 \\
 0 & 0 & 1 & -1 & 1 & -1 & \frac{1}{2} & -1 & \frac{1}{2} & -1 & 0 & 0 & 0 \\
 0 & 0 & 0 & 0 & 0 & 0 & 0 & \frac{1}{2} & \frac{1}{2} & -1 & 1 & 0 & 0 \\
 0 & 0 & 0 & 0 & 0 & 0 & 0 & 0 & 0 & 0 & 0 & 0 & 0 \\
 0 & 0 & 0 & 0 & 0 & 0 & 0 & 0 & 0 & 0 & 0 & 0 & 0 \\
\end{array}
\right)$\end{adjustbox}
		\end{tabular}
	\end{center}
	\caption{Local density $\rho_{1,0,9}$ for the case of $(1,1,5)$-slanted initial data}
	\label{T109_profile_115}
\end{figure}
	In Fig \ref{T109_profile_115}, we observe the usual alternance of red/blue hexagons along diagonals in the direction $(1,-1)$, now spaced by 5 units.
In addition,  we have pairs of consecutive red (maximally occupied) square faces along vertical lines spaced by 2 units, alternating with pairs of consecutive blue (empty) squares. 

More generally, the number of square faces between two hexagonal faces along vertical lines is $t-1$ in the case of $(1,1,t)$-slanted initial data. Inbetweeen two consecutive pinned (red) hexagons along a vertical, say at positions $(i,j)$ and $(i,j+t)$, 
there is a sequence of $t-3$ maximally occupied (red) square faces (with only one frozen configuration, 
with all their horizontal edges occupied) at positions
$(i,j+2), \cdots, (i,j+t-2)$, while between the two (blue) hexagons at positions $(i+1,j-1)$ and $(i+1,j-1+t)$
there is a sequence of $t-3$ empty (brown) squares at positions $(i+1,j+1), \cdots, (i+1,j+t-1)$. The pattern is repeated on a lattice generated by $(2,-2)$
and $(0,t)$. The only variations in the configurations are determined by the 2 choices for each pinned hexagonal face. We expect a similar ``pinned" structure to hold within all the facet phases observed above.

\subsection{3D view of the Holographic principle}

\begin{figure}
\begin{center}
\begin{minipage}{0.5\textwidth}
        \centering
        \includegraphics[width=4cm]{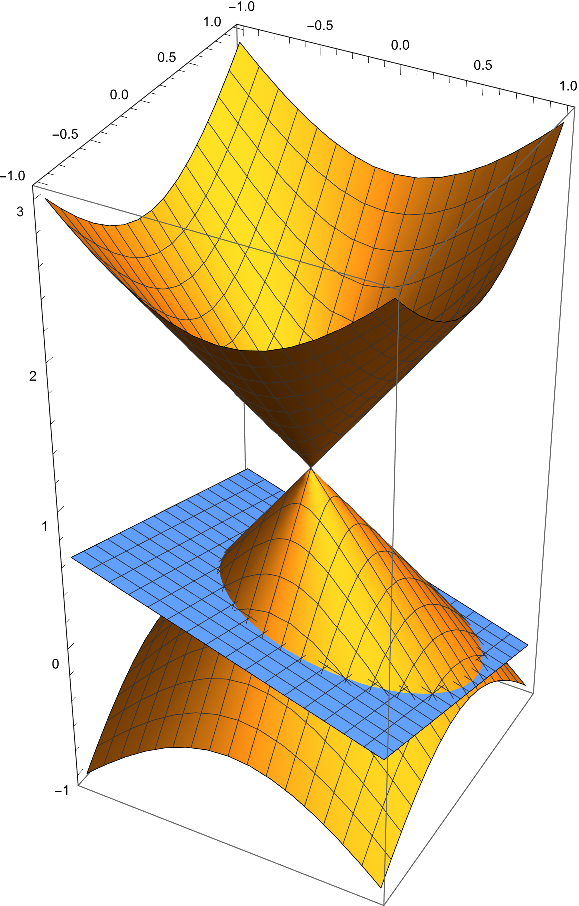} % first figure itself
        %\caption{first figure}
    \end{minipage}\hfill
    \begin{minipage}{0.5\textwidth}
        \centering
        \includegraphics[width=4cm]{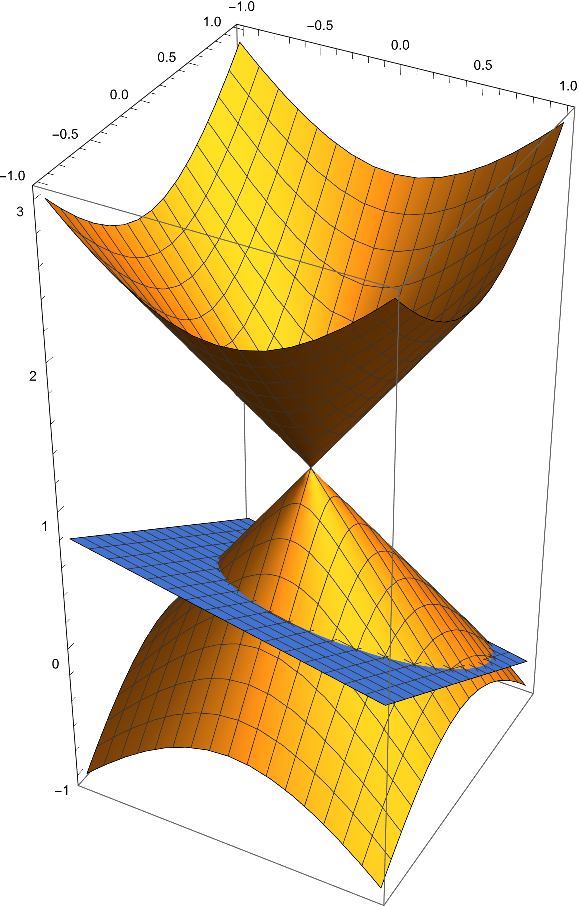}  % second figure itself
        %\caption{second figure}
    \end{minipage}
\end{center}
\caption{\small The surface for the $T$-system with (1,1,3)-slanted uniform initial data (depicted in orange/brown), together with
Left: the original slanted initial data plane ${\mathcal P}_{1,1,3}$ (depicted in blue); Right: the holographic section in the direction $(1,3,5)$, i.e. the plane ${\mathcal P}_{1,3,5}$ (depicted in blue). }
\label{fig:surface113unif}
\end{figure}

In Section \ref{sec:hologen}, we introduced a holographic principle which allows to re-interpret in dimer language any given solution of the $T$-system with an $(r,s,t)-$slanted initial data giving rise to an arctic phenomenon, in terms of any other slanted direction $(\tilde r,\tilde s, \tilde t)$. We argue now that the two arctic curves pertaining to the same solution of the $T$-system are simply intersections of a single two-dimensional surface in three dimensions
with the corresponding slanted planes. This is easy to see on the uniform case of Section \ref{sec:holo-uniform}. 
Indeed, the holographic arctic ellipse equation \eqref{arctictttilde} may indeed be interpreted as the intersection in 3D space with coordinates $(u,v,w)$
of the slanted plane
${\mathcal P}_{\tilde  r,\tilde s,\tilde t}: \ \tilde r u+\tilde s v+\tilde t w=0$ with the curve
$${\mathcal C}_{r,s,t}: \ \ (1-A_{r,s,t})\, u^2 +A_{r,s,t}\, v^2 -A_{r,s,t}\,(1-A_{r,s,t})\, (w-1)^2 =0 $$
The latter is a cone\footnote{This property is easily seen from the homogeneity of the surface equation in the variables $(u,v,w-1)$.} with apex $(0,0,1)$ which is parameterized by the initial data direction $(r,s,t)$, and contains the original arctic curve
of the uniform $(1,1,3)-$slanted model in the plane ${\mathcal P}_{r,s,t}$ (as the original arctic ellipse is the intersection of the plane ${\mathcal P}_{r,s,t}$ with the surface ${\mathcal C}_{r,s,t}$). In fact, the surface ${\mathcal C}_{r,s,t}$ is also defined as the family of lines through the apex $(0,0,1)$ that intersect the $(r,s,t)$ arctic curve in the plane ${\mathcal P}_{r,s,t}$, defined by:
$$ r u +s v+t w=0,\quad {\rm and} \quad (1-A_{r,s,t})\, t^2u^2 +A_{r,s,t}\, t^2v^2 -A_{r,s,t}\,(1-A_{r,s,t})\, (r u+s v+t)^2 =0 .$$
We may therefore think of the curve \eqref{arctictttilde} as the 2D holographic view of the 
surface ${\mathcal C}_{r,s,t}$ in 3D (see Fig.~\ref{fig:surface113unif} for the example $r=1,s=1,t=3$). Note finally that the domain for the dimer models
corresponds to the inside of the pyramid $|u|+|v|=|w-1|$, which is tangent to the surface ${\mathcal C}_{r,s,t}$ along four lines.

We suspect the surface ${\mathcal C}_{r,s,t}$ may have a physical meaning as the singularity
locus of some 3D statistical model inside the pyramid $|u|+|v|=|w-1|$, where the surface corresponds to sharp phase separations like in the 2D interpretation.

\begin{figure}
\begin{center}
\begin{minipage}{0.5\textwidth}
        \centering
        \includegraphics[width=4cm]{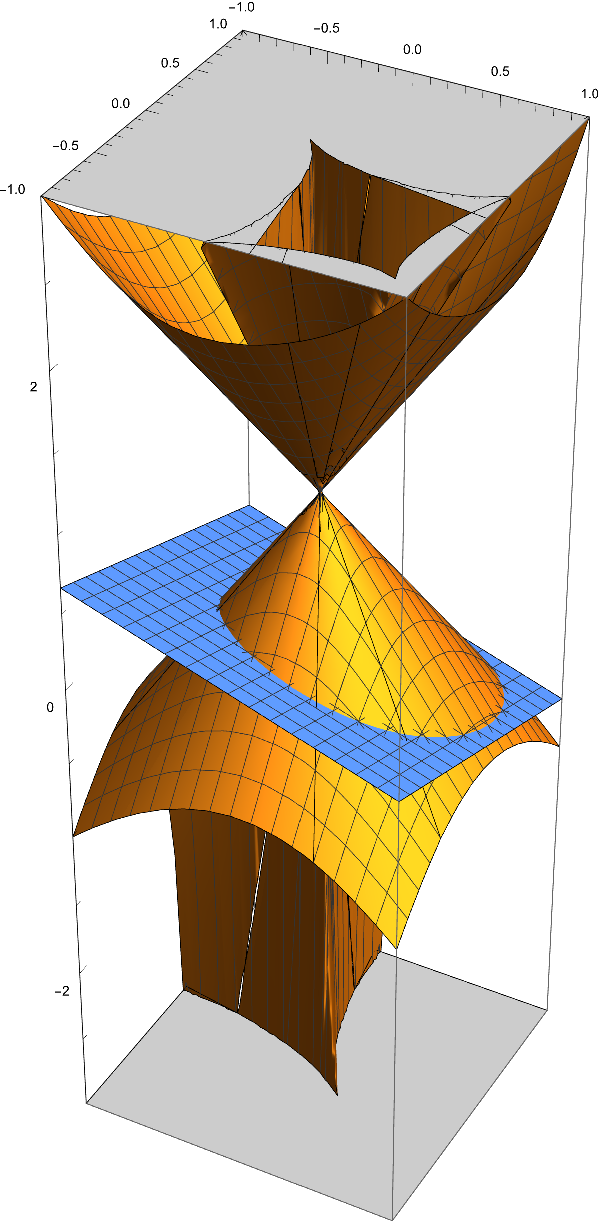} % first figure itself
        %\caption{first figure}
    \end{minipage}\hfill
    \begin{minipage}{0.5\textwidth}
        \centering
        \includegraphics[width=4cm]{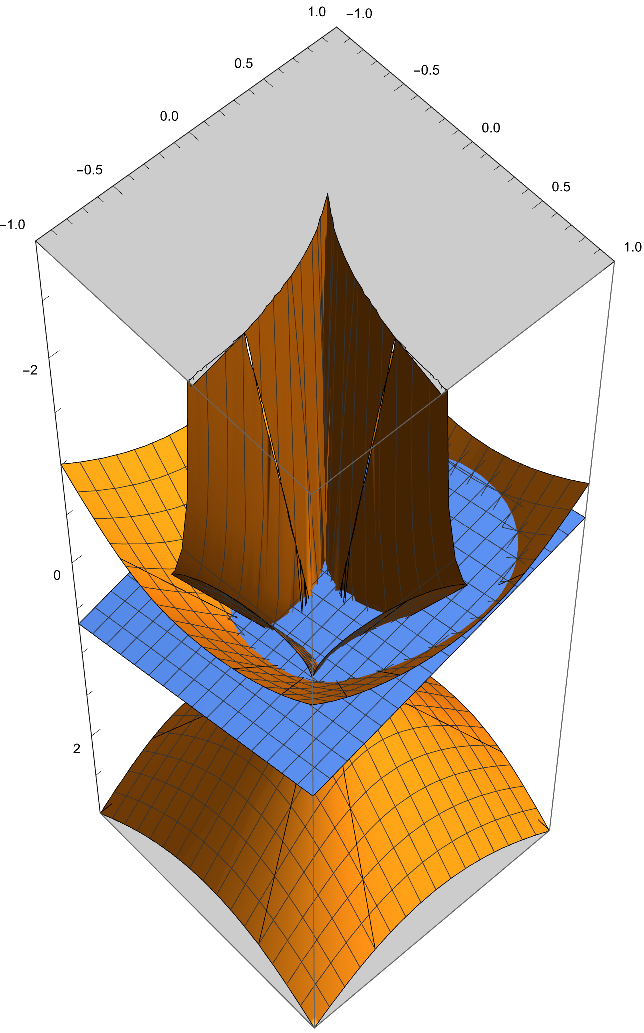}  % second figure itself
        %\caption{second figure}
    \end{minipage}
\end{center}
\caption{\small The surface for the $T$-system with (1,1,3)-slanted 2x2-periodic initial data with $\sigma=\tau=1/4$ (depicted in orange/brown), and the original slanted initial data plane ${\mathcal P}_{1,1,3}$ (depicted in blue). Left: upper side view. Right: Lower side view. }
\label{fig:surface113per}
\end{figure}

We expect this phenomenon to be general, namely that all holographic views of any given $(r,s,t)$-model studied in this paper are obtained as
the intersection of a suitable cone in 3D (conjecturally defined by the family of lines through the apex $(0,0,1)$ that intersect the original $(r,s,t)$-arctic curve in the plane ${\mathcal P}_{r,s,t}$), with the corresponding view-planes. 
%The existence of this surface and the fact that it is a cone are a consequence
%of the scaling equation \eqref{holo}: $ \Delta(x,y,z)=D(z^a/x,z^b/x,z)$ for $a={\tilde r}/{\tilde t}$, $b={\tilde s}/{\tilde t}$, leading to
%$$ \Delta(e^{\epsilon x} ,e^{\epsilon y},e^{-\epsilon (xu+yv)})=D(e^{-\epsilon(b(y-x)v -x (w-1))},e^{-\epsilon(a(x-y)u -y (w-1))},e^{-\epsilon (xu+yv)})
%= \epsilon^{\theta} (K(x,y\vert u,v,w)+O(\epsilon))$$
%where $K$ is a homogeneous function of the variable $u,v,w=-(au+bv)$.

As an illustration, we have represented in 
Fig.~\ref{fig:surface113per} the surface for the
$(1,1,3)-$ slanted 2x2 periodic case for $\sigma=\tau=1/4$ (in orange/brown) in two different views showing the above and below parts. 
The blue plane is the original slanted plane ${\mathcal P}_{1,1,3}$, and the intersection with the surface was depicted in Fig. 20 (C). 
The actual equation of the surface (a cone of homogeneous degree $8$) is available on demand from the authors.

\subsection{Conclusion and perspectives} 

This paper has introduced new solutions of the $T$-system and interpreted them in terms of dimer partition functions with special initial data.
This study is by no means exhaustive and would deserve a more systematic approach, leading possibly to a classification of exact solutions. 
However, the study has allowed us to find explicit arctic curves for a large class of suitably weighted pinecone dimer models, thus extending
widely the results of \cite{DiFrancesco1}. In particular, we have identified the structure of the new included ``facet" phase forming bubbles inside 
the liquid phase, as being pinned on some sublattice of hexagonal faces of the pinecones, while keeping a non-zero entropy. 

We also introduced a holographic principle allowing for re-interpreting exact solutions from different points of view, and eventually exhibiting an
underlying three-dimensional structure. 

Dimer models have many different formulations, and it would be interesting to investigate the non-intersecting lattice path/network formulation 
associated to the $(r,s,t)$ pinecones. This formulation has the advantage of giving an alternative route to access to thermodynamic properties 
of the models, and in particular the arctic phenomenon: we may hope to be able to use the so-called tangent method of Colomo and Sportiello \cite{Colomo_2016,DFLapa,DFGUI,DFG2,DFG3,DFref,DebRu,DebinDiFrancescoGuitter_2020,DFTRI}, 
and compare the results to those obtained in the present paper. Some advances in this direction were perfomed in  \cite{Ruelle} for the case 
of the two-periodic Aztec diamond.

Finally, beyond the case of dimers, arctic curves have been derived for interacting fermion models such as the Six or Twenty Vertex models \cite{copro,Colomo_2016,DebinDiFrancescoGuitter_2020,DFref}, and display new features, such as non-analyticity of arctic curves in non-free femion cases. These cases escape Kasteleyn theory, although the partition functions are still determinantal. It is known however that determinants obey Pl\"ucker relations, of which the $T$-system is one particular example. 
Such relations were already used in \cite{copro,DFref} to obtain the thermodynamic free energy and boundary one-point functions, but a general approach to path density is still lacking.
It would be very interesting to mimick the approach of the present paper in these more involved cases.
\bibliographystyle{plain}
\bibliography{Tslanted.bib}

\end{document}